\documentclass[11pt,oneside,openany,notitlepage]{report}
\usepackage[margin=1in]{geometry}
\usepackage{amsmath,amsfonts,amssymb,amsthm,thmtools,thm-restate}
\usepackage{graphicx}
\usepackage{enumerate}
\usepackage{bbm}
\usepackage{upgreek}
\usepackage{verbatim}
\usepackage{hyperref,color}
\usepackage[dvipsnames]{xcolor}
\hypersetup{
	colorlinks=true,
	pdfpagemode=UseNone,
    citecolor=OliveGreen,
    linkcolor=NavyBlue,
    urlcolor=Magenta,
	pdfstartview=FitW
}
\usepackage{appendix}
\usepackage{tikz}
\usepackage{array}
\usepackage[ruled,linesnumbered]{algorithm2e}
\SetKwComment{Comment}{/* }{ */}
\usepackage{multicol}
\allowdisplaybreaks

\usepackage{newtxtext}

\theoremstyle{plain}
\newtheorem{thm}{Theorem}[chapter]
\newtheorem{theorem}[thm]{Theorem}

\newtheorem{lemma}[thm]{Lemma}

\newtheorem{corollary}[thm]{Corollary}

\newtheorem{claim}[thm]{Claim}

\theoremstyle{definition}

\newtheorem{definition}[thm]{Definition}

\newtheorem*{ass*}{Assumption}
\newtheorem*{assumption*}{Assumption}
\newtheorem{exercise}[thm]{Exercise}

\theoremstyle{remark}

\newtheorem{remark}[thm]{Remark}



\makeatletter

\newcommand{\fnttcs@crefprefix}[2]{%
  \expandafter\def\csname fnttcs@crefprefix@#1\endcsname{#2}}
\fnttcs@crefprefix{thm}{theorem}
\fnttcs@crefprefix{lem}{lemma}
\fnttcs@crefprefix{cor}{corollary}
\fnttcs@crefprefix{claim}{claim}
\fnttcs@crefprefix{clm}{claim}
\fnttcs@crefprefix{def}{definition}
\fnttcs@crefprefix{defn}{definition}
\fnttcs@crefprefix{rem}{remark}
\fnttcs@crefprefix{sec}{section}
\fnttcs@crefprefix{sub}{section}
\fnttcs@crefprefix{subsec}{section}
\fnttcs@crefprefix{eq}{equation}
\fnttcs@crefprefix{eqn}{equation}
\fnttcs@crefprefix{induct}{equation}
\fnttcs@crefprefix{ineq}{equation}
\fnttcs@crefprefix{matroid}{equation}
\fnttcs@crefprefix{step}{equation}
\fnttcs@crefprefix{TD}{equation}

\newcommand{\fnttcs@creflabel}[2]{%
  \expandafter\def\csname fnttcs@creflabel@#1\endcsname{#2}}
\fnttcs@creflabel{DL1}{lemma}
\fnttcs@creflabel{assume111}{item}
\fnttcs@creflabel{assume222}{item}
\fnttcs@creflabel{colorings-triangle-free}{theorem}
\fnttcs@creflabel{edge-colorings}{theorem}
\fnttcs@creflabel{exc1}{exercise}
\fnttcs@creflabel{exc2}{exercise}
\fnttcs@creflabel{exc3}{exercise}
\fnttcs@creflabel{sec:introduction}{chapter}
\fnttcs@creflabel{sec:properties}{chapter}
\fnttcs@creflabel{sec:MC}{chapter}
\fnttcs@creflabel{sec:rapid}{chapter}
\fnttcs@creflabel{sec:RW}{chapter}
\fnttcs@creflabel{sec:matroids}{chapter}
\fnttcs@creflabel{sec:trickledown}{chapter}
\fnttcs@creflabel{sec:methods}{chapter}
\fnttcs@creflabel{sec:critical-point}{chapter}
\fnttcs@creflabel{sec:hard-core-random}{chapter}
\fnttcs@creflabel{sec:other-models}{chapter}
\fnttcs@creflabel{section-proof-relax}{chapter}
\fnttcs@creflabel{TD:blue}{lemma}
\fnttcs@creflabel{def:Q}{equation}
\fnttcs@creflabel{defn:K}{equation}
\fnttcs@creflabel{defn:Poincare}{equation}
\fnttcs@creflabel{defn:var-conservation}{equation}
\fnttcs@creflabel{lem:AT-gap}{corollary}
\fnttcs@creflabel{lem:CI-HC}{theorem}
\fnttcs@creflabel{lem:impr-RW-thm}{theorem}
\fnttcs@creflabel{lem:opt-relax-SI}{theorem}
\fnttcs@creflabel{lem:plan}{theorem}
\fnttcs@creflabel{lem:TD1-Dir}{equation}
\fnttcs@creflabel{lem:TD1-Var}{equation}
\fnttcs@creflabel{rem:local-downup}{lemma}

\newcommand{\fnttcs@crefnames}[3]{%
  \expandafter\def\csname fnttcs@crefsingular@#1\endcsname{#2}%
  \expandafter\def\csname fnttcs@crefplural@#1\endcsname{#3}}
\fnttcs@crefnames{theorem}{Theorem}{Theorems}
\fnttcs@crefnames{lemma}{Lemma}{Lemmas}
\fnttcs@crefnames{corollary}{Corollary}{Corollaries}
\fnttcs@crefnames{claim}{Claim}{Claims}
\fnttcs@crefnames{definition}{Definition}{Definitions}
\fnttcs@crefnames{remark}{Remark}{Remarks}
\fnttcs@crefnames{chapter}{Chapter}{Chapters}
\fnttcs@crefnames{section}{Section}{Sections}
\fnttcs@crefnames{equation}{Eq.}{Eqs.}
\fnttcs@crefnames{item}{Item}{Items}
\fnttcs@crefnames{exercise}{Exercise}{Exercises}

\def\fnttcs@getcreftype#1{%
  \ifcsname fnttcs@creflabel@#1\endcsname
    \csname fnttcs@creflabel@#1\endcsname
  \else
    \expandafter\fnttcs@getcrefprefix#1:\@nil
  \fi}
\def\fnttcs@getcrefprefix#1:#2\@nil{%
  \ifcsname fnttcs@crefprefix@#1\endcsname
    \csname fnttcs@crefprefix@#1\endcsname
  \else
    equation%
  \fi}

\def\fnttcs@crefnumbertext#1#2{%
  \def\fnttcs@tmp{#1}%
  \def\fnttcs@equation{equation}%
  \ifx\fnttcs@tmp\fnttcs@equation
    (\ref*{#2})%
  \else
    \ref*{#2}%
  \fi}
\def\fnttcs@crefnumber#1#2{%
  \hyperref[#2]{\fnttcs@crefnumbertext{#1}{#2}}}

\newcommand{\fnttcs@crefpair}[4]{%
  \hyperref[#3]{#1~\fnttcs@crefnumbertext{#2}{#3}} and~%
  \fnttcs@crefnumber{#2}{#4}}
\newcommand{\fnttcs@creftriple}[5]{%
  \hyperref[#3]{#1~\fnttcs@crefnumbertext{#2}{#3}}, %
  \fnttcs@crefnumber{#2}{#4} and~%
  \fnttcs@crefnumber{#2}{#5}}
\newcommand{\fnttcs@crefrange}[4]{%
  \hyperref[#3]{#1~\fnttcs@crefnumbertext{#2}{#3}} to~%
  \fnttcs@crefnumber{#2}{#4}}
\newcommand{\fnttcs@crefmulti}[2]{%
  \expandafter\def\csname fnttcs@crefmulti@#1\endcsname{#2}}

\fnttcs@crefmulti
  {sub:rowsum,sub:symmetrized-influence,sub:semidefinite-ordering}{%
  \fnttcs@crefrange
    {Sections}{section}
    {sub:rowsum}{sub:semidefinite-ordering}}

\fnttcs@crefmulti
  {thm:SI-constant-relax,thm:SI-constant-mix}{%
  \fnttcs@crefpair
    {Theorems}{theorem}
    {thm:SI-constant-relax}{thm:SI-constant-mix}}

\fnttcs@crefmulti
  {thm:SI-mixing,thm:SI-constant-relax,thm:SI-constant-mix}{%
  \fnttcs@creftriple
    {Theorems}{theorem}
    {thm:SI-mixing}
    {thm:SI-constant-relax}
    {thm:SI-constant-mix}}

\fnttcs@crefmulti
  {lem:diff-var,lem:improved-technical}{%
  \fnttcs@crefpair
    {Lemmas}{lemma}
    {lem:diff-var}{lem:improved-technical}}

\fnttcs@crefmulti
  {eqn:Var-A,eqn:Var-B}{%
  \fnttcs@crefpair
    {Eqs.}{equation}
    {eqn:Var-A}{eqn:Var-B}}

\fnttcs@crefmulti
  {almost-111,almost-222}{%
  \fnttcs@crefpair
    {Eqs.}{equation}
    {almost-111}{almost-222}}

\fnttcs@crefmulti
  {eq:ratio-pf,eq:ratio-pf2}{%
  \fnttcs@crefpair
    {Eqs.}{equation}
    {eq:ratio-pf}{eq:ratio-pf2}}

\fnttcs@crefmulti
  {eq:cc-CI-term0,eq:cc-CI-term1,eq:cc-CI-term2}{%
  \fnttcs@crefrange
    {Eqs.}{equation}
    {eq:cc-CI-term0}{eq:cc-CI-term2}}

\fnttcs@crefmulti
  {lem:cc-CI,lem:hc-small-lambda-coupling}{%
  \fnttcs@crefpair
    {Lemmas}{lemma}
    {lem:cc-CI}{lem:hc-small-lambda-coupling}}

\fnttcs@crefmulti
  {thm:CV-colorings,colorings-triangle-free}{%
  \fnttcs@crefpair
    {Theorems}{theorem}
    {thm:CV-colorings}{colorings-triangle-free}}

\DeclareRobustCommand{\cref}[1]{%
  \begingroup
  \ifcsname fnttcs@crefmulti@#1\endcsname
    \csname fnttcs@crefmulti@#1\endcsname
  \else
    \edef\fnttcs@type{\fnttcs@getcreftype{#1}}%
    \hyperref[#1]{%
      \csname fnttcs@crefsingular@\fnttcs@type\endcsname
      \nobreakspace
      \expandafter\fnttcs@crefnumbertext
        \expandafter{\fnttcs@type}{#1}}%
  \fi
  \endgroup}

\makeatother

\newtheoremstyle{consequence}{}{}{\normalfont}{}{\bfseries}{.}{.5em}{#1\thmnote{ #3}}
\theoremstyle{consequence}

\newcommand{\blambda}{\boldsymbol{\lambda}}
\newcommand{\sra}{\text{\scriptsize{$\,\rightarrow\,$}}}

\newcommand{\M}{\mathcal{M}}
\newcommand{\Pinning}{\mathcal{P}}
\newcommand{\up}[1]{\mathrm{P}_{#1}^{\uparrow}}
\newcommand{\down}[1]{\mathrm{P}_{#1}^{\downarrow}}

\newcommand{\updown}[1]{\mathrm{P}_{#1}^{\wedge}}

\newcommand{\downup}[1]{\mathrm{P}_{#1}^{\vee}}
\newcommand{\Glauber}{\mathrm{P_{\textsc{gd}}}}
\newcommand{\Glaubert}{\mathrm{P^\tau_{Gl}}}

\newcommand{\HBt}{\mathrm{P_{HB}^{\tau}}}

\newcommand{\ind}{\mathbbm{1}}
\newcommand{\dif}{\,\mathrm{d}}
\newcommand{\E}{\mathbb{E}}
\newcommand{\Exp}{\E}

\newcommand{\Var}{\mathrm{Var}}
\renewcommand{\Pr}{\mathrm{Pr}}
\newcommand{\Prob}[1]{\mathrm{Pr}\left( #1 \right)}

\newcommand{\Cov}{\mathrm{Cov}}

\newcommand{\Ent}{\mathrm{Ent}}

\newcommand{\allone}{\mathbf{1}}

\newcommand{\norm}[1]{\left\lVert #1 \right\rVert}

\newcommand{\ip}[2]{\left\langle #1 , #2 \right\rangle}
\newcommand{\grad}{\nabla}

\newcommand{\diag}{\mathrm{diag}}
\newcommand{\poly}{\mathrm{poly}}

\newcommand{\sgn}{\mathrm{sgn}}
\newcommand{\est}{\mathsf{EST}}
\newcommand{\fpras}{\mathsf{FPRAS}}

\newcommand{\fptas}{\mathsf{FPTAS}}
\newcommand{\eps}{\varepsilon}

\def\cI{\mathcal{I}}

\def\cF{\mathcal{F}}

\newcommand{\N}{\mathbb{N}}

\newcommand{\R}{\mathbb{R}}
\newcommand{\C}{\mathbb{C}}

\newcommand{\CC}{\mathcal{C}}
\newcommand{\DD}{\mathcal{D}}

\newcommand{\PP}{\mathcal{P}}

\newcommand{\Dirichlet}{\mathcal{E}}

\newcommand{\ra}{\rightarrow}

\newcommand{\Trelax}{T_{\mathrm{relax}}}
\newcommand{\Tmix}{T_{\mathrm{mix}}}

\title{Spectral Independence and Local-to-Global Techniques for Optimal Mixing of Markov Chains}

\author{
Zongchen Chen\thanks{School of Computer Science, Georgia Institute of Technology.
Email: chenzongchen@gatech.edu.} \and
 	Daniel \v{S}tefankovi\v{c}\thanks{Department of Computer Science, University of Rochester. 
	Email: stefanko@cs.rochester.edu. Research supported in part by NSF grant CCF-1563757.
	}
  	\and Eric Vigoda\thanks{Department of Computer Science, University of California, Santa Barbara. 
	Email: vigoda@ucsb.edu. 	Research supported in part by NSF grant CCF-2147094.
	}
}

\date{\today}

\begin{document}
\maketitle

\begin{abstract}
This monograph is an exposition on an exciting new technique known as spectral independence, which has been instrumental in analyzing the convergence rate of Markov Chain Monte Carlo (MCMC) algorithms.  For a high-dimensional distribution defined on labelings of the vertices of an $n$-vertex graph, the spectral independence condition, introduced by Anari, Liu, and Oveis Gharan (2020), is a bound on the maximum eigenvalue of the $n\times n$ influence matrix whose entries capture the influence between pairs of vertices (closely related to the covariance between the variables).  In the first part of the monograph, we present results showing that spectral independence (and related techniques) imply fast mixing of simple Markov chains such as the Glauber dynamics (also known as the Gibbs sampler). 
These proofs rely on local-to-global theorems relating local walks encoding pairwise correlations to variance decay and mixing properties of Markov chains.
  
  We focus on two applications: the hard-core model on independent sets of a graph (which is a combinatorial example of a binary graphical model) and random bases of a matroid. We demonstrate several methods for establishing spectral independence, including the correlation decay approach introduced by Weitz (2006) and Bandyopadhyay and Gamarnik (2006) and the analytic approach introduced by Barvinok (2015) using zero-freeness of the associated complex-valued partition function.  We apply the techniques presented in this monograph to show recent results of fast mixing of the Glauber dynamics on general graphs in the so-called tree-uniqueness region, polynomial-time mixing on general graphs at the critical point for the uniqueness threshold, and polynomial-time mixing on random regular graphs beyond the uniqueness threshold.  We also utilize the local-to-global framework of spectral independence and the Trickle-Down theorem of Oppenheim (2018) to prove fast mixing of the bases-exchange walk for generating a random basis of an arbitrary matroid, which was established by Anari, Liu, Oveis Gharan and Vinzant (2019).

  Our focus in this monograph is on the analysis of the spectral gap of the associated Markov chains from a functional analysis perspective; we present proofs of the associated local-to-global theorems and the Trickle-Down Theorem from this same Markov chain perspective. The monograph is self-contained and aims to present the proofs in a unified fashion.

\end{abstract}

\thispagestyle{empty}

\newpage

\tableofcontents

\newpage

\chapter{Introduction}
\label{sec:introduction}

This monograph aims to present the new technique of spectral independence together with its implications on the mixing time of the single-site update Markov chain known as the Glauber dynamics, which is also known as the Gibbs sampler in some communities.
Spectral independence is a simple, yet quite powerful technique to obtain optimal upper bounds on the convergence rate of Markov chains to their stationary distributions.  We will introduce the spectral independence technique, explain the basic properties of the associated influence matrix, and discuss connections to simplicial complexes.  We then prove that spectral independence implies fast mixing of the Glauber dynamics.

We also present the prominent result of Anari, Liu, Oveis Gharan, and Vinzant~\cite{ALOV19} which shows fast mixing of a Markov chain, known as the bases-exchange walk, for randomly sampling bases of a matroid.   This utilizes the local-to-global theorems presented for spectral independence and the Trickle-Down Theorem of Oppenheim~\cite{Opp18} to bound the spectral gap of the associated local walks.  

Our aim in this work is to present the analyses from a
Markov chain perspective, and to be as self-contained as possible.

\section{Setting}

Our setting is high-dimensional probability distributions on the discrete hypercube.
Let $V=[n]=\{1,\dots,n\}$ be a finite ground set, and we study distributions supported on a subset of $\{0,1\}^n$.
Often, such distributions are defined by a graphical representation; this setting corresponds to the equilibrium distribution for \emph{spin systems} in statistical physics and \emph{undirected graphical models} in machine learning. 
In such scenarios, the distribution $\mu$ is defined with respect to a graph $G=(V,E)$ and a sample $\sigma \in \{0,1\}^n$ corresponds to assigning a spin/label/color to every vertex. The probability mass function is then defined by the interactions of spins assigned to neighboring vertices.

Let $\mu$ be a distribution on $\{0,1\}^n$, and let $\Omega \subseteq \{0,1\}^n$ be the support of $\mu$, that is, 
\[
\Omega = \left\{ \sigma\in\{0,1\}^n:\mu(\sigma)>0 \right\}.
\]
We refer to the above setting as a binary system since each vertex has an assignment in $\{0,1\}$. 
Many of the results in these notes generalize to multispin systems which are distributions defined on $\{1,\dots,q\}^n$ for an integer $q\geq 2$; we will comment on the generalization to $q\geq 3$ spins at the end of this monograph in \cref{sec:multispin}, and outside that section of the text we only consider the binary $q=2$ case.

An example to keep in mind is the so-called {\em hard-core model} that we now define.

\section{Running example: Hard-core Model}

The hard-core model is defined by a graph $G=(V,E)$ and an {\em activity} $\lambda>0$.  The motivation of the hard-core model in statistical physics is as an idealized model of a gas where $\lambda$ corresponds to the fugacity of the gas (the fugacity is related to the density).
Configurations of the model are independent sets of $G$: an independent set is a subset $I\subset V$ of vertices which does not contain an edge, which means that for all $\{u,v\}\in E$ either $u\notin I$ and/or $v\notin I$.  
Let $\Omega=\Omega_G$ denote the collection of independent sets of $G$ (regardless of their sizes).

For an independent set $\sigma\in\Omega$, we can view
$\sigma$ as an $n$-dimensional indicator vector in $\{0,1\}^n$ where the $i$-th coordinate is assigned 1 if the $i$-th vertex is in $\sigma$ and is assigned 0 otherwise.  Thus, 
\[  \Omega=\Omega_G=\{\sigma\in\{0,1\}^n: \forall \{v,w\}\in E, \sigma(v)+\sigma(w)\leq 1\}.
\]  
For an independent set $\sigma$, when a vertex $v\in\sigma$ (corresponding to $\sigma(v)=1$) then we often refer to $v$ as \emph{occupied}, whereas for $w\notin\sigma$ (corresponding to $\sigma(w)=0$) we refer to $w$ as \emph{unoccupied}. 
Each independent set $\sigma\in\Omega$ is assigned the following weight:
\[ w(\sigma) = \lambda^{|\sigma|},
\]
where $|\sigma|$ is the number of vertices in the independent set $\sigma$, in other words, $|\sigma|=|\{v\in V: \sigma(v)=1\}|$.
The Gibbs distribution $\mu=\mu_{G,\lambda}$ is the probability distribution on independent sets that is proportional to the weights. Specifically, for $\sigma\in\Omega$,
\[ \mu(\sigma) = \frac{\lambda^{|\sigma|}}{Z},
\]
where the  normalizing factor $Z=Z_{G}(\lambda)=\sum_{\eta\in\Omega} \lambda^{|\eta|}$ is known as the \emph{partition function}.

A second major application, to matroids, will be discussed in \cref{sub:matroids-intro}.

\section{Glauber dynamics/Gibbs sampler}
The following Markov chain is a simple single-site update Markov chain designed to sample from high-dimensional distributions.  The chain is known as the \emph{Glauber dynamics} in the statistical physics community and as the \emph{Gibbs sampler} in the machine learning community.

For a distribution $\mu$ on $\Omega\subset \{0,1\}^n$, the Glauber dynamics is defined as follows.  
From a state $X_t\in\Omega$,
\begin{enumerate}
    \item Choose a coordinate $i$ uniformly at random from $[n]$.
    \item For all $j\neq i$, let $X_{t+1}(j)=X_t(j)$.
    \item Choose $X_{t+1}(i)$ from the conditional distribution $\mu(\sigma(i)=\cdot \mid \sigma(j)=X_{t}(j) \mbox{ for all }j\neq i)$. 
    In words, fix the value at all other coordinates except for the $i$-th, and resample the value at the $i$-th coordinate conditional on the fixed configuration at the rest of the coordinates.
    \end{enumerate}

The Glauber dynamics takes a simpler form when specialized to the hard-core model due to its combinatorial nature (the reader new to the field can further consider the case $\lambda=1$ for simplicity).
In particular, the conditional distribution in step 3 above only depends on the configuration of the vertices neighboring $i$.
Recall that for the hard-core model on a graph $G$, the state space $\Omega$ consists of all independent sets of $G$.

From a state $X_t\in\Omega$, the transitions $X_t\rightarrow X_{t+1}$ of the Glauber dynamics for the hard-core model are as follows:
\begin{enumerate}
    \item Choose a vertex $v$ uniformly at random from $V$.
    \item Let \[ X' = \begin{cases}
    X_t\cup\{v\} & \mbox{with probability }\lambda/(1+\lambda) \\
      X_t\setminus\{v\} & \mbox{with probability }1/(1+\lambda) 
      \end{cases}
      \]
     \item If $X'\in\Omega$ (i.e., $X'$ is a valid independent set) then set $X_{t+1}=X'$, 
     and otherwise set $X_{t+1}=X_t$.
\end{enumerate}

Returning to the general setting of a distribution $\mu$ on $\Omega\subset\{0,1\}^n$, let $P$ denote the transition matrix of the Glauber dynamics.  Note that $P$ is an $N\times N$ stochastic matrix where $N=|\Omega|$.  

A Markov chain is \emph{irreducible} if every state can reach every other state using transitions of the Markov chain. Moreover, a Markov chain is \emph{aperiodic} if for all sufficiently large $t\ge t_0$, there is a positive probability for the chain to get back to the starting state after $t$ steps. Finally, a Markov chain is \emph{ergodic} if and only if it is both irreducible and aperiodic.  Throughout this monograph we restrict attention to distributions $\mu$ where the Glauber dynamics is ergodic.

For the special case of the hard-core model it is easy to see that the Glauber dynamics is irreducible (via the empty set) and aperiodic, and hence the Glauber dynamics is ergodic and has a unique stationary distribution.  The distribution $\mu$ satisfies the detailed balance conditions, and hence, the Glauber dynamics is reversible (see \cref{sub:reversibility} for more details).  Thus, the distribution $\mu$ is the unique stationary distribution.

The {\em mixing time} of a Markov chain, denoted as $\Tmix(\eps)$ for $\eps>0$, is the minimum number of steps $t$ such that, for the worst initial state $X_0$, the distribution of $X_t$ is within total variation distance $\leq \eps$ of the stationary distribution $\mu$;
see \cref{sub:mixingtime-defn} for a more formal definition.
   For the Glauber dynamics, we say the mixing time is \emph{optimal} if $\Tmix=O(n\log{n})$ as this matches the lower bound established by Hayes and Sinclair~\cite{HayesSinclair} for any graph of constant maximum degree $\Delta$.

In contrast to the mixing time, a weaker notion of convergence is the relaxation time. For $N=|\Omega|$, let $\uplambda_1=1> \uplambda_2\geq\uplambda_3\geq \cdots \geq \uplambda_N>-1$ denote the eigenvalues of the transition matrix $P$ of a reversible, ergodic Markov chain (recall, a Markov chain is reversible if it satisfies the detailed balance conditions, and eigenvalues of reversible Markov chains are real). Note, $\uplambda_1$ corresponds to the principal (row) eigenvector which is the stationary distribution.  Moreover, $\uplambda_2<1$ if the Markov chain is irreducible and hence the stationary distribution is unique.  In addition, $\uplambda_N>-1$ since the Markov chain is aperiodic.

The {\em relaxation time} is defined as $\Trelax:=(1-\uplambda_*)^{-1}$ where $\uplambda_*=\max\{\uplambda_2,|\uplambda_N|\}$.  The quantity $1-\uplambda_*$ is known as the absolute spectral gap, see \cref{sub:spectral-gap} for further discussion. 
The relaxation time measures the decay rate of the variance of the chain with respect to the stationary distribution; this will be formally presented in~\cref{sec:MC}.  
For the Glauber dynamics, the transition matrix is positive semidefinite, i.e., $\uplambda_N \ge 0$ (see \cref{eq:relax-gap} in \cref{sub:spectral-gap} for further discussion). Hence, $\Trelax:=(1-\uplambda_2)^{-1}$, where $1-\uplambda_2$ is called the spectral gap.
The optimal relaxation time for the Glauber dynamics is $O(n)$, whereas the optimal mixing time is $O(n\log{n})$.  An optimal mixing time of the form $O(n\log(n/\eps))$ implies an optimal relaxation time (see \cref{rem:connection-mixing-to-relax} in \cref{sub:coupling-SI}), but the reverse implication does not necessarily hold (see \cref{sub:mixingviagap} for related results).  
In fact, the relaxation time is sufficient for the purposes of approximate counting algorithms (e.g., estimating the partition function of the hard-core model), see~\cite[Section 7]{SVV:annealing} (see also~\cite{Huber,Kolmogorov} for further improvements).

Further properties and fundamental definitions regarding Markov chains are presented in \cref{sec:MC}.  We lower bound the spectral gap of the Glauber dynamics (and hence upper bound the relaxation time) using a functional analytic approach to analyze the decay rate of the variance of the Markov chain with respect to the stationary distribution; consequently, in \cref{sub:Dirichlet} 
we introduce the variance functional and the Dirichlet form which is a measure of the local variation of a Markov chain.  In \cref{sub:app-tensorization} we present the notion of approximate tensorization of variance, which is a useful technique for lower bounding the spectral gap of the Glauber dynamics.  The relations between the relaxation time and the mixing time of a Markov chain are established in \cref{sub:mixingviagap}.

\section{Key Definitions: Spectral Independence and Influence Matrix}
\label{sub:SI-definition}

We now introduce the central concept of spectral independence, a powerful framework for analyzing the convergence of the Glauber dynamics.
At a high level, our monograph proceeds in three stages. First, we define the influence matrix, which captures correlations between variables, and use it to define spectral independence. Second, we show that spectral independence implies rapid mixing of the Glauber dynamics through local-to-global principles relating local walks encoding pairwise interactions to global convergence.
 Third, we present several techniques for establishing spectral independence in important models, including the hard-core model and matroids.

Spectral independence was introduced by Anari, Liu, and Oveis Gharan~\cite{ALO20}.  It is defined by an $n\times n$ influence matrix which captures the pairwise influence or correlations between pairs of vertices.

\begin{definition}[Influence Matrix]
\label{defn:inf-matrix}
Let $\mu$ be a distribution on $\Omega$ where $\Omega\subseteq\{0,1\}^n$.
Let $\Psi$ be the following real-valued $n\times n$ matrix; we refer to $\Psi$ as the {\em influence matrix}. For $1\le i,j\le n$, let
\begin{align}
\label{eqn:Psi-empty}
\Psi_\mu(i,j) = \Psi(i\rightarrow j) 
&:= \mu\left(\sigma(j)=1 \mid \sigma(i)=1\right)
- 
\mu\left(\sigma(j)=1 \mid \sigma(i)=0\right) \\
&= \Pr_{\sigma \sim \mu} \left[ \sigma(j)=1 \mid \sigma(i)=1 \right]
- \Pr_{\sigma \sim \mu} \left[ \sigma(j)=1 \mid \sigma(i)=0 \right].\nonumber
\end{align}
\end{definition}
As indicated above, here and throughout this monograph, $\mu(\mathcal{E})$ denotes the probability of an event $\mathcal{E}$, and $\sigma\sim\mu$ denotes that $\sigma$ is a random variable with distribution $\mu$.
Note, when $i=j$ then $\Psi(i,i)=1$ and hence the diagonal entries of $\Psi$ are defined to be $1$ (see \cref{rem:diagonals} below for further discussion regarding an alternative definition of the diagonal entries).

A closely related notion is the {\em modified influence matrix} which we formally define in \cref{sub:CorrelationMatrix}.  The modified influence matrix differs from \cref{eqn:Psi-empty} in that the last term changes from 
$\mu\left(\sigma(j)=1 \mid \sigma(i)=0\right)$ to 
$\mu\left(\sigma(j)=1\right)$.  This leads to a weaker notion than spectral independence (we will later refer to the alternative notion as weak spectral independence); the connections between the influence matrix and modified influence matrix are discussed in further detail in \cref{sub:CorrelationMatrix}.

We require the generalization of the above definitions to arbitrary ``pinnings,'' which are fixed assignments to subsets of vertices.  
For a subset of vertices $S\subset [n]=\{1,\dots,n\}$, a {\em pinning} on $S$ is an assignment $\tau:S\rightarrow\{0,1\}$.
Let $\Omega_\tau$ denote the assignments on all of $V$ which are consistent with $\tau$; more formally, for $S\subset V$ and a pinning $\tau:S\rightarrow\{0,1\}$, let
\[
\Omega_\tau:=\{\sigma\in\Omega: \sigma(S)=\tau(S)\},
\]
where $\sigma(S)$ is the restriction of $\sigma$ to $S$, and hence $\sigma(S)=\tau(S)$ iff  $\sigma(v)=\tau(v)$ for all $v\in S$.

We only consider pinnings $\tau$ that can be extended to a valid configuration $\sigma$ on the entire vertex set~$V$.  Hence, 
we say a pinning $\tau$ is valid if $\Omega_\tau\neq\emptyset$. 
Let $\Pinning$ denote the collection of all valid pinnings (see \cref{sub:pinnings} for further introduction).

For a pinning $\tau\in\Pinning$, let $\mu_{\tau}$ denote the conditional distribution, in other words, $\mu_\tau$ is the distribution~$\mu$ conditional on the fixed assignment $\tau$ on $S$. For $\sigma\in\Omega_\tau$, we have 
\[ \mu_{\tau}(\sigma)=\mu(\sigma\mid\sigma(S)=\tau) 
= \frac{\mu(\sigma(\overline{S})\cup\tau(S))}{\mu(\Omega_\tau)},
\]
\[ \mbox{where} \ \ \ \mu(\Omega_\tau) = \sum_{\eta\in\Omega_\tau} \mu(\eta)
= \sum_{\substack{\eta\in\Omega: \\ \eta(S)=\tau(S)}} \mu(\eta).
\]

For the case of the hard-core model on a graph $G=(V,E)$, pinnings correspond to induced subgraphs of $G$ in the following manner.  In the hard-core model, for $S\subset V$, a valid pinning $\tau$ on $S$ is an independent set on $S$.  Since $\tau$ is an independent set, an equivalent realization of the pinning $\tau$ is to consider the induced subgraph on $V'\subset V$ where $V'$ is formed by starting from $V$ and then proceeding as follows: for every vertex $v\in S$ where $\tau(v)=0$ then remove $v$ (this is because fixing $\tau(v)=0$ means $v$ is fixed to be outside of the independent set), and for every vertex $v\in S$ where $\tau(v)=1$  then remove $v$ and all $w\in N(v)$ (since $\tau(v)=1$ then $v$ is fixed to be in the independent set and hence all neighbors of $v$ are not in the independent set); let $V'$ be the remaining vertices.  Let $H$ be the induced subgraph on $V'$. The distribution $\mu_H$ (which is the Gibbs distribution on $H$ without any pinning) is identical to the distribution $\mu_\tau$ (which is the conditional Gibbs distribution on $G$ with pinning $\tau$).  Hence, for the case of the hard-core model, the collection of valid pinnings for $G$ can be realized by considering all induced subgraphs of $G$ (without any pinnings).

Returning to the general setting, for a pinning $\tau\in\Pinning$ 
where $\tau:S\rightarrow\{0,1\}$ for $S\subset V$, 
let $T\subset V\setminus S$ denote the set of vertices which are ``free'' in the following sense:
\[
i\in T \iff \mu_\tau(\sigma(i)=1)>0 \mbox{ and } \mu_\tau(\sigma(i)=0)>0.
\]
Note, a vertex $i\in V\setminus (S\cup T)$ is ``frozen'', i.e., it can only attain one spin and hence we can fix the configuration on that vertex and remove the corresponding two rows/columns from the matrix $\Psi_\tau$; henceforth, we only define the influence matrix on the free vertices.

For $i,j\in T$, let 
\begin{align}
\label{eq:inf-pin}
\Psi_{\mu_\tau}(i\rightarrow j) &:= 
 \mu_\tau\left(\sigma(j)=1 \mid \sigma(i)=1\right) - 
\mu_\tau\left(\sigma(j)=1 \mid \sigma(i)=0\right)
\\
& = \mu\left(\sigma(j)=1 \mid \sigma(i)=1,\sigma(S)=\tau\right) - 
\mu\left(\sigma(j)=1 \mid \sigma(i)=0,\sigma(S)=\tau\right).
\nonumber
\end{align}
Since the empty pinning $\emptyset$ is always valid, $\Psi=\Psi_\emptyset$, which agrees with \cref{eqn:Psi-empty}.
As in~\cref{eqn:Psi-empty}, observe that the diagonal entries of $\Psi_\tau$ are defined to be $1$, see~\cref{rem:diagonals} for further discussion.

The matrix $\Psi$ can be asymmetric, and the entries of $\Psi$ can be positive or negative.  Nevertheless, all of the eigenvalues of $\Psi$ are real and non-negative; this is established in~\cref{DL1} in~\cref{sub:PSD}.  Intuitively, the fact that $\Psi$ has real non-negative eigenvalues is because it is closely related to the covariance matrix $\Cov$.  In particular, the influence matrix $\Psi$ is equal to a positive diagonal matrix times the covariance matrix $\Cov$, and then the real non-negative eigenvalues of $\Psi$ follow from the positive semidefiniteness of the covariance matrix.

Since all eigenvalues of $\Psi$ are real valued we can denote the maximum eigenvalue by $\uplambda_{\max}(\Psi)$.  Now we can define the spectral independence condition as follows.

\begin{definition}[Spectral Independence]
\label{defn:SI}
For $\eta\ge 0$, we say that $\mu$ is $\eta$-spectrally independent if for all pinnings $\tau\in\Pinning$, 
\begin{equation}\label{eqn:SI-defn} 
\uplambda_{\max}(\Psi_{\mu_\tau})\leq 1+\eta.
\end{equation}
\end{definition}

Note that the spectral independence condition only depends on the distribution $\mu$, and it has no explicit dependence on the Glauber dynamics.  
The definition was extended to non-binary spin spaces, such as the Potts model and colorings, in~\cite{FGYZ21,CGSV21} (see also~\cite{CLV21,BCCPSV22} for a general formulation); we present these generalizations and discuss the related results in \cref{sec:multispin}.

When $\mu$ is a product distribution then 
$\eta=0$.  Our goal is to show that $\eta$ is constant, and hence in some sense it is not far from a product distribution.

\begin{remark}
\label{rem:diagonals}
Note the diagonals of the influence matrix $\Psi$ are 1 since if $i=j$ then conditioning on $i$ prescribes $j$.  We could have defined the influence matrix so that the off-diagonal entries remain the same and the diagonals are 0; this would decrease all of the eigenvalues by 1, and hence with this alternative definition we would change the spectral independence requirement from $1+\eta$ to $\eta$.
\end{remark}

For the special case of the hard-core model, the configurations of the model are independent sets.  Hence, we denote a configuration by the set of vertices in the independent set $I\subset V$.  We refer to those vertices $v\in I$ as occupied and those $w\notin I$ as unoccupied.  
For the hard-core model, representing a configuration by its subset of occupied vertices simplifies the notation.
 For example, returning to the definition of the influence matrix in \cref{defn:inf-matrix}, for the hard-core model on $G=(V,E)$, for $u,v\in V$, the $(u,v)$-entry of the influence matrix $\Psi$ can be compactly expressed as the following:
 $$
 \Psi(u\ra v) := \mu(v|u) - \mu(v|\overline{u}),
 $$
 where $u,v$ denote the events $u\in I, v \in I$ respectively, and $\overline{u}$ denotes the event $u \notin I$.

As discussed earlier, for the hard-core model, since the configurations correspond to independent sets then a pinning $\tau$ on $S$ can be realized by considering the appropriate induced subgraph $H$ of $G$.   Consequently, for a graph $G=(V,E)$ if we show for all induced subgraphs $H$ of $G$ that $\uplambda_{\max}(\Psi_H)\leq 1+\eta$ (without any pinnings) then that implies the spectral independence condition \cref{defn:SI} holds for $G$ (since every pinning~$\tau$ can be realized by an appropriate induced subgraph $H$).

In \cref{sec:properties} we prove several fundamental properties of the influence matrix~$\Psi$, including the fact that the eigenvalues of $\Psi$ are real and non-negative; we also present several basic techniques for bounding the maximum eigenvalue of the influence matrix~$\Psi$ in \cref{sub:rowsum,sub:symmetrized-influence,sub:semidefinite-ordering}, and introduce the notion of weak spectral independence in \cref{sub:CorrelationMatrix}.

\section{Main Results: Fast Mixing via Spectral Independence}
\label{sub:mixing-via-SI}

Anari, Liu, and Oveis Gharan~\cite{ALO20} proved that spectral independence, together with an easily satisfied technical condition, implies polynomial mixing time of the Glauber dynamics.  
To be precise, the following statement is a convenient consequence of the level-dependent
local-to-global result of Anari, Liu, and Oveis Gharan~\cite{ALO20}.

Before stating the theorem, we require one minor technical condition in our version of the theorem from \cite{ALO20}.  
We require that the Glauber dynamics on sufficiently small subgraphs (namely, those induced from large pinnings) has constant relaxation time.   The upper bound $A$ on the relaxation time can be exponential in the constant graph size since this only affects the leading factor in the subsequent \cref{thm:SI-mixing}. 
\begin{definition}[Small conditional gap]
\label{defn:small-subgraph-gap}
For $A\geq 1$ and $\eta\geq 0$, we say that  $\mu$ satisfies
$(A,\eta)$-small-conditional-gap if, for every valid pinning
$\tau$ leaving $m$ vertices unpinned, where
\[
2\leq m\leq\lfloor\eta\rfloor+1,
\]
the Glauber dynamics for $\mu_\tau$ has relaxation time at most
$A$.
\end{definition}

\begin{restatable}[cf.~\cite{ALO20}]{theorem}{ALOpolymixing}
\label{thm:SI-mixing} 
Let $\mu$ be a distribution supported on
$\Omega\subseteq\{0,1\}^n$ such that the Glauber dynamics for
$\mu$ is ergodic. Suppose that $\mu$ is $\eta$-spectrally
independent and satisfies $(A,\eta)$-small-conditional-gap for
constants $\eta>0$ and $A\geq1$. Then the relaxation time of the Glauber dynamics for $\mu$ satisfies:
\[ \Trelax \leq Cn^{1+\eta},
\]
for a constant $C=C(\eta,A)$.
Therefore, the mixing time satisfies $T_{\mathrm{mix}}=O(n^{1+\eta} \log(1/\mu^*))$ where $\mu^*:= \min_{\sigma \in \Omega} \mu(\sigma)$.
\end{restatable}

For the hard-core model with fugacity $\lambda$, note that $\log(1/\mu^*) = O_\lambda(n)$ and the
small-conditional-gap assumption is automatic with
$A=A(\eta,\lambda)$ (since, after any pinning, the conditional
distribution is another hard-core model at the same fugacity on an induced subgraph with only a bounded number of vertices), hence \cref{thm:SI-mixing} yields a mixing time upper bound of $O(n^{2+\eta})$.

The above result was improved by Chen, Liu, and Vigoda~\cite{CLV21} to show optimal bounds on the relaxation time when the maximum degree is constant.

\begin{restatable}[\cite{CLV21}]{theorem}{CLVoptimalrelax}
\label{thm:SI-constant-relax}
For all constant $\Delta\geq 2$, for any $n$-vertex graph $G=(V,E)$ of maximum degree $\Delta$ and fugacity $\lambda > 0$, if the hard-core model on $G$ is $\eta$-spectrally independent for a constant $\eta>0$, 
then the relaxation time of the Glauber dynamics satisfies:
\[ \Trelax \leq O(n).
\]
In particular, there exists a constant $C(\eta,\Delta,\lambda)$
such that the relaxation time satisfies:
\[ \Trelax \leq C(\eta,\Delta,\lambda)n.
\]
Therefore, the mixing time satisfies $T_{\mathrm{mix}}=O(n^{2})$.
\end{restatable}

The above result shows that spectral independence implies  optimal relaxation time.  In fact, spectral independence is a necessary condition for optimal relaxation time:
Anari, Jain, Koehler, Pham, and Vuong~\cite{relax-optimal} proved that optimal relaxation time of the Glauber dynamics implies spectral independence;
see also \cite[Theorem 1.1]{Lee2024parallelising} for closely related results.

\begin{restatable}{theorem}{optrelaxSI}
\label{lem:opt-relax-SI}
Let $\mu$ be a distribution on $\{0,1\}^n$. 
If the Glauber dynamics for $\mu$ has relaxation time at most $Cn$ then $\uplambda_{\max}(\Psi_\mu) \le C$. 
Furthermore, if for any pinning $\tau$ on a subset $S$ of vertices, the Glauber dynamics for $\mu_\tau$ has relaxation time at most $C\cdot (n-|S|)$, then $\mu$ is $(C-1)$-spectrally independent. 
\end{restatable}

If one further assumes a lower bound on the marginal probability for every vertex and every valid pinning, then Chen, Liu and Vigoda~\cite{CLV21} proved an optimal upper bound on the mixing time, rather than the relaxation time.

\begin{definition}[Marginally bounded]
\label{defn:marg-bound}
For $b>0$, we say that $\mu$ is $b$-marginally bounded if, for every pinning $\tau\in\Pinning$, every $v\in V$, and every $s\in\{0,1\}$, the following holds:
\[
\mu_\tau(\sigma(v)=s) \ge b
\quad\mbox{or}\quad \mu_\tau(\sigma(v)=s)=0.
\]
\end{definition}

\begin{restatable}[\cite{CLV21}]{theorem}{CLVoptimalmix}
\label{thm:SI-constant-mix}
For all constant $\Delta\geq 2$, for any $n$-vertex graph $G=(V,E)$ of maximum degree $\Delta$ and fugacity $\lambda > 0$, if the hard-core model on $G$ is $\eta$-spectrally independent and $b$-marginally bounded for constants $\eta,b>0$, 
then the mixing time of the Glauber dynamics satisfies:
\[ \Tmix \leq O(n\log{n}).
\]
In particular, there exists a constant $C(\eta,\Delta,b)$ such that, for all $\eps>0$, the mixing time satisfies
\[ \Tmix(\eps) \leq C(\eta,\Delta,b)n\log(n/\eps).
\]
\end{restatable}

As noted in the theorem statement, the setting of \cref{thm:SI-mixing} is any distribution on $\Omega\subseteq\{0,1\}^n$ where the Glauber dynamics is ergodic. In contrast, \cref{thm:SI-constant-relax,thm:SI-constant-mix} are stated for the hard-core model and analogous statements hold for any 2-spin system (i.e., binary graphical model with pairwise interactions); the generalizations to $q>2$ spins are presented in \cref{sec:multispin}.
Examples of 2-spin systems include the hard-core model, the Ising model (see \cref{sub:Ising}), and matchings (when viewed as independent sets of the line graph).

As pointed out earlier, \cref{thm:SI-constant-mix} establishes an optimal bound on the mixing time of $O(n\log{n})$ as Hayes and Sinclair~\cite{HayesSinclair} proved that for any graph with constant maximum degree $\Delta$, the mixing time is $\geq C'(\Delta)n\log{n}$ for some constant $C'(\Delta)$.

Closely connected to the influence matrix is a notion of a local walk, which is introduced in \cref{sec:rapid}.
The local walk is a Markov chain which is defined on spin assignments for individual vertices in such a way that the spectral gap of the local walk is closely related to the maximum eigenvalue of the influence matrix, and hence the spectral independence property yields a lower bound on the spectral gap of the local walk.

The convergence results stated above (\cref{thm:SI-mixing,thm:SI-constant-relax,thm:SI-constant-mix})
are established using a so-called local-to-global theorem of Alev and Lau~\cite{AL20}, which is known as the Random Walk Theorem, presented in \cref{sec:rapid}.  
The Random Walk Theorem relates the spectral gap of the local walk to the spectral gap of the Glauber dynamics, which corresponds to the global walk in the framework presented in \cref{sec:rapid}.
Applying the Random Walk Theorem, in \cref{sec:poly-mixing-proof}, we prove \cref{thm:SI-mixing} of Anari, Liu, and Oveis Gharan~\cite{ALO20}, which shows that spectral independence, together with small-conditional-gap, implies polynomial mixing of the Glauber dynamics where the exponent in the polynomial depends on the spectral independence constant.

The Random Walk Theorem is proved in \cref{sec:RW}, where we introduce the associated up and down walks and explore their connection to the local walk.  Subsequently, in \cref{section-proof-relax}, we present the proof of \cref{thm:SI-constant-relax} of Chen, Liu, and Vigoda~\cite{CLV21}, which shows that spectral independence implies an optimal bound on the spectral gap of the Glauber dynamics; this proof utilizes an improved version of the Random Walk Theorem which is presented in \cref{sub:improved-RW-statement} and proved in \cref{sub:improved-RW-proof}.
Finally, in \cref{sec:optimal} we present the high-level ideas used to further extend the proof from variance contraction to entropy contraction and obtain optimal mixing time.  \cref{lem:opt-relax-SI} is proved in \cref{sub:coupling-SI}.

\section{Application: Hard-core Model}
\label{sub:hard-core-results}

An exciting consequence of the optimal mixing result stated in~\cref{thm:SI-constant-mix} is a beautiful computational phase transition for the hard-core model.
In particular, for every $\Delta\geq 3$ there is a critical point $\lambda_c(\Delta)$ where a computational phase transition occurs on graphs of maximum degree $\Delta$.  Namely, for all $\lambda<\lambda_c(\Delta)$, the Glauber dynamics has optimal mixing time on any graph of maximum degree $\Delta$; this is established using~\cref{thm:SI-constant-mix}.  On the other side, when $\lambda>\lambda_c(\Delta)$ then it is computationally hard to generate samples from a distribution which is close to the Gibbs distribution on all graphs of maximum degree $\Delta$.  

The computational critical point $\lambda_c(\Delta)$ is a consequence of the following phase transition for the Gibbs distribution on $\Delta$-regular trees.  The critical point $\lambda_c(\Delta)$ is defined as follows:
\[
\lambda_c(\Delta):= \frac{(\Delta-1)^{\Delta-1}}{(\Delta-2)^\Delta}.
\]
Note, $\lambda_c(\Delta)$ is a decreasing function, and for all $\Delta\leq 5$ then $\lambda_c(\Delta)>1$, while for $\Delta\geq 6$ then $\lambda_c(\Delta)<1$ ($\lambda_c(\Delta)=1$ at $\Delta\approx 5.14104$).

The function $\lambda_c(\Delta)$ corresponds to the critical point for the following decay of correlations property on trees.  Fix an integer $\Delta\geq 3$.  For an integer $\ell\geq 1$, let $T_\ell$ denote the complete $\Delta$-regular tree of height~$\ell$.  More precisely, $T_\ell$ is the tree where the root is at level $0$, all leaves are at level~$\ell$, and every non-leaf vertex has $\Delta-1$ children.  (Whether the root vertex has degree $\Delta$ or $\Delta-1$ makes negligible difference to the following discussion.)

Let $\mu_\ell$ denote the Gibbs distribution on the tree $T_\ell$, and let $p_\ell$ denote the probability that the root vertex is in an independent set drawn from the distribution $\mu_\ell$.  Note, the maximum sized independent set in~$T_\ell$ contains all of the leaves and then every other level (as we move up the tree).  Thus, for even $\ell$ the root is in the maximum sized independent set, and for odd $\ell$ the root is not in the maximum sized independent set.  Consequently, for large $\lambda$, we expect that $p_{2\ell}>p_{2\ell+1}$ as the root has a bias to be in the independent set for even heights and a bias to be out of the independent set for odd heights. 

There is a phase transition on complete $\Delta$-regular trees at $\lambda_c(\Delta)$ in the following sense.  More formally, for all $\Delta\geq 3$ the following holds:
\begin{align*}
    \mbox{For all }\lambda\leq\lambda_c(\Delta), & \lim_{\ell\rightarrow\infty} p_{2\ell+1}=\lim_{\ell\rightarrow\infty} p_{2\ell}
    \\
      \mbox{For all }\lambda>\lambda_c(\Delta), & \lim_{\ell\rightarrow\infty} p_{2\ell+1}<\lim_{\ell\rightarrow\infty} p_{2\ell}
\end{align*}
This phase transition is derived and explained in more detail in \cref{sec:hard-core}.  

In statistical physics terminology, $\lambda_c(\Delta)$ is the critical point for the uniqueness/non-uniqueness phase transition on the infinite $\Delta$-regular tree, 
which we now briefly describe and we refer the 
interested reader to \cite{Georgii,FV:book} for a more comprehensive introduction.
A Gibbs measure for the hard-core model on the infinite $\Delta$-regular tree is any measure such that for every finite subset of vertices and any pinning outside the subset, the conditional distribution gives a hard-core model with the corresponding pinning. 
In particular, we can define two extremal Gibbs measures, $\mu_{\mathrm{even}}$ and $\mu_{\mathrm{odd}}$, which can be obtained by taking the limit of finite-volume Gibbs distributions on $T_{2\ell}$ and $T_{2\ell+1}$ as $\ell$ tends to infinity, respectively.

When $\lambda\leq\lambda_c(\Delta)$ then the extremal measures are the same $\mu_{\mathrm{even}} = \mu_{\mathrm{odd}}$ and hence there is a unique infinite-volume Gibbs measure; this is called the {\em tree uniqueness region}.  On the other side, when $\lambda>\lambda_c(\Delta)$ then $\mu_{\mathrm{even}}\neq\mu_{\mathrm{odd}}$ and hence there are multiple infinite-volume Gibbs measures; this is called the {\em tree non-uniqueness region}.  

The amazing aspect is that the decay of correlations property on $\Delta$-regular trees directly relates to the computational phase transition on all graphs of maximum degree $\Delta$.  We can now formally state the optimal mixing result in the tree uniqueness region.

\begin{theorem}
    \label{thm:hard-core-constant-degree}
    For all $\Delta\geq 3$, all $\delta>0$, there exists $C(\Delta,\delta)$, for all $\lambda\leq(1-\delta)\lambda_c(\Delta)$, for any $n$-vertex graph $G=(V,E)$ of maximum degree $\Delta$, the mixing time of the Glauber dynamics for the hard-core model satisfies for all $\eps>0$,
    \[ \Tmix(\eps)\leq C(\Delta,\delta)n\log(n/\eps).
    \]
\end{theorem}

\begin{remark}
\label{rem:unbounded-degree}
    The independent works of Chen, Feng, Yin, and Zhang \cite{CFYZ22} and Chen and Eldan \cite{CE25} improved \cref{thm:hard-core-constant-degree} to obtain $\Tmix \leq C(\delta)n\log(n/\eps)$ and thus extend the theorem to classes of graphs where the maximum degree $\Delta$ grows as a function of $n$.  We will show in \cref{sec:critical-point} (see \cref{rem:big-degree})
    how the tools presented for the proof of the following theorem (\cref{thm:hard-core-critical}) yield $\Trelax\leq C(\delta)n$, thus extending \cref{thm:hard-core-constant-degree} to unbounded degrees for the relaxation time.
\end{remark}

A remarkable result of Chen, Chen, Yin, and Zhang~\cite{CCYZ-critical-hard-core} establishes polynomial mixing time for the Glauber dynamics all the way up to {\em and including} the critical point.

\begin{restatable}[{\cite[Theorem 1.1]{CCYZ-critical-hard-core}}]{theorem}{criticalpoint}
\label{thm:hard-core-critical}
    There exist constants $C,C'>0$, such that for all $\Delta\geq 3$, for any graph $G=(V,E)$ of maximum degree $\Delta$, the relaxation time of the Glauber dynamics for the hard-core model at fugacity $\lambda\leq\lambda_c(\Delta)$ satisfies
    \[
    \Trelax \leq C'n^C.
    \]
\end{restatable}

\cref{thm:hard-core-critical} is proved in \cref{sec:critical-point}.

On the other side, it is computationally hard to approximately sample from the Gibbs distribution and approximate the partition function in the tree non-uniqueness region; in fact, not only is it hard to obtain an $\fpras$ (i.e., an arbitrarily close approximation) for the partition function, it is hard to approximate the partition function within some exponential factor even when restricted to triangle-free graphs.  The result below is not proved in this monograph.  We refer the reader to \cite{GSV:colorings,BIS-hardness} for more results in this context (including so-called \#BIS-hardness on bipartite graphs in the tree non-uniqueness region).

\begin{theorem}[\cite{Sly10,SlySun,GSV16,GSV:colorings}]
    \label{thm:hard-core-hardness}
    For all $\Delta\geq 3$, all $\delta>0$, there exists $\eps=\eps(\Delta,\delta)$, for all $\lambda\geq(1+\delta)\lambda_c(\Delta)$, assuming $\mathsf{RP}\neq \mathsf{NP}$, there is no algorithm which for any triangle-free graph $G=(V,E)$ of maximum degree $\Delta$ produces an estimate $\est$ where:
\[ \Prob{Z_{G}(\lambda)2^{-\eps n} \leq \est \leq Z_{G}(\lambda)2^{\eps n}} \geq 3/4,
\]
and runs in time polynomial in $n=|V|$.
\end{theorem}

The above hardness result is for general graphs, or even triangle-free graphs.  If we further restrict attention to a particular class of graphs then the threshold for rapid mixing of the Glauber dynamics might occur at a different point.  A natural class of graphs to consider is random $\Delta$-regular graphs.  Since these graphs are locally tree-like (in particular, they have few short cycles with high probability over the choice of graphs) then we might expect that the tree uniqueness threshold $\lambda_c(\Delta)$ is still the threshold for rapid mixing.  Surprisingly, Chen, Chen, Chen, Yin and Zhang~\cite{CCCYZ-random-regular} showed that when restricted to random $\Delta$-regular graphs then we can obtain fast mixing beyond the uniqueness threshold.

\begin{restatable}[{\cite[Theorem 1.1]{CCCYZ-random-regular}}]{theorem}{randomregularthm}
\label{thm:hard-core-random-graphs}
For all $\Delta\geq 3$, there exists $C(\Delta)$ where the following holds.  For $G$ chosen uniformly at random from all $\Delta$-regular graphs,  with probability $1-o(1)$ over the choice of $G$, for all $\lambda<\frac{1}{3\sqrt{\Delta}}$ the Glauber dynamics for the hard-core model on $G$ at fugacity $\lambda$ satisfies:
\[
\Tmix(\eps) \leq C(\Delta)n\log(n/\eps).
\]
\end{restatable}

\cref{thm:hard-core-random-graphs} is proved in \cref{sec:hard-core-random}.   Inspired by \cref{thm:hard-core-random-graphs} it is natural to conjecture that the rapid mixing threshold on random $\Delta$-regular graphs is closely related to the reconstruction threshold, which considers the influence of typical boundary conditions on complete $\Delta$-regular trees, as opposed to worst-case boundaries as considered in the uniqueness threshold; we refer the interested reader to \cite{BST,RSVVY} for further discussion regarding the reconstruction threshold in the hard-core model.

We outline results for other spin systems in \cref{sec:multispin}, including the Ising model results in \cref{sub:Ising} and colorings results in \cref{sub:colorings}.

\section{Methods for Establishing Spectral Independence}
\label{sub:intro-methods}

In \cref{sec:methods} we present general techniques for establishing spectral independence, focusing on three approaches as outlined below.  As a consequence, for the hard-core model, we prove that spectral independence holds up to the critical point $\lambda_c(\Delta)$ on any graph of maximum degree $\Delta$.

\begin{restatable}{theorem}{SIhardcore}
\label{thm:main-SI-hard-core-uniq}
    For all $\Delta\geq 3$, all $\delta>0$, for all $\lambda\leq(1-\delta)\lambda_c(\Delta)$, for any $n$-vertex graph $G=(V,E)$ of maximum degree $\Delta$, the Gibbs distribution $\mu_{G,\lambda}$ is $\eta(\delta)$-spectrally independent where $\eta(\delta)\leq C/\delta$ for an absolute constant $C>0$.
\end{restatable}

Combining \cref{thm:main-SI-hard-core-uniq} with the results from \cref{sub:mixing-via-SI}, which show that spectral independence implies fast mixing, yields the majority of the fast mixing results for the hard-core model stated in \cref{sub:hard-core-results}.
Towards proving \cref{thm:main-SI-hard-core-uniq}, we present several techniques for establishing spectral independence.

In \cref{sec:Weitz} we present the correlation-decay approach of Weitz~\cite{Wei06} and Bandyopadhyay and Gamarnik~\cite{Gamarnik} and how it yields spectral independence.  At a high level, we show that to upper bound the row-sum of the influence matrix, which is the sum of the influences of a fixed vertex, one can consider influences in a tree, where the tree corresponds to self-avoiding walks in the graph.  Then one can use the potential function method to prove that the sum of the influences in the tree decays exponentially fast (in the distance) when the fugacity is in the tree-uniqueness region.  
In \cref{sec:stability} we show that stability of the partition function, so-called zero-freeness, also implies spectral independence; such conditions were used in the approximate counting algorithm introduced by Barvinok~\cite{Bar17book}.
Finally, in \cref{sub:coupling-SI}, we show that a contractive coupling, and more generally optimal relaxation time for a local chain, implies spectral independence.

\section{Matroids and Trickle-Down Theorem}
\label{sub:matroids-intro}

We utilize the spectral independence approach to establish fast mixing of a Markov chain, known as the bases-exchange walk, for generating a random basis of a matroid, see \cref{sec:matroids} for a more comprehensive introduction and the associated proofs.  Anari, Liu, Oveis Gharan, and Vinzant~\cite{ALOV19} presented the first proof of fast mixing of the bases-exchange walk for arbitrary matroids.  In fact, the work of~\cite{ALOV19} for matroids inspired the idea of spectral independence, however we present these results in reverse chronological order.  We do not directly utilize spectral independence to prove fast mixing of the bases-exchange walk, instead we utilize the framework of down-up chains and the associated Random Walk Theorem (\cref{thm:RW}) which are presented in \cref{sec:RW}.

One additional tool utilized for the analysis of the bases-exchange walk is  the Trickle-Down Theorem of  Oppenheim~\cite{Opp18} which is used to bound the spectral gap of the local walks considered in the Random Walk Theorem, see \cref{sec:trickledown}.
A notable feature of this approach is that it yields a proof based
on local-to-global Markov chain arguments, elementary matroid
structure, and linear algebra.  In particular, it avoids the
sophisticated Hodge-theoretic and algebraic-geometric machinery
underlying alternative approaches through strongly log-concave
or Lorentzian polynomials; see \cite{AHK18,BH20}.

We refer the reader to \cref{sec:matroids} for an introduction to matroids and associated terminology, including bases of a matroid.  A matroid  $\M=(E,\cI)$  consists of a ground set $E$ and a collection of independent sets $\cI$.  For the example of a graphic matroid the ground set is the set of edges $E$ of a graph $G=(V,E)$.  The independent sets are subsets of the ground set  with the requirement that every subset of an independent set is also an independent set (hence the collection $\cI$ is downward closed), and the collection of independent sets also needs to satisfy the so-called exchange property, see \cref{sec:matroids}.
The exchange property implies that all maximal independent sets are of the same size; these maximal independent sets are referred to as the bases of the matroid and their size is the rank of the matroid.
For the example of a graphic matroid, the independent sets are acyclic subgraphs (viewed as subsets of edges) of a graph $G=(V,E)$ and the bases are spanning forests of $G$.

The bases-exchange walk is a Markov chain on the bases of a matroid where the transitions utilize the exchange property in the following manner: from a basis $B_t$ at time $t$, we choose an element $e\in B_t$ uniformly at random, and an element $f$ uniformly from those elements such that $B'=B_t\cup \{f\}\setminus \{e\}$ is a basis, and then we set $B_{t+1}=B'$ as the state at time $t+1$.
The main result presented in \cref{sec:matroids} is a bound on the relaxation time of the bases-exchange walk.

\begin{restatable}[{\cite{ALOV19}}]{theorem}{matroidmain}
\label{thm:matroid-main}
For a matroid $\M=(E,\cI)$, the relaxation time of the bases-exchange walk satisfies $\Trelax \leq r(\M)$, where $r(\M)$ is the rank of the matroid.
Consequently, the mixing time of the bases-exchange walk satisfies $\Tmix=O(r(\M) \log(|\cI|)) = O(r(\M)^2\log n)$ where $n=|E|$.
\end{restatable}
The mixing time bound was improved by Cryan, Guo, and Mousa~\cite{CGM19} to $O(r(\M)\log{\log(|\cI|)}) = O(r(\M) \log r(\M) + r(\M) \log{\log{n}})$; moreover, the $\log\log n$ term was subsequently removed in \cite{ALOVV21}, achieving $O(r(\M) \log r(\M))$ mixing time.

The results for generating a random basis of a matroid are presented beginning in~\cref{sec:matroids} where we review the relevant definitions for matroids and present the main result for the mixing time of the bases-exchange walk.  The analysis of the bases-exchange walk is presented in \cref{sec:trickledown}; the proof utilizes the Random Walk Theorem, presented earlier in \cref{sec:RW}, and the Trickle-down Theorem, which is presented in~\cref{sub:trickledown-statement}.  The Trickle-down Theorem is used to obtain an optimal bound on the spectral gap of the associated local walks.  The proof of fast mixing of the bases-exchange walk is presented in~\cref{sec:basesexchange}, and the proof of the Trickle-down Theorem is presented in~\cref{sec:proof-trickledown}.

\chapter{Properties of the Influence Matrix}
\label{sec:properties}

Here we establish several basic properties of the influence matrix $\Psi$, defined in \cref{sub:SI-definition}.  In \cref{sub:PSD} we prove that the eigenvalues of the influence matrix $\Psi$ are real-valued and non-negative.  In \cref{sub:rowsum,sub:symmetrized-influence,sub:semidefinite-ordering} we present several useful techniques for bounding the maximum eigenvalue of the influence matrix and hence establishing spectral independence.
We first present in \cref{sub:rowsum} a simple approach for bounding the maximum eigenvalue by the maximum absolute rowsum.  Since the influence matrix is not symmetric it is often convenient to consider a symmetrized version of the influence matrix, which we present in \cref{sub:symmetrized-influence}.  In \cref{sub:semidefinite-ordering} we introduce a semidefinite inequality which is equivalent to spectral independence.  

In \cref{sub:CorrelationMatrix} we present the modified influence matrix $\widetilde{\Psi}$ which is closely related to the definition of the influence matrix.  The modified influence matrix leads to a weaker notion, hence referred to as weak spectral independence; the connections are discussed in \cref{sub:CorrelationMatrix}.

\section{Nonnegative eigenvalues}
\label{sub:PSD}

Here we will prove that the influence matrix has non-negative real eigenvalues by relating it to the covariance matrix.

For a matrix $A$, its signature is the following triple of integers: number of positive, negative, and zero eigenvalues. Sylvester's law of inertia
states that for a symmetric matrix $A$ and a non-singular matrix $B$, the signature of $A$ and $BAB^T$ are the same. We will utilize the following basic consequence of Sylvester's law.

\begin{lemma}
\label{lem:sylvesters}
For a symmetric matrix $A$ and a positive diagonal matrix $D$, we have that both $A$ and $DA$ have real eigenvalues and the signatures of $A$ and $DA$ are the same. 
\end{lemma}

\begin{proof}
Note that all eigenvalues of $A$ are real since it is symmetric. Observe that $DA$ is similar to the matrix $D^{-1/2}(DA)D^{1/2} = D^{1/2} A D^{1/2}$ 
which is a symmetric matrix and hence also has real eigenvalues. Note that, by Sylvester's law of inertia, $A$ and $D^{1/2} A D^{1/2}$ (and hence $DA$)
have the same signature.
\end{proof}

A simpler version of~\cref{lem:sylvesters} that states that $A$ is positive semidefinite if and only if $DA$ is positive semidefinite does not
require Sylvester's law of inertia, see~\cref{sub:semidefinite-ordering}.

\begin{lemma}\label{DL1}
All eigenvalues of $\Psi$ are non-negative real numbers.
\end{lemma}

\begin{proof}
    For every $i,j \in [n]$, the covariance of $\sigma(i),\sigma(j)$ is given by
    \begin{align}
        \Cov_{\mu}(i,j) 
        = \Cov_{\sigma \sim \mu} (\sigma(i), \sigma(j))
        &= 
        \E_{\mu}[\sigma(i) \sigma(j)]
        - \E_{\mu}[\sigma(i)]\cdot\E_\mu[\sigma(j)]
        \label{cov-defn} \\
        \nonumber
        &=
        \mu\left( \sigma(i)=\sigma(j)=1 \right)
        - \mu(\sigma(i)=1)\cdot\mu(\sigma(j)=1)
        \\&=
        \mu(\sigma(i)=1)\cdot\Big(\mu(\sigma(j)=1 \mid \sigma(i)=1) - \mu(\sigma(j)=1)\Big).
\nonumber
\end{align}
    Plugging in
    \begin{align*}
        \mu(\sigma(j) = 1) = \mu(\sigma(j) = 1 \mid \sigma(i) = 1) \cdot \mu(\sigma(i) = 1) 
        + \mu(\sigma(j) = 1 \mid \sigma(i) = 0) \cdot \mu(\sigma(i) = 0), 
    \end{align*}
    we obtain that
    \begin{align}
\lefteqn{
\Cov_{\mu}(i,j) 
}
\nonumber
\\
\nonumber
        &=  
        \mu(\sigma(i)=1)\cdot\Big(\mu(\sigma(j)=1 \mid \sigma(i)=1) - 
        \mu(\sigma(j) = 1 \mid \sigma(i) = 1) \cdot \mu(\sigma(i) = 1) 
         \\ 
         \nonumber
         & \hspace*{1.25in} - \mu(\sigma(j) = 1 \mid \sigma(i) = 0) \cdot \mu(\sigma(i) = 0)
        \Big)
        \\
        \nonumber
             &=  
        \mu(\sigma(i)=1)\cdot\Big(\mu(\sigma(j)=1 \mid \sigma(i)=1) \cdot(1-\mu(\sigma(i) = 1)) 
 - \mu(\sigma(j) = 1 \mid \sigma(i) = 0) \cdot \mu(\sigma(i) = 0)
        \Big)
             \\
             \label{last-cov-calc}
             &=  
        \mu(\sigma(i)=1)\cdot\mu(\sigma(i) = 0) \cdot\Big(\mu(\sigma(j)=1 \mid \sigma(i)=1)
 - \mu(\sigma(j) = 1 \mid \sigma(i) = 0)
        \Big).
    \end{align}
Let $\Var_\mu$ denote the vector of variances, with 
\begin{align*}
    \Var_\mu(i) = \Var_{\sigma\sim\mu}(\sigma(i)) = \mu(\sigma(i) = 1) \cdot \mu(\sigma(i) = 0)
\end{align*}
for all $i\in[n]$.
Plugging in the definition of the influence matrix $\Psi$ into \cref{last-cov-calc} then we have for all $i,j \in [n]$,
\begin{align}\label{eq:cov-inf}
    \Cov_{\mu}(i,j) = \Var_\mu(i) \cdot \Psi_{\mu}(i,j).
\end{align}

    Let $D = \diag(\Var_\mu)$ denote the diagonal matrix of variances, and hence $D^{-1}$ is the diagonal matrix with $D^{-1}(i,i)=1/D(i,i)$.  Then we have the matrix identity: 
    \begin{align*}
        \Psi_{\mu} = D^{-1}\Cov_{\mu}.
    \end{align*}
    Since $\Cov_{\mu}$ is symmetric positive semidefinite (see \cite[Chapter 7]{horn2012matrix}), and $D$ (and hence $D^{-1}$ as well) is a diagonal matrix with positive diagonal entries (this follows from the definition of the set $T$ of free vertices in \cref{sub:SI-definition}), then, by \cref{lem:sylvesters}, we have that $\Psi_{\mu}$ has non-negative real eigenvalues.
\end{proof}

\section{Bounding Spectral Radius of Influence Matrix}
\label{sub:rowsum}   

To upper bound the maximum eigenvalue of the influence matrix, most approaches use the maximum absolute row sum.  Recall, the influence matrix $\Psi$ is defined to be a signed matrix, meaning the entries can be positive, negative, or zero.
Taking the entrywise absolute value of $\Psi$ and then taking the maximum row sum yields an upper bound on the maximum eigenvalue of $\Psi$.

\begin{lemma}
    \label{lem:rowsum}
For any $n\times n$ influence matrix $\Psi$, 
\[ \uplambda_{\max}(\Psi) \leq \max_{i} \sum_{j} |\Psi(i,j)|.
\]
\end{lemma}

The below proof and statement of \cref{lem:rowsum} hold for arbitrary matrices $M$ in place of the influence matrix $\Psi$ with the largest eigenvalue replaced by the spectral radius of $M$.

\begin{proof}
Consider an $n\times n$ matrix $M$ with all real eigenvalues.
Let $v$ be the eigenvector corresponding to the largest (in absolute value) eigenvalue $\uplambda$ of $M$.  Thus, 
\begin{equation}
    \label{eigen}
    Mv = \uplambda v.
\end{equation} Let $1\le i\le n$ be the position in $v$ where $|v_i|$ is maximized (note that $|v_i|>0$). We have
\begin{align*}
    |\uplambda|\,|v_i| & = |(M v)_i|   & \mbox{by \cref{eigen}}
    \\ & =
    \big| \sum_j M_{ij} v_j \big| 
    \\ & \leq 
    \sum_j |M_{ij}|\,|v_j| & \mbox{by the triangle inequality}
    \\ & \leq  
    |v_i| \sum_j |M_{ij}| & \mbox{since $|v_i|$ is the maximum.}
\end{align*}
Dividing both sides by $|v_i|$ yields the lemma.
\end{proof}

More generally, one can consider any matrix $p$-norm $\|\Psi\|_p$ for any $p\geq 1$.  The absolute row sum in \cref{lem:rowsum} corresponds to $p=\infty$, and the absolute column sum corresponds to $p=1$.
Recall, the matrix $p$-norm is defined as follows:
\[
  \|\Psi\|_p := \max_{\mathbf{x}: \|\mathbf{x}\|_p=1} \|\Psi \mathbf{x}\|_p,
\]
where $\mathbf{x}$ is an $n$-dimensional real-valued unit column vector, and for $\mathbf{y}$, which is an $n$-dimensional real-valued column vector, the vector $p$-norm of $\mathbf{y}$ is $\|\mathbf{y}\|_p=\left(\sum_{i} |\mathbf{y}_i|^p\right)^{1/p}$.
\begin{lemma}
    \label{lem:p-norm}
For any $n\times n$ influence matrix $\Psi$, for any $p\geq 1$,
\[ \uplambda_{\max}(\Psi) \leq  \|\Psi\|_p. 
\]
\end{lemma}

\begin{proof}
As in the proof of \cref{lem:rowsum}, we prove the stronger
statement for an arbitrary $n\times n$ matrix $M$ with real
eigenvalues.
Let $v$ be the eigenvector corresponding to the largest (in absolute value) eigenvalue $\uplambda$ of~$M$.  Thus, 
\[
    Mv = \uplambda v.
\]
We have
\begin{align*}
    |\uplambda|\,\|v\|_p & = \|\uplambda v\|_p   & 
    \\ & =
    \big\| M v\big\|_p 
    \\ & \leq 
    \|M\|_p \|v\|_p, &
\end{align*}
and dividing both sides by $\|v\|_p$ completes the proof of the lemma.
\end{proof}

\section{Correlation Matrix}
\label{sub:symmetrized-influence}

As pointed out earlier, the influence matrix $\Psi$ is not necessarily symmetric.  
It is often convenient to work with a symmetric matrix, and in this section we introduce a symmetric version of the influence matrix.

 Recall from \cref{defn:inf-matrix}, the influence matrix $\Psi$ has $(u,v)$-entry:
 $$
 \Psi(u\ra v) := \mu(\sigma(v)=1|\sigma(u)=1) - \mu(\sigma(v)=1|\sigma(u)=0).
 $$
 Also recall from \cref{eq:cov-inf} that the following holds for all $u,v\in V$:
 \begin{equation}\label{e3o}
 \Psi(v\ra u) \cdot \Var_\mu(v)
 = \Cov_\mu(u,v)
 = \Psi(u\ra v) \cdot \Var_\mu(u).
 \end{equation}

 Let $D = \diag(\Var_\mu)$ be the diagonal matrix with entries $\Var_\mu(v) = \mu(\sigma(v)=1) \cdot \mu(\sigma(v)=0)$. Equation~\eqref{e3o} means that $D \Psi = \Cov_\mu$ is symmetric.  
 Since $D\Psi$ is symmetric then $D^{1/2} \Psi D^{-1/2}$ is also symmetric, as
 it is obtained from $D \Psi$ by multiplying by the same diagonal matrix on the left and on
 the right.  Now we can define the symmetric version of the influence matrix.

\begin{definition}\label{def:correlation-matrix}
 Let the {\em correlation matrix} be defined as $D^{1/2} \Psi D^{-1/2}$ where $D = \diag(\Var_\mu)$ is the diagonal matrix of variances, with entries $D(i,i) = \Var_\mu(i) = \mu(\sigma(i)=1) \cdot \mu(\sigma(i)=0)$.
\end{definition}

An important feature of the correlation matrix is that the largest eigenvalue is the same as for the influence matrix.  

\begin{lemma} \label{lem:correlation-matrix}
For an $n\times n$ influence matrix $\Psi$, the correlation matrix 
$D^{1/2} \Psi D^{-1/2}$ has $(u,v)$-entries $\Cov_\mu(u,v)/\sqrt{\Var_\mu(u) \cdot \Var_\mu(v)}$, whose absolute value equals $\sqrt{\Psi(u\ra v)\Psi(v\ra u)}$.
Moreover, the correlation matrix $D^{1/2} \Psi D^{-1/2}$ is similar to the influence matrix $\Psi$, and thus we have:
 \[
     \uplambda_{\max}(\Psi) = \uplambda_{\max}(D^{1/2} \Psi D^{-1/2}).
 \]
 \end{lemma}

 \begin{proof}
 The $(u,v)$-entry in $D^{1/2} \Psi D^{-1/2}$ is
 \begin{align*}
     \sqrt{\Var_\mu(u)} \cdot \Psi(u\ra v) \cdot \frac{1}{\sqrt{\Var_\mu(v)}}
     = \sqrt{\Var_\mu(u)} \cdot \frac{\Cov_\mu(u,v)}{\Var_\mu(u)} \cdot \frac{1}{\sqrt{\Var_\mu(v)}}
     = \frac{\Cov_\mu(u,v)}{\sqrt{\Var_\mu(u) \cdot \Var_\mu(v)}}.
 \end{align*}
 Notice that the square of the $(u,v)$-entry is
 \begin{align*}
     \frac{\left(\Cov_\mu(u,v)\right)^2}{\Var_\mu(u) \cdot \Var_\mu(v)} = \Psi(u\ra v) \cdot \Psi(v \ra u),
 \end{align*}
 establishing the lemma. 
 \end{proof}

 \section{Semidefinite Ordering}
 \label{sub:semidefinite-ordering}

Recall that an $n\times n$ symmetric matrix $A$ is positive semidefinite (written $A\succeq 0$) if for all $x\in\R^n$ we have $x^T A x\geq 0$. Semidefiniteness induces an ordering
on $n\times n$ symmetric matrices where $A\preceq B$ iff $B-A\succeq 0$. 

Let $D = {\mathrm{diag}}(\Var_\mu)$ be the diagonal matrix of variances, with entries  $D(i,i) = \Var_\mu(i) = \mu(\sigma(i)=1) \cdot \mu(\sigma(i)=0)$.
 Bounding the largest eigenvalue of the influence matrix $\Psi$ by $\uplambda$ is equivalent to the following statement:
 \[
 D^{1/2} \Psi D^{-1/2} \preceq \uplambda I,
 \]
 which is equivalent to the following:
 \[
 D \Psi \preceq \uplambda D.
 \]
This is because of the following fact: for any invertible $U$ we have $A\preceq B$ is equivalent to $U^T A U \preceq U^T B U$, since for any 
$x\in\R^n$ we have $x^T (U^T B U - U^T A U) x = (Ux)^T (B-A) (Ux)$; we used $U=D^{1/2}$ for the above equivalence.

 Recall the $n\times n$ covariance matrix $\Cov_\mu$ from \cref{cov-defn}, and $\Cov_\mu=D\Psi$.
Therefore, we have that $\uplambda_{\max}(\Psi)\leq 1+\eta$ is equivalent to the following semidefinite inequality:
 \begin{equation}
  \label{cov-SI-connection}
 \Cov_\mu\preceq (1+\eta) D.
 \end{equation}
 Recall that the condition in \cref{eqn:SI-defn}, namely $\uplambda_{\max}(\Psi)\leq1+\eta$, is equivalent to \cref{cov-SI-connection}.

 \section{Modified Influence Matrix}
 \label{sub:CorrelationMatrix}

 \begin{definition}[Modified Influence Matrix]
\label{defn:correlation-matrix}
Let $G=(V,E)$ be a graph where $V=\{1,\dots,n\}$, and
$\mu$ be a distribution on $\Omega$ where $\Omega\subset\{0,1\}^n$.
Let $\widetilde{\Psi}$ be the following real-valued $n\times n$ matrix; we refer to $\widetilde{\Psi}$ as the \emph{modified influence matrix}. For $1\le i,j\le n$, let
\begin{equation}
\label{eqn:CM-empty}
\widetilde{\Psi}(i\rightarrow j) = \widetilde{\Psi}(i,j) := 
    \mu\left(\sigma(j)=1 \mid \sigma(i)=1\right)
- 
\mu\left(\sigma(j)=1\right).
\end{equation}
\end{definition}

\begin{remark}
    We remark that the modified influence matrix was first studied in \cite{AASV21} and was called the correlation matrix. However, we define the correlation matrix in \cref{def:correlation-matrix} which matches the standard definition in statistics, as can be seen from \cref{lem:correlation-matrix}.
\end{remark}

We consider an analog of spectral independence for the modified influence matrix, which we refer to as \emph{weak spectral independence}.

\begin{definition}[Weak Spectral Independence]
\label{defn:weak-SI}
For $\eta\geq 0$, we say that $\mu$ is $\eta$-weak spectrally independent if for all pinnings $\tau\in\Pinning$,
\begin{equation}
    \label{cond-weak-max}
\uplambda_{\max}(\widetilde{\Psi}_\tau)\leq 1+\eta.
\end{equation}
\end{definition}

Recall, for $\eta$-spectral independence, $\uplambda_{\max}(\Psi)\leq 1+\eta$ is equivalent to $\Cov_\mu\preceq (1+\eta) D$, see \cref{cov-SI-connection}.  
Let $\widetilde{D}={\mathrm{diag}}(m_\mu)$ be the diagonal matrix of means, with entries  
\begin{align*}
    \widetilde{D}(i,i) = m_\mu(i) = \Exp_{\sigma \sim \mu}[\sigma(i)] = \mu(\sigma(i)=1).
\end{align*}
Analogous to \cref{cov-SI-connection}, we have that $\uplambda_{\max}(\widetilde{\Psi})\leq 1+\eta$ is equivalent to:
\begin{equation}
    \label{CM-semidefinite}
    \Cov_\mu\preceq (1+\eta) \widetilde{D}.
\end{equation}

\begin{remark}
\label{rem:weak-vs-nonweak}
Note that~\eqref{cov-SI-connection} implies~\eqref{CM-semidefinite}, since for a random variable $X\in [0,1]$ we have $\Var(X)\leq\Exp(X)$. Consequently, $\eta$-spectral independence implies $\eta$-weak spectral independence.
Moreover, under the assumption of $b$-marginal boundedness the notions of $\eta$-spectral independence and $\eta'$-weak spectral independence are equivalent for some constants $\eta,\eta'>0$; namely, $\eta'$-weak spectral independence implies $((1+\eta')/b-1)$-spectral independence. 
\end{remark}

Note that for the rapid mixing results presented in \cref{sec:introduction} we need $\eta$-spectral independence for a constant $\eta>0$.  In particular, to prove \cref{thm:SI-constant-relax} which establishes optimal relaxation time we need $\eta$-spectral independence for constant $\eta>0$.  For \cref{thm:SI-constant-mix} which establishes optimal mixing time, as noted above in \cref{rem:weak-vs-nonweak}, under the assumption of $b$-marginal boundedness the notions of $\eta$-spectral independence and $\eta'$-weak spectral independence are equivalent for some constants $\eta,\eta'>0$.

In \cref{sec:matroids}, for the application to the bases-exchange walk for matroids, we implicitly show that $0$-weak spectral independence implies optimal relaxation time, see \cref{rem:weakSI-mixing}.
Moreover, in \cref{sub:logconcave} we show that the log-concavity of the partition function is closely related to $0$-weak spectral independence.

\chapter{Markov Chain Fundamentals}
\label{sec:MC}

To prove \cref{thm:SI-constant-relax} we will bound the spectral gap of the Glauber dynamics.  To do that, we consider a functional analysis formulation of the spectral gap.  In particular, we will bound the decay rate of the variance of the chain with respect to the stationary distribution for an arbitrary real-valued function on the state space.  This yields tight bounds on the spectral gap and relaxation time of the chain.  Subsequently in \cref{sec:optimal}, we will extend the notions presented in this chapter from variance to entropy, and present the main ideas for bounding the (modified) log-Sobolev constant which captures the decay rate of relative entropy; this yields optimal upper bounds on the mixing time, rather than the relaxation time.

\section{Mixing Time and Total Variation Distance}
\label{sub:mixingtime-defn}

    For a pair of distributions $\mu$ and $\nu$ on a finite state space $\Omega$, their total variation distance is defined as:
    \[ \|\mu-\nu\|_{\mathrm{TV}} = \frac12\sum_{\sigma\in\Omega}|\mu(\sigma)-\nu(\sigma)| = \max_{S\subset\Omega} |\mu(S) - \nu(S)| = \max_{S\subset\Omega} \mu(S) - \nu(S).
    \]
For a Markov chain $(X_t)$ with transition matrix $P$, state space $\Omega$, and unique stationary distribution $\mu$, its mixing time is the number of steps $\Tmix$ from the worst-case initial state $X_0\in\Omega$ to reach a distribution which is within total variation distance $\leq 1/4$ of $\mu$:
\[
\Tmix = \min\{t:\forall X_0\in\Omega, \|P^t(X_0,\cdot)-\mu\|_{\mathrm{TV}}\leq 1/4 \}.
    \]
    More generally, for any $\eps>0$, let 
    \[
\Tmix(\eps) = \min\{t:\forall X_0\in\Omega, \|P^t(X_0,\cdot)-\mu\|_{\mathrm{TV}}\leq \eps \}.
    \]

 We say the mixing time is \emph{optimal} if $\Tmix=O(n\log{n})$ as this matches the lower bound established by Hayes and Sinclair~\cite{HayesSinclair} for any graph of constant maximum degree $\Delta$.  Note an upper bound of $\Tmix=O(n\log{n})$ yields a more general bound of $\Tmix(\eps)=O(n\log{n}\log(1/\eps))$ via a standard boosting result.  
 Specifically, an optimal mixing time bound of $\Tmix=O(n\log{n})$ implies the following mixing time bound for any $\eps>0$ (see~\cite[Section 4.5]{LevinPeresWilmer}):
   \[ \Tmix(\eps)\leq \Tmix(1/4)\times\lceil\log_2(1/\eps)\rceil.
   \]  
 However, the proofs of optimal mixing in this monograph will yield a stronger upper bound of $\Tmix(\eps)=O(n\log(n/\eps))$, and hence we will state this improved bound when relevant.

   Closely connected to variation distance is the coupling technique.  
For a pair of distributions $\mu,\nu$, both of which are on the same state space $\Omega$, a coupling $\omega$ of $\mu$ and $\nu$ is a distribution on $\Omega\times\Omega$ such that the projections on each coordinate follow the respective distributions.  More formally, $\omega$ is a valid coupling of $\mu$ and $\nu$ if the following hold:
\begin{itemize}
    \item For all $\tau\in\Omega$, \[ \sum_{\sigma\in\Omega} \omega(\sigma,\tau) = \nu(\tau),
    \]
        \item For all $\sigma\in\Omega$, \[ \sum_{\tau\in\Omega} \omega(\sigma,\tau) = \mu(\sigma).
    \]
\end{itemize}

The goal is to design a coupling $\omega$ which minimizes the disagreement probability which is the probability of $\sigma\neq \tau$ where $(\sigma,\tau)$ is a sample from the coupled distribution $\omega$.  This disagreement probability provides an upper bound on the total variation distance of the coupled distributions $\mu$ and $\nu$, and there always exists a coupling where the disagreement probability equals the total variation distance; this is summarized in the following lemma, often known as the coupling lemma (see, e.g., \cite[Lemma 3.6]{Aldous} 
and \cite[Proposition 4.7 and Remark 4.8]{LevinPeresWilmer}).

\begin{lemma}
    For a pair of distributions $\mu$ and $\nu$ on a finite state space $\Omega$, let $\PP(\mu,\nu)$ denote the family of all valid couplings of $\mu$ and $\nu$.
    \begin{itemize}
        \item For any $\omega\in\PP(\mu,\nu)$, for a sample $(\sigma,\tau)\sim\omega$, 
        \[
        \|\mu-\nu\|_{\mathrm{TV}} \leq \Prob{\sigma\neq\tau}.
        \]
        \item There exists a coupling $\omega\in\PP(\mu,\nu)$ where:
         \[
        \|\mu-\nu\|_{\mathrm{TV}} = \Prob{\sigma\neq\tau}.
        \]
    \end{itemize}
\end{lemma}

For a more comprehensive introduction to coupling methods in the MCMC setting see \cite{LevinPeresWilmer,Jerrum:notes}.

\section{Reversibility}
\label{sub:reversibility}

We describe an ergodic Markov chain by $(\Omega,P,\mu)$, where $\Omega$ is the finite state space, $P$ is the $N\times N$ transition matrix for $N=|\Omega|$, and $\mu$ is the unique stationary distribution (we only consider ergodic Markov chains in this text).  For some of the upcoming algebraic statements in terms of eigenvalues we will need to restrict attention to reversible chains; these are Markov chains whose transition matrix $P$ and stationary distribution $\mu$ satisfy the following condition known as the reversibility condition or detailed balance conditions, namely, for all pairs $\sigma,\sigma'\in\Omega$, 
\begin{equation}
    \label{reversibility}
\mu(\sigma)P(\sigma,\sigma')=\mu(\sigma')P(\sigma',\sigma).
\end{equation}
It is straightforward to verify the reversibility condition \cref{reversibility} for the Glauber dynamics since $P(\sigma,\sigma')>0$ iff $\sigma = \sigma'$ or $\sigma\oplus\sigma'=\{v\}$ for some $v\in V$ (i.e., the Glauber dynamics changes at most one vertex in each transition).
For a discussion of non-reversible chains, we refer the interested reader to Montenegro and Tetali's monograph~\cite{MontenegroTetali}.

\section{Dirichlet Form and Variance}
\label{sub:Dirichlet}

Consider a function $f:\Omega\rightarrow\R$.  The expectation of $f$ with respect to the stationary distribution $\mu$ is denoted as follows:
\[
\mu(f) := \Exp_\mu[f] := \sum_{\eta\in\Omega}\mu(\eta)f(\eta).
\]
We will also study the variance of $f$ with respect to $\mu$ which has the following convenient, equivalent formulations that can be easily verified by simple algebra:
\begin{align*}
\nonumber
\Var_\mu(f) & 
:= \mu\left( \left(f - \mu(f)\right)^2 \right)
\\
& = \sum_{\sigma\in\Omega}\mu(\sigma)\left(f(\sigma)-\mu(f)\right)^2
\\
& = \frac12\sum_{\sigma,\eta\in\Omega}\mu(\sigma)\mu(\eta)(f(\sigma)-f(\eta))^2,
\end{align*}
see \cite[(5.4-5.5)]{Jerrum:notes} for a proof of the last identity.
Notice that the last formulation of variance is a sum over all pairs of states of the difference in their functional value (squared to measure the absolute value).

It will be convenient later to use inner product and $2$-norm to represent the variance functional. 
For two functions $f,g:\Omega\rightarrow\R$, the \emph{inner product} of $f$ and $g$ with respect to $\mu$ is defined as
\begin{align*}
    \ip{f}{g}_\mu := \sum_{\eta\in\Omega} \mu(\eta) f(\eta)g(\eta) = \mu(fg),
\end{align*}
and the \emph{$2$-norm} of $f$ with respect to $\mu$ is defined as
\begin{align*}
    \norm{f}_{2,\mu} = \sqrt{\ip{f}{f}_\mu}.
\end{align*}

We remark that a transition matrix $P$ is reversible with respect to the distribution $\mu$ over $\Omega$ if and only if for all functions $f,g: \Omega \to \R$, it holds
\begin{align*}
    \ip{Pf}{g}_\mu = \ip{f}{Pg}_\mu,
\end{align*}
where $Pf: \Omega \to \R$ is defined as $Pf(\sigma) = \sum_{\eta \in \Omega} P(\sigma, \eta)f(\eta)$ and can be intuitively understood as the matrix-vector product in a finite state space.
Namely, $P$ is reversible with respect to $\mu$ iff it is a self-adjoint operator under the inner product defined by $\mu$.

We can represent the variance of $f$ as
\begin{align*}
    \Var_\mu(f) = \norm{f-\mu(f)}_{2,\mu}^2.
\end{align*}

The variance functional is closely related to the $\chi^2$-divergence. 
Specifically, for any distribution $\nu$ over $\Omega$, the \emph{$\chi^2$-divergence} of $\nu$ from $\mu$ is defined as
\begin{align*}
    D_{\chi^2}(\nu \,\Vert\, \mu) := \sum_{\eta \in \Omega} \mu(\eta) \left( \frac{\nu(\eta)}{\mu(\eta)} - 1 \right)^2
    = \sum_{\eta \in \Omega} \frac{\left(\nu(\eta)-\mu(\eta)\right)^2}{\mu(\eta)}.
\end{align*}
Observe that $D_{\chi^2} (\nu \,\Vert\, \mu) \ge 0$, with equality if and only if $\nu = \mu$.
The $\chi^2$-divergence is not a metric (it does not satisfy symmetry or the triangle inequality), and it upper bounds the total variation distance (see \cref{sub:mixingviagap}).
If we let $f = \nu/\mu$ be the relative density of $\nu$ with respect to $\mu$, i.e., $f(\eta) = \nu(\eta)/\mu(\eta)$ for each $\eta \in \Omega$, then it holds $\mu(f) = 1$ and $\Var_\mu(f) = D_{\chi^2} (\nu \,\Vert\, \mu)$.

In contrast to the variance which measures the ``global variation'' of the function $f$, the {\em Dirichlet form} is a measure of the ``local variation'' for a specified Markov chain.  The terms global vs. local are with respect to the graph $(\Omega,P)$.
The Dirichlet form measures the difference in the functional value over pairs of adjacent states with respect to the graph $(\Omega,P)$ as formulated in the following definition.
\begin{definition}[Dirichlet Form] For a Markov chain on state space $\Omega$ with transition matrix $P$, and for a function $f:\Omega\rightarrow\R$, define the \emph{Dirichlet form} as follows:
    \[
        \Dirichlet_P(f) := \frac{1}{2}
        \sum_{\sigma,\eta\in\Omega}\mu(\sigma)P(\sigma,\eta)(f(\sigma)-f(\eta))^2.
    \]
    \label{defn:Dirichlet}
\end{definition}
Equivalently, we can write the Dirichlet form as
\begin{align*}
    \Dirichlet_P(f) = \ip{f}{(I-P)f}_\mu,
\end{align*}
where $I$ is the identity matrix, and $(I-P)f$ is a function from $\Omega$ to $\R$ obtained as a matrix-vector product of $I-P$ and $f$ (viewed as a column vector).

 \section{Spectral Gap}
 \label{sub:spectral-gap}
 
For an ergodic, reversible Markov chain $(\Omega,P,\mu)$, recall that $P$ denotes the transition matrix and $\mu$ is the stationary distribution which means that $\mu$ is a row eigenvector of $P$ with eigenvalue $1$.  Since the Markov chain is reversible, the eigenvalues of $P$ are real-valued; this follows from the fact that $P$ is similar to 
$\diag(\mu)^{1/2} P \diag(\mu)^{-1/2}$ and the latter matrix is symmetric by reversibility. We denote the eigenvalues of $P$ as follows:
 \[
 1> \uplambda_2\geq\uplambda_3\geq \cdots \geq \uplambda_N>-1.
 \] 
 The fact that $\uplambda_N>-1$ follows from the aperiodicity of an ergodic Markov chain.  
The \emph{spectral gap} is defined as $\gamma=1-\uplambda_2$, and the \emph{absolute spectral gap} is $1-\uplambda_*$ where $\uplambda_*=\max\{\uplambda_2,|\uplambda_N|\}$ is the second largest eigenvalue in absolute value. 

Dyer, Greenhill, and Ullrich~\cite{DGU} proved that the heat-bath block dynamics is positive semidefinite (PSD), that is, all of the eigenvalues of $P$ are non-negative.  The heat-bath block dynamics is defined by a collection of blocks $B_1,\dots,B_L$ where each $B_i\subset V$ and $\cup_i B_i= V$; the steps $X_t\rightarrow X_{t+1}$ of the block dynamics operate by first choosing a random $B_i$ and then choosing the configuration $X_{t+1}(B_i)$ on $B_i$ from the Gibbs distribution conditional on $X_t(V\setminus B_i)$ (the configuration on $V\setminus B_i$ stays the same).  
The transition matrix of heat-bath update for one block is positive semidefinite (since it is a block diagonal matrix with positive semidefinite blocks of rank 1). The positive semidefiniteness now follows since the heat bath chain is a mixture of block updates.
The Glauber dynamics is the special case of the heat-bath block dynamics where the blocks are single vertices.
Hence, for the Glauber dynamics, the absolute spectral gap is the same as the spectral gap.

The absolute spectral gap governs the decay rate of variance for
the discrete-time chain (which is the setting of this monograph).   For a general reversible chain, the ordinary spectral gap
alone does not control variance decay when negative eigenvalues
are present. One can avoid this issue by passing to the lazy
chain, at the cost of a factor of $2$ in the spectral gap.
For the Glauber dynamics considered here, this distinction is
irrelevant because its transition matrix is positive
semidefinite.

\begin{lemma}[Spectral gap]
\label{lem:spectral-gap-Poincare}
    For a reversible Markov chain defined by $(\Omega,P,\mu)$, the spectral gap $\gamma$ is the largest constant such that the following inequality holds for any function $f:\Omega\rightarrow\R$,
    \begin{equation}
    \label{defn:Poincare}
        \gamma\cdot \Var_\mu(f) \leq \Dirichlet_P(f)\, .
    \end{equation}
\end{lemma}

We refer the reader to \cite[Lemma 13.7]{LevinPeresWilmer} for a complete proof of \cref{lem:spectral-gap-Poincare} but only sketch the main steps to obtain this lemma.
Observe that the optimal $\gamma$ satisfying \cref{defn:Poincare} for all $f$ is given by
\begin{align*}
    \gamma = \inf_{f: \Var_\mu(f) \neq 0} \frac{\Dirichlet_P(f)}{\Var_\mu(f)}
    = 
    \inf_{f:\mu(f)=0}
    \frac{\ip{f}{(I-P)f}_\mu}{\ip{f}{f}_\mu}
    = \text{second smallest eigenvalue of }(I-P),
\end{align*}
where the last equality follows from the Courant--Fischer theorem (since $P$ is reversible/self-adjoint), observing that $\allone$ is the eigenvector of $I-P$ corresponding to the smallest eigenvalue $0$.

For non-reversible Markov chains, we can define $\gamma$ as in \cref{defn:Poincare} which is referred to as the Poincar\'{e} constant (see~\cite[Remark 5.3]{Jerrum}); however it is not necessarily equivalent to the spectral gap $1-\uplambda_2$ in the case of non-reversible chains.

Recall, the
relaxation time is defined as the inverse of the absolute spectral gap:
\[ \Trelax :=(1-\uplambda_*)^{-1}.
\] 
Let $P_{\textsc{gd}}$ denote the transition matrix of the Glauber dynamics.  As noted above, for the Glauber dynamics the transition matrix $P_{\textsc{gd}}$ is PSD and hence the absolute spectral gap is the same as the spectral gap, and thus we have:
\begin{equation}\label{eq:relax-gap}
    \Trelax(P_{\textsc{gd}}) =  \gamma^{-1}.
\end{equation}

\section{Mixing Time Bounds via Spectral Gap}
\label{sub:mixingviagap}

We now show how the spectral gap implies an upper bound on the mixing time.   For the Glauber dynamics~$(X_t)$, let $P_{\textsc{gd}}$ denote the transition matrix, let $\nu_t$ denote the distribution of the random variable~$X_t$ at time~$t$, and let~$\mu$ denote the stationary distribution.  Let 
\[
\mu^*=\min_{x\in\Omega} \mu(x).
\]
We will show the following bound:
 \begin{equation}\label{eq:mix-gap}
 \Tmix(P_{\textsc{gd}}) \leq \frac{1}{2\gamma}\ln(4/\mu^*).
\end{equation}
We will also establish a more general bound which accounts for a so-called warm start when the initial distribution $\nu_0$ is somewhat close to the stationary distribution $\mu$.

Let $f_t = \nu_t/\mu$ denote the relative density at time $t$, i.e., $f_t(x) = \nu_t(x)/\mu(x)$ for each $x \in \Omega$.
Notice that $\mu(f_t)=1$.
Recall the $\chi^2$-divergence from stationarity defined as:
\[
D_{\chi^2}(\nu_t \,\Vert\, \mu) =
\Var_\mu(f_t) = 
  \left\| f_t - 1\right\|_{2,\mu}^2 
  = \sum_{x\in\Omega}\mu(x)\left(\frac{\nu_t(x)}{\mu(x)} - 1\right)^2.
\]
Meanwhile, the $L^1$-distance from stationarity is related to the total variation distance:
\[
 \left\| f_t - 1\right\|_{1,\mu} = 
 \sum_{x\in\Omega} |\nu_t(x) - \mu(x)| = 2 \left\| \nu_t - \mu\right\|_{\mathrm{TV}}.
\]

We can bound the convergence time to the stationary distribution $\mu$ as a function of the initial distance.
Let 
\begin{equation}
    \label{initial-dist}
    M:= D_{\chi^2}(\nu_0 \,\Vert\, \mu)
\end{equation} 
denote the initial $\chi^2$-divergence from the stationary distribution.
Note that $M$ is typically exponentially large in $n$:
\begin{equation}
    \label{M-worst}
  M    = 
  \left(\sum_x \frac{\nu_0(x)^2}{\mu(x)}\right)-1 
  \leq \frac{1}{\mu^*}\sum_x \nu_0(x)^2
  \leq  \frac{1}{\mu^*}.
\end{equation}
In contrast, when $M=O(1)$ (e.g., if $\nu_0(x)\leq 2\mu(x)$ for all $x\in\Omega$) then we refer to $\nu_0$ as a {\em warm start}.
We will prove the following bound on the distance to stationarity of the Glauber dynamics after $t$ steps as a function of the initial distance~$M$:
\begin{equation}
\label{eq:warm-start}
 4\left\| \nu_t - \mu\right\|_{\mathrm{TV}}^2 
 \leq D_{\chi^2}(\nu_t \,\Vert\, \mu)
 \leq M\exp(-2\gamma t).
\end{equation}
The above bound holds for the Glauber dynamics, and more generally holds for any reversible chain whose transition matrix is positive semidefinite (consequently, the absolute spectral gap is the same as the spectral gap);
we will prove it below.
Plugging in the worst case bound for the initial distance $M$ from \cref{M-worst} into \cref{eq:warm-start} yields the upper bound on the mixing time stated earlier in \cref{eq:mix-gap}.

For general Markov chains (not necessarily reversible), defined by a transition matrix $P$, recall $\gamma$ is defined in \cref{defn:Poincare}.  
Note that for non-reversible chains the eigenvalues of the transition matrix are not necessarily real.
For general chains (not necessarily reversible) a similar result as in \cref{eq:warm-start} holds with $\gamma/2$ in place of $2\gamma$ for the lazy version $P_{zz} := \frac{1}{2}(P+I)$ (which transitions as in the original chain $P$ with probability $1/2$ and stays in the same state with probability $1/2$), see \cite{Mihail} and \cite[Theorem 5.6]{Jerrum:notes}.  

Notice that if we have a warm start for the initial distance then  
the mixing time is on the order of the inverse spectral gap.  Consequently, in many approximate counting algorithms, such as the simulated annealing algorithm of \cite{SVV:annealing,Huber,Kolmogorov,HK}, the running time avoids the extra $O(\log(1/\mu^*))$ factor.

The main task is to prove \cref{eq:warm-start}; then \cref{eq:mix-gap} follows as a corollary as noted above by plugging in the worst-case bound on $M$ from \cref{M-worst}.  We refer the interested reader to \cite[Chapter 5]{Jerrum:notes} for an alternative (and somewhat simpler) proof of \cref{eq:warm-start} for the lazy chain $P_{zz}$.  The below proof is stated for the Glauber dynamics but holds for any PSD transition matrix.

\begin{proof}[Proof of \cref{eq:warm-start}]
We first utilize the following key fact for the decay of variance/$\chi^2$-divergence of the Glauber dynamics chain with transition matrix $P_{\textsc{gd}}$:
\begin{align}
    & \Var_\mu(P_{\textsc{gd}}f)\leq (1-\gamma)^2 \Var_\mu(f), \quad \forall f: \Omega \to \R; \nonumber \\
    \quad\iff\quad
    & D_{\chi^2}(\nu P_{\textsc{gd}} \,\Vert\, \mu) \le (1-\gamma)^2 D_{\chi^2}(\nu \,\Vert\, \mu), \quad \text{$\forall$ distribution $\nu$ on $\Omega$}. \label{eqn:decay-variance}
\end{align}
We will establish \cref{eqn:decay-variance} later in the proof.
Note, for the analog of \cref{eqn:decay-variance} for the lazy version of general chains (including non-reversible chains), 
see \cite{Mihail} and, e.g., \cite[Corollary 5.7]{Jerrum:notes}.

To complete the proof of \cref{eq:warm-start}, assuming \cref{eqn:decay-variance} holds, we will begin by showing the first inequality in \cref{eq:warm-start}.  Namely, by the Cauchy--Schwarz inequality
the $L^1$-distance is upper bounded by the $L^2$-distance (see also \cite[(7.1)]{MontenegroTetali}):
\begin{align*}
    2 \left\| \nu_t - \mu\right\|_{\mathrm{TV}}
    = \left\| f_t - 1\right\|_{1,\mu}
    = \mu\left( |f_t-1| \right)
    \le \sqrt{ \mu\left( (f_t-1)^2 \right) }
    = \left\| f_t - 1\right\|_{2,\mu}
    = \sqrt{ D_{\chi^2}(\nu_t \,\Vert\, \mu) }.
\end{align*}
For the second inequality in \cref{eq:warm-start}, we deduce from \cref{eqn:decay-variance} that
\begin{equation*}
D_{\chi^2}(\nu_t \,\Vert\, \mu)
= D_{\chi^2}(\nu_0 P_{\textsc{gd}}^t \,\Vert\, \mu)
\leq (1-\gamma)^{2t} D_{\chi^2}(\nu_0 \,\Vert\, \mu)
\leq M\exp(-2\gamma t),
\end{equation*}
where $M$ is the $\chi^2$-divergence of the initial distribution $\nu_0$ to the stationary distribution $\mu$, see \cref{initial-dist}.

It remains to prove \cref{eqn:decay-variance}.  We claim the following holds:
\begin{equation}
    \label{ineq:var-dirichlet2}
      \Var_\mu(P_{\textsc{gd}}f) = \Var_\mu(f) - \Dirichlet_{P_{\textsc{gd}}^2}(f).
\end{equation}
Without loss of generality, we assume $\mu(f) = 0$.
We observe that
\begin{align*}
    \Var_\mu(f) = \norm{f}_{2,\mu}^2 = \langle f, f \rangle_\mu
\end{align*}
and
\begin{align*}
    \Var_\mu(P_{\textsc{gd}}f) 
    = \norm{P_{\textsc{gd}}f}_{2,\mu}^2
    = \langle P_{\textsc{gd}}f, P_{\textsc{gd}}f \rangle_\mu
    = \langle f, P_{\textsc{gd}}^2f \rangle_\mu,
\end{align*}
where the last equality follows from the reversibility of $P_{\textsc{gd}}$.
Therefore, we deduce that
\begin{align*}
    \Var_\mu(f) - \Var_\mu(P_{\textsc{gd}}f) = \langle f, (I-P_{\textsc{gd}}^2)f \rangle_\mu = \Dirichlet_{P_{\textsc{gd}}^2}(f).
\end{align*}
Since the Glauber dynamics $P_{\textsc{gd}}$ is positive semidefinite and the spectral gap of $P_{\textsc{gd}}$ is $\gamma$, we deduce that the spectral gap of $P^2_{\textsc{gd}}$ is $1-(1-\gamma)^2$. Hence, we obtain the following Poincar\'{e} inequality for $P^2_{\textsc{gd}}$:
\[
(1-(1-\gamma)^2) \Var_\mu(f)\leq \Dirichlet_{P_{\textsc{gd}}^2}(f).
\]
Plugging this into \cref{ineq:var-dirichlet2}, we obtain:
\begin{align*}
    \Var_\mu(P_{\textsc{gd}}f)\leq (1-\gamma)^2 \Var_\mu(f),
\end{align*}
which completes the proof of \cref{eqn:decay-variance}, and hence concludes the proof of \cref{eq:warm-start}.
\end{proof}

We showed in \cref{eq:mix-gap} how a lower bound on the spectral gap of the Glauber dynamics (which is equivalent to an upper bound on the relaxation time) implies an upper bound on the mixing time with an extra $O(\log(1/\mu^*))$ factor.  In \cref{rem:connection-mixing-to-relax} in \cref{sub:coupling-SI} we discuss the reverse connection: namely, an optimal upper bound on the mixing time of the Glauber dynamics of the form $\Tmix(\eps)=O(n\log(n/\eps))$ for all $\eps>0$ implies an optimal upper bound on the relaxation time of $\Trelax=O(n)$.

\section{Approximate Tensorization of Variance}
\label{sub:app-tensorization}

A useful tool for bounding the spectral gap of the Glauber dynamics is the following property of the distribution known as approximate tensorization of variance.  

\begin{definition}
\label{defn:approx-tensorization-var}
    Let $\mu$ be a probability distribution on $\Omega\subseteq\{0,1\}^n$.
    For a constant $C\geq 1$, we say that
    $\mu$ satisfies \emph{$C$-approximate tensorization of variance} if the following holds for all $f:\Omega\rightarrow\R$,
    \begin{equation}
    \label{eq:AT}
        \Var_\mu(f) \leq C\sum_{i=1}^n \Exp_{\mu}[\Var_i(f)].
    \end{equation}
\end{definition}

Let us decipher the right-hand side of \cref{eq:AT} before proceeding: for the term $\Exp_{\mu}[\Var_i(f)]$, we first sample a configuration $\tau$ on $V\setminus\{i\}$ where $\tau$ is drawn from the projection of $\mu$ on $V\setminus\{i\}$, then we consider the variance of the function $f$ with respect to the marginal distribution of $\mu$ on $i$ conditional on $\tau$.  

It is important to stress that approximate tensorization of variance is a property of the distribution and does not explicitly involve the Glauber dynamics, but as we will see shortly it is equivalent to bounding the spectral gap.  
When $\mu$ is a product distribution, i.e., $\mu=\otimes_{i=1}^n \mu(i)$, then \cref{eq:AT} holds with $C=1$, which is the optimal constant (it is necessary that $C\ge 1$, as can be seen by taking functions depending only on a single coordinate);
this is known as a tensorization inequality and holds beyond variance, for example, for entropy (Shearer's inequality), e.g., see \cite[Proposition 3.1]{chafai2004entropies} and \cite[Section 4.2]{Caputo-notes}.

We now give a more formal description of the right-hand side of \cref{eq:AT}.
For a configuration $\tau\in\{0,1\}^{V\setminus\{i\}}$ on all coordinates except $i$, let $\mu(\tau)$ denote the marginal probability of $\tau$ in the distribution~$\mu$, i.e., $\mu(\tau) = \sum_{\sigma\in\Omega: \sigma(V\setminus\{i\})=\tau} \mu(\sigma)$.  Conversely, let $\mu_\tau$ denote the conditional distribution $\mu$ conditional on $\tau$, and thus for $c\in\{0,1\}$, 
\[ \mu_\tau(\eta(i)=c) = \Pr_{\eta \sim \mu_\tau}(\eta(i)=c) = \Pr_{\eta \sim \mu}(\eta(i)=c \mid \eta(V\setminus\{i\})=\tau).
\]
Then we have the following:
\begin{equation}
\label{eq:AT-expanded}
\Exp_{\mu}[\Var_i(f)]
= 
\sum_{\tau\in\{0,1\}^{V\setminus\{i\}}}\mu(\tau)\mu_\tau(\eta(i)=0)    \mu_\tau(\eta(i)=1)\left(f(\tau_{0})-f(\tau_{1})\right)^2,
\end{equation}
where, for $c\in\{0,1\}$, $\tau_{c}$ is the extension of $\tau$ to a vector in $\{0,1\}^n$ where for $j \in [n]$:
\[
\tau_{c}(j) := \begin{cases} c, & \mbox{if } j=i; \\
\tau(j), & \mbox{if } j\neq i.
\end{cases}
\]

For a configuration $\sigma\in\Omega$, for $i \in [n]$, let $\sigma^i$ denote the configuration which flips the value at the $i$'th coordinate and keeps all other coordinates the same.  Let $P_{\textsc{gd}}$ denote the transition matrix for the Glauber dynamics.  Using the observation that $\mu_{\sigma(V\setminus\{i\})}(\eta(i)=1-\sigma(i)) = nP_{\textsc{gd}}(\sigma,\sigma^i)$ and recalling the definition of the Dirichlet form (see \cref{defn:Dirichlet}), we can rewrite \cref{eq:AT-expanded} as follows:
\begin{equation}
\label{eq:AT-expanded-further}
\sum_{i=1}^n\Exp_{\mu}[\Var_i(f)]
 = 
\frac{n}{2}\sum_{i=1}^n \sum_{\sigma\in\Omega}\mu(\sigma)P_{\textsc{gd}}(\sigma,\sigma^i)\left(f(\sigma)-f(\sigma^i)\right)^2
 = 
  n\Dirichlet_{P_{\textsc{gd}}}(f).
\end{equation}

Therefore, 
     \[  \Var_\mu(f) \leq C\sum_{i=1}^n \Exp_{\mu}[\Var_i(f)] \quad\iff\quad  
     \frac{1}{Cn} \cdot \Var_\mu(f) \le \Dirichlet_{P_{\textsc{gd}}}(f) ,
     \]
     and hence $C$-approximate tensorization of variance is equivalent to the spectral gap satisfying $\gamma\geq 1/(Cn)$.

\begin{corollary}
    \label{lem:AT-gap}
    A distribution $\mu$ on $\{0,1\}^n$ satisfies $C$-approximate tensorization of variance iff the Glauber dynamics for $\mu$ has spectral gap $\gamma\geq 1/(Cn)$ (or equivalently, the Glauber dynamics has relaxation time $\Trelax\leq Cn$).
\end{corollary}

The analog of approximate tensorization of variance with entropy in place of variance is presented in \cref{sec:optimal}, see \cref{def:ent-tensorization}.  Whereas approximate tensorization of variance implies optimal bounds of $O(n)$ time on the relaxation time, the stronger condition of approximate tensorization of entropy implies optimal bounds of $O(n\log{n})$ time on the mixing time.  We refer the interested reader to \cref{sec:optimal} for further details.

\section{Approximate Subadditivity of Variance}

Here we present a functional characterization of spectral independence, utilizing the variance functional.

\begin{definition}
We say that a measure $\mu$ on $\{0,1\}^n$ satisfies \emph{approximate subadditivity of variance} with constant $C$ if for every $f: \Omega \to \R$
we have
\begin{align}\label{eq:subadditivity}
\sum_{i=1}^n \Var_{\mu} (\E_{V \setminus \{i\}}[f]) \leq C\cdot \Var_\mu(f)
\end{align}
\end{definition}

We note that, similar to \cref{eq:AT}, the right-hand side represents the variance of $f(X)$ where $X$ is generated from $\mu$, and the $i$'th term of the left-hand side represents the variance of $\E[f(X) | X_i]$ for $X$ sampled from $\mu$; specifically, 
\begin{align*}
    \Var_\mu(f) = \Var_{X \sim \mu}[f(X)] 
    \quad\text{and}\quad
    \Var_{\mu} (\E_{V \setminus \{i\}}[f]) = \Var_{\mu} [\E[f | X_i]].
\end{align*}

\begin{lemma}
\label{lem:SI-subadditivity}
A measure $\mu$ satisfies $\uplambda_{\max}(\Psi_\mu) \le C$
if and only if $\mu$ satisfies approximate subadditivity of variance with constant $C$.
\end{lemma}

\begin{proof}
Without loss of generality, we can assume $\E_\mu[f]=0$. Let $p_i=\Pr(X_i = 1) \in (0,1)$ for $i\in [n]$. Let $a_i = \E[f | X_i= 1]$ and note that  $\E[f | X_i=0]= - a_i p_i / (1-p_i)$, since $\E_\mu[f]=0$. In terms of $p_i$ and $a_i$'s the LHS of~\eqref{eq:subadditivity} is 
\begin{equation}\label{eq:LHS}
\sum_{i=1}^n a_i^2 \frac{p_i}{1-p_i}.
\end{equation}
We only need to verify~\eqref{eq:subadditivity} for $f$ that minimizes $\Var_\mu[f]$ subject to  $\E[f | X_i=1]=a_i$ for $i\in [n]$ and $\E_\mu[f]=0$. From the method of Lagrange multipliers
it follows that such $f$ is linear, that is, for some $c_0,c_1,\dots,c_n$ we have
$$
f = c_0 + \sum_{i=1}^n c_i X_i.
$$
Let $c=(c_1,\dots,c_n)$. In terms of $c_i$'s the main term on the RHS of~\eqref{eq:subadditivity} is
$$\Var_\mu[f] = c^T \Cov_\mu c.$$
The $c_i$'s, $a_i$'s, and $p_i$'s satisfy the following equation (for $i\in [n]$):
\begin{equation}\label{eq:ac}
(\Cov_\mu c)_i = \E[f X_i] - \E[f] \E[X_i] = \Pr[X_i=1]\E[f|X_i=1] - 0 = p_i a_i.
\end{equation}
Let $D = \diag(\Var(X_i))$ be the variance matrix.  We have $D_{ii}=p_i(1-p_i)$ for $i\in [n]$ and recall
that the bound $\uplambda_{\max}(\Psi_\mu)\leq C$
is equivalent to $\Cov_\mu\preceq C D$, see \cref{cov-SI-connection}. 
Thus we can express the LHS of~\eqref{eq:subadditivity} as $c^T \Cov_\mu D^{-1} \Cov_\mu c$ and~\eqref{eq:subadditivity}
holding for every $f$ is equivalent to the following
\begin{equation}\label{eq:laststep}
\Cov_\mu D^{-1} \Cov_\mu \preceq C\, \Cov_\mu .
\end{equation}
Let
\[
M=D^{-1/2}\Cov_\mu D^{-1/2}.
\]
Note that $M$ is positive semidefinite (since $\Cov_\mu$ is positive semidefinite), and~\eqref{eq:laststep}
is equivalent, after congruence by $D^{-1/2}$, to
\[
M^2\preceq CM.
\]
Since $M$ is positive semidefinite, this holds if and only if
every eigenvalue of $M$ is at most $C$.  Since $M$ is similar to
$\Psi_\mu$, this is equivalent to
\[
\uplambda_{\max}(\Psi_\mu)\leq C.
\]
\end{proof}

\chapter{Rapid Mixing from Spectral Independence}
\label{sec:rapid}

The main result of this chapter is \cref{thm:SI-mixing}, which states that if $\mu$ is $\eta$-spectrally independent and satisfies $(A,\eta)$-small-conditional-gap for constants~$\eta,A$, then the Glauber dynamics has polynomial mixing time.  We restate the theorem for convenience.

\ALOpolymixing*

\section{Pinnings}
\label{sub:pinnings}
Recall that a pinning is a fixed assignment of spins to a subset of vertices.  For $S\subset V$ and a pinning $\tau:S\rightarrow\{0,1\}$, denote the assignments consistent with $\tau$ as $\Omega_\tau:=\{\sigma\in\Omega: \sigma(S)=\tau(S)\}$. 

We will partition the pinnings based on the number of  pinned vertices.
For $0\leq k\leq n$, denote the set of valid pinnings of size $k$ by 
\[ \Pinning_k:=\{\tau\in\{0,1\}^S: \mbox{ for some }S\subseteq V, \mbox{ such that } |S|=k, \Omega_\tau\neq\emptyset\}.
\]

Recall, 
the set of valid configurations is denoted by
\[
\Omega=\{\sigma\in\{0,1\}^n:\mu(\sigma)>0\}.
\]
Observe that $\Omega=\Pinning_n$, that is the set of valid configurations is the same as the set of pinnings on all $n$ vertices (this follows from the observation that for $\tau\in\Omega$, $\Omega_\tau=\{\tau\}\neq\emptyset$ and
hence $\tau\in\Pinning_n$).

Finally, denote the collection of all (valid) pinnings by:
\[ \Pinning = \bigcup_{k=0}^n \Pinning_k.\]
 
For a pinning $\tau\in\Pinning$, let $\mu_{\tau}$ denote the conditional Gibbs distribution (that is, the distribution $\mu$ conditional on the fixed assignment $\tau$ on $S$).

\section{Rapid Mixing Theorem}
In the following sections, we will bound the mixing time of the Glauber dynamics by considering the spectral gap.
We will first prove that if the Gibbs distribution $\mu$ is $\eta$-spectrally independent (note, \cref{defn:SI} requires it for all pinnings $\tau\in\Pinning$) and it satisfies $(A,\eta)$-small-conditional-gap, then the spectral gap satisfies:
\begin{equation}
\label{eqn:gap-SI}
\gamma(\Glauber) \geq \frac{C(\eta,A)}{n^{1+\eta}},
\end{equation}
for a constant $C(\eta,A)>0$.
By \cref{eq:relax-gap} this implies relaxation time $O(n^{1+\eta})$
and mixing time $O(n^{1+\eta}\log(1/\mu^*))$, and hence 
this proves \cref{thm:SI-mixing}. 

\section{Local to Global: Random Walk Theorem of Alev-Lau}
\label{sec:local-walks}

The key technical tool in the proof of \cref{eqn:gap-SI} is the
following local-to-global theorem of Alev and Lau~\cite{AL20},
which improved upon Kaufman and Oppenheim~\cite{KO20}.
Our input graph $G=(V,E)$ has $n=|V|$ vertices, and each vertex
has two possible spin assignments, $0$ or $1$.  The state space
of the local walk consists of the at most $2n$ vertex--spin pairs
$(i,s_i)\in V\times\{0,1\}$ that have positive marginal
probability under the Gibbs distribution $\mu$. The transition
probabilities are defined in terms of pairwise conditional
probabilities. The stationary distribution of the local walk is
the normalized marginal distribution; namely,
\[
\pi((i,s_i))
=
\frac{1}{n}\mu(\sigma(i)=s_i).
\]

Let $\mathcal L\subseteq V\times\{0,1\}$ denote the state space
of the local walk. The matrix $Q$ is a real-valued
$|\mathcal L|\times|\mathcal L|$ matrix, where
$|\mathcal L|\leq 2n$. For distinct vertices $i,j\in V$ and
states $(i,s_i),(j,s_j)\in\mathcal L$, let
\[
Q((i,s_i),(j,s_j))
=
\frac{1}{n-1}
\Pr_{\sigma\sim\mu}
[\sigma(j)=s_j\mid \sigma(i)=s_i].
\]
For states $(i,s),(i,s')\in\mathcal L$ associated with the same
vertex, set
\[
Q((i,s),(i,s'))=0.
\]

Moreover, we define the following conditioned local walk.  Let $0\leq k\leq n-2$, let $S\subset V$ with
$|S|=k$, and let $\tau:S\rightarrow\{0,1\}$ be a valid pinning.
Define the state space of the corresponding conditioned local
walk $Q_\tau$ by
\[
\mathcal L_\tau
=
\left\{
(i,s):i\in V\setminus S,\ s\in\{0,1\},
\ \mu_\tau(\sigma(i)=s)>0
\right\}.
\]
The matrix $Q_\tau$ is a real-valued
$|\mathcal L_\tau|\times|\mathcal L_\tau|$ matrix. For distinct
vertices $i,j\in V\setminus S$ and states
$(i,s_i),(j,s_j)\in\mathcal L_\tau$, let
\begin{eqnarray}
\nonumber
Q_{\tau}((i,s_i),(j,s_j))
&=&
\frac{1}{n-k-1}
\Pr_{\sigma\sim\mu}
\bigl(
\sigma(j)=s_j
\mid
\sigma(i)=s_i,\sigma(S)=\tau
\bigr)
\\
\label{def:Q}
&=&
\frac{1}{n-k-1}
\Pr_{\sigma\sim\mu_\tau}
\bigl(
\sigma(j)=s_j
\mid
\sigma(i)=s_i
\bigr).
\end{eqnarray}
For states $(i,s),(i,s')\in\mathcal L_\tau$ associated with the
same vertex, set
\[
Q_\tau((i,s),(i,s'))=0.
\]
The stationary distribution of $Q_\tau$ is
\[
\pi_\tau((i,s))
=
\frac{1}{n-k}\mu_\tau(\sigma(i)=s).
\]
Notice that $Q=Q_\emptyset$.

The Markov chain $Q$ is called the ``local walk'' because it captures local information of the Gibbs distribution, namely the marginal probabilities of (vertex,spin) pairs.  In contrast, the Glauber dynamics is referred to as a ``global walk'' because it captures the probability of a configuration on the entire graph. 

The upcoming Random Walk Theorem of Alev and Lau~\cite{AL20} is referred to as a local-to-global theorem because it 
relates the behavior of the
local chains $Q_\tau$ to the Glauber dynamics which is a global chain.
To do this we need to consider the local walk on each level.  Namely, for $k$ where $0\leq k\leq n-2$, for $\tau\in\Pinning_k$, consider the local walk $Q_\tau$.  Let $\gamma_k$ be the minimum spectral gap for a local walk $Q_\tau$ where $\tau\in\Pinning_k$.

We can now formally state the Random Walk Theorem of Alev and Lau~\cite{AL20} (which improves upon earlier work of~\cite{KO20}).
\begin{restatable}[Random Walk Theorem \cite{AL20}]{theorem}{RWtheorem}
\label{thm:RW}
\begin{equation}
    \label{eqn:RW-thm-simplified}
 \gamma(P_{\mathrm{Glauber}}) \geq \frac{1}{n}\prod_{k=0}^{n-2}\gamma_k,
\end{equation}
where $$\gamma_k=\min_{\tau\in\Pinning_k}\gamma(Q_{\tau})$$ is the spectral gap
for the local walk with a worst-case pinning of $k$ vertices.
\end{restatable}

We will present the proof of the Random Walk Theorem in \cref{sec:RW}.

\section{Local Walk Connection to Influence Matrix}
\label{sec:local-influence}

We begin by bounding $\uplambda_2$ of the local walk $Q_{\tau}$ in terms of the influence matrix $\Psi_{\tau}$ for the spectral independence technique.  In fact, we can relate the entire spectrum of $Q_{\tau}$ with the spectrum of $\Psi_{\tau}$ in the following manner.

Let us first extend the definition of the influence matrix in \cref{defn:inf-matrix} to a fixed pinning.  For a subset $S\subset V$ where $|S|=k$, for a pinning $\tau:S\rightarrow\{0,1\}$ where $\tau\in\Pinning$, recall (see equation~\eqref{eq:inf-pin}) the influence matrix $\Psi_\tau$ is defined as:
  \[
\Psi_\tau(i\rightarrow j)  := \mu\left[\sigma(j)=1 \mid \sigma(i)=1,\sigma(S)=\tau\right]
- 
\mu\left[\sigma(j)=1 \mid \sigma(i)=0,\sigma(S)=\tau\right].
\]

\begin{lemma}
For a pinning $\tau\in\Pinning_k$, where $0\leq k\leq n-2$,
\label{lem:QandPsi}
\begin{equation}
\label{eq:Inf-to-local}
\uplambda_2(Q_{\tau}) = \frac{1}{n-k-1}\left(\uplambda_{\max}(\Psi_{\tau}) - 1\right).
\end{equation}
Moreover, the spectrum of $Q_{\tau}$ is given as follows:
\begin{equation}\label{er3}
\mathsf{spec}(Q_{\tau}) = \mathsf{spec}\left(\frac{1}{n-k-1}\left(\Psi_{\tau} - I\right)\right)\cup\{1\}\cup\left\{n-k-1\mbox{ copies of } \frac{-1}{n-k-1}\right\}.
\end{equation}
\end{lemma}

As a consequence of the above lemma we have that 
$\gamma_k\geq 1-\eta/(n-k-1)$, see \cref{sec:poly-mixing-proof} for more details.

\begin{proof}[Proof of \cref{lem:QandPsi}]
We first prove the lemma for the case without a pinning (which is the case $k=0$ and $\tau=\emptyset$). Recall
that $\mu$ denotes the Gibbs distribution.

For each vertex $i\in V$, let $\mathcal{L}_i$ denote the states
of the local walk associated with $i$:
\[
\mathcal{L}_i
=
\{(i,s):\Pr_{\sigma\sim\mu}(\sigma(i)=s)>0 \mbox{ where } s\in\{0,1\}\}.
\]
There are no transitions within a vertex part. Moreover, from any
state $(i,s_i)$, the sum of the transition probabilities to the
states associated with a fixed vertex $j\neq i$ is
\[
\sum_{(j,s_j)\in\mathcal{L}_j}
Q((i,s_i),(j,s_j))
=
\frac{1}{n-1}.
\]
In other words, the walk first chooses a vertex $j\neq i$
uniformly at random and then chooses its spin according to the
conditional marginal distribution given $\sigma(i)=s_i$.

For each $i\in V$, let $\pi_i^T$ denote the row vector supported
on $\mathcal{L}_i$ defined by
\[
\pi_i((i,s))
=
\Pr_{\sigma\sim\mu}(\sigma(i)=s),
\]
and let
\[
\pi^T=\frac1n\sum_{i=1}^n\pi_i^T.
\]
For $i\neq j$ and $(j,s_j)\in\mathcal{L}_j$, the law of total
probability gives
\begin{align*}
(\pi_i^TQ)((j,s_j))
&=
\frac{1}{n-1}
\sum_{(i,s_i)\in\mathcal{L}_i}
\Pr(\sigma(i)=s_i)
\Pr(\sigma(j)=s_j\mid\sigma(i)=s_i)\\
&=
\frac{1}{n-1}\Pr(\sigma(j)=s_j).
\end{align*}
Since there are no transitions from the $i$th part to itself,
we obtain
\begin{equation}
\label{eq:pi-i-Q}
\pi_i^TQ
=
\frac{1}{n-1}\sum_{j\neq i}\pi_j^T.
\end{equation}
In particular, $\pi^TQ=\pi^T$, and hence $\pi$ is a stationary
distribution for $Q$, and thus $\pi$ is a left eigenvector with eigenvalue $1$.

Let $\mathcal{U}$ denote the set of row vectors whose restriction
to each vertex part $\mathcal{L}_i$ is a scalar multiple of
$\pi_i^T$; equivalently,
\[
\mathcal{U}
=
\left\{
\sum_{i=1}^n a_i\pi_i^T:
a_1,\dots,a_n\in\R
\right\}.
\]
By \eqref{eq:pi-i-Q}, the restriction of $Q$ to $\mathcal{U}$
is the simple random walk on the complete graph on $n$ vertices.
Consequently, its eigenvalues on $\mathcal{U}$ are $1$ and
$-1/(n-1)$, where the latter has multiplicity $n-1$. In
particular, for all $i\in V$,
\[
(\pi_i^T-\pi^T)Q
=
-\frac{1}{n-1}(\pi_i^T-\pi^T).
\]

Let $\mathbf{1}$ be the all-ones vector, and let $\mathbf{1}_j$ be the vector that has ones at $(j,0)$ and $(j,1)$ 
with the remaining entries set to~$0$.

In summary, $Q$ has a few trivial eigenvalues: $1$ is the trivial largest eigenvalue with right eigenvector~$\mathbf{1}$ and left eigenvector $\pi^T$; and as mentioned above, the other trivial eigenvalue is $-1/(n-1)$ of multiplicity $n-1$, with corresponding right eigenvectors $\frac{1}{n}\mathbf{1} - \mathbf{1}_i$ and left eigenvectors $\pi^T - \pi_i^T$, for $i \in [n]$ (note, they actually span an $(n-1)$-dimensional subspace).
We have shown that the spectrum of $Q$ contains
$1$ and $n-1$ copies of $-1/(n-1)$ (more generally, $(n-k-1)$ copies of $\frac{-1}{n-k-1}$).
It remains to 
show that it also contains $\mathsf{spec}\left(\frac{1}{n-k-1}\left(\Psi - I\right)\right)$.

We will ``zero-out'' these trivial eigenvalues $1$ and $-1/(n-1)$ in the following manner.
 Define
\[
R
=
\frac{n}{n-1}\mathbf{1}\pi^T
-
\frac{1}{n-1}
\sum_{i=1}^n\mathbf{1}_i\pi_i^T
\]
and
\[
M_\pi=Q-R.
\]

Note that $R$ has eigenvalue $1$ with column eigenvector~$\mathbf{1}$ and row eigenvector $\pi^T$ and eigenvalues $-1/(n-1)$ of multiplicity $n-1$, with corresponding column eigenvectors $\frac{1}{n}\mathbf{1} - \mathbf{1}_i$ and row eigenvectors $\pi^T - \pi_i^T$, for $i \in [n]$. The remaining eigenvalues of $R$ are zero since $R$ has rank $n$ (the column space of $R$ is generated
by $1_i$ for $i\in [n]$). Thus subtracting $R$ from $Q$ moves the eigenvalues $1,-1/(n-1),\dots,-1/(n-1)$ to zero and 
it does not affect the other eigenvalues.

Then we notice that $M_{\pi}$ has the following block structure (we assume that the indices are in the following order: $(1,0),\dots,(n,0),(1,1),\dots,(n,1)$ where $V=\{1,\dotsm,n\}$):
\begin{equation}\label{er1}
M_{\pi} = \begin{pmatrix} 
A_{\pi}  & -A_{\pi} \\ B_{\pi} & -B_{\pi}\end{pmatrix},
\end{equation}
where 
\begin{equation}\label{er2}
A_{\pi}-B_{\pi} = \frac{1}{n-1}(\Psi-I).
\end{equation}
To see~\eqref{er1} and~\eqref{er2} note that for $i\neq j$ we have
$$
M_\pi((i,s_i),(j,s_j)) = \frac{1}{n-1}\Big(
\Pr_{\sigma\sim\mu}[\sigma(j)=s_j\mid \sigma(i)=s_i]- \Pr_{\sigma\sim\mu}[\sigma(j)=s_j]
\Big),
$$
and 
$$
M_\pi((i,s_i),(i,s_i)) = M_\pi((i,s_i),(i,1-s_i)) = 0.
$$

Note that if $w^T$ is a left-eigenvector of $A_\pi-B_\pi$ then $(w^T\ -w^T)$ is a left-eigenvector of $M_\pi$ with the same eigenvalue. Moreover, vectors of the form $\bigl( \begin{smallmatrix}v \\ v\end{smallmatrix}\bigr)$ (a space of dimension $n$) are both: right-eigenvectors of $M_\pi$ with eigenvalue $0$, and are perpendicular to the left-eigenvectors of the form $(w^T\ -w^T)$. These vectors 
of the form $\bigl( \begin{smallmatrix}v \\ v\end{smallmatrix}\bigr)$ yield right-eigenvectors of $\frac{n}{n-1}\mathbf{1}\pi^T - \frac{1}{n-1}\sum_{i=1}^n\mathbf{1}_i\pi_i^T$, where $\mathbf{1}$ has eigenvalue $1$ and the subspace perpendicular to $\pi$ (a space of dimension $n-1$)
has eigenvectors with eigenvalue $-\frac{1}{n-1}$. 

We have identified eigenvalues $1$ ($1$ trivial copy) and $-1/(n-1)$ ($n-1$ trivial copies) of $Q$ and the corresponding left and right eigenvectors. We constructed a matrix $R$ that contains exactly this part of $Q$ (if we write $Q=\sum_i \lambda_i u_i v_i^T$ where $u_i$'s are the right and $v_i$'s are the left eigenvectors then $R$ is the
part of the sum ranging over the $n$ identified eigenvalues). 
The rest of the spectrum of $Q$ is then contained in the spectrum of $M_\pi = Q-R$ and we have shown that the spectrum
of $M_\pi$ corresponds to the spectrum of $\Psi$.
This implies~\cref{er3}, and \cref{eq:Inf-to-local} follows as a consequence of \cref{er3}.

This completes the proof of the lemma for the case without pinning.  The proof easily generalizes to an arbitrary pinning $\tau\in\Pinning_k$ with $(n-1)$ replaced by $(n-k-1)$.  
\end{proof}

One of the key tasks in the rest of the monograph is to use a
lower bound on the spectral gaps of the local walks to obtain a
lower bound on the spectral gap of the Glauber dynamics.  Here we
establish a comparison in the reverse direction: a lower bound on
the spectral gap of the conditioned Glauber dynamics implies a
lower bound on the spectral gap of the corresponding local walk.
This will be used in the proof of \cref{thm:SI-mixing} to convert
the small-conditional-gap assumption (see \cref{defn:small-subgraph-gap}) into a uniform lower bound
on the final local-walk gaps.

\begin{lemma}
\label{lem:local-gap-from-Glauber}
For a pinning $\tau\in\Pinning_k$, where $0\leq k\leq n-2$, let $\Glaubert$ denote the Glauber dynamics for
$\mu_\tau$, and let $Q_\tau$ denote the corresponding local walk.
Then,
\[
\gamma(Q_\tau)
\geq
\frac{n-k}{n-k-1}\gamma(\Glaubert).
\]
\end{lemma}

\begin{proof}
Fix a function $f:\Omega_\tau\rightarrow\R$.
Let $\tau:S\rightarrow\{0,1\}$, where $S\subset V$ and
$|S|=k$. Let $m=n-k$ be the number of unpinned vertices, and,
for notational convenience, relabel the unpinned vertices so that the unpinned vertices are
$V\setminus S=\{1,\dots,m\}$.

For $1\leq j\leq m$, let $X_{-j}$ denote the assignment of all
unpinned vertices except $j$. For distinct $1\leq i,j\leq m$,
the value of $X_i$ is determined by $X_{-j}$. Hence,
conditioning on $X_{-j}$ gives at least as much information as
conditioning on $X_i$, and therefore
\[
\Exp_{\mu_\tau}[\Var(f\mid X_i)]
\geq
\Exp_{\mu_\tau}[\Var(f\mid X_{-j})].
\]
Averaging over $j\neq i$ and summing over $i$, we obtain
\begin{align*}
\sum_{i=1}^m \Exp_{\mu_\tau}[\Var(f\mid X_i)]
&\geq
\frac{1}{m-1}
\sum_{i=1}^m\sum_{j\neq i}
\Exp_{\mu_\tau}[\Var(f\mid X_{-j})] \\
&=
\sum_{j=1}^m
\Exp_{\mu_\tau}[\Var(f\mid X_{-j})] \\
&=
\sum_{j=1}^m \Exp_{\mu_\tau}[\Var_j(f)] \\
&=
m\Dirichlet_{\Glaubert}(f) \\
&\geq
m\gamma(\Glaubert)\Var_{\mu_\tau}(f),
\end{align*}
where the penultimate equality is the conditional-variance
formulation of the Glauber Dirichlet form, and the last inequality
follows from the Poincar\'e inequality.

Consequently, by the law of total variance,
\begin{align*}
\sum_{i=1}^m
\Var_{\mu_\tau}\!\left(\Exp[f\mid X_i]\right)
&=
m\Var_{\mu_\tau}(f)
-
\sum_{i=1}^m
\Exp_{\mu_\tau}[\Var(f\mid X_i)]\\
&\leq
m\bigl(1-\gamma(\Glaubert)\bigr)
\Var_{\mu_\tau}(f).
\end{align*}
Frozen unpinned vertices contribute zero to the left-hand side,
so the same inequality holds after restricting the sum to the
free vertices, which are the vertices indexing $\Psi_\tau$.
Since this holds for every $f:\Omega_\tau\rightarrow\R$, the
measure $\mu_\tau$ satisfies approximate subadditivity of variance
with constant $m\bigl(1-\gamma(\Glaubert)\bigr)$, and hence by \cref{lem:SI-subadditivity},
\[
\uplambda_{\max}(\Psi_\tau)
\leq
m\bigl(1-\gamma(\Glaubert)\bigr).
\]
Finally, we apply \cref{lem:QandPsi} to bound the spectral gap of the local walk:
\begin{align*}
\gamma(Q_\tau)
&=
1-\frac{\uplambda_{\max}(\Psi_\tau)-1}{m-1}\\
&\geq
1-\frac{m(1-\gamma(\Glaubert))-1}{m-1}\\
&=
\frac{m}{m-1}\gamma(\Glaubert)\\
&=
\frac{n-k}{n-k-1}\gamma(\Glaubert).
\end{align*}
\end{proof}

\section{Proof of Rapid Mixing (Theorem \ref{thm:SI-mixing})}
\label{sec:poly-mixing-proof}
We can now utilize~\cref{lem:QandPsi} and \cref{thm:RW} to conclude $\poly(n)$ relaxation time for the Glauber dynamics and thereby prove \cref{thm:SI-mixing}.

\begin{proof}[Proof of \cref{thm:SI-mixing}]
Recall, the definition of spectral independence in~\cref{defn:SI}; it states that the maximum eigenvalue of the influence matrix $\Psi$  is at most $1+\eta$.  Hence, from~\cref{lem:QandPsi}, we get that 
$\uplambda_2(Q)\leq\eta/(n-1)$.  
Moreover, \cref{defn:SI} is for the worst-case pinning and hence we can apply \cref{lem:QandPsi} to any pinning.  Therefore, we have that 
\begin{equation}
\label{eqn:gammak}
    \gamma_k\geq 1-\eta/(n-k-1),
\end{equation} where $\gamma_k$ is defined below~\eqref{eqn:RW-thm-simplified}.
However, the above bound is not useful when $k\geq n-\eta-1$, but in this case we have the following bound.

When $k\geq n-\eta-1$ then by the assumption of $(A,\eta)$-small-conditional-gap, the spectral gap for the Glauber dynamics satisfies $\gamma(\Glaubert)\geq 1/A$ for any $\tau\in\Pinning_{k}$, and hence by \cref{lem:local-gap-from-Glauber}, the local walk satisfies the following:
\begin{equation} \label{eqn:gap-bigk}
\gamma_k=\min_{\tau\in\Pinning_{k}}\gamma(Q_\tau)\geq \frac{n-k}{A(n-k-1)}\geq \frac{\eta+1}{A\eta} \geq \frac{1}{A}.
\end{equation}

Now we can apply \eqref{eqn:RW-thm-simplified} in the Random Walk Theorem (\cref{thm:RW}),  
and we obtain the following for $n>\eta+1$ (the case of $n\leq \eta+1$ follows directly from the $(A,\eta)$-small-conditional-gap assumption):
\begin{align*}
 \gamma(P_{\mathrm{Glauber}}) &\geq \frac{1}{n}\prod_{k=0}^{n-2}\gamma_k
 \\
 &\geq A^{-\lceil\eta\rceil}\times\frac{1}{n}\prod_{k=0}^{\lceil n-\eta-2\rceil}\gamma_k
 & \mbox{by \cref{eqn:gap-bigk} }
 \\
 &\geq A^{-\lceil \eta\rceil}\times\frac{1}{n}\prod_{k=0}^{\lceil n-\eta-2\rceil} \left(1-\frac{\eta}{n-k-1}\right)
 & \mbox{by \cref{eqn:gammak} }
 \\
&\geq \frac{C'}{n}\prod_{k=0}^{\lfloor n-\eta-2\rfloor} \left(1-\frac{\eta}{n-k-1}\right) 
\\
&\geq
\frac{C'}{n}\exp\left(-\sum_{k=0}^{\lfloor n-\eta-2\rfloor}\frac{\eta}{n-k-1-\eta}\right)
& \mbox{since $1-x\geq\exp\left(-\frac{x}{1-x}\right)$} \\
&\geq
\frac{C'}{n}\exp\left(-\eta\sum_{i=1}^{\lfloor n-\eta-1\rfloor}\frac{1}{i}\right)
\\
&\geq \frac{C''}{n}\exp(-\eta(1+\ln(n-\eta))) 
\\
&\geq \frac{C'''}{n}(n-\eta)^{-\eta}
\\
&\geq  \frac{C'''}{n^{1+\eta}},
\end{align*}
where $C',C'',C'''$ are positive constants (which may depend on $\eta,A$).

We have shown that the spectral gap of the Glauber dynamics is
$\Omega_{\eta,A}(n^{-1-\eta})$. Applying
\cref{eq:relax-gap}, we obtain
\[
\Trelax=O_{\eta,A}(n^{1+\eta}).
\]
The stated mixing-time bound
\[
\Tmix
=
O_{\eta,A}\left(
n^{1+\eta}\log(1/\mu^*)
\right)
\]
then follows from \cref{eq:mix-gap}.
\end{proof}

\chapter{Random Walk Theorem: Proof of Local-to-Global}
\label{sec:RW}

In this chapter we will prove the Random Walk Theorem (\cref{thm:RW}), which we restate here for convenience. We will use random walks between adjacent levels of the associated simplicial complex which corresponds to walks between sets $\Pinning_k$ and $\Pinning_{k+1}$.  These will be denoted as up and down walks, and their composition as up-down and down-up walks.

Our proof of the Random Walk Theorem (\cref{thm:RW}) follows the approach in \cite{AL20} (which is based on \cite{KO20}); alternative proof approaches are presented in \cite{KM20,AR24}, see also \cite{AASV21,Liu23,CE25}.

\RWtheorem*

\section{Up and Down Walks}\label{sec:chains}
We begin with the definitions of the down-up and up-down chains which are at the heart of the proof.  We do not explicitly utilize simplicial complexes, though readers familiar with simplicial complexes will understand the natural connection.

We start by defining a probability distribution $\pi_k$ on the collection of partial assignments to a set of $k$ vertices where $0\leq k\leq n$.
Specifically, a partial assignment sampled from $\pi_k$ is obtained by first randomly picking a set $S$ of size $k$ and independently sampling $\eta \sim \mu$, and then outputting the corresponding pinning $\eta$ restricted to $S$.
Recall, $\tau\in\Pinning_k$ is an assignment of spins $\{0,1\}$ to a subset $S\subseteq V$ where $|S|=k$.  We define the distribution $\pi_k$ on $\Pinning_k$ as follows:
\[ \mbox{for }\tau\in\Pinning_k, \ \ \ \ \pi_k(\tau) = \frac{1}{ \binom{n}{k}}\mu(\tau),
\]
where 
\begin{equation}\label{mudef}
\mu(\tau) := \sum_{\substack{\eta\in\{0,1\}^n: \\ \tau(S)=\eta(S)}} \mu(\eta).
\end{equation}

For any $S\subset V$, observe that $\sum_{\tau:S\rightarrow\{0,1\}} \mu(\tau) = 1$, and hence $\sum_{\tau\in\Pinning_k}\pi_k(\tau)=1$.  Moreover, for $k=n$, $\Pinning_n=\Omega$ and  $\pi_n = \mu$.   (Note that $\pi_1$ is the same as the distribution $\pi$ defined in the proof of \cref{lem:QandPsi} in~\cref{sec:local-influence}.)

Consider $0\leq k\leq r\leq n$ and $\sigma\in\Pinning_k$ where $\sigma:S\rightarrow \{0,1\}$ for $S\subset V$ where $|S|=k$.
Observe that
\[
\sum_{\eta\in\Pinning_r: \sigma\subset\eta} \mu(\eta)
 = \binom{n-k}{r-k}\mu(\sigma).
\]
We can view the first step of the process of sampling from $\pi_k$ as follows. To obtain a random subset of $V$ of size $k$ we can first sample a random subset $S$ of $V$ of size $r$ and then pick a random subset of $S$ of size $k$. This intuition implies 
the following relationship between $\pi_k$ and $\pi_r$:
\begin{equation}
\label{matroid:first-step}
\pi_k(\sigma) = 
\frac{\mu(\sigma)}{{\binom{n}{k}}} = \frac{1}{\binom{n-k}{r-k}} \sum_{\eta\in\Pinning_r: \sigma\subset\eta} \frac{\mu(\eta)}{{\binom{n}{k}}} 
= \frac{1}{\binom{r}{k}} \sum_{\eta\in\Pinning_r: \sigma\subset\eta} \frac{\mu(\eta)}{\binom{n}{r}} = 
{\binom{r}{k}}^{-1}\sum_{\eta\in\Pinning_r: \sigma\subset\eta}
\pi_r(\eta).
\end{equation}

We generalize the above definition to a fixed pinning $\eta\in\Pinning$. For $0\leq\ell\leq n$, consider $\eta\in\Pinning_\ell$ where $\eta$ is on $S$ (i.e., for $S\subset V$, $\eta:S\rightarrow\{0,1\}$). 
For $1\leq j\leq n-\ell$, let $\Pinning_{\eta,j}$ be the set of assignments $\tau$ of spins $\{0,1\}$ to a subset $S'\subseteq V\setminus S$ where $|S'|=j$ and $\tau\cup\eta\in\Pinning_{j+\ell}$.

We define the probability distribution for $j$ levels above $\eta$ with respect to the conditional Gibbs distribution~$\mu_\eta$.  
Hence, for $j$ where $1\leq j\leq n-\ell$, define $\pi_{\eta,j}$ as follows:
\[ \mbox{for }\tau\in\Pinning_{\eta,j} \mbox{ where } \tau:S'\rightarrow\{0,1\} \mbox{ with } S'\subset V\setminus S, \ \ \ \ \pi_{\eta,j}(\tau) = \frac{1}{ \binom{n-\ell}{j} }\mu(\sigma(S')=\tau\mid\sigma(S)=\eta).
\]
We will only use the above definitions of $\Pinning_{\eta,j}$ and $\pi_{\eta,j}$ for the case $j=1$, in which case it is closely related to the local walk $Q_\eta$, see \cref{rem:local-downup} below.

For every $1\leq k\leq n$, we define the following random walks which are known as 
{\em Up Walk} and 
{\em Down Walk}.
\begin{description}
\item[Down Walk:] The down-walk is denoted by $\down{k}$.  From $\tau\in \Pinning_k$ we remove an element of $\tau$ chosen uniformly at random.  Note the elements of $\tau$ are (vertex,spin) pair.
Hence, for $(j,s_j)\in\tau$, 
\[ \down{k}(\tau,\tau\setminus(j,s_j))=1/k.
\]
\item[Up Walk:] The up-walk is denoted by $\up{k}$ and corresponds to the following (random) step from 
$\Pinning_k$ to $\Pinning_{k+1}$. Starting from $\tau \in \Pinning_k$ where $\tau$ is an assignment on $S\subset V, |S|=k$, then choose a random $j\not\in S$ and spin $s_j\in\{0,1\}$ where the probability of picking $(j,s_j)$ is proportional to $\pi_{k+1}(\tau\cup(j,s_j))$. Hence, 
\[ \up{k}(\tau,\tau\cup (j,s_j)) =  
\frac{\pi_{k+1}(\tau\cup(j,s_j))}{(k+1)\pi_k(\tau)} = \frac{\mu(\tau\cup(j,s_j))}{(n-k)\mu(\tau)} = \frac{1}{(n-k)}\mu(\sigma(j)=s_j\mid \sigma(S)=\tau),
\]
where the denominator in the first identity comes from the observation that for $S\subset V, |S|=k$:
\begin{multline*}
    \sum_{\substack{(j',s_{j'}): \\ j'\notin S,s_{j'}\in\{0,1\}}} \pi_{k+1}(\tau\cup(j',s_{j'})) = \frac{1}{\binom{n}{k+1}} \sum_{j'\notin S}\sum_{s_{j'}\in\{0,1\}} \mu(\tau\cup(j',s_{j'}))   \\
= \frac{1}{\binom{n}{k+1}} \sum_{j'\notin S} \mu(\tau) = \frac{n-k}{\binom{n}{k+1}} \mu(\tau) = (k+1)\pi_k(\tau). 
\end{multline*}
From the definition of the up and down chains, we have that the reversibility condition is satisfied:
\begin{align}\label{eq:down-up-adjoint}
    \pi_k(\tau)\up{k}(\tau,\tau\cup (j,s_j)) =  
\pi_{k+1}(\tau\cup(j,s_j))\down{k+1}(\tau\cup (j,s_j),\tau).
\end{align}

	\item[Up-Down Walk:] The up-down walk is denoted by $\updown{k}$ and corresponds to the following Markov chain from $\Pinning_k$ to $\Pinning_k$:
 	\[ \updown{k} = \up{k}\down{k+1}.\]
Consider $\sigma\in\Pinning_{k-1}$ where $\sigma:S\rightarrow\{0,1\}$ for $S\subset V$.  For $i,j\notin S$ and $s_i,s_j\in\{0,1\}$, the transition 
$\updown{k}(\sigma\cup(i,s_i),\sigma\cup(j,s_j))$ operates as follows.
Starting from $\tau=\sigma\cup (i,s_i) \in \Pinning_k$, in the up-step, we choose a random $(j,s_j)$ where the probability of picking $(j,s_j)$ is proportional to $\pi_{k+1}(\tau\cup(j,s_j))$.  Then, in the down-step, we remove $(i,s_i)$, which is chosen
uniformly at random from $\tau\cup\{(j,s_j)\}$.
More formally, we have the following for $i\neq j$:
\begin{equation}
\label{eqn:up-down}
\updown{k}(\sigma\cup(i,s_i),\sigma\cup(j,s_j))
=
\frac{\pi_{k+1}(\sigma\cup(i,s_i)\cup(j,s_j))}
{(k+1)^2\pi_k(\sigma\cup(i,s_i))}.
\end{equation}
	\item[Down-Up Walk:] The down-up walk is denoted by $\downup{k}$.  From $\tau\in\Pinning_k$, we first remove a uniformly random element $(i,s_i)$ from $\tau$ and then add an element $(j,s_j)$ with probability proportional to the weight of the resulting set $\pi_k(\tau\cup(j,s_j)\setminus(i,s_i))$,
	\[ \downup{k} = \down{k}\up{k-1}.\]
\end{description}

Observe that the stationary distribution of $\updown{k}$ and $\downup{k}$ is $\pi_k$.
The reversibility condition \cref{eq:down-up-adjoint} shows that the two operators $\down{k},\up{k-1}$ are adjoint to each other, which implies that both up-down walks and down-up walks are reversible and positive semidefinite (PSD); see \cite[Section 2]{CCGP25} for more discussion.

The above definitions can be considered as walks on levels of the simplicial complex defined by the Gibbs distribution $\mu$ on $\Omega\subset \{0,1\}^n$ where the ground set of the simplicial complex is the $2n$ pairs of (vertex,spin) assignments $(v,\sigma(v))\in V\times\{0,1\}$.

Let us illustrate the above definitions for the special case of independent sets (this is the hard-core model with $\lambda=1$); here the Gibbs distribution is uniformly distributed over all independent sets (of any size) of~$G$. This corresponds to a simplicial complex over subsets of size $n=|V|$ since for every independent set $I$ we have 
$\{(v,1)\,|\, v\in I\}\cup \{(v,0)\,|\, v\not\in I\}$
in $\Pinning_{n}$.  The Glauber dynamics which updates the spin at a randomly chosen vertex is equivalent to $\downup{n}$.

In general, we have the following connection between the Glauber dynamics and the down-up walk.

\begin{lemma}\label{rem:glauber} For a spin system, the down-up chain $\downup{n}$  is equivalent to the Glauber dynamics.  
\end{lemma}

The local walk $Q$ defined in \cref{sec:local-walks}
is closely related to the up-down chain, in particular, the local walk is the non-lazy version of the up-down chain $\updown{1}$.

\begin{lemma}\label{rem:local-downup} 
The transition matrices of the up-down chain $\updown{1}$ and the local walk satisfy the following:
\[
\updown{1} = (Q + I)/2.
\]
More generally, for any $1\leq k<n$, consider $\eta\in\Pinning_{k-1}$ where $\eta$ is an assignment on $S\subset V$.  For distinct states $a,b$:
\[
\updown{k}(\eta\cup\{a\},\eta\cup\{b\})
=
\frac{1}{k+1}Q_\eta(a,b).
\]
\end{lemma}

\begin{proof}
For simplicity of notation we first consider the case without a pinning, which is the first statement in \cref{rem:local-downup}. Observe that $\pi_1$ is the stationary distribution for both $\updown{1}$ and $Q$. The factor $1/2$ in the RHS of the first statement of \cref{rem:local-downup} comes from the step of the down-walk which drops an element chosen uniformly at random from $(i,s_i),(j,s_j)$, whereas the local walk $Q$ corresponds to the {\em non-lazy walk} which chooses $(j,s_j)$ to avoid the self-loop.
\begin{align*}
\updown{1}((i,s_i),(j,s_j)) & = 
\up{1}((i,s_i),(i,s_i)\cup(j,s_j))\down{2}
((i,s_i)\cup(j,s_j),(j,s_j))
\\ & =
\frac{\pi_{2}((i,s_i)\cup(j,s_j))}{2\pi_1((i,s_i))}\times\frac{1}{2} 
\\ & = \frac{1}{2}\frac{n}{\binom{n}{2}}\times\frac{\mu((i,s_i)\cup(j,s_j))}{2\mu(i,s_i)}
\\ & = \frac{1}{2}\frac{1}{n-1}\times\frac{\mu((i,s_i)\cup(j,s_j))}{\mu(i,s_i)} 
\\ & = \frac{1}{2}Q((i,s_i),(j,s_j)),
\end{align*}
this proves the first statement in \cref{rem:local-downup}.

We now consider $1\leq k<n$ and a pinning $\eta\in\Pinning_{k-1}$. 
For $i\neq j$,
\begin{align}
\nonumber
\updown{k}(\eta\cup(i,s_i),\eta\cup(j,s_j)) & = 
\up{k}(\eta\cup(i,s_i),\eta\cup(i,s_i)\cup(j,s_j))
\down{k+1}(\eta\cup(i,s_i)\cup(j,s_j),\eta\cup(j,s_j))
\\ 
\nonumber
& =
\frac{\pi_{k+1}(\eta\cup(i,s_i)\cup(j,s_j))}{(k+1)\pi_k(\eta\cup(i,s_i))}\times\frac{1}{(k+1)} 
\\ 
\nonumber
& = \frac{\mu(\eta\cup(i,s_i)\cup(j,s_j))}{(k+1)(n-k)\mu(\eta\cup(i,s_i))} 
\\ & = \frac{1}{(k+1)}Q_\eta((i,s_i),(j,s_j)),
\label{step:rem:local-downup}
\end{align}
which proves the second statement in \cref{rem:local-downup}.
\end{proof}

\section{Spectrum of Up-Down Chains}
We begin by pointing out that the spectral gaps for the up-down chain and down-up chain are the same.
\begin{lemma}For $1\leq k\leq n$,
    \label{lem:updown-downup}
    \[
    \gamma(\updown{k-1}) = \gamma(\downup{k}).
    \]
    Moreover, they have the same non-zero spectrum of eigenvalues:
    \[
    \mathsf{spec}_{\neq 0}(\updown{k-1}) = \mathsf{spec}_{\neq 0}(\downup{k}).
    \]
\end{lemma}

\begin{proof}
Note that $\updown{k-1} = P^\uparrow_{k-1}P^\downarrow_{k}$ and $\downup{k} = P^\downarrow_{k}P^\uparrow_{k-1}$.
    Hence, the lemma follows from the general linear algebra fact that for an $n\times m$ matrix $A$ and $m\times n$ matrix $B$ the non-zero spectrum of $AB$ is the same as the non-zero spectrum of $BA$.
    For a more detailed proof see \cite[Corollary 3.3.2]{Mousa-thesis}.
\end{proof}

\section{Proof Setup}
We can now state a more general version of the Random Walk Theorem (\cref{thm:RW}) which bounds the spectral gap of the down-up chain at any level.
We will prove, for all $2\leq k\leq n$:
\begin{equation}
    \label{eqn:RW-thm}
 \gamma(\downup{k}) \geq \frac{1}{k}\prod_{i=0}^{k-2}\gamma_i,
\end{equation}
where $$\gamma_i=\min_{\tau\in\Pinning_i}\gamma(Q_{\tau})$$ is the spectral gap
for the local walk with a worst-case pinning of $i$ vertices.
Note, the claim trivially holds for $k=1$ as $\gamma(\downup{1})=1$.
Since the Glauber dynamics is equivalent to $\downup{n}$, \cref{eqn:RW-thm-simplified} is the special case of \cref{eqn:RW-thm} with $k=n$.

\section{Key Technical Lemma}
The following lemma will be the key tool in the inductive proof of~\cref{thm:RW}.  
It relates the Dirichlet form for the up-down walk $\updown{k}$ at level $k$ with the down-up walk $\downup{k}$ at level $k$.  Recall that the spectral gap of the $\updown{k}$ is equal to the spectral gap of $\downup{k+1}$ (see~\cref{lem:updown-downup}).  Hence, the lemma is in essence an inductive statement as it relates the chain $\updown{k}$ to $\downup{k}$, rather than $\downup{k+1}$.  In this manner the lemma will immediately yield an inductive proof of the Random Walk Theorem.

\begin{lemma}
\label{lem:technical-RW}
For all $1\leq k< n$, and any $f: \Pinning_k \to\mathbb{R}$, the following holds:
\[
\Dirichlet_{\updown{k}}(f) \geq \frac{k}{k+1}\gamma_{k-1}\Dirichlet_{\downup{k}}(f).
\]
    \end{lemma}

    At an intuitive level, the reader can understand the factor $k/(k+1)$ because in the LHS the down walk is from level $k+1$ to $k$ and hence a transition occurs with probability $1/(k+1)$, whereas in the RHS the down walk is from level $k$ to $k-1$ and hence occurs with probability $1/k$.

We will utilize the following two claims to prove \cref{lem:technical-RW}.  The first claim captures the local-to-global nature of the process.  In particular, for the chain $\updown{k}$ we can express its Dirichlet form in terms of local walks for pinnings at level $k-1$.  Recall that the Dirichlet form for a chain $P$ is the local variation of a functional $f$ over pairs of states connected by a transition of $P$.  

For the chain $\updown{k}$, consider a pair of distinct states $\sigma,\tau$ 
where $\updown{k}(\sigma,\tau)>0$.  Notice that it must be the case that $|\sigma\oplus\tau|=2$ where $\oplus$ is the symmetric difference, i.e., $\sigma\oplus\tau=(\sigma\setminus\tau)\cup(\tau\setminus\sigma)$. This is because in the up-step we can add one element $(i,s_i)$ to $\sigma$ and then in the down-step we can remove one element $(j,s_j)$ from $\sigma\cup(i,s_i)$.  Hence, $\sigma\cap\tau=\eta$ where $\eta\in\Pinning_{k-1}$.  Then the transition 
$\updown{k}(\sigma,\tau)$ is closely related to the transition $Q_{\eta}((i,s_i),(j,s_j))$, which  yields the following identity.

\begin{claim} For all $1\leq k< n$, and all $f: \Pinning_k \to\mathbb{R}$,
    \label{claim:AAA}
   \[    \Dirichlet_{\updown{k}}(f)  =
    \frac{k}{k+1}\sum_{\eta\in\Pinning_{k-1}} \pi_{k-1}(\eta) \Dirichlet_{Q_{\eta}}(f_\eta), \]
where for $\eta:S\rightarrow\{0,1\}$ and $(j,s_j)\in (V\setminus S)\times\{0,1\}$, then $f_\eta((j,s_j)) = f(\eta\cup(j,s_j))$.
\end{claim}

The second claim works in the reverse manner and relates the variance on level $1$ conditional on $\eta\in\Pinning_{k-1}$ to the global chain $\downup{k}$.

\begin{claim}For all $0< k\leq n$, and all $f: \Pinning_k\to\mathbb{R}$,
\label{claim:DDD}
\[
   \Dirichlet_{\downup{k}}(f) = \sum_{\eta\in\Pinning_{k-1}} \pi_{k-1}(\eta) \Var_{\pi_{\eta,1}}(f_\eta).
   \]
\end{claim}

Using these two claims it is straightforward to prove the key technical lemma.

\begin{proof}[Proof of \cref{lem:technical-RW}]
For $1\leq k<n$, and $\eta\in\Pinning_{k-1}$ we know that     $\gamma(Q_\eta) \geq \gamma_{k-1}$,
    and hence since $\pi_{\eta,1}$ is the stationary distribution for $Q_\eta$ then for any 
    $g:(V\setminus S)\times\{0,1\}\rightarrow\mathbb{R}$
    we have that 
    \begin{equation}
        \label{eq:DDD}
    \Dirichlet_{Q_{\eta}}(g)\geq \gamma_{k-1}\Var_{\pi_{\eta,1}}(g).
    \end{equation}

We can now prove the lemma:
\begin{align*}
\nonumber
    \Dirichlet_{\updown{k}}(f)
&=
    \frac{k}{k+1}\sum_{\eta\in\Pinning_{k-1}} \pi_{k-1}(\eta) \Dirichlet_{Q_{\eta}}(f_\eta)
   & \mbox{by \cref{claim:AAA}} \\
    &\geq 
    \frac{k}{k+1}\gamma_{k-1} \sum_{\eta\in\Pinning_{k-1}} \pi_{k-1}(\eta) \Var_{\pi_{\eta,1}}(f_\eta)
       & \mbox{by \cref{eq:DDD} with $g=f_\eta$}
       \\
      &=
   \frac{k}{k+1}\gamma_{k-1} \Dirichlet_{\downup{k}}(f),
      & \mbox{by \cref{claim:DDD}} 
    \end{align*}
    as claimed.
\end{proof}

\begin{proof}[Proof of \cref{claim:AAA}] 
As a warm-up let us consider the simpler case where $k=1$.  Recall that $\pi_1$ is the stationary distribution for $\updown{1}$ and $Q$.  In \cref{rem:local-downup} we showed that for $(i,s_i),(j,s_j)\in V\times\{0,1\}$ where $i\neq j$, \[ \updown{1}((i,s_i),(j,s_j))=\frac{1}{2}Q((i,s_i),(j,s_j)).
    \]
    For any $f:\Pinning_1\to\mathbb{R}$ (which is $f:V\times\{0,1\}\rightarrow\mathbb{R}$), 
    we have the following:
       \begin{align*}
           \Dirichlet_{\updown{1}}(f)  & =
           \frac{1}{2}\sum_{(i,s_i),(j,s_j)} \pi_1((i,s_i))\updown{1}((i,s_i),(j,s_j))(f((i,s_i))- f((j,s_j)))^2 \\
           & = \frac{1}{4}\sum_{(i,s_i),(j,s_j)} \pi_1((i,s_i))Q((i,s_i),(j,s_j))(f((i,s_i))- f((j,s_j)))^2 \\
           & = \frac{1}{2}\Dirichlet_{Q}(f),
           \end{align*}
           where the summation is over $(i,s_i),(j,s_j)\in V\times\{0,1\}$. 

    Since $\Pinning_{0}=\{\emptyset\}$, this proves the claim for the case $k=1$.

Before proving the general claim, consider the following observation.
For $\eta\cup a\in\Pinning_k$, note that:
\begin{equation}
    \label{step:claim1}
\pi_k(\eta\cup a) = \frac{1}{\binom{n}{k}}\mu(\eta\cup a) = \frac{k}{(n-k+1)}\pi_{k-1}(\eta)\mu_\eta(a) = k\pi_{k-1}(\eta)\pi_{\eta,1}(a).
\end{equation}

In general, for $k\geq 1$ we have the following for any $f:\Pinning_k\to\mathbb{R}$:
      \begin{align*}
           \Dirichlet_{\updown{k}}(f)  & =
           \frac{1}{2}\sum_{\eta\in\Pinning_{k-1}}\sum_{\substack{a,b\in (V\setminus S)\times\{0,1\}: \\ \eta:S\rightarrow \{0,1\}}} \pi_k(\eta\cup a)\updown{k}(\eta\cup a,\eta\cup b)(f(\eta\cup a)- f(\eta\cup b))^2 \\
        & =
           \frac{1}{2(k+1)}\sum_{\eta\in\Pinning_{k-1}}\sum_{a,b} \pi_k(\eta\cup a)Q_\eta(a,b)(f_\eta(a)- f_\eta(b))^2 & \mbox{by \cref{step:rem:local-downup}}
           \\
                        & =
           \frac{k}{2(k+1)}
           \sum_{\eta\in\Pinning_{k-1}}\pi_{k-1}(\eta)\sum_{a,b} \pi_{\eta,1}(a)Q_\eta(a,b)(f_\eta(a)- f_\eta(b))^2 & \mbox{by \cref{step:claim1}}
           \\
           & = \frac{k}{k+1}\sum_{\eta\in\Pinning_{k-1}}\pi_{k-1}(\eta)\Dirichlet_{Q_\eta}(f_\eta). \qedhere
           \end{align*}
\end{proof}

   \begin{proof}[Proof of \cref{claim:DDD}]
\begin{align*}
\lefteqn{
\sum_{\eta\in\Pinning_{k-1}} \pi_{k-1}(\eta) \Var_{\pi_{\eta,1}}(f_\eta)
} \hspace{1in}
\\
               &= 
    \frac{1}{2}\sum_{\eta\in\Pinning_{k-1}} \pi_{k-1}(\eta) \sum_{\substack{a,b\in (V\setminus S)\times\{0,1\}:\\ \eta:S\rightarrow\{0,1\}}} \pi_{\eta,1}(a)\pi_{\eta,1}(b)(f_\eta(a)-f_\eta(b))^2
    \\
                       &= 
   \frac{1}{2}\sum_{\eta,a,b} \pi_{k-1}(\eta)\frac{\mu_\eta(a)}{(n-k+1)}\times\frac{\mu_\eta(b)}{(n-k+1)}(f(\eta\cup\{a\})-f(\eta\cup\{b\}))^2
   \\
   &= 
   \frac{1}{2}\sum_{\eta,a,b} \pi_{k}(\eta\cup\{a\})\times\frac{1}{k}\times\frac{\mu_\eta(b)}{(n-k+1)}(f(\eta\cup\{a\})-f(\eta\cup\{b\}))^2
      \\
      &= 
   \frac{1}{2}\sum_{\eta,a,b} \pi_{k}(\eta\cup\{a\})P^{\downarrow}_{k}(\eta\cup\{a\},\eta)P^{\uparrow}_{k-1}(\eta,\eta\cup\{b\})(f(\eta\cup\{a\})-f(\eta\cup\{b\}))^2
      \\
      &=
   \Dirichlet_{\downup{k}}(f). \qedhere
    \end{align*}
   \end{proof}

\section{Inductive Proof of Random Walk Theorem (\texorpdfstring{\cref{thm:RW}}{Theorem 4.1})}
We can now prove \cref{eqn:RW-thm}, which implies the Random Walk Theorem (\cref{thm:RW}).

\begin{proof}[Proof of \cref{thm:RW}]

Our goal is to prove by induction that for all $2\leq\ell\leq n$:
\begin{equation}
    \label{eqn:ind-hyp-RW}
\gamma(\downup{\ell}) \geq \frac{1}{\ell}\prod_{i=0}^{\ell-2}\gamma_i.
\end{equation}
This is equivalent to the following statement in terms of the Dirichlet form and variance:
\begin{equation}
\label{eqn:ind-hyp-Dir-downup}
\forall f:\Pinning_{\ell}\rightarrow\R,\qquad
\Dirichlet_{\downup{\ell}}(f)
\geq
\frac{1}{\ell}
\prod_{i=0}^{\ell-2}\gamma_i
\Var_{\pi_{\ell}}(f).
\end{equation}
Recall $\gamma(\updown{\ell-1})=\gamma(\downup{\ell})$ was established in \cref{lem:updown-downup}.  Hence, \cref{eqn:ind-hyp-RW} 
is equivalent to the following statement:
\begin{equation}
\label{eqn:ind-hyp-Dir-updown}
\forall f:\Pinning_{\ell-1}\rightarrow\R,\qquad
\Dirichlet_{\updown{\ell-1}}(f)
\geq
\frac{1}{\ell}
\prod_{i=0}^{\ell-2}\gamma_i
\Var_{\pi_{\ell-1}}(f).
\end{equation}
We will prove \cref{eqn:ind-hyp-Dir-updown} by induction and use \cref{eqn:ind-hyp-Dir-downup} in the inductive step.

Now let us assume the inductive hypothesis \cref{eqn:ind-hyp-RW} for all $\ell<k$ and we will establish it for $\ell=k$.
\begin{align*}
\Dirichlet_{\updown{k-1}}(f) 
& \geq 
\frac{(k-1)\gamma_{k-2}}{k}\Dirichlet_{\downup{k-1}}(f)
&\mbox{by \cref{lem:technical-RW}}
\\
&\geq 
\frac{(k-1)\gamma_{k-2}}{k}\times\frac{1}{k-1}\prod_{i=0}^{k-3}\gamma_i\Var_{\pi_{k-1}}(f)
&\mbox{by \cref{eqn:ind-hyp-Dir-downup} for $\ell=k-1$}
\\
& = \frac{1}{k}\prod_{i=0}^{k-2}\gamma_i\Var_{\pi_{k-1}}(f),
\end{align*}
which establishes \cref{eqn:ind-hyp-Dir-updown} (and hence \cref{eqn:ind-hyp-RW}) by induction.
\end{proof}

\chapter{Optimal Mixing Time for Glauber from Spectral Independence}
\label{section-proof-relax}

The goal of this chapter is to prove that when the system is $\eta$-spectrally independent, then the relaxation time of the Glauber dynamics is $O(n)$, thereby proving~\cref{thm:SI-constant-relax} of~\cite{CLV21}, which we restate here for convenience.

Our proof follows the same approach as in~\cite{CLV21}.  In \cref{sec:optimal} we comment on how the proof changes to obtain \cref{thm:SI-constant-mix} establishing an optimal $O(n\log{n})$ mixing time, instead of relaxation time.

\CLVoptimalrelax*

\section{Uniform Block Dynamics}
Consider the following dynamics, parameterized by $0<\alpha<1$, which we will refer to as the {\em $\alpha n$-uniform block dynamics}.  The dynamics updates $\alpha n$ random vertices $S$ in each step.     Assume $\alpha n$ is an integer.  From $Y_t\in\Omega$, the transitions $Y_t\rightarrow Y_{t+1}$ of the $\alpha n$-uniform block dynamics are as follows:
\begin{enumerate}
    \item 
    Select $\alpha n$ random vertices, chosen uniformly at random from all $\binom{n}{\alpha n}$ subsets.   Let $S$ denote the chosen set.
    \item For all $w\notin S$, let $Y_{t+1}(w)=Y_t(w)$.
    \item Sample $Y_{t+1}(S)$ from the conditional Gibbs distribution $\mu(\sigma(S)\mid \sigma(w)=Y_{t+1}(w) \mbox{ for all } w\notin S).$
\end{enumerate}

Let $\downup{n,\ell}=P^\downarrow_{n}\cdots P^\downarrow_{\ell+1} P^\uparrow_{\ell}\cdots P^\uparrow_{n-1}$, that is, we apply the corresponding down walk $n-\ell$ times and then the corresponding up walk $n-\ell$ times.
Note the $\alpha n$-uniform block dynamics is identical to $\downup{n,(1-\alpha)n}$.  Using the Random Walk Theorem framework we will first prove fast mixing of the $\alpha$-block dynamics.  In particular we will show that the relaxation time is $O(1)$ where the constant in the big-$O()$ notation is a function of $1/\alpha$.  
We will then show that this implies fast mixing of the Glauber dynamics when $\alpha=O(1/\Delta)$ where $\Delta$ is the maximum degree of the graph.  

\section{Multi-level Random Walk Theorem}
\label{sub:improved-RW-statement}
The Random Walk Theorem (\cref{thm:RW}) is not strong enough to establish fast mixing of the $\alpha n$-uniform block dynamics; we will need the following improved local-to-global theorem which was presented by \cite[Theorem 5.4]{CLV21} and independently \cite[Theorem 4]{GM20}.  For simplicity we only present the spectral gap version, the stronger version for entropy contraction is presented in~\cite{CLV21,GM20}.

Recall that, $\gamma_j$ is the worst-case spectral gap of the local walk over all possible pinnings on $j$ vertices.
For $1\leq i\leq n-1$, let 
\[ 
\Gamma_i := \prod_{j=0}^{i-1}\max\{2\gamma_j-1,0\}
\] and let $\Gamma_0=1$.
Note that since, for all $j$, $\gamma_j\leq 1$ as the spectral gap is at most $1$, then $(2\gamma_j-1)\leq 1$ and hence $\Gamma_i\leq 1$ for all $i$.  Also, note that the quantity $\max\{2\gamma_j-1,0\}$ is only useful (otherwise the results will be trivially true) if $\gamma_j>1/2$ for all $j$, and this will be the case in our application of the upcoming Random Walk Theorem.  Moreover, one could remove this implicit assumption that $\gamma_j>1/2$ by replacing $(2\gamma_j-1)$ by $\gamma_j/(2-\gamma_j)$ in the definition of $\Gamma_i$; see the proof of~\cref{lem:improved-technical} for details. 

We will prove, for $0<k\leq n$, \begin{equation}
    \label{eqn:RW-one-improved}
 \gamma(\downup{k}) \geq \frac{\Gamma_{k-1}} {\sum_{i=0}^{k-1} \Gamma_i} 
\end{equation}
The above inequality will be a special case of the following multi-level version of the Random Walk Theorem. 
\begin{theorem}[Multi-level RW Theorem {\cite[Theorem A.9]{CLV21},\cite[Theorem 4]{GM20}}]
\label{lem:impr-RW-thm} 
    For $0\leq\ell<n$, 
    \begin{equation}
    \label{eqn:RW-improved-general}
 \gamma(\downup{n,\ell}) \geq \frac{\sum_{i=\ell}^{n-1} \Gamma_i}{\sum_{i=0}^{n-1} \Gamma_i} 
\end{equation}
\end{theorem}

Before proving the above theorem let us first compare it to the Random Walk Theorem (\cref{thm:RW}) stated earlier.
Looking at \cref{eqn:RW-one-improved} for $k=n$ (or \eqref{eqn:RW-improved-general} for $\ell=n-1$) we have:
\[
 \gamma(\downup{n}) \geq \frac{\Gamma_{n-1}}{\sum_{i=0}^{n-1}\Gamma_i} = \frac{\prod_{j=0}^{n-2}(2\gamma_j-1)}{\sum_{i=0}^{n-1}\Gamma_i} \geq 
\frac{1}{n} \prod_{j=0}^{n-2}(2\gamma_j-1)  \ \ \ \mbox{ since $\Gamma_i\leq 1$.}
\]
In contrast, \cref{thm:RW} has $2\gamma_j-1$ replaced by $\gamma_j$ in the lower bound on the spectral gap of $\downup{n}$, which is always better.
However, \cref{thm:RW} can only handle one-level down-up walks $\downup{n}$, while the advantage of \cref{lem:impr-RW-thm} is to deal with multi-level down-up walks $\downup{n,\ell}$.

\section{Fast Mixing of Uniform Block Dynamics}
Applying \cref{lem:impr-RW-thm} we establish fast mixing of the $\alpha n$-uniform block dynamics.

\begin{lemma}
\label{lem:gap-global-block}
For every $\eta>0$ and $0<\alpha<1$, there exist constants
$n_0=n_0(\eta,\alpha)$ and $C=C(\eta,\alpha)>0$ such that the
following holds. For every $n\geq n_0$, if the system is
$\eta$-spectrally independent, then
\[
\gamma(\downup{n,(1-\alpha)n})\geq C.
\]
\end{lemma}

\begin{proof}
For simplicity, we suppress integer rounding in the block sizes;
using floors or ceilings changes only the constants.
Let 
\[ m=(1-\alpha/2)n.
\]
For every $0\leq i\leq m-1$,
  \[ n-i-1\geq n-m = \frac{\alpha n}{2},
  \]
  and hence by \cref{eqn:gammak},
  \[ \gamma_i\geq 1-\frac{\eta}{n-i-1} \geq 1 - \frac{2\eta}{\alpha n} \geq 3/4,
  \] when $n\geq 8\eta/\alpha$.  Therefore, $2\gamma_i-1\geq 1/2 > 0$.  
  
We begin with a useful lower bound on $\Gamma_{m}$:
\begin{align}   
\Gamma_{m}
& = \prod_{i=0}^{m-1}\max\{2\gamma_i-1,0\} 
\tag{\mbox{by definition}}
\nonumber\\
& = \prod_{i=0}^{m-1}(2\gamma_i-1) 
\nonumber
\\
& \geq \prod_{i=0}^{m-1} \left(1-\frac{2\eta}{n-i-1}\right) 
\tag{by \cref{eqn:gammak}}
\nonumber\\
& \geq \left(1 - \frac{4\eta}{\alpha n}\right)^m
\nonumber\\
& \geq \left(1 - \frac{4\eta}{\alpha n}\right)^n
\tag{\mbox{since $m\leq n$}}
\nonumber\\ 
& \geq \exp(-8\eta/\alpha),
\label{eqn:alpha-bound} 
\end{align}
for $n\geq 8\eta/\alpha$, where the last inequality uses that $1-x\geq \exp(-2x)$ for $x\leq 1/2$.

We can now proceed to lower bound the spectral gap of the $\alpha n$-uniform block dynamics:
    \begin{align*}
    \label{gap:global-block}
 \gamma(\downup{n,(1-\alpha)n}) 
&\geq
\frac{\sum_{i=(1-\alpha)n}^{n-1} \Gamma_i}{\sum_{i=0}^{n-1} \Gamma_i} 
 & \mbox{by \cref{eqn:RW-improved-general}}
 \\
 &\geq
\frac{\sum_{i=(1-\alpha)n}^{(1-\alpha/2)n} \Gamma_i}{n} 
 & \mbox{since $\Gamma_i\leq 1$}
 \\
   &\geq
 \frac{(\alpha/2)n\Gamma_{(1-\alpha/2)n}}{n} 
 & \mbox{since $\Gamma_i\geq\Gamma_{i+1}$}
 \\
    &\geq
 (\alpha/2)\exp(-8\eta/\alpha) 
 & \mbox{by \cref{eqn:alpha-bound}}
 \\
 & = C(\alpha,\eta). \qedhere
\end{align*}
\end{proof}

\section{Shattering}
\label{sec:shattering}

We established fast mixing of the $\alpha n$-uniform block dynamics, see \cref{lem:gap-global-block}, which updates $\alpha n$ random vertices in each step.  We will use that to establish fast mixing of the Glauber dynamics.  

The first step is to look at the properties of the updated vertices in a step of the block dynamics.  In particular, \cite{CLV21} showed that a random subset of $\alpha n$ vertices is ``shattered'' in the sense that the expected size of each component is $O(1)$ when $\alpha<1/(6\Delta)$ where $\Delta$ is the maximum degree of the graph $G$.
The shattering occurs because $\alpha\Delta<1$ and hence this corresponds to a sub-critical branching process.

For a subset $S\subset V$, let $\CC_S$ denote the collection of connected components in the induced subgraph on $S$.
For $v\in S$, let $T_v$ denote the component containing $v$,
and if $v\notin S$, set $T_v=\emptyset$.

\begin{lemma}{\cite[Lemma 4.3]{CLV21}} 
\label{lem:shattering}
For a graph $G=(V,E)$ of maximum degree $\Delta$, for $\alpha>0$, choose a random subset $S\subset V$ where $|S|=\alpha n$.  Then for every integer $k\geq 1$, for every $v\in V$,
\[ \Prob{|T_v|=k} \leq \alpha(6\Delta\alpha)^{k-1},\]
where $T_v$ is the component containing $v$ in the induced subgraph on $S$ (with the convention above that $T_v=\emptyset$ when $v\notin S$).
\end{lemma}

We will use the following combinatorial result in the proof of \cref{lem:shattering}, see \cite[Lemma 2.1]{BCKL13} for a slightly stronger bound.
\begin{lemma}[{\cite[Lemma 2.1]{BCKL13}}]
\label{lem:subgraphs-count}
    Let $G=(V,E)$ be a graph of maximum degree $\Delta \ge 3$ and $v \in V$ be an arbitrary vertex. The number of distinct connected subgraphs of size $k$ that contain $v$ is at most $\binom{\Delta k}{k-1}$.
\end{lemma}
\begin{proof}
    We can encode a connected subgraph $H$ of $G$ of size $k$ containing $v$ by a binary string of length $\Delta k$ with exactly $k-1$ ones as follows: we run a modified DFS algorithm that only traverses edges in $H$ starting from~$v$. For each encountered edge output $1$ if it is traversed (for that the edge must be in $H$ and the other endpoint has not been yet visited) and $0$ otherwise. At the end we pad the string with zeros to make it length $\Delta k$. Note that we can recover the execution of the modified DFS (and hence $H$) from this sequence. 
    Therefore, the number of such $H$ is at most the number of length-$\Delta k$ binary strings with $k-1$ ones, which is $\binom{\Delta k}{k-1}$.
\end{proof}

\begin{proof}[Proof of \cref{lem:shattering}]
For $k=1$,
\[
\Prob{|T_v|=1}
\leq
\Prob{v\in S}
=
\alpha.
\]
For $k\geq2$, the probability that a particular set of size $k$ is contained in $S$ is $\binom{n-k}{\alpha n-k}/\binom{n}{\alpha n}\leq \alpha^k$. Combining \cref{lem:subgraphs-count} and the binomial upper bound $\binom{n}{k}\leq (n{\mathrm{e}}/k)^k$, we obtain the following:
\[
\Prob{|T_v|=k}
\leq
\binom{\Delta k}{k-1}\alpha^k
\leq
\left(
\frac{\Delta\mathrm{e}k}{k-1}
\right)^{k-1}\alpha^k
\leq
\alpha(6\Delta\alpha)^{k-1}.
\qedhere
\]
\end{proof}

\section{Optimal Relaxation Time of Glauber: Proof of \texorpdfstring{\cref{thm:SI-constant-relax}}{Theorem 1.6}}
\label{sec:proof-relax}

Using the above shattering result (\cref{lem:shattering}) with the fast mixing result for the $\alpha n$-uniform block dynamics (\cref{lem:gap-global-block}), we can prove an optimal upper bound on the relaxation time of the Glauber dynamics and hence prove~\cref{thm:SI-constant-relax}.

The high-level proof idea for \cref{thm:SI-constant-relax} is the following.  Let $S\subset V$ be a random subset where $|S|=\alpha n$ as in \cref{lem:shattering}.  Let $\mathcal{C}_S$ denote the connected components in the induced subgraph on $S$.  Conditional on the configuration on $\overline S$, the Gibbs
distribution on $S$ is a product distribution over the components
$T\in\mathcal{C}_S$. Hence, by tensorization of variance and
averaging over the boundary configuration we have the following
 (see \cref{sub:app-tensorization} for discussion of approximate tensorization of variance for product measures):
\[
\sum_{\tau\in\Omega}\mu(\tau)\Var_S[F\mid\tau]
\leq
\sum_{\tau\in\Omega}\mu(\tau)
\sum_{T\in\mathcal{C}_S}\Var_T[F\mid\tau].
\]
Within a component $T\in\mathcal{C}_S$ we can apply \cref{thm:SI-mixing} to conclude a polynomial bound on the spectral gap of the Glauber dynamics within $T$ (with fixed configuration $\tau$ on the boundary of $T$), and hence 
\begin{equation}
    \label{var-not-ent}
\Var_T[F|\tau]
\leq
C'|T|^{\eta+1}
\Dirichlet_{\Glaubert(T)}(f).
\end{equation}
Finally, we apply \cref{lem:shattering} to utilize the upper bound on the size of $T$ and conclude the proof.

The bound in \cref{var-not-ent} is the only step in the proof where the analogous bound does not hold with respect to entropy in place of variance.  This is an issue when we aim to upper bound the mixing time instead of the relaxation time.  To obtain a bound on the entropy analog of \cref{var-not-ent} we need the additional assumption of marginal boundedness; see \cref{sec:optimal} for further details.

\begin{proof}[Proof of \cref{thm:SI-constant-relax}]
In the following, for a subset $S\subset V$, we use $\HBt(S)$ to denote the heat-bath block dynamics on $S$;
this dynamics updates the configuration on all of $S$ in one step from the Gibbs distribution conditional on the fixed configuration $\tau$ on $\overline{S}$.
Similarly, we use $\Glaubert$ to denote the Glauber dynamics on a set $S$ with a fixed configuration $\tau$ on $\overline{S}$; in particular, for $v\in S$, with probability $1/|S|$ we update the configuration at $v$ from the Gibbs distribution conditional on the fixed configuration for the other vertices. 

For a function $f:\Omega\rightarrow\R$, we will consider the conditional variance.  For $\sigma\sim\mu$, let \[ F = f(\sigma). \]
Note that when we write $\Var(f)$ we are implicitly writing $\Var[F]$ since we first draw a sample $\sigma$ from the Gibbs distribution $\mu$ and then look at the variance of $F=f(\sigma)$.  

Consider a configuration $\tau\in\Omega$, and a subset $S\subset V$.  In the upcoming analysis we will consider the variance of $f$ within $S$ where we fix $\tau$ on $\overline{S}=V\setminus S$. We denote this as:
\[ \Var_S[F\mid\tau] := \Var[F\mid \sigma(\overline{S})=\tau(\overline{S})].
\]
Once again, in words $\Var_S[F\mid \tau]$ is the variance of the function $f$ over the choice of $\sigma$ from the Gibbs distribution $\mu$ conditional on $\sigma(\overline{S})=\tau(\overline{S})$.

Again fix $\tau\in\Omega$ and $S\subset V$.  Consider the heat-bath dynamics on $S$ with a fixed $\tau$ on $\overline{S}$.  The dynamics chooses the configuration on $S$ from $\mu_{\tau(\overline{S})}$, which is the Gibbs distribution $\mu$ conditional on the configuration on $\overline{S}$ being $\tau(\overline{S})$.  Hence,
\begin{equation}
\label{eqn:Cond-Var-Dirichlet}
  \Dirichlet_{\HBt(S)}(f) 
 = \Var[F\mid \sigma(\overline{S})=\tau(\overline{S})]  
\end{equation}
Moreover, if $S=\{v\}$ for $v\in V$ then this is heat-bath on the single vertex $v$.  The Glauber dynamics operates by choosing a random vertex $v$ and then applying the heat-bath dynamics on $v$.  Hence, we have the following (analogous to \cref{eq:AT-expanded-further}):
\begin{align}
\nonumber
  \Dirichlet_{\Glauber}(f) 
 & = \frac{1}{2}\sum_{\xi\in\Omega}\sum_{\sigma\in\Omega}\mu(\xi)\Glauber(\xi,\sigma)(f(\xi)-f(\sigma))^2
 \\
 \nonumber
  & = \sum_{\xi\in\Omega}\mu(\xi)\frac{1}{n}\sum_{v\in V}\left[\frac12\sum_{\substack{\sigma\in\Omega:\\ \sigma(V\setminus\{v\})=\xi(V\setminus\{v\})}}
  \mu\!\left(
\sigma\mid
\sigma(V\setminus\{v\})=\xi(V\setminus\{v\})
\right)(f(\xi)-f(\sigma))^2\right]
 \\
 \nonumber
&=
\frac1n\sum_{\xi\in\Omega}\mu(\xi)
\sum_{v\in V}
\Var\!\left[
F\mid
\sigma(V\setminus\{v\})
=
\xi(V\setminus\{v\})
\right]
\\
&=
\frac{1}{n}
\Exp_{\xi\sim\mu}\left[
\sum_{v\in V}\Var_v[F\mid\xi]
\right].
\label{eqn:Glauber-Cond-Var}
\end{align}

Recall, for a subset $S\subset V$, $\CC_S$ denotes the collection of connected components in the induced subgraph on $S$.
For any function $f:\Omega\rightarrow\R$,
\begin{align*}
\Var(f) 
& \leq
\frac{1}{C}
\Dirichlet_{\downup{n,(1-\alpha)n}}(f)
& \mbox{by \cref{lem:gap-global-block}}
\\
& =
\frac{1}{C}\frac{1}{\binom{n}{\alpha n}}\sum_{\substack{S\subset V: \\ |S|=\alpha n}} 
\sum_{\tau\in\Omega} \mu(\tau)\Dirichlet_{\HBt(S)}(f)
& \mbox{(see below)}
\\
& =
\frac{1}{C}\frac{1}{\binom{n}{\alpha n}}\sum_{\substack{S\subset V: \\ |S|=\alpha n}} 
\sum_{\tau\in\Omega} \mu(\tau)\Var_S[F\mid \tau]
& \mbox{by \cref{eqn:Cond-Var-Dirichlet}}
\\
& \leq 
\frac{1}{C}\frac{1}{\binom{n}{\alpha n}}\sum_{\substack{S\subset V: \\ |S|=\alpha n}} 
\sum_{\tau\in\Omega} \mu(\tau) \sum_{T\in\CC_S}\Var_T[F\mid \tau]
\tag*{\mbox{by variance tensorization over $T\in\CC_S$}}
\\
& \leq 
 \frac{1}{C}\frac{1}{\binom{n}{\alpha n}}\sum_{\substack{S\subset V: \\ |S|=\alpha n}} \sum_{\tau\in\Omega} \mu(\tau) 
 \sum_{T\in\CC_S}
  C'|T|^{\eta+1}\Dirichlet_{\Glaubert(T)}(f)
\tag*{\mbox{by \cref{thm:SI-mixing}, (see below)}}
\\
& =
 \frac{1}{C}\frac{1}{\binom{n}{\alpha n}}
 \sum_{\substack{S\subset V:\\ |S|=\alpha n}}
 \sum_{\tau\in\Omega}\mu(\tau)
 \sum_{T\in\CC_S}
 \frac{C'|T|^{\eta+1}}{|T|}
 \Exp_{\xi\sim\mu_{\tau(V\setminus T)}}\left[
     \sum_{v\in T}\Var_v[F\mid\xi]
 \right]
 \\[-.3ex]
& \tag*{\mbox{by \cref{eqn:Glauber-Cond-Var},
applied conditionally on $V\setminus T$}}
\\
& =
 \frac{1}{C}\frac{1}{\binom{n}{\alpha n}}\sum_{\substack{S\subset V: \\ |S|=\alpha n}} \sum_{\tau\in\Omega} \mu(\tau) 
 \sum_{T\in\CC_S}
  C'|T|^{\eta+1}\frac{1}{|T|}\sum_{v\in T}\Var_{v}[F\mid \tau]
\\[-.3ex]
& \tag*{\mbox{by averaging over the boundary and the conditional configuration on $T$}}
\\
& =
 \frac{1}{C} \sum_{\tau\in\Omega} \mu(\tau) 
 \sum_{v\in V}\Var_{v}[F\mid \tau]
 \sum_{k=1}^{\alpha n}\Prob{|T_v|=k}
\times C'k^{\eta}
& \mbox{rearranging}
\\
& \leq 
 \frac{1}{C} \sum_{\tau\in\Omega} \mu(\tau) 
 \sum_{v\in V}\Var_{v}[F\mid \tau]
 \sum_{k=1}^{\alpha n}
 C'(6\Delta\alpha)^{k-1} k^{\eta}
  & \mbox{by \cref{lem:shattering}}
\\
& \leq 
 \frac{1}{C} \sum_{\tau\in\Omega} \mu(\tau) 
 \sum_{v\in V}\Var_{v}[F\mid \tau]
 \sum_{k=1}^{\alpha n}
 C' 2^{-(k-1)}
 \tag*{\mbox{for $\alpha<\exp(-\eta)/(100\Delta)$}}
  \\
& \leq 
 \frac{1}{C''} \sum_{\tau\in\Omega} \mu(\tau) 
 \sum_{v\in V}\Var_{v}[F\mid \tau]
\\
& =
\frac{n}{C''}\Dirichlet_{\Glauber}(f)
 & \mbox{by \cref{eqn:Glauber-Cond-Var}.}
\end{align*}
In the second line we are implementing the $(n,(1-\alpha)n)$ down-up dynamics by first choosing the subset $S$ of size $\alpha n$ for update and then applying the heat-bath dynamics on this set $S$ with those vertices outside $S$ fixed to $\tau$ chosen from the Gibbs distribution $\mu$.
In the fifth line, we are applying \cref{thm:SI-mixing} for the Glauber dynamics which operates on the vertices in this component $T$, where the vertices outside $T$ have the fixed configuration $\tau$.
To apply \cref{thm:SI-mixing} we need $\eta$-spectral independence, which we have by assumption, and we also need $(A,\eta)$-small-conditional-gap, for which one can easily establish a bound of $A\leq (1+\lambda)^t\poly(t)$, and since $t\leq \eta+1$ in the definition of $(A,\eta)$-small-conditional-gap we obtain a constant $A=A(\eta,\lambda)$.  Consequently, from the application of \cref{thm:SI-mixing} in this step, we obtain the claimed bound of $C'|T|^{\eta+1}$ where $C'=C'(\eta,\lambda)$.
Note that \cref{lem:gap-global-block}, used in the first line, applies only for $n\geq n_0(\eta,\alpha)$; however,  this suffices for the proof of \cref{thm:SI-constant-relax} as the  finitely many smaller values of $n$ can be absorbed into the constant in the $O(n)$ bound.

This proves that the spectral gap of the Glauber dynamics is $\geq C''/n$ for a constant $C''=C''(\eta,\Delta,\lambda)$.  Hence, by applying \cref{eq:relax-gap}, the relaxation time of the Glauber dynamics is $\leq Cn$ for some constant $C=C(\eta,\Delta,\lambda)$. 
For the hard-core model at fixed fugacity $\lambda$, we have:
\[ \mu^* =\min_{x\in\Omega} \mu(x) \geq \left(\frac{\min\{1,\lambda\}}{1+\lambda}\right)^n,
\] and hence $\log(1/\mu^*)=O_\lambda(n)$.  Consequently, by \cref{eq:mix-gap},
\[
\Tmix=O_{\eta,\Delta,\lambda}(n^2).
\]
This establishes the bounds on the relaxation time and mixing time stated in \cref{thm:SI-constant-relax}.
\end{proof}

\section{Multi-level Random Walk Theorem: Proof of \texorpdfstring{\cref{lem:impr-RW-thm}}{Theorem 6.1}}
\label{sub:improved-RW-proof}

In this section we will prove the Multi-level Random Walk Theorem presented in~\cref{lem:impr-RW-thm} in \cref{sub:improved-RW-statement}.
The proof will involve statements comparing Dirichlet forms and variance on different levels and hence we will need to define the projection of a function $f$ on lower levels.  

Consider an arbitrary function $f:\Omega\rightarrow\R$.  Recall $\Pinning_n=\Omega$, and let $f^{(n)}=f$.  We will define $f^{(k)}$ for $0\leq k<n$ in the following inductive manner.

For $f^{(k+1)}:\Pinning_{k+1}\rightarrow\R$, let 
\[ f^{(k)}=\up{k}f^{(k+1)}.
\]
Hence, for $\sigma\in\Pinning_{k}$ we have that $f^{(k)}(\sigma) = \sum_{\tau\in\Pinning_{k+1}}\up{k}(\sigma,\tau)f^{(k+1)}(\tau)$.

We will use the two identities involving the (global) variance and the local variance, measured by the Dirichlet form.  The following is an interesting and useful identity for relating the Dirichlet form for the down-up walk from level $i$ to $j$ in terms of the difference of the variance at levels $i$ and $j$.

\begin{lemma}
    \label{lem:diff-var}
For all $n\geq i>j\geq 0$,  for all $f^{(i)}:\Pinning_i\rightarrow \R$, the following holds:
\begin{equation}
   \label{eqn:general-basic-fact}
    \Dirichlet_{\downup{i,j}}(f^{(i)}) = \Var_{\pi_i}(f^{(i)}) - \Var_{\pi_{j}}(f^{(j)}).
\end{equation}
\end{lemma}

In the proof of the Random Walk Theorem (\cref{eqn:RW-thm}), the key technical inequality stated in \cref{lem:technical-RW} was the following:
\[
\Dirichlet_{\updown{k}}(f) \geq \frac{k}{k+1}\gamma_{k-1}\Dirichlet_{\downup{k}}(f).
\]

For the multi-level result we will use the following variant.  Recall, $f^{(k)}=\up{k}f^{(k+1)}$.

\begin{lemma}
    \label{lem:improved-technical}
    For $1\leq k<n$ and any
$f^{(k+1)}:\Pinning_{k+1}\rightarrow\R$, the following holds:
\begin{equation}\label{eqn:NEW-D}
\Dirichlet_{\downup{k+1}}(f^{(k+1)}) \geq (2\gamma_{k-1}-1)\Dirichlet_{\downup{k}}(f^{(k)}).
\end{equation}
\end{lemma}

We defer the proofs of \cref{lem:diff-var,lem:improved-technical} to~\cref{sec:missing}, and we proceed with the proof of~\cref{lem:impr-RW-thm}.

\begin{proof}[Proof of \cref{lem:impr-RW-thm}]
We begin by proving \cref{eqn:RW-one-improved} which states that the spectral gap of $\downup{k}$ is at least $\Gamma_{k-1}/\left(\sum_{i=0}^{k-1}\Gamma_i\right)$ where $\Gamma_i=\prod_{j=0}^{i-1}\max\{2\gamma_{j}-1,0\}$.  We will then use this result for the one-level down-up walk to derive the more general statement in \cref{lem:impr-RW-thm} for the down-up walk $\downup{n,\ell}$.

Our inductive hypothesis is the following statement about the spectral gap of $\downup{k}$:
\begin{equation}
\label{induct:simpler}
\left(\sum_{i=0}^{k-1}\Gamma_i\right)\Dirichlet_{\downup{k}}(f^{(k)}) 
 \geq 
\Gamma_{k-1}\Var_{\pi_{k}}(f^{(k)}).
\end{equation}

The base case $k=1$ is immediate, since $\downup{1}$ completely
resamples from $\pi_1$ and hence has spectral gap~$1$.
For the inductive step, if $2\gamma_{k-1}-1\leq0$, then
$\Gamma_k=0$ and the desired inequality for $k+1$ is immediate.
Hence, assume that $2\gamma_{k-1}-1>0$.

We can now complete the proof by induction as follows:
\begin{align*}
\lefteqn{
\left(\sum_{i=0}^{k}\Gamma_i\right)\Dirichlet_{\downup{k+1}}(f^{(k+1)}) 
} \hspace{.3in}
\\& = 
\Gamma_{k}\Dirichlet_{\downup{k+1}}(f^{(k+1)}) + \left(\sum_{i=0}^{k-1}\Gamma_i\right)\Dirichlet_{\downup{k+1}}(f^{(k+1)})
\\
& \geq 
\Gamma_{k}\Dirichlet_{\downup{k+1}}(f^{(k+1)}) + \left(\sum_{i=0}^{k-1}\Gamma_i\right)(2\gamma_{k-1}-1)\Dirichlet_{\downup{k}}(f^{(k)})
&\mbox{by \cref{eqn:NEW-D}}
\\
& \geq 
\Gamma_{k}\Dirichlet_{\downup{k+1}}(f^{(k+1)}) + \Gamma_{k-1}(2\gamma_{k-1}-1)\Var_{\pi_{k}}(f^{(k)})
&\mbox{by inductive hypothesis}
\\
& =
\Gamma_{k}\Dirichlet_{\downup{k+1}}(f^{(k+1)}) + \Gamma_{k}\Var_{\pi_{k}}(f^{(k)})
&\mbox{since } 
\Gamma_k = (2\gamma_{k-1}-1)\Gamma_{k-1}
\\
& =
\Gamma_{k}\Var_{\pi_{k+1}}(f^{(k+1)})
&\mbox{by \cref{lem:diff-var} with $i=k+1,j=k$.}
\end{align*}
This completes the proof of \cref{eqn:RW-one-improved}.

Let us now prove the more general statement in~\cref{eqn:RW-improved-general}.
Establishing the bound stated in \cref{eqn:RW-improved-general} on the spectral gap of $\downup{n,\ell}$ is equivalent to proving the following for all $0<\ell<n$ and all $f:\Omega\rightarrow\R$:
\begin{equation}
    \label{induct:RW-improved-general}
\left(\sum_{i=0}^{n-1} \Gamma_i\right) \Dirichlet_{\downup{n,\ell}}(f) \geq 
\left( \sum_{i=\ell}^{n-1} \Gamma_i \right)\Var_{\pi_{n}}(f).
\end{equation}

Now we can prove \cref{induct:RW-improved-general} using \cref{induct:simpler}. 
From \cref{lem:diff-var} we have that $\Dirichlet_{\downup{k}}(f^{(k)}) = \Var_{\pi_k}(f^{(k)}) - \Var_{\pi_{k-1}}(f^{(k-1)})$, and hence \cref{induct:simpler} is equivalent to the following:
\begin{equation}
\label{induct:AAA-simpler}
\frac{\Var_{\pi_{k}}(f^{(k)}) }
{\left(\sum_{i=0}^{k-1}\Gamma_{i}\right)}
 \geq 
\frac{\Var_{\pi_{k-1}}(f^{(k-1)})}{\left(\sum_{i=0}^{k-2}\Gamma_i\right)}.
\end{equation}
Then by chaining these inequalities (with $f=f^{(n)}$) we obtain:
\begin{equation}
\label{induct:GGG-simpler}
\frac{\Var_{\pi_{n}}(f) }
{\left(\sum_{i=0}^{n-1}\Gamma_{i}\right)}
 \geq 
\frac{\Var_{\pi_{\ell}}(f^{(\ell)})}{\left(\sum_{i=0}^{\ell-1}\Gamma_i\right)}.
\end{equation}
Finally, applying \cref{lem:diff-var} with $i=n,j=\ell$, we have that \cref{induct:GGG-simpler} is equivalent to \cref{induct:RW-improved-general}, which completes the proof of~\cref{lem:impr-RW-thm}.
\end{proof}

\section{Proofs of Basic Facts for Dirichlet Form and Variance}
\label{sec:missing}

Here we provide the missing proofs of \cref{lem:diff-var,lem:improved-technical}.
We begin with the proof of \cref{lem:diff-var}.


\begin{proof}[Proof of \cref{lem:diff-var}]
 Without loss of generality we can assume that $\Exp_{\pi_i}[f^{(i)}] = \Exp_{\sigma\sim\pi_i}[f^{(i)}(\sigma)]=0$ and hence:
\begin{equation}
    \label{eqn:Var-A}
    \Var_{\pi_i}(f^{(i)}) = \sum_{\sigma\in\Pinning_i}\pi_i(\sigma)f^{(i)}(\sigma)^2.
\end{equation}

Since we have $\Exp_{\pi_i}[f^{(i)}]=0$, it follows that $\Exp_{\eta\sim\pi_{i-1}}[f^{(i-1)}(\eta)]=
\Exp_{\sigma\sim\pi_i}[f^{(i)}(\sigma)]=0$ and hence by induction we have that 
$\Exp_{\eta\sim\pi_{j}}[f^{(j)}(\eta)]=0$ for $j<i$.

For $0\leq j< i\leq n$, $\eta\in\Pinning_j$ and $\sigma\in\Pinning_i$ with
$\eta\subseteq\sigma$, the transition probability for the
multi-level up walk from $\eta$ to $\sigma$ is
\begin{equation}
\label{eqn:up-jtoi}
\up{j,i}(\eta,\sigma)
=
\frac{\pi_i(\sigma)}
{\binom{i}{j}\pi_j(\eta)}.
\end{equation}
To see this, note that for each fixed ordering of the $i-j$ elements in $\sigma\setminus\eta$, the product of the corresponding single-level up-transition probabilities is $
\frac{\pi_i(\sigma)}
{(j+1)(j+2)\cdots i\,\pi_j(\eta)}$, and there are $(i-j)!$ possible orderings.

Now we can consider the variance at level $j$:
\begin{align}
\nonumber
\Var_{\pi_j}(f^{(j)})
&=
\sum_{\eta\in\Pinning_j}
\pi_j(\eta)f^{(j)}(\eta)^2
\\
\nonumber
&=
\sum_{\eta\in\Pinning_j}\pi_j(\eta)
\left[
\sum_{\sigma\in\Pinning_i}
\up{j,i}(\eta,\sigma)f^{(i)}(\sigma)
\right]^2
\\
\nonumber
&=
\sum_{\eta\in\Pinning_j}\pi_j(\eta)
\sum_{\substack{\tau_1,\tau_2:\\
\eta\cup\tau_1,\eta\cup\tau_2\in\Pinning_i}}
\up{j,i}(\eta,\eta\cup\tau_1)
f^{(i)}(\eta\cup\tau_1)
\up{j,i}(\eta,\eta\cup\tau_2)
f^{(i)}(\eta\cup\tau_2)
\\
\nonumber
&=
\sum_{\eta\in\Pinning_j}\pi_j(\eta)
\sum_{\tau_1,\tau_2}
\frac{
\pi_i(\eta\cup\tau_1)
\pi_i(\eta\cup\tau_2)
}{
\binom{i}{j}^2\pi_j(\eta)^2
}
f^{(i)}(\eta\cup\tau_1)
f^{(i)}(\eta\cup\tau_2)
\tag*{\mbox{by \cref{eqn:up-jtoi}}}
\\
\nonumber
&=
\sum_{\sigma_1,\sigma_2\in\Pinning_i}
\sum_{\substack{\eta\in\Pinning_j:\\
\eta\subseteq\sigma_1,\ \eta\subseteq\sigma_2}}
\frac{
\pi_i(\sigma_1)\pi_i(\sigma_2)
}{
\binom{i}{j}^2\pi_j(\eta)
}
f^{(i)}(\sigma_1)f^{(i)}(\sigma_2)
\\
&=
\sum_{\sigma_1,\sigma_2\in\Pinning_i}
\pi_i(\sigma_1)
\downup{i,j}(\sigma_1,\sigma_2)
f^{(i)}(\sigma_1)f^{(i)}(\sigma_2).
\label{eqn:Var-B}
\end{align}

We can complete the proof by  decomposing the Dirichlet form of the down-up walk into two terms as follows:
\begin{align*}
 \Dirichlet_{\downup{i,j}}(f^{(i)}) 
  &= \frac12\sum_{\sigma_1,\sigma_2\in\Pinning_i}\pi_i(\sigma_1)\downup{i,j}(\sigma_1,\sigma_2)(f^{(i)}(\sigma_1)-f^{(i)}(\sigma_2))^2.
\\
 & = \sum_{\sigma\in\Pinning_i}\pi_i(\sigma)f^{(i)}(\sigma)^2
- \sum_{\sigma_1,\sigma_2\in\Pinning_i}\pi_i(\sigma_1)\downup{i,j}(\sigma_1,\sigma_2)f^{(i)}(\sigma_1)f^{(i)}(\sigma_2)
\\
& = \Var_{\pi_i}f^{(i)} - \Var_{\pi_{j}}f^{(j)} & \mbox{by \cref{eqn:Var-A,eqn:Var-B}}
\end{align*}
where the second equality follows from the reversibility condition $\pi_i(\sigma)\downup{i,j}(\sigma,\tau)=\pi_i(\tau)\downup{i,j}(\tau,\sigma)$ and the fact that  $\sum_{\tau\in\Pinning_i}\downup{i,j}(\sigma,\tau)=1$.  This completes the proof of the lemma.
\end{proof}


We now prove \cref{lem:improved-technical}.  We will use the following identity in the proof of~\cref{lem:improved-technical}. 

\begin{claim}
\label{claim:first-step}
\begin{equation}
\label{eqn:first-step}
\Var_{\pi_{k+1}}(f^{(k+1)}) - \Var_{\pi_{k-1}}(f^{(k-1)}) = \sum_{\tau\in\Pinning_{k-1}}\pi_{k-1}(\tau)\Var_{\pi_{\tau,2}}(f_\tau^{(2)}).
\end{equation}    
\end{claim}

\begin{proof}
Let us begin by examining each of the terms in the LHS of \cref{eqn:first-step}.  Once again we can assume $\Exp_{\pi_{k+1}}[f^{(k+1)}]=\Exp_{\pi_{k-1}}[f^{(k-1)}]=0$.
\begin{equation}
\label{eqn:Var-first-k+1}
\Var_{\pi_{k+1}}(f^{(k+1)}) 
 = 
\sum_{\sigma\in\Pinning_{k+1}}\pi_{k+1}(\sigma)f^{(k+1)}(\sigma)^2
\end{equation}
And for the variance at level $k-1$ we have:
\begin{equation}
\Var_{\pi_{k-1}}(f^{(k-1)}) 
= 
\sum_{\eta\in\Pinning_{k-1}}\pi_{k-1}(\eta)f^{(k-1)}(\eta)^2
\label{eqn:Var-second-k-1}
\end{equation}

Now let us analyze the variance of the local walk on the RHS of \cref{eqn:first-step}.
 First note that for $\eta\in\Pinning_{k-1}$ where $\eta:S\rightarrow\{0,1\}$ for $S\subset V$, for $\tau:T\rightarrow\{0,1\}$ where $T\subset V\setminus S, |T|=2$, then $f_\eta^{(2)}(\tau)=f^{(k+1)}(\eta\cup\tau)$.
 In the following equation array, the sum over $\tau$ is over all $T\subset V\setminus S, |T|=2$, where $\eta$ is on $S$, and all $\tau:T\rightarrow \{0,1\}$. 

 \begin{align*}
\lefteqn{
\sum_{\eta\in\Pinning_{k-1}}\pi_{k-1}(\eta)\Var_{\pi_{\eta,2}}(f_\eta^{(2)})
} 
\\
 & =
 \sum_{\eta\in\Pinning_{k-1}}\pi_{k-1}(\eta)\sum_{\tau}\pi_{\eta,2}(\tau)(f^{(k+1)}(\eta\cup\tau) - 
 f^{(k-1)}(\eta))^2
 \\
  & =
 \sum_{\eta\in\Pinning_{k-1}}\pi_{k-1}(\eta)\sum_{\tau}\pi_{\eta,2}(\tau)f^{(k+1)}(\eta\cup\tau)^2 - 2 \sum_{\eta\in\Pinning_{k-1}}\pi_{k-1}(\eta)f^{(k-1)}(\eta)\sum_{\tau}\pi_{\eta,2}(\tau)f^{(k+1)}(\eta\cup\tau)
 \\ & \hspace{.5in}
 +  \sum_{\eta\in\Pinning_{k-1}}\pi_{k-1}(\eta)f^{(k-1)}(\eta)^2\sum_{\tau}\pi_{\eta,2}(\tau)
   \\
   & = \sum_{\eta\in\Pinning_{k-1}}\pi_{k-1}(\eta)\sum_{\tau}\pi_{\eta,2}(\tau)f^{(k+1)}(\eta\cup\tau)^2 
 - \sum_{\eta\in\Pinning_{k-1}}\pi_{k-1}(\eta)f^{(k-1)}(\eta)^2
 \\
   & = \sum_{\sigma\in\Pinning_{k+1}}\pi_{k+1}(\sigma)f^{(k+1)}(\sigma)^2 
 - \sum_{\eta\in\Pinning_{k-1}}\pi_{k-1}(\eta)f^{(k-1)}(\eta)^2
\\
& = \Var_{\pi_{k+1}}(f^{(k+1)}) 
-  \Var_{\pi_{k-1}}(f^{(k-1)}),
\end{align*}
where we used the following two identities:
\begin{equation}
    \label{eqn:step111}
\mbox{for every $\eta\in\Pinning_{k-1}$, } 
f^{(k-1)}(\eta) = \sum_{\tau}\pi_{\eta,2}(\tau)f^{(k+1)}(\eta\cup\tau)
\end{equation}
 \begin{equation}
 \label{eqn:step222}
 \mbox{for every $\sigma\in\Pinning_{k+1}$, } \sum_{\substack{\eta\in\Pinning_{k-1}: \\ \eta\subset\sigma}}\pi_{k-1}(\eta)\pi_{\eta,2}(\sigma\setminus\eta)
 = \pi_{k+1}(\sigma)
 \end{equation}
This completes the proof of the claim.
\end{proof}

Now we can complete the proof of \cref{lem:improved-technical}.

\begin{proof}[Proof of \cref{lem:improved-technical}]
Recall, our aim is to prove that for any $f^{(k+1)}:\Pinning_{k+1}\rightarrow\R$, the following holds:
\[
\tag{\ref{eqn:NEW-D}}
\Dirichlet_{\downup{k+1}}(f^{(k+1)}) \geq (2\gamma_{k-1}-1)\Dirichlet_{\downup{k}}(f^{(k)}).
\]

Fix $\tau\in\Pinning_{k-1}$.  Note that by the definition of $\gamma_j$ we have that the spectral gap of the local walk $Q_\tau$ satisfies $\gamma(Q_\tau)\geq\gamma_{k-1}$ and hence $\gamma(\updown{\tau,1})\geq\gamma_{k-1}/2$; therefore, by~\cref{lem:updown-downup}, we also have that
$\gamma(\downup{\tau,2})\geq\gamma_{k-1}/2$ which means that for any $f$:
\begin{equation}
\label{miss-222}
\Dirichlet_{\downup{\tau,2}}(f_\tau^{(2)})\geq \frac12\gamma_{k-1}\Var_{\pi_{\tau,2}}(f_\tau^{(2)}).
\end{equation}
\cref{lem:diff-var} yields that
    $\Dirichlet_{\downup{\tau,2}}(f^{(2)}_\tau) = \Var_{\pi_{\tau,2}}(f^{(2)}_\tau) - \Var_{\pi_{\tau,1}}(f^{(1)}_\tau)$, and hence combining with~\cref{miss-222}
we have the following:
    \begin{align}
\nonumber
\Var_{\pi_{\tau,2}}(f_\tau^{(2)})
&\geq 
\frac{1}{1-\gamma_{k-1}/2}\Var_{\pi_{\tau,1}}(f_\tau^{(1)})
\\
&\geq 
2\gamma_{k-1}\Var_{\pi_{\tau,1}}(f_\tau^{(1)}).
\label{missing-step}
\end{align}

We can now complete the proof of \cref{eqn:NEW-D}:
\begin{align*}
    \Dirichlet_{\downup{k+1}}(f^{(k+1)}) + 
      \Dirichlet_{\downup{k}}(f^{(k)})
       & = \Var_{\pi_{k+1}}(f^{(k+1)}) - \Var_{\pi_{k-1}}(f^{(k-1)})
       & \mbox{by \cref{lem:diff-var}}
       \\
       & = 
 \sum_{\tau\in\Pinning_{k-1}}\pi_{k-1}(\tau)\Var_{\pi_{\tau,2}}(f_\tau^{(2)})
       &\mbox{by \cref{claim:first-step}}
\\
&\geq  2\gamma_{k-1}\sum_{\tau\in\Pinning_{k-1}}\pi_{k-1}(\tau)\Var_{\pi_{\tau,1}}(f_\tau^{(1)})
& \mbox{ by \cref{missing-step}}
\\
& = 2\gamma_{k-1}\Dirichlet_{\downup{k}}(f^{(k)}),
& \mbox{ by \cref{claim:DDD}}
\end{align*}
which proves that $  \Dirichlet_{\downup{k+1}}(f^{(k+1)})\geq  (2\gamma_{k-1} -1 )\Dirichlet_{\downup{k}}(f^{(k)})$, and hence proves \cref{eqn:NEW-D} and \cref{lem:improved-technical}.

If we used the bound $1/(1-\gamma_{k-1}/2)$ in \eqref{missing-step} instead of $2\gamma_{k-1}$ then the statement of \cref{lem:improved-technical} has $\gamma_{k-1}/(2-\gamma_{k-1})$ in place of $(2\gamma_{k-1}-1)$; this matches the statement of \cite[Fact A.8 and Theorem~A.9]{CLV21}.
\end{proof}

\section{Optimal Mixing Time via Entropy Decay}
\label{sec:optimal}

In the previous section we established \cref{thm:SI-constant-relax} which says that spectral independence implies optimal $O(n)$ relaxation time of the Glauber dynamics.  In contrast, 
\cref{thm:SI-constant-mix} says that spectral independence and marginal boundedness imply optimal $O(n\log{n})$ mixing time.  Both \cref{thm:SI-constant-relax,thm:SI-constant-mix} were proved by Chen, Liu, and Vigoda~\cite{CLV21}, and our proofs of these theorems follow their proof approach. 

We begin by recalling the statement of \cref{thm:SI-constant-mix}.

\CLVoptimalmix*

 To obtain \cref{thm:SI-constant-mix}, we follow the same proof approach as in the proof of \cref{thm:SI-constant-relax} with the general objective of replacing the variance functional $\Var$ by the entropy functional $\Ent$.
We first recall the definition of entropy of a functional.

\begin{definition}[Entropy]The entropy of a function 
$f:\Omega\to\mathbb{R}_{\geq 0}$ with respect to $\mu$ is the following:
\[
    \Ent[f] := \mu[f \log f] - \mu[f] \log \left(\mu[f]\right)
     = \sum_{\sigma\in\Omega} \mu(\sigma)f(\sigma)\log(f(\sigma)) - \sum_{\sigma\in\Omega}\mu(\sigma)f(\sigma)\log[\sum_{\eta\in\Omega}\mu(\eta)f(\eta)],
\]
where $0\log{0}=0$.
\end{definition}

For a vertex $x$, denote the expected entropy at $x$ by:
\[
\mathrm{Ent}_{x}[f]=  \sum_{\tau}\mu(\tau)\mathrm{Ent}^{\tau}_{x}[f]  = \sum_{\tau}\mu(\tau)\Ent[F\mid \sigma(V\setminus\{x\})=\tau(V\setminus\{x\})].\]

We utilize entropy tensorization, which is the analog of approximate tensorization of variance, see~\cref{defn:approx-tensorization-var}.

\begin{definition}[Entropy Tensorization]\label{def:ent-tensorization}
A distribution $\mu$ with support $\Omega \subseteq \{0, 1\}^n$ satisfies approximate tensorization of 
entropy with constant $C>0$, if for 
all $f:\Omega\to\mathbb{R}_{\geq 0}$ we have that
\begin{align}\nonumber
\mathrm{Ent}(f)\leq C\sum_{v\in V}\mu\left( \mathrm{Ent}_v(f)\right).
\end{align}
\end{definition}
Note, the corresponding notion of variance tensorization is equivalent to the spectral gap of the (heat-bath) Glauber dynamics as the RHS in the above is equivalent to the Dirichlet form, see~\cref{lem:AT-gap}.

Entropy tensorization with constant $C$ yields an optimal mixing time bound for the Glauber dynamics by bounding the decay rate of entropy of the Glauber dynamics with respect to the Gibbs distribution, see~\cite{CMT,Caputo-notes}:
\begin{equation}\label{eq:mix-ET}
   \Tmix \leq  C n\log(\log(1/\mu^*)).  
\end{equation}       

As suggested above, to establish approximate tensorization of entropy we aim to replace $\Var$ by $\Ent$ in the proof of \cref{thm:SI-constant-relax}.  The first is that the proof of \cref{thm:SI-constant-relax} uses the spectral gap of the Glauber dynamics.  Recall the formulation of the spectral gap in the Poincar\'{e} inequality, see \cref{defn:Poincare}, where the spectral gap $\gamma = \min_f \Dirichlet(f)/\Var(f)$.  Moreover, for the down-up chain $\downup{i,j}$ one can replace the $\Dirichlet_{\downup{i,j}}(f)$ by $\Var_{\pi_i}(f)-\Var_{\pi_j}(f)$, see \cref{lem:diff-var}. 
Notice that in \cref{section-proof-relax} for the proof of \cref{thm:SI-constant-relax} we only considered the down-up chain and hence we can rephrase the proof only in terms of $\Var(f)$, without use of $\Dirichlet(f)$.  Therefore we can aim to replace all occurrences of $\Var(f)$ with $\Ent(f)$ in order to establish entropy decay instead of variance decay.

The only technical obstacle in replacing $\Var$ by $\Ent$ is in the application of \cref{thm:SI-mixing} in the proof of \cref{thm:SI-constant-relax} in \cref{sec:proof-relax}; in particular, see \cref{var-not-ent}.  \cref{thm:SI-mixing} established a lower bound on the spectral gap of the Glauber dynamics but it does not establish a bound on its entropy decay.  To get an analogous bound on entropy decay we need to use an additional property known as {\em marginal boundedness}, see \cref{defn:marg-bound}.  Using $b$-marginal boundedness, \cite[Lemma 4.2]{CLV21} proved a corresponding bound on $\Ent_T(f)$ in terms of $\sum_{v\in T}\mu(\Ent_v(f))$ which depends on $b$ and $|T|$:
\begin{equation}
\label{entropy-marginal}
\Ent_T(F|\tau) \leq C'\sum_{v\in T}\mu_T^\tau(\Ent_v[F|\tau]).
\end{equation}

In \cite[Lemma 4.2]{CLV21} they established $C'=C'(k,b)=\frac{3k^2\log(1/b)}{2b^{2k+2}}$ where $k=|T|$ (note $b$ typically depends on $\Delta$, and also depends on $\lambda$ or $\eta$).  Below, we provide a simple proof that $C'=C'(|T|,b,\Delta,\eta)= 3C(\eta,\Delta,b)k\log(1/b)$ where $C(\eta,\Delta,b)$ is a constant we obtain based on the constant $C(\eta,\Delta,\lambda)$ in \cref{thm:SI-constant-relax}.  
The bound in \cref{entropy-marginal} (with either constant $C'$) can be used in the analogous step of \cref{var-not-ent} to instead establish entropy decay, and then conclude
approximate tensorization of entropy of the Gibbs distribution (see~\cref{def:ent-tensorization} above) which implies an optimal $O(n\log{n})$ bound on the mixing time, via~\cref{eq:mix-ET}.

This outlines the approach of \cite{CLV21} to obtain optimal mixing time bounds.
Further improvements are obtained in \cite{CFYZ22,CE25}; these works upper bound the mixing time by $C(\eta)n\log{n}$ instead of $C(\eta,\Delta)n\log{n}$ based on a stronger form of spectral independence.
The work of Chen, Feng, Yin, and Zhang~\cite{CFYZ22} introduces the field dynamics which is presented in~\cref{sec:critical-point}.  In contrast, the work of Chen and Eldan~\cite{CE25} applies the framework of stochastic localization.

\begin{proof}[Proof of \cref{entropy-marginal}]
Consider the graph $G=(V,E)$ where $T\subseteq V$.  Let $k=|T|$.  We are running the Glauber dynamics on $T$ where for every $v\notin T$ the configuration at $v$ is fixed to $\tau(v)$.  Let $\mu=\mu_\tau$ denote the Gibbs distribution on $T$ with fixed configuration $\tau$ outside $T$.

From Diaconis and Saloff-Coste \cite{DiaSal96} (see the proof of Corollary A.4), we know the following inequality:
\begin{equation}
    \label{ent-var}
    \Ent_\mu(f) \leq \left[\frac{\log(1/\mu^*-1)}{1-2\mu^*}\right]\Var_\mu(\sqrt{f}),
\end{equation}
where $\mu^*=\min_{\sigma\in\Omega}\mu(\sigma)$.
From \cref{thm:SI-constant-relax} and \cref{lem:AT-gap}, we have approximate tensorization of variance:
\begin{equation}
    \label{AT-induct}
    \Var_\mu(f)\leq C(\eta,\Delta,\lambda)\sum_{v\in T} \Exp_{\mu}[\Var_v(f)].
\end{equation}
We will use the following elementary fact (e.g., see \cite{DiaSal96}), which we will prove momentarily:
\begin{equation}
\label{Z3}
    \Var_v(\sqrt{f}) \leq \Ent_v(f).
\end{equation}

Therefore, by combining \cref{ent-var} and \cref{AT-induct}, and since $1-b^k\geq\mu^*\geq b^k$ we have:
\begin{align*}
    \Ent_\mu(f)
   & \leq C(\eta,\Delta,\lambda)\left[\frac{\log(1/\mu^*-1)}{1-2\mu^*}\right] \sum_{v\in T} \Exp_{\mu}[\Var_v(\sqrt{f})] &\mbox{by \cref{ent-var} and \cref{AT-induct}}
   \\
   & \leq 3C(\eta,\Delta,\lambda)k\log(1/b)\sum_{v\in T} \Exp_{\mu}[\Var_v(\sqrt{f})]
   & \mbox{since $1-b^k\geq\mu^*\geq b^{k}$}
   \\
   &
   \leq 3C(\eta,\Delta,\lambda)k\log(1/b)\sum_{v\in T} \Exp_{\mu}[\Ent_v(f)]  & \mbox{by \cref{Z3}.}
\end{align*}
This completes the proof of \cref{entropy-marginal} with constant
$C'=3C(\eta,\Delta,\lambda)k\log(1/b)$.
Moreover, since
\[
\frac{b}{1-b}\leq\lambda\leq\frac{1-b}{b},
\]
the factor $C(\eta,\Delta,\lambda)$ may be replaced by
$C(\eta,\Delta,b)$, and hence $C'=3C(\eta,\Delta,b)k\log(1/b)$.

For completeness we now prove \cref{Z3}.
Without loss of generality, we assume $f$ satisfies $\Exp_{\mu}[f] = 1$, and that $f$ is positive-valued (i.e., $f:\Omega\rightarrow\R_{>0}$); the extension for nonnegative $f$ follows by a continuity argument.
We then have the following:
\begin{align}
\nonumber
    \Ent_\mu(f) & = \Exp_\mu[f\log(f)]
    \\
    \nonumber
    & = 2\Exp_\mu[f\log(\sqrt{f})] \\
    & \geq 2\Exp_\mu\left[f\left(1-\frac{1}{\sqrt{f}}\right)\right]
    \label{BBB1}
    \\
    \nonumber
    & = 2(1-\Exp_\mu[\sqrt{f}])
    \\ 
    &\geq 1-(\Exp_\mu[\sqrt{f}])^2 
\label{CCC2}
    \\ & = \Var_\mu(\sqrt{f}),
    \nonumber
\end{align}
where \cref{BBB1} follows from the fact that $\log(x)=\log(1+x-1)\leq x-1$ and hence $\log(1/x)\geq 1-x$, 
and \cref{CCC2} follows from $2(1-x)\geq 1-x^2$.
\end{proof}

\chapter{Matroids}
\label{sec:matroids}

We now utilize the spectral independence approach to establish fast mixing of a Markov chain for generating a random basis of a matroid.  
We will define the bases-exchange walk for walking on the set of bases of a given matroid.  We then present the beautiful result of Anari, Liu, Oveis Gharan, and Vinzant~\cite{ALOV19} proving fast mixing of the bases-exchange walk for any matroid. 

The bases-exchange walk was introduced by Feder and Mihail~\cite{FM}, who proved fast mixing for balanced matroids.  It was a major open problem to establish fast mixing for all matroids (this is a special case of a more general conjecture of Mihail and Vazirani, see \cite{FM,Mihail-polytopes,Kaibel}); this was resolved by the work of Anari, Liu, Oveis Gharan, and Vinzant~\cite{ALOV19} which we present here.  Independent work of Br{\"a}nd{\'e}n and Huh \cite{BH20} also implies fast mixing using completely different approaches. 

The work of~\cite{ALOV19} for matroids inspired the idea of spectral independence, as we presented earlier, which was subsequently introduced in~\cite{ALO20}.  In this monograph we present these results in the opposite order.

The general approach to prove fast mixing of the bases-exchange walk on the bases of a matroid is to apply the Random Walk Theorem (\cref{thm:RW}).  To bound the spectral gap of the local walks we will use a result of Oppenheim~\cite{Opp18} known as the Trickle-Down Theorem.  For a matroid of rank $r$, the Trickle-Down Theorem bounds the spectral gap of the local walk $Q_S$ for a pinning $S$ by induction on $|S|$.  The base case will be a bound on the gap of the local walk $Q_S$ when $k:=|S|=r-2$ which will correspond to rank-2 matroids; this is a significantly easier case to analyze as it corresponds to an unweighted walk.  We will then use the Trickle-Down Theorem to deduce a bound on the spectral gap of $Q_S$ for pinnings of size $k-1$ and hence any pinning $S$ by induction.  

\section{Matroids: Definitions}

A matroid $\M=(E,\cI)$ consists of the following:
\begin{itemize}
    \item {\em Ground Set} denoted by $E$;
    \item Collection $\cI$ of subsets of the ground set.  We refer to each $S\in \cI$ as an {\em independent set}.
\end{itemize}
To be a matroid we require the following properties:
\begin{description}
    \item[P1: Downward closure:]
    \label{matroid:P1}
    Every subset of an independent set is also an independent set.  This means that $\cI$ is downward closed, or equivalently that $\cI$ forms a simplicial complex.  To be clear, this property says that for every $S \in \cI$, if $T\subset S$ then  $T\in \cI$.
    \item[P2: Exchange property:] If $S\in \cI$ and $T\in \cI$, and if $|S|<|T|$ then there exists $e\in T\setminus S$ where $S\cup \{e\}\in \cI$.
\end{description}

Note, the empty set $\emptyset$ is always an independent set by the downward closure property.  A set $R\subset E$ which is {\em not} an independent set is called  dependent.

One important consequence of the exchange axiom is that all maximal  independent sets are of the same size (a maximal independent set is an independent set not contained in any larger independent set).

\begin{definition}
Consider a matroid $\M=(E,\cI)$. 
A {\em basis} is a maximal independent set.  In words, a basis is an independent set $S\in\cI$ where for all $e\not\in S$, we have that $S\cup \{e\}\not\in \cI$.  Note that the set of bases uniquely determines a matroid. We will use $\cF$ to denote the collection of bases of a matroid.

All bases of a matroid have the same size which is called the {\em rank} of the matroid $\M$.  Let $r=r(\M)$ denote the rank of matroid $\M$.
\end{definition}

\begin{definition}[Truncation]
Let $M$ be a matroid. Let $k$ be a positive integer. The set $\{S\in \cI: |S| \le k\}$
is called the truncation of $M$. Note that the truncation of a matroid is also a matroid.
\end{definition}

We state a few exercises that detail important, fundamental properties of matroids.

\begin{exercise}\label{exc1}
Let $\M=(E,\cI)$ be a matroid, and let $\cF$ denote its
collection of bases. Let $\cF^*$ be the set of complements, that is,
\[
\cF^*=\{E\setminus B:B\in\cF\}.
\]
Show that $\cF^*$ is the collection of bases of a matroid
$\M^*$ on the ground set $E$.  This matroid is called the
\emph{dual} of $\M$.
\end{exercise}

\begin{exercise}\label{exc2}
Let $\M=(E,\cI)$ be a matroid, and let $S\subseteq E$. 
Let $\cI''$ be the set of independent 
sets contained in $S$, that is,
\[
\cI''
=
\{T\subseteq S:T\in\cI\}.
\]
Show that
\[
\M|S=(S,\cI'')
\]
is a matroid.  This matroid is called the \emph{restriction} of
$\M$ to $S$.
\end{exercise}

\begin{exercise}\label{exc3}
Let $\M=(E,\cI)$ be a matroid, let $S\subseteq E$, and let $B$
be any basis of the restriction $\M|S$.  Let $\cI'$ be the set of subsets $T\subseteq E\setminus S$ such that $T\cup B$ is independent:
\[
\cI'
=
\{T\subseteq E\setminus S:T\cup B\in\cI\}.
\]
Show that $\cI'$ does not depend on the choice of the basis $B$
of $\M|S$, and that
\[
\M/S=(E\setminus S,\cI')
\]
is a matroid.  This matroid is called the \emph{contraction} of
$\M$ by $S$.
\end{exercise}

\begin{remark}
A matroid $\M'$ that can be obtained from $\M$ by a sequence of
restrictions and contractions is called a \emph{minor} of $\M$.
\end{remark}

\section{Matroids: Examples}

Here are some natural examples of matroids.

\begin{description}
    \item[{\bf Graphic matroid:}]
Let $G=(V,E)$ be a connected graph.  Let $\cI=\{S\subset E: (V,S) \mbox{ is acyclic}\}$ be the collection of all forests of $G$.  
    The collection $\M=(E,\cI)$ is a matroid. 
    The bases of $\M$ are maximal acyclic subgraphs.  Since $G$ is assumed to be connected then the bases of $\M$ are the spanning trees of $G$, and hence this is also called the spanning tree matroid.
\item[{\bf Transversal Matroid:}]
Let $G=(V,E)$ be a bipartite graph where  $V=L\cup R$ is the bipartition.  Let $\cI$ denote the subsets of $L$ that are the endpoints of some matching of $G$.  Note, $S\in \cI$ is a set of vertices and it can correspond to multiple matchings.  

The collection $\M=(L,\cI)$ is a matroid.  The bases are the subsets of $L$ covered by a maximum matching.  
\item[{\bf Linear matroid:}]
Let $V = \{v_1,\ldots, v_n\}$ be vectors in some vector space.
Let 
\[
\cI = \{S \subseteq V : S \text{ is linearly independent} \}.
\]
Then $\M=(V,\cI)$ is a matroid.
\end{description}

\section{Bases-exchange Walk}

Our central question is, for a given matroid $\M=(E,\cI)$, can we generate a basis, uniformly at random from the set $\cF$ of all bases, in time polynomial in $n=|E|$.  We are given the matroid in implicit form, so we can efficiently check whether a given subset $S\subset E$ is independent or dependent.
We will use the following natural Markov chain known as the bases-exchange walk to generate a random basis.

The bases-exchange walk is the following Markov chain $(B_t)$ on the set of bases of a given matroid $\M$.
Let $\cF$ denote the collection of bases of a matroid $\M$.  From $B_t\in\cF$, the transitions $B_t\rightarrow B_{t+1}$ are defined as follows:
\begin{itemize}
    \item Choose an element $e\in B_t$ uniformly at random.
    \item Let $F=\{f\in E: B_t\cup \{f\}\setminus \{e\}\in\cI\}$ denote the set of elements that we can add to $B_t\setminus e$ while remaining independent.
    \item Choose an element $f\in F$ uniformly at random.  Let $B_{t+1} = B_t\cup \{f\}\setminus \{e\}$. 
\end{itemize}

Note the element $e$ is in $F$ and hence $P(B_t,B_t)>0$; thus, the chain is aperiodic. 

\begin{exercise}
Using the exchange-axiom, prove by induction that the chain is irreducible.
\end{exercise}

Therefore, the bases-exchange walk is an ergodic Markov chain.  Since the chain is symmetric, its unique stationary distribution is uniform over $\cF$, the collection of bases of the matroid $\M$.
In the terminology from the spectral independence section, the bases-exchange walk is equivalent to the down-up walk on the bases of a matroid. 

\section{Main Theorem: Fast Mixing of Bases-Exchange Walk}

The main result is a bound on the spectral gap of the bases-exchange walk.

\matroidmain*

By proving entropy contraction, \cite{CGM19} improves the mixing time bound to $O(r(M)\log{\log(|\cF|)}) = O(r(M) \log r(M) + r(M) \log{\log{n}})$.
Furthermore, the $\log\log n$ term was subsequently removed in \cite{ALOVV21}, achieving $O(r(M) \log r(M))$ mixing time.
This is tight because we need $\Omega(r(M) \log r(M))$ time to replace every element due to coupon collector.

\section{Application: Reliability}

 Before proving fast convergence of the bases-exchange walk we discuss an important application for sampling and approximate counting of bases of a matroid.
 
An interesting application is to compute reliability.
Let $X\subseteq [n]$ be a random set where every element is included independently with probability $p$.
The {\em reliability} is defined as $\Pr(X\supseteq \text{basis})$,
and the {\em unreliability} is defined as $\Pr(X\not\supseteq \text{basis})$.

For a graphic matroid these problems are known as network reliability/unreliability as the problem corresponds to whether $X$ is connected/disconnected.  Karger~\cite{Karger} presented an $\fpras$ for network unreliability,
and Guo and Jerrum~\cite{GJ19} presented an $\fpras$ for network reliability.

\cref{thm:matroid-main} yields an $\fpras$ for the reliability of a matroid.
In particular, rapid mixing of the bases-exchange walk yields an efficient sampling algorithm for generating a basis of a matroid (almost) uniformly at random;
the classical result of \cite{JVV86} then yields 
an $\fpras$ for computing the number of bases.

Consider a matroid $\M$ for which we would like to compute the reliability, which is the probability that the random subset $X$ contains a basis of $\M$.
Let $\M^*$ denote the dual of $\M$, see \cref{exc1}; note $\M^*$ is a matroid.
Note,
\begin{align*}
    \Prob{X\supseteq \text{basis of $\M$}} = 
    \Prob{\overline{X}\subseteq \text{basis of $\M^*$}}=
     \sum_k (1-p)^k p^{n-k} \text{\#\{independent set of $\M^*$ of size $k$\}},
\end{align*}
and we can obtain an $\fpras$ for the number of independent sets (of a matroid $\M^*$) of size $k$ since the truncation of a matroid is a matroid.  This yields an $\fpras$ for the reliability of the matroid $\M$.

\chapter{Trickle-Down Theorem}
\label{sec:trickledown}

Here we present Oppenheim's Trickle-Down Theorem~\cite{Opp18} which is the new tool needed to prove \cref{thm:matroid-main}.  
To prove \cref{thm:matroid-main} we follow the proof approach of \cite{ALOV19} which relies on Oppenheim's Trickle-Down Theorem presented in \cref{thm:trickle-down}.  

The conceptual
significance of Oppenheim's Trickle-Down Theorem~\cite{Opp18} in the matroid application is that it
propagates an optimal spectral-gap bound from the smallest
nontrivial local walks to all local walks.  For matroids, these
smallest local walks correspond to rank~2 matroids, whose
associated graphs are complete multipartite graphs and whose
spectral gaps can be bounded by an elementary linear-algebraic
argument.  Combined with the Random Walk Theorem (\cref{thm:RW}), the
Trickle-Down Theorem therefore yields the fast-mixing result for
the bases-exchange walk through Markov chain and linear-algebraic
methods. 
This provides a strikingly accessible alternative to approaches
through strongly log-concave or Lorentzian polynomials, whose
matroidal theory is rooted in deep Hodge-theoretic and
algebraic-geometric results; see \cite{AHK18,BH20}.  

We present Oppenheim's Trickle-Down Theorem in the general context of simplicial complexes; this enables one to apply it both in the context of spin systems and for bases of a matroid.

\section{Simplicial Complexes}

Let $\Lambda$ be a finite set.  Let $\Omega$ denote a collection of subsets of $\Lambda$ that is 1) closed 
under taking subsets (that is, if $S\in\Omega$
and $Z\subseteq S$ then $Z\in\Omega$), and 2) all 
maximal (under inclusion) sets in $\Omega$
have the same size $r$. Such a pair $(\Lambda,\Omega)$
is called a pure abstract simplicial complex. 

For spectral independence on a (binary) spin system, let $\Lambda=V\times\{0,1\}$ for an input graph $G=(V,E)$.  Here $r=|V|$ while $|\Lambda|=2r$.  The distribution $\mu$ is the Gibbs distribution on assignments $\sigma\in \{0,1\}^n$. Now we construct $\Omega$. The maximal sets in $\Omega$ are the support of $\mu$. Defining the distributions $\pi_k$
as in \cref{sec:chains}, the union of the supports of $\pi_0,\dots,\pi_r$ yields $\Omega$.

For the example of bases of a matroid, let $\Lambda=E$, the ground set of the matroid $\M=(E,\cI)$.  Then let $\Omega=\cI$ be the independent sets of $\M$.  The maximal sets in $\Omega$ are the bases of $\M$ and hence $r$ is the rank of the matroid.  All bases in $\Omega$ have the same weight, thus we set $\mu$ to be the uniform distribution over the set of bases $\cF$. Defining the distributions $\pi_k$
as in \cref{sec:chains}, the union of the supports of $\pi_0,\dots,\pi_r$ is $\Omega$.

More generally, let $\mu$ be a distribution on sets of size $r$. Define the distributions $\pi_k$ and the up/down chains exactly as in \cref{sec:chains}.  
The union of the supports of distributions $\pi_0,\pi_1,\dots,\pi_r$ is a pure abstract simplicial complex.

\section{Trickle-Down Statement}
\label{sub:trickledown-statement}

Define the distributions $\pi_k$ and the up/down chains exactly as in \cref{sec:chains}, and define the local walks $Q_S$ as in \cref{sec:local-walks}. Let $(\Lambda,\Omega)$ be the simplicial complex that is 
the union of the supports of $\pi_0,\ldots,\pi_r$. We can now state the Trickle-Down theorem.

\begin{theorem}[\cite{Opp18}]
\label{thm:trickle-down}
    Consider $S\in\Omega$.  Let $i=|S|$ and assume $0\leq i<r-2$.  Let $\gamma_i=\gamma(Q_S)$.  Assume that $Q_S$ is irreducible and hence $\gamma_i>0$.  Suppose there exists $\gamma_{i+1}>0$ where for all $Z\in\Omega$ such that  $S\subset Z$ and $|Z|=|S|+1$, 
    \begin{equation}
        \label{TD:assume}
        \gamma(Q_{Z})\geq \gamma_{i+1}.
    \end{equation} Then the following holds:
    \[  \gamma_i \geq 2 - \frac{1}{\gamma_{i+1}}.\]
    \end{theorem}

As a consequence, if there is optimal spectral gap, namely, $\gamma(Q_{S'})\geq 1$, for all $S'\in\Omega$ where $|S'|=r-2$, then by induction we obtain, via the Trickle-Down Theorem, that there is optimal spectral gap of $\gamma(Q_S)\geq 1$ for all $S\in\Omega$.

    \begin{corollary}[Trickle-Down Without Loss]
    \label{cor:trickle-without-loss}
        If for all $S'\in\Omega$ with $|S'|=r-2$ we have $\gamma(Q_{S'})\geq 1$, and if for all $S\in\Omega$ we have that $Q_S$ is irreducible, then for all $S\in\Omega$ it holds that 
        \[ \gamma(Q_S)\geq 1.
        \]
    \end{corollary}

\section{Rapid Mixing of Bases-Exchange Walk: Proof of \texorpdfstring{\cref{thm:matroid-main}}{Theorem 1.16}}
\label{sec:basesexchange}

The proof of \cref{thm:matroid-main} uses the Trickle-down Theorem (\cref{thm:trickle-down}), which bounds the spectral gap of the local walks, and the Random Walk Theorem (\cref{thm:RW}), which is the local-to-global theorem.  

Our goal is to bound the spectral gap of the bases-exchange walk using the Random Walk Theorem.  To that end we need to bound the spectral gap of the local walk (this is the up-down walk $\updown{1}$) for every pinning (or link in the terminology of simplicial complexes). The pinnings are defined for every independent set~$S$.  In particular, for an independent set~$S$ the pinning is defined as $\{B\setminus S : S\subseteq B, B\in \cF\}$, which is the collection of bases of the contraction $\M/S$. \cref{exc3} establishes that the contraction of a matroid is also a matroid.

These local walks for a particular pinning 
are weighted because they are projections of the original simplicial complex.
 To handle these weighted local walks we utilize the Trickle-down Theorem.  The Trickle-down Theorem says that if we have an ``optimal'' bound on the spectral gaps for all local walks on pinnings of size $r-2$ then we obtain an optimal bound on the spectral gap for the local walks for all pinnings, and then we can apply the Random Walk Theorem.  This greatly simplifies matters because the local walks on pinnings of size $r-2$ are unweighted,
and hence they simply correspond to (unweighted) rank~2 matroids.

\begin{lemma}\label{lem:spectralgap-rank2}
For a rank 2 matroid $\M$, the spectral gap of the local walk satisfies
\[ 
\gamma(Q_{\M}) \geq 1.
\]
\end{lemma}

\begin{proof}
Consider a rank 2 matroid $\M$.  Consider a graph $G_\M$ whose vertex set is the set of  elements 
of the matroid $\M$ and whose edges are the bases of the matroid $\M$. We will argue that this graph consists of a 
complete multipartite graph plus some isolated vertices.
The isolated vertices $v$ are those for which $\{v\}$ is not independent. Remove the isolated vertices.  
Now, by the Exchange Axiom, for any three vertices $i,j,k$ with $\{j,k\}\in E$, either $\{i,j\}\in E$ or $\{i,k\}\in E$. Equivalently, if $\{i,j\}\not\in E$ and $\{i,k\}\not\in E$ then 
$\{j,k\}\not\in E$. This means that we have a multipartite graph (where the partitions are maximal sets of pairwise dependent elements). 

The adjacency matrix $A$ of a complete multipartite graph is the all ones matrix (that has only one non-zero eigenvalue) minus a block diagonal matrix, so it has the second eigenvalue $\le 0$. 
More precisely, we can write $A = \mathbf{1} \mathbf{1}^T - \sum_i \mathbf{1}_{S_i} \mathbf{1}_{S_i}^T\preceq \mathbf{1} \mathbf{1}^T$, where $S_i$'s are the parts of the graph. Note that $\mathbf{1} \mathbf{1}^T$ has rank $1$ and hence all the eigenvalues except one are $0$.
(Adding the isolated vertices does not change the non-zero part of the spectrum.)

The local walk $Q_\M$ is equivalent to the random walk on the graph $G_\M$. The transition matrix of the random walk on $G_\M$ is the matrix $D^{-1} A$ where $D$ is the diagonal matrix with the degrees.
Note, by \cref{lem:sylvesters}, $D^{-1} A$ has the same signature as $A$; hence, it has all eigenvalues, except the top eigenvalue $1$, less than or equal to $0$. This completes the proof.
\end{proof}

\begin{lemma}\label{lem:matroid-local}
Let $\M=(E,\cI)$ be a matroid of rank $r$. For any $S\in\cI$, $|S|\leq r-2$, we have that the spectral gap of the local walk $Q_S$ satisfies $\gamma(Q_S)\geq 1$.
\end{lemma}

\begin{proof}[Proof of~\cref{lem:matroid-local}]
This lemma follows from \cref{lem:spectralgap-rank2} and \cref{cor:trickle-without-loss} (since the pinnings of size $r-2$ correspond to rank $2$ matroids).
\end{proof}

We can now prove the main theorem establishing rapid mixing of the bases-exchange walk.

\begin{proof}[Proof of Theorem~\ref{thm:matroid-main}]
 The random walk theorem \cref{thm:RW} implies that the spectral gap $\gamma$ for the transition matrix $P$ of the bases-exchange walk satisfies:
\begin{equation}
    \label{eqn:RW-thm222}
 \gamma(P) \geq \frac{1}{r}\prod_{k=0}^{r-2}\gamma_k,
\end{equation}
where $$\gamma_k=\min_{S\in\cI:|S|=k}\gamma(Q_{S})$$ and $\gamma(Q_S)$ is the spectral gap for the local walk $Q_S$, which by \cref{lem:matroid-local} we have that $\gamma_k\geq 1$ for every~$k$.
\end{proof}

\section{Connections to Spectral Independence and Log-Concavity}
\label{sub:logconcave}

In this section, we explain the equivalence between 0-weak spectral independence and log-concavity of the generating polynomial; this connection is deeply rooted in the study of strongly log-concave homogeneous polynomials (also called Lorentzian polynomials) \cite{ALOV19,BH20}.

Recall, the definition of the modified influence matrix $\widetilde{\Psi}$ in \cref{defn:correlation-matrix}, which is slightly different than the influence matrix $\Psi$ in \cref{defn:inf-matrix}; the former $\widetilde{\Psi}$ leads to the notion of weak spectral independence (see \cref{defn:weak-SI}) and the latter $\Psi$ leads to spectral independence (see \cref{defn:SI}).  Both weak spectral independence and spectral independence upper bound the maximum eigenvalue $\uplambda_{\max}$ of the corresponding matrix but do so for all valid pinnings (which corresponds to all induced subgraphs in the case of the hard-core model on independent sets).

We first observe that the local walks having an optimal spectral gap is equivalent to optimal spectral independence.  

\begin{remark}
\label{rem:weakSI-mixing}
The following are equivalent by \cref{lem:QandPsi}:
\begin{enumerate}
    \item For every pinning $S$, the local walk $Q_S$ satisfies $\gamma(Q_S)\geq 1$.
    \item $0$-spectral independence holds. 
\end{enumerate}
\end{remark}

Recall, $0$-spectral independence implies $0$-weak spectral independence.  Now we relate $0$-weak spectral independence with log-concavity.

Given a distribution $\mu$ on $\{0,1\}^n$ consider the following generating polynomial:
$$
f(x_1,\dots,x_n) = \sum_{\sigma\in \{0,1\}^n} \mu(\sigma) \prod_{i=1}^n x_i^{\sigma_i}.
$$

\begin{definition}[log-concavity at a point]
We say that $f$ is log-concave at a point $x$ if $(\nabla^2 \log f)(x)\preceq 0$.
\end{definition}

As noted above, $0$-weak spectral independence is defined as $\uplambda_{\max}(\widetilde{\Psi}_\tau)\leq 1$ for all pinnings $\tau$.  The following lemma shows that $\uplambda_{\max}(\widetilde{\Psi})\leq~1$ is equivalent to log-concavity of $f$.

\begin{theorem}
    Log-concavity of $f$ at the point $(1,\dots,1)$ is equivalent to the influence matrix $\widetilde{\Psi}$ having $\uplambda_{\max}(\widetilde{\Psi})\leq~1$. 
\end{theorem}

\begin{proof}

For $1\leq i\neq j\leq n$ we have
\begin{align}
x_i x_j (\nabla^2 \log f(1,\dots,1))_{ij} & =  x_i x_j \left(\frac{(\frac{\partial^2}{\partial x_i \partial x_j} f ) }{f}(1,\dots,1) - \frac{ (\frac{\partial}{\partial x_i} f)(\frac{\partial}{\partial x_j} f)  }{f^2}(1,\dots,1) \right) 
\nonumber
\\ 
\label{eq:off-diag}
& = \Pr_{\sigma\sim\mu}(\sigma(i)=\sigma(j)=1) - \Pr_{\sigma\sim\mu}(\sigma(i)=1)\Pr_{\sigma\sim\mu}(\sigma(j)=1),
\end{align}
and for $1\leq i\leq n$ we have
\begin{align}\label{eq:diag}
x_i^2 (\nabla^2 \log f(1,\dots,1))_{ii} & = - x_i^2 \left(\frac{ \frac{\partial}{\partial x_i} f }{f}(1,\dots,1) \right)^2 = - \Pr_{\sigma\sim\mu}(\sigma(i)=1)^2.
\end{align}

Recall from \cref{cov-defn} the definition of the covariance matrix:
   \begin{align*}
        \Cov_{\mu}(i,j) 
        &= 
               \Pr_{\sigma\sim\mu}[\sigma(i)=\sigma(j)=1] 
        - \Pr_{\sigma\sim\mu}[\sigma(i)=1]\cdot\Pr_{\sigma\sim\mu}[\sigma(j)=1].
\end{align*}
Let $\widetilde{D}={\mathrm{diag}}(\Exp(\mu))$ be the diagonal matrix with entries  $\widetilde{D}(i,i)=\mu(\sigma(i)=1)$.
Hence $(\nabla^2\log f)(1,\dots,1)\preceq0$ is equivalent to
\begin{equation}
\label{cov-diag}
\Cov_\mu-\widetilde D
=
\diag(x_i)
(\nabla^2\log f(1,\dots,1))
\diag(x_i)
\preceq0.
\end{equation}

Recall from \cref{CM-semidefinite} that $\uplambda_{\max}(\widetilde{\Psi})\leq 1$ is equivalent to $\Cov_\mu\preceq  \widetilde{D}$.  Therefore, \cref{cov-diag} establishes the condition \cref{cond-weak-max} in the definition of $0$-weak spectral independence, see \cref{defn:weak-SI}.
\end{proof}

We refer interested readers to the works \cite{ALOV19,BH20} for more discussion on (strongly) log-concave polynomials
(for example, see \cite{BH20} for  interesting connections between matroids and the support of $\mu$ with strongly log-concave generating polynomials).

\section{Proof of the Trickle-Down Theorem (\texorpdfstring{\cref{thm:trickle-down}}{Theorem 8.1})}
\label{sec:proof-trickledown}

In this section, we complete the proof of the Trickle-Down Theorem (\cref{thm:trickle-down}); this is the only remaining task to complete the proof of fast mixing of the bases-exchange walk (\cref{thm:matroid-main}).

\begin{remark}
    For an alternative proof of the Trickle-Down Theorem (\cref{thm:trickle-down}) in the language of covariance matrices (see \cref{cov-SI-connection}), see \cite{Frederic-note,AKV24}. We highlight that \cite{AKV24} established a trickle-down theorem for general localization schemes, which can be applied to the analysis of the field dynamics for the hard-core model; see \cref{sec:hard-core-random} for more details.
    \end{remark}

For notational convenience, throughout the remainder of this section
we denote the ground set $\Lambda$ by $E$.

Fix $S\in\Omega$ with $i:=|S|<r-2$.
Let $\pi_S=\pi_{S,1}$, and hence $\pi_S$ is a distribution over
$a\in E\setminus S$.
Recall that the local walk $Q_S$ has stationary distribution $\pi_S$.

We will use the following technical lemma.

\begin{lemma}
\label{lem:TD111}
Let $S\in\Omega$ with $i:=|S|<r-2$.
For any function $f:E\rightarrow\R$,
\begin{align}
\label{lem:TD1-Dir}
\Dirichlet_{Q_S}(f) & = \sum_{a\in E\setminus S}\pi_S(a)\Dirichlet_{Q_{S\cup a}}(f)
 \\
 \label{lem:TD1-Var}
 \Exp_{\pi_S}(f) 
& = \sum_{a\in E\setminus S}\pi_S(a)
\Exp_{\pi_{S\cup a}}(f).
\end{align}
\end{lemma}

We note that \cref{lem:TD1-Var} can be intuitively explained as follows: To find the expectation of $f(T)$ where $T \sim \pi_S$, we can generate $T$ by first choosing $a \sim \pi_S$ and then sampling $T \sim \pi_{S \cup a}$, and thus \cref{lem:TD1-Var} follows immediately from the law of total expectation. \cref{lem:TD1-Dir} can be established in a similar fashion. The complete proof of \cref{lem:TD111} can be found in \cref{subsec:pf-TD111}.

We can now prove the Trickle-Down Theorem.

\begin{proof}[Proof of \cref{thm:trickle-down}]
Our goal is to bound $\gamma_i=\gamma(Q_S)$ which is the spectral gap for the 
local walk $Q_S$ with pinning $S\subset E$.  Recall, 
for all $f:E\rightarrow\R$,
\begin{equation}
    \label{TD:dirichlet}
\gamma_i = \min_{f} \frac{\Dirichlet_{Q_S}(f)}{\Var_{\pi_S}(f)},
\end{equation}
where the minimization is over all $f:E\rightarrow\R$ where $f$ is not constant on $E\setminus S$ and hence $\Var_{\pi_S}(f)\neq 0$.
Let $f^*$ be the function $f$ which achieves the minimum in \cref{TD:dirichlet}, and hence
\[ \gamma_i\Var_{\pi_S}(f^*) = \Dirichlet_{Q_S}(f^*).
\]

The minimizer $f^*$ of~\eqref{TD:dirichlet} has to be an ``eigenvector'' in the following sense (this follows more directly by 
viewing Dirichlet forms through eigenvalues; we include
a self-contained proof in \cref{subsec:pf-TD111}).
\begin{lemma}
\label{TD:blue}
For any $a\in E\setminus S$,
\[
    \Exp_{\pi_{S\cup a}}(f^*) = (1-\gamma_i)f^*(a).
\]    
\end{lemma}

Without loss of generality, in \cref{TD:dirichlet} we can restrict attention to functions $f$ where $\Exp_{\pi_S}[f] = 0$ (since translating $f$ by a constant function does not change the numerator and the denominator), and hence for $f^*$ we can also assume that $\Exp_{\pi_S}[f^*]=0$.

\begin{align*}
    \gamma_i \Var_{\pi_S}(f^*) & = \Dirichlet_{Q_S}(f^*) 
\\
& = \sum_{a\in E\setminus S}\pi_S(a)\Dirichlet_{Q_{S\cup a}}(f^*) & \mbox{by \cref{lem:TD1-Dir}}
\\
& \geq \gamma_{i+1}\sum_{a\in E\setminus S}\pi_{S}(a)\Var_{\pi_{S\cup a}}(f^*)
& \mbox{by \cref{TD:assume}}
\\
& = \gamma_{i+1}\sum_a\pi_{S}(a)\Exp_{\pi_{S\cup a}}[(f^*)^2]
- \gamma_{i+1}\sum_a\pi_{S}(a)[\Exp_{\pi_{S\cup a}}(f^*)]^2
\\
& = \gamma_{i+1}\sum_a\pi_{S}(a)\Exp_{\pi_{S\cup a}}[(f^*)^2]
- \gamma_{i+1}(1-\gamma_i)^2\sum_a\pi_{S}(a)(f^*(a))^2
&\mbox{by \cref{TD:blue}}
\\
& = \gamma_{i+1}\sum_a\pi_{S}(a)\Exp_{\pi_{S\cup a}}[(f^*)^2]
- \gamma_{i+1}(1-\gamma_i)^2\Exp_{\pi_S}[(f^*)^2]
\\
& = \gamma_{i+1}\Exp_{\pi_S}[(f^*)^2]
- \gamma_{i+1}(1-\gamma_i)^2\Exp_{\pi_S}[(f^*)^2]&\mbox{by \cref{lem:TD1-Var}}
\\
& = \gamma_{i+1}\Var_{\pi_{S}}(f^*)
- \gamma_{i+1}(1-\gamma_i)^2\Var_{\pi_S}(f^*).
\end{align*}

Now, since $\Var_{\pi_S}(f^*)>0$ we have
$$
\gamma_i\geq \gamma_{i+1}(1-(1-\gamma_i)^2) = \gamma_{i+1}\gamma_i (2- \gamma_i).
$$
Since $\gamma_i>0$ (we assumed that $Q_S$ is irreducible) we have
$$
1 \geq \gamma_{i+1} (2- \gamma_i),
$$
which is equivalent to:
$$
\gamma_i \geq 2 - \frac{1}{\gamma_{i+1}},$$
which completes the proof.
\end{proof}

\section{Proofs of Technical Lemmas}
\label{subsec:pf-TD111}

It remains to prove \cref{lem:TD111} and \cref{TD:blue}.

\begin{proof}[Proof of \cref{lem:TD111}]
Recall that $i=|S|$ where $i<r-2$, and that $\mu$ is a distribution on
$r$-element subsets of the ground set~$E$ (e.g., in the specific case of matroids, $\mu$ is the uniform distribution over the bases of the matroid).
For $T\subset E$, let
$\mu(T)=\Pr_{B\sim\mu}(T\subseteq B)$ denote the
probability that a set sampled from $\mu$ contains~$T$.
We begin with two elementary facts. First, observe that
\begin{equation}
    \label{matroid:basic-fact}
\pi_{S}(a)
= \pi_{S,1}(a)
= \frac{\mu(S\cup a)}{(r-i)\mu(S)}
= \frac{\pi_{i+1}(S\cup a)\binom{r}{i+1}}
       {(r-i)\pi_i(S)\binom{r}{i}}
= \frac{\pi_{i+1}(S\cup a)}
       {(i+1)\pi_i(S)}.
\end{equation}
Second, for every integer $0\leq k<r$ and every
$T\in\Omega$ with $|T|=k$, applying
\cref{matroid:first-step} with upper level $k+1$, we have
\begin{equation}
\label{matroid:sum-a}
\sum_{a\in E\setminus T}\pi_{k+1}(T\cup a)
=(k+1)\pi_k(T).
\end{equation}

We can now proceed with the proof of \cref{lem:TD1-Var}:
\begin{align*}
\Exp_{\pi_S}(f)  & = \sum_{b\in E\setminus S}\pi_S(b)f(b)
 \\
& = \sum_{b\in E\setminus S}\frac{\pi_{i+1}(S\cup b)}{(i+1)\pi_i(S)} f(b)
&\mbox{by \cref{matroid:basic-fact}}
 \\
& =  \sum_{b\in E\setminus S} \sum_{a\in E\setminus (S\cup b)} \frac{\pi_{i+2}(S\cup b\cup a)}{(i+1)(i+2)\pi_i(S)} f(b)
&\mbox{by \cref{matroid:sum-a} with $T=S\cup b$}
 \\
& =  \sum_{a\in E\setminus S} \sum_{b\in E\setminus (S\cup a)}  \frac{\pi_{i+2}(S\cup b\cup a)}{(i+2)\pi_{i+1}(S\cup a)} \frac{\pi_{i+1}(S\cup a)}{(i+1)\pi_i(S)} f(b)
\\
& =  \sum_{a\in E\setminus S} \pi_S(a) \sum_{b\in E\setminus (S\cup a)}  \frac{\pi_{i+2}(S\cup b\cup a)}{(i+2)\pi_{i+1}(S\cup a)}  f(b)
&\mbox{by \cref{matroid:basic-fact}}
\\
& =  \sum_{a\in E\setminus S} \pi_S(a) \sum_{b\in E\setminus (S\cup a)}  \pi_{S\cup a}(b) f(b)
&\mbox{by \cref{matroid:basic-fact}}
\\
& = \sum_{a\in E\setminus S}\pi_S(a)
\Exp_{\pi_{S\cup a}}(f).
\end{align*}

We can establish \cref{lem:TD1-Dir} in a similar manner.
For distinct $b,c\in E\setminus S$, we have
\begin{equation}
\label{eq:inter-distinct=bc}
\pi_{S,2}(\{b,c\})
=
\pi_S(b)\pi_{S\cup b}(c)
+
\pi_S(c)\pi_{S\cup c}(b)
=
2\pi_S(b)\pi_{S\cup b}(c)
=
2\pi_S(c)\pi_{S\cup c}(b),
\end{equation}
where the first equality follows by considering the two possible orders
in which the unordered pair $\{b,c\}$ can be generated. 
In the following calculations, all sums over $b,c$ are restricted to $b\neq c$.

First, the left-hand side of \cref{lem:TD1-Dir} is given by the following:
\begin{align*}
\Dirichlet_{Q_S}(f)
&= \frac12
\sum_{\substack{b,c\in E\setminus S:\\ b\neq c}}
\pi_S(b)Q_S(b,c)(f(b)-f(c))^2
\\
&= \frac12\sum_{b,c\in E\setminus S}
\pi_S(b)(i+2)
\updown{i+1}(S\cup b,S\cup c)(f(b)-f(c))^2
&\mbox{by \cref{step:rem:local-downup}}
\\
&= \frac12
\sum_{b,c\in E\setminus S}
\pi_S(b)
\frac{\pi_{i+2}(S\cup b\cup c)}
{(i+2)\pi_{i+1}(S\cup b)}
(f(b)-f(c))^2
&\mbox{by \cref{eqn:up-down}}
\\
&= \frac12
\sum_{b,c\in E\setminus S}
\pi_S(b)\pi_{S\cup b}(c)(f(b)-f(c))^2
\\
&= \frac14\sum_{b,c\in E\setminus S}
\pi_{S,2}(\{b,c\})(f(b)-f(c))^2
&\mbox{by \cref{eq:inter-distinct=bc}.}
\end{align*}

Second, the right-hand side of \cref{lem:TD1-Dir} can be calculated in the following manner:
\begin{align*}
    \sum_{a\in E\setminus S}\pi_S(a) \Dirichlet_{Q_{S\cup a}}(f)
&= \frac14\sum_{a\in E\setminus S}\pi_S(a) \sum_{b,c\in E\setminus (S \cup a)}\pi_{S \cup a,2}(\{b,c\}) (f(b)-f(c))^2 
\\ 
&= \frac14\sum_{b,c\in E\setminus S} \pi_{S,2}(\{b,c\}) (f(b)-f(c))^2 \sum_{a \in E\setminus (S\cup \{b,c\})} \pi_{S \cup \{b,c\}}(a)
\\
&= \frac14\sum_{b,c\in E\setminus S} \pi_{S,2}(\{b,c\}) (f(b)-f(c))^2,
\end{align*}
where the first equality follows from the preceding identity for the LHS with $S$ replaced by $S\cup a$, the second equality follows from
\[
\pi_S(a)\pi_{S\cup a,2}(\{b,c\})
=
\pi_{S,2}(\{b,c\})\pi_{S\cup\{b,c\}}(a),
\]
and the final equality follows from
\[
\sum_{a\in E\setminus(S\cup\{b,c\})}
\pi_{S\cup\{b,c\}}(a)=1.
\]
\end{proof}

We now prove \cref{TD:blue}, which will complete the proof of the Trickle-Down Theorem (\cref{thm:trickle-down}).

\begin{proof}[Proof of \cref{TD:blue}]
Since $Q_S$ is reversible, the variational characterization of
the spectral gap implies that the mean-zero minimizer $f^*$ in
\eqref{TD:dirichlet} is an eigenvector of $Q_S$ with eigenvalue
$1-\gamma_i$.  Hence, for every $a\in E\setminus S$,
\[
\sum_{b\in E\setminus(S\cup\{a\})}
Q_S(a,b)f^*(b)
=
(1-\gamma_i)f^*(a).
\]
By the definition of the local walk,
\[
Q_S(a,b)
=
\pi_{S\cup a}(b)
=
\frac{\pi_{i+2}(S\cup a\cup b)}
     {(i+2)\pi_{i+1}(S\cup a)}.
\]
Therefore,
\[
\Exp_{\pi_{S\cup a}}(f^*)
=
\sum_{b\in E\setminus(S\cup\{a\})}
\pi_{S\cup a}(b)f^*(b)
=
(1-\gamma_i)f^*(a),
\]
which proves the lemma.
\end{proof}

\chapter{Methods for Establishing Spectral Independence}
\label{sec:methods}

In this chapter, we present general techniques for establishing spectral independence, using the hard-core model over independent sets as an example and thereby establish \cref{thm:main-SI-hard-core-uniq}, which we restate here for convenience.

\SIhardcore*

As we saw previously, spectral independence immediately implies $O(n\log{n})$ mixing time of the Glauber dynamics for constant-degree graphs.  
We focus on three approaches. In \cref{sec:Weitz}, we show that correlation decay approaches as used in Weitz's algorithm~\cite{Wei06} imply spectral independence.
In \cref{sec:stability}, we show that stability of the partition function, so-called zero-freeness, also implies spectral independence; such conditions were used in the approximate counting algorithm introduced by Barvinok~\cite{Bar17book}.
Finally, in \cref{sec:CI} we discuss the coupling independence method introduced in \cite{CZ23} which gives a simple and intuitive way to establish spectral independence.
In \cref{sub:coupling-SI}, we show that an optimal relaxation time for a local chain implies spectral independence, implying the necessity of spectral independence for optimal mixing of the Glauber dynamics.

\section{Preliminaries}
\label{sec:hard-core}
\subsection{Hard-core model}

For a graph $G$, recall $\Omega(G)$ is the collection of all independent sets of $G$. 
For an independent set $I\in\Omega(G)$, we say a vertex $v \in V$ is \emph{occupied} if $v \in I$, and \emph{unoccupied} otherwise; we often denote $v\in I$ as $v$ and $v\notin I$ as $\overline{v}$.
Consequently, for the hard-core model we can simplify the notation in the following manner.  For a pair of vertices $u,v \in V$, the influence of $u$ on $v$ is denoted as follows:
\begin{align*}
\Psi_{G,\lambda}(u \sra v) &
= \mu_{G,\lambda}(v\in I \mid u\in I) - \mu_{G,\lambda}(v\in I \mid u\notin I)
\\
& = \mu_{G,\lambda}(v \mid u) - \mu_{G,\lambda}(v \mid \bar{u}).
\end{align*}

We define the \emph{occupancy ratio} $R_{G,\lambda}(v)$ at $v$ as follows:
\begin{equation}
\label{occ-ratio-initial}
R_{G,\lambda}(v) = \frac{\mu_{G,\lambda}(v)}{\mu_{G,\lambda}(\bar{v})}.
\end{equation}

\subsection{Tree-uniqueness threshold}
\label{sub:tree-uniq-threshold}
Fix an integer $d \ge 2$ and a real $\lambda > 0$. 
Consider the hard-core model on a complete $d$-ary tree of height $h$ denoted by $T_h$; note that the root has degree $d$ while all other internal vertices have degree $\Delta = d+1$.
(The model is sometimes also analyzed on $\Delta$-regular trees; this does not change any of the results since the two trees only differ at the root.)

Recall that $Z_G(\lambda)$ is the partition function of the model. 
For $v\in V$, denote the reduced partition function obtained after
fixing $v$ to be occupied (equivalently,
$Z_{G-(\{v\}\cup N(v))}(\lambda)$) by
\begin{equation}
\label{eqn:partition-in}
Z_G^v(\lambda) :=\sum_{I \in \Omega(G): v \in I} \lambda^{|I|-1}
\end{equation}
Similarly, denote the partition function for $v$ unoccupied as:
 \begin{equation}
\label{eqn:partition-out}
Z_G^{\bar{v}}(\lambda) :=\sum_{I \in \Omega(G): v \notin I} \lambda^{|I|}
 \end{equation}
Recall the occupancy ratio defined in \cref{occ-ratio-initial}.  For $v\in V$, note the occupancy ratio at $v$ in terms of $Z_G^{v}$ and $Z_G^{\bar{v}}$ is
\begin{equation}\label{eq:ratio-pf}
R_{G,\lambda}(v) = \frac{\lambda Z_G^v(\lambda)}{Z_G^{\bar{v}}(\lambda)}.
\end{equation}

The \emph{tree recursion} is a function $F = F_{d,\lambda}$ that can be used to compute the occupancy ratio at the root, defined as
\[
F: \R_{\ge 0} \to \R_{\ge 0}, \quad F(R) = \frac{\lambda}{(1+R)^d}. 
\]
Denote by $R_h = R_{d,\lambda,h}$ the root occupancy ratio for $T_h$; e.g., $R_0 = \lambda$, $R_1 = \lambda/(1+\lambda)^d$. 
Then one can easily show that $R_{h} = F(R_{h-1})$.
A natural and important question is whether the sequence $\{R_h\}$ converges when $h$ tends to infinity, 
which is closely related to the Gibbs measure on the \emph{infinite $d$-ary tree}. 
The answer to this question is determined by whether the (unique) fixed point of $F$ is attractive or repulsive. 
Denote the unique positive fixed point of $F$ by $R^*$, i.e., $R^* (1+R^*)^d = \lambda$. 
The \emph{critical fugacity} is given by
\[
\lambda_c(\Delta) = \frac{(\Delta-1)^{\Delta-1}}{(\Delta-2)^\Delta},
\]
where $\Delta = d+1$ is the maximum degree of complete $d$-ary trees. 
It can be shown that if $\lambda \le \lambda_c(\Delta)$ then the fixed point $R^*$ is attractive and $R_h \to R^*$ as $h\to \infty$, and if instead $\lambda > \lambda_c(\Delta)$ then the fixed point $R^*$ is repulsive and $R_{2h-1} \to R'$, $R_{2h} \to R''$ as $h\to \infty$ for some $R' < R^* < R''$. 

Let $\Delta \ge 3$ be an integer. 
The critical fugacity $\lambda_c(\Delta)$ captures phase transitions for the hard-core model in multiple aspects. 
\begin{itemize}
\item When $\lambda \le \lambda_c(\Delta)$ there exists a unique Gibbs measure on the infinite $d$-ary tree; 
meanwhile, when $\lambda > \lambda_c(\Delta)$ there are multiple Gibbs measures. For this reason the critical value $\lambda_c(\Delta)$ is called the \emph{tree-uniqueness threshold}. 

\item When $\lambda < \lambda_c(\Delta)$, for complete $d$-ary trees we have $|R_{h} - R_{h-1}| = \exp(-\Theta(h))$, which can be viewed as the difference of root occupancy ratios on $T_{h+1}$ between  
fixing all leaves to be unoccupied (corresponding to $R_h$ on $T_h$) and fixing all leaves to be occupied (corresponding to $R_{h-1}$ on $T_{h-1}$). This describes a spatial mixing/correlation decay property with exponential decay rate, which fails when $\lambda > \lambda_c(\Delta)$.  

More generally, when $\lambda < \lambda_c(\Delta)$, for any graph $G$ of maximum degree at most $\Delta$, for any vertex $v \in V$ and any two pinnings $\sigma,\tau$ on a subset of vertices $\Lambda \subseteq V \setminus v$, it holds $|R^\sigma(v) - R^\tau(v)| = \exp(-\Omega(\ell))$ where $\ell$ is the distance from $v$ to a closest vertex $u \in \Lambda$ such that $\sigma(u) \neq \tau(u)$. This is known as the \emph{strong spatial mixing} property with exponential decay rate; see \cite{Wei06}. 

\item There exists an open set $\Gamma$ of complex numbers containing the interval $[0,\lambda_c(\Delta))$ such that, for every graph $G$ of maximum degree at most $\Delta$, one has $Z_G(\lambda) \neq 0$ whenever $\lambda \in \Gamma$. 
Meanwhile, the (complex) zeros of $Z_G(\lambda)$ can be arbitrarily close to $\lambda_c(\Delta)$ for a graph $G$ of maximum degree at most $\Delta$. See \cite{PR19}.

\item For constant $\Delta \ge 3$ and $\delta \in (0,1)$, when $\lambda \le (1-\delta) \lambda_c(\Delta)$, there exists a \emph{fully polynomial-time approximation scheme (FPTAS)} for the partition function $Z_G(\lambda)$ for all graphs $G$ of maximum degree at most $\Delta$ \cite{Wei06,Bar15,PR17}, and the Glauber dynamics for sampling from $\mu_{G,\lambda}$ converges in $O(n\log n)$ steps (see \cref{lem:plan} below). 
Meanwhile, when $\lambda > \lambda_c(\Delta)$ there is no $\fptas/\fpras$ for estimating the partition function for a graph $G$ of maximum degree at most $\Delta$ assuming $\mathsf{RP}\neq \mathsf{NP}$ \cite{Sly10,SlySun,GSV16}, see \cref{thm:hard-core-hardness}; this hardness reduction utilizes as a gadget that the Gibbs distribution on random $\Delta$-regular bipartite graphs is bimodal when $\lambda > \lambda_c(\Delta)$ and hence the Glauber dynamics has exponential mixing time.  
\end{itemize}

\subsection{Spectral independence}
Consider the hard-core model on a graph $G=(V,E)$ with fugacity $\lambda>0$. 
Recall, for two distinct vertices $u,v \in V$, the \emph{(pairwise) influence} of $u$ on $v$ is defined by:
\[
\Psi_{G,\lambda}(u \sra v)
 = \mu_{G,\lambda}(v \mid u) - \mu_{G,\lambda}(v \mid \bar{u}).
\]
And for the diagonal entries we set $\Psi_{G,\lambda}(u \sra u) = 1$.
For a pinning $\tau$, we also define the influence matrix $\Psi_{G,\lambda}^\tau$ for the conditional Gibbs distribution $\mu_{G,\lambda}^\tau$, where we let $\Psi_{G,\lambda}^\tau(u \sra v) = 0$ if $\tau$ forces $u$ to be unoccupied (note that in this case $\Psi_{G,\lambda}^\tau(v \sra u) = 0$ by definition). 
Finally, recall the Gibbs distribution $\mu_{G,\lambda}$ is said to be $\eta$-spectrally independent if for any pinning $\tau$, the maximum eigenvalue of the influence matrix $\Psi_{G,\lambda}^\tau$ is at most $1+\eta$. 

The main aim in this section is to establish the following spectral independence result for the hard-core model in the tree-uniqueness regime. 

\begin{theorem}
\label{lem:plan}
For every integer $\Delta \ge 3$ and real $\delta \in (0,1)$, there exists $\eta = \eta(\delta) = O(1/\delta)$, such that 
for any graph $G=(V,E)$ of maximum degree at most $\Delta$, any vertex $u \in V$, and any $\lambda \le (1-\delta) \lambda_c(\Delta)$, it holds
\begin{equation}\label{eq:main_goal}
\sum_{v \in V} \left| \Psi_{G,\lambda}(u \sra v) \right| \le 1+\eta. 
\end{equation}
As a consequence, for any graph $G$ of maximum degree at most $\Delta$ the Gibbs distribution $\mu_{G,\lambda}$ is $\eta$-spectrally independent, and the mixing time of the Glauber dynamics is $O_{\Delta,\delta}(n \log n)$. 
\end{theorem}

\cref{lem:plan} in fact extends to unbounded-degree cases where $\Delta = \omega_n(1)$. We assume $\Delta = O(1)$ for simplicity and refer interested readers to \cite{CFYZ24,AJKPV22,CFYZ22,CE25}.

Note, \cref{thm:main-SI-hard-core-uniq} is a corollary of \cref{lem:plan}.
In this monograph, we prove \cref{lem:plan} with a constant $1+\eta = O(1/\delta)$ independent of the maximum degree $\Delta$.
A tight bound was given by \cite{ChenJiang} which shows that spectral independence holds with $1+\eta \le 2(1+2/(\Delta-2))/\delta$.
The constant here matters as it determines the mixing time upper bound of the Glauber dynamics for the \emph{critical} hard-core model, i.e., when $\lambda=\lambda_c(\Delta)$; see \cite{CCYZ-critical-hard-core}.

\subsection{Relating influences and occupancy ratios}
Here we give a lemma relating the influences of a vertex $u$ and the occupancy ratio at $u$. 
It is helpful to consider a more general setting where every vertex $v$ has a distinct fugacity $\lambda_v$. 
Let $\blambda = (\lambda_v)_{v \in V}$ be a vector of fugacities, and the hard-core distribution is then defined by, for $\sigma\in\Omega(G)$,
\begin{equation}\label{eq:gibbs}
\mu_{G,\blambda}(\sigma) = \frac{\prod_{v \in \sigma} \lambda_v}{Z_G(\blambda)},
\end{equation}
where the multivariate partition function (independence polynomial) is defined as
\begin{equation}\label{eq:partition-function}
Z_G(\blambda) = \sum_{\sigma \in \Omega(G)} \prod_{v \in \sigma} \lambda_v.
\end{equation}
Generalizing \cref{eqn:partition-in} and \cref{eqn:partition-out}, let
\begin{equation}
\label{eqn:partition-in2}
Z_G^v(\blambda) :=\sum_{I \in \Omega(G): v \in I} \prod_{w\in I\setminus\{v\}} \lambda_w,
\end{equation}
and
 \begin{equation}
\label{eqn:partition-out2}
Z_G^{\bar{v}}(\blambda) :=\sum_{I \in \Omega(G): v \notin I} \prod_{w\in I}\lambda_w.
 \end{equation}

Viewing the influences and the occupancy ratios as rational functions of $\blambda$, we have the following relationship. 
\begin{claim}
\label{clm:inf-derivative}
For two distinct vertices $u,v \in V$, we have
\[
\Psi_{G}(u \sra v; \blambda) = \frac{\partial \log R_G(u; \blambda)}{\partial \log \lambda_v} = \frac{\lambda_v}{R_G(u; \blambda)} \frac{\partial R_G(u; \blambda)}{\partial \lambda_v}.
\]
\end{claim}

\begin{proof}
Similarly as \cref{eq:ratio-pf}, we have
\begin{equation}\label{eq:ratio-pf2}
R_G(u; \blambda) = \frac{\lambda_u Z_G^u(\blambda)}{Z_G^{\bar{u}}(\blambda)}, 
\end{equation}
and hence
\[
\frac{\partial \log R_G(u; \blambda)}{\partial \log \lambda_v}
= \frac{\partial}{\partial \log \lambda_v} \log\left( Z_G^u(\blambda) \right)
- \frac{\partial}{\partial \log \lambda_v} \log\left( Z_G^{\bar{u}}(\blambda) \right).
\]
We compute that
\begin{align*}
\frac{\partial}{\partial \log \lambda_v} \log\left( Z_G^u(\blambda) \right)
&= \frac{\lambda_v}{Z_G^u(\blambda)} \frac{\partial}{\partial \lambda_v} Z_G^u(\blambda) \\
&= \frac{\lambda_v}{Z_G^u(\blambda)} \sum_{I \in \Omega(G):\, u \in I} \frac{\partial}{\partial \lambda_v} \prod_{w \in I \setminus \{u\}} \lambda_w \\
&= \frac{\lambda_v}{Z_G^u(\blambda)} \sum_{I \in \Omega(G):\, u,v \in I} \prod_{w \in I \setminus \{u,v\}} \lambda_w \\
&= \frac{\lambda_v Z_G^{uv}(\blambda)}{Z_G^u(\blambda)}
= \mu_{G,\blambda}(v \mid u).
\end{align*}
Similarly, we have
\[
\frac{\partial}{\partial \log \lambda_v} \log\left( Z_G^{\bar{u}}(\blambda) \right) = \mu_{G,\blambda}(v \mid \bar{u}).
\]
Therefore, we conclude that
\[
\frac{\partial \log R_G(u; \blambda)}{\partial \log \lambda_v} 
= \mu_{G,\blambda}(v \mid u) - \mu_{G,\blambda}(v \mid \bar{u}) 
= \Psi_{G}(u \sra v; \blambda),
\]
as claimed.
\end{proof}

Since the occupancy ratios were intensively studied in previous works for establishing properties like correlation decay or zero-freeness, by \cref{clm:inf-derivative} we can transform these properties or their proof approaches into results for influences and thus establish spectral independence.

\section{Spectral Independence via Correlation Decay}
\label{sec:Weitz}

In this section, we show \cref{lem:plan} establishing $\eta$-spectral independence with $\eta = O(1/\delta)$ when $\lambda\leq(1-\delta)\lambda_c(\Delta)$ using an approach based on the strong spatial mixing (correlation decay) property, which appeared in \cite{ALO20,CLV20} and was based on techniques in \cite{Wei06,LLY13}.

\subsection{Proof approach}
\begin{itemize}
\item We need to show that for any graph $G=(V,E)$ of maximum degree at most $\Delta$ and any vertex $u \in V$ it holds
\[
\sum_{v \in V} \left| \Psi_{G}(u \sra v) \right| = O(1/\delta). 
\]

\item Fix a graph $G=(V,E)$ of maximum degree at most $\Delta$ and a vertex $u \in V$.  We construct a tree rooted at $u$ called the \emph{self-avoiding walk tree} $T = T_{\textsc{saw}}(G,u)$, which enumerates all self-avoiding walks (in $G$) starting from $u$. 
The tree $T$ is, in general, exponentially large in $n=|V|$, and each vertex of $G$ can appear multiple times in $T$.
The maximum degree of $T$ is at most $\Delta$.

We will define a hard-core model on $T$, such that the following properties hold:
\begin{itemize}
\item The occupancy ratio at $u$ is preserved: $$R_G(u) = R_T(u).$$
\item The influence from $u$ to another vertex $v$ is preserved: $$\Psi_G(u \sra v) = \sum_{w \in \CC_T(v)} \Psi_T(u \sra w),$$ 
where $\CC_T(v)$ denotes the set of all copies of $v$ in $T$. 
\end{itemize}

\item Then, it suffices to show that for any \emph{tree} $T=(V,E)$ of maximum degree at most $\Delta$ and any vertex $u \in V$ it holds
\[
\sum_{v \in V} \left| \Psi_T(u \sra v) \right| = O(1/\delta). 
\]
This will be proved via the so-called \emph{potential function method} \cite{RSTVY13,LLY13}. 
\end{itemize}

\subsection{Self-avoiding walk tree}
\label{sub:Tsaw}

We now define the self-avoiding walk tree more formally. 
Fix an arbitrary total order ``$<$'' on the vertices of~$G$. 

\begin{figure}
\centering
\includegraphics[width=0.9\textwidth]{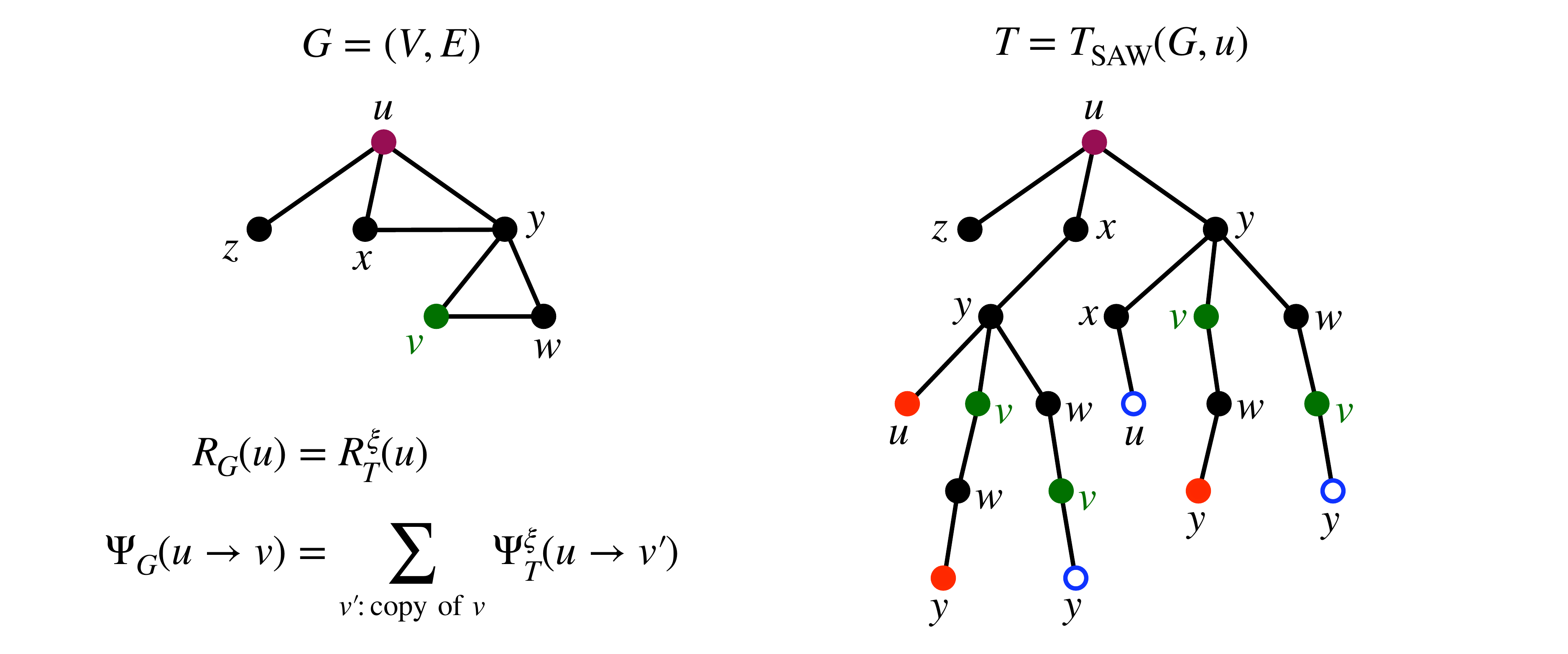}
\caption{An example of the self-avoiding walk tree defined in \cref{def:SAW} and the pinning $\xi$ defined in \cref{def:hc-saw}. Solid red vertices are those fixed to be occupied in the pinning $\xi$, and hollow blue vertices are those fixed to be unoccupied. 
Note that the path $u$-$z$-$u$ does not appear in the self-avoiding walk tree since it corresponds to a reduced cycle of length 2.
}
\end{figure}

\begin{definition}[Self-avoiding walk tree]
\label{def:SAW}
Let $G = (V,E)$ be a connected graph and $u \in V$ be a vertex. The \emph{self-avoiding walk (SAW) tree} $T = T_{\textsc{saw}}(G,u)$ of $G$ rooted at $u$ is a tree constructed as follows. 
We start with a tree consisting of root $u$ with children $N(u)$. Then we repeat the following operation (until the tree stops changing). For each leaf $v$ in $T$ such that the path from the root to $v$ does not contain
a repeated vertex we add below $v$ vertices $N(v)\setminus\{w\}$, where $w$ is the parent of $v$ in $T$.
\end{definition}

Observe that every path in $T$ from the root $u$ to a leaf corresponds to a self-avoiding walk in $G$. More precisely, if $u=v_0$-$v_1$-$\cdots$-$v_{\ell-1}$-$v_{\ell} = v$ is a path from $u$ to a leaf $v$, then it corresponds to a walk in $G$ (i.e., $\{v_{i-1},v_i\} \in E$) such that:
\begin{itemize}
\item either $u=v_0, v_1,\dots,v_{\ell-1}, v_\ell = v$ are all distinct vertices (so they form a self-avoiding walk), and $\deg_G(v) = 1$ (so the self-avoiding walk is maximal);
\item or $u=v_0, v_1,\dots,v_{\ell-1}$ are all distinct vertices (so they form a self-avoiding walk), and $v= v_\ell = v_i$ for some $i \le \ell-3$; 
hence, the last vertex $v$ makes a cycle (recall that in an undirected graph a cycle has at least $3$ vertices). 
\end{itemize}

\begin{remark}
Observe the following basic properties of $T=T_{\textsc{saw}}(G,u)$:
\begin{enumerate}[(1)]
\item The maximum degree of $T$ is the same as that of $G$. 
\item Leaves in $T$ correspond to either ``pendant vertices'' (those of degree $1$) in $G$ or to the completion of a cycle.
\item Each vertex of $G$ possibly appears multiple times in $T$. We denote the set of all copies in $T$ of a vertex $v \in V$ by $\CC_T(v)$. 
\item If $G$ itself is a tree, then $T = G$. However, in general $T$ can be exponentially larger than $G$.  
\end{enumerate}
\end{remark}

Now for the hard-core model defined on $G$ with the fugacity vector $\blambda = (\lambda_v)_{v \in V}$, we consider an associated hard-core model on the self-avoiding walk tree with a specific pinning on some leaves.  

\begin{definition}[Fugacities and pinning on $T_{\textsc{saw}}(G,u)$]
\label{def:hc-saw}
Let $G = (V,E)$ be a connected graph and $u \in V$ be a vertex. 
Suppose $\blambda = (\lambda_v)_{v \in V}$ is a fugacity vector on $G$.
Let $T = T_{\textsc{saw}}(G,u)$ be the SAW tree of $G$ rooted at $u$. 
Define the fugacities on $T$ and a pinning $\xi$ on $T$ as follows.
\begin{itemize}
\item For each vertex $v \in V$, every copy $w \in \CC_T(v)$ of $v$ has the same fugacity $\lambda_w = \lambda_v$. 
\item We define a pinning $\xi$ on a subset of leaves in the following way. 
Let $u=v_0$-$v_1$-$\cdots$-$v_{\ell-1}$-$v_{\ell} = v$ be a path in $T$ from $u$ to a leaf $v$.
\begin{itemize}
\item If $\deg_G(v) = 1$, then $v$ is not pinned;
\item Otherwise $v = v_\ell = v_i$ for some $i \le \ell-3$, and we fix $v_\ell$ 
to be occupied if $v_{i+1} < v_{\ell-1}$, and unoccupied if instead $v_{i+1} > v_{\ell-1}$.
\end{itemize} 
\end{itemize}
\end{definition}

\begin{remark}
For a path $u$--$v_1$--$\cdots$--$v_{i-1}$--$v$--$v_{i+1}$--$\cdots$--$v_{\ell-1}$--$v$ from the root $u$ to a leaf $v$ such that $v = v_\ell = v_i$, 
there is also a path from $u$ to another copy of $v$ in the SAW tree given by $u$--$v_1$--$\cdots$--$v_{i-1}$--$v$--$v_{\ell-1}$--$\cdots$--$v_{i+1}$--$v$, i.e., the order of the cycle is reversed. 
The pinnings at the two copies of $v$ (both are leaves) are opposite of each other. 
\end{remark}

For each leaf $v$ pinned to be occupied we can remove $v$ and its parent, and for each leaf $v$ pinned to be unoccupied we can remove $v$.  In this way we obtain an equivalent hard-core model on a tree $T'$ without any pinned vertices, and the root vertex $u$ only appears once in $T'$ since any additional occurrences of $u$ in $T$ must have been as leaf vertices since they corresponded to closing a cycle. 

Below we still use $\blambda$ to denote the fugacity vector for the hard-core model on $T$. 
The following lemma relates the hard-core models on $G$ and on the corresponding SAW tree $T$. 

\begin{lemma}[\cite{CLV20}]
\label{lem:weitz}
Let $T = T_{\textsc{saw}}(G,u)$ be the SAW tree of $G$ rooted at $u$. 
Consider the hard-core model on $T$ with the pinning $\xi$ as defined in \cref{def:hc-saw}. 
\begin{enumerate}[(1)]
\item $Z_G(\blambda)$ divides $Z_T^\xi(\blambda)$. 
Moreover, there exists a polynomial $P(\blambda)$ \emph{independent of $\lambda_u$}, such that
\[
Z_T^\xi(\blambda) = Z_G(\blambda) P(\blambda). 
\]
\item {\cite{Wei06}} The occupancy ratio at $u$ is preserved:
\[
R_G(u; \blambda) = R_T^\xi(u; \blambda). 
\]
\item The influence of $u$ on another vertex $v$ is preserved:
\[
\Psi_G(u \sra v; \blambda) = \sum_{w \in \CC_T(v)} \Psi_T^\xi(u \sra w; \blambda).
\]
\end{enumerate}
\end{lemma}

\begin{proof}[Proof of ``(1)'']
We will prove (1) by induction on the lexicographic ordering of the following pair of numbers: number of vertices of degree more than one, number of vertices of degree one. 

First suppose $u$ has degree one. Let $w$ be the neighbor of $u$. Let $G'$ be $G$ with vertex $u$ removed. Let $T'=T_{\textsc{saw}}(G',w)$. By the induction hypothesis there exists a polynomial 
$P'(\blambda)$ independent of $\lambda_w$ such that 
\begin{equation}\label{eq:indhyp}
Z_{T'}^\xi(\blambda) = Z_{G'}(\blambda) P'(\blambda). 
\end{equation}
Let $\blambda'$ be $\blambda$ with $\lambda_w$ replaced by $\lambda_w/(1+\lambda_u)$. Note that
\begin{equation}\label{eq:indstep}
Z_{T}^\xi(\blambda) = (1+\lambda_u) Z_{T'}^\xi(\blambda')
\quad\mbox{and}\quad
Z_{G}(\blambda) = (1+\lambda_u) Z_{G'}(\blambda').
\end{equation}
Plugging~\eqref{eq:indstep} into~\eqref{eq:indhyp} we obtain (1) for $G$ and $T$ (since $P'(\blambda)$
is independent of $\lambda_w$ we have $P'(\blambda) = P'(\blambda')$).

Now suppose $u$ has degree $d>1$. Let $w_1,\dots,w_d$ be the neighbors of $u$. Let $G'$ be the graph obtained by removing $u$ from $G$ and by adding $u_1,\dots,u_d$ where $u_i$ is connected to $w_i$ for $i\in [d]$. For $i\in\{0,\dots,d+1\}$ let $\xi'_i$ be a partial assignment of spins to $u_1,\dots,u_d$ where $u_j$ for $j<i$ are assigned unoccupied and $u_j$ for $j>i$ are assigned occupied. Let $G_i$ be the graph that corresponds to the assignment $\xi'_i$ in $G'$ (we obtain $G_i$ from $G'$ by removing $u_j$ for $j<i$ and removing
$u_j$ and $N(u_j)=\{w_j\}$ for $j>i$). Let $T_i=T_{\textsc{saw}}(G',u_i)$. Note that $T_i$ is a subtree of $T$ (it includes the root, one child of the root, and the entire subtree rooted at the child) and the pinning $\xi_i$ on $T_i$ defined by~\cref{def:hc-saw} is compatible with $\xi'_i$. 
By the induction hypothesis there exist polynomials $P_i$ independent of $\lambda_{u_i}$ such that
\begin{equation}\label{eq:eee0}
Z_{T_i}^{\xi_i}(\blambda) = Z_{G'}^{\xi_i'}(\blambda) P_i(\blambda) = Z_{G_i}(\blambda) P_i(\blambda). 
\end{equation}

Recall that $Z_G = \lambda_uZ_G^u + Z_G^{\overline{u}}$, where $Z_G^u(\blambda)$ and $Z_G^{\overline{u}}(\blambda)$ are the parts of $Z_G(\blambda)$ with $u$ occupied and unoccupied, respectively.
Since $Z_G = Z_G^{\bar{u}} + \lambda_u Z_G^{u}$ and $Z_T^{\xi} = Z_T^{\xi,\bar{u}} + \lambda_u Z_T^{\xi,u}$, 
our goal is to show
$Z_T^{\xi,\bar{u}} = Z_G^{\bar{u}}P(\blambda)$ and
$Z_T^{\xi,u} = Z_G^{u}P(\blambda)$
for some polynomial $P(\blambda)$ independent of $\lambda_u$.

Note that we have
\begin{equation}\label{eq:eee1}
Z_G^{u}(\blambda) = Z_{G'}^{\xi'_0}(\blambda) = Z_{G'}^{\xi'_1,u_1}(\blambda) \quad\mbox{and}\quad
Z_G^{\bar{u}}(\blambda) = Z_{G'}^{\xi'_{d+1}}(\blambda) = Z_{G'}^{\xi'_{d},\bar{u_d}}(\blambda).
\end{equation}
We also have
\begin{equation}\label{eq:eee2}
Z_T^{\xi,u} = \prod_{i=1}^d Z_{T_i}^{\xi_i,u_i}(\blambda),\quad
 Z_T^{\xi,\bar{u}} = \prod_{i=1}^d Z_{T_i}^{\xi_i,\bar{u_i}}(\blambda).
\end{equation}
Since $P_i$ does not depend on $\lambda_{u_i}$ from~\eqref{eq:eee0} we have
\begin{equation}\label{eq:eee3}
Z_{T_i}^{\xi_i,u_i}(\blambda) = Z_{G'}^{\xi_i',u_i}(\blambda) P_i(\blambda)\quad\mbox{and}\quad
Z_{T_i}^{\xi_i,\bar{u_i}}(\blambda) = Z_{G'}^{\xi_i',\bar{u_i}}(\blambda) P_i(\blambda).
\end{equation}
Plugging~\eqref{eq:eee3} into~\eqref{eq:eee2} we obtain 
\[
Z_T^{\xi,u}
=
\prod_{i=1}^d
Z_{G'}^{\xi'_i,u_i}(\blambda)P_i(\blambda)
=
Z_G^u
\prod_{i=2}^d
Z_{G'}^{\xi'_i,u_i}(\blambda)
\prod_{i=1}^d P_i(\blambda).
\]
Similarly, we have
\[
Z_T^{\xi,\bar{u}}
=
\prod_{i=1}^d
Z_{G'}^{\xi'_i,\bar{u_i}}(\blambda)P_i(\blambda)
=
Z_G^{\bar{u}}
\prod_{i=1}^{d-1}
Z_{G'}^{\xi'_i,\bar{u_i}}(\blambda)
\prod_{i=1}^d P_i(\blambda)
=
Z_G^{\bar{u}}
\prod_{i=2}^{d}
Z_{G'}^{\xi'_i,u_i}(\blambda)
\prod_{i=1}^d P_i(\blambda).
\]
Letting
\[
P(\blambda)
=
\prod_{i=2}^{d}
Z_{G'}^{\xi'_i,u_i}(\blambda)
\prod_{i=1}^d P_i(\blambda),
\]
then we have the following:
\begin{align*}
    Z_T^{\xi} = Z_T^{\xi,\bar{u}} + \lambda_u Z_T^{\xi,u} = \left(Z_G^{\bar{u}} + \lambda_u Z_G^{u}\right) P(\blambda) = Z_G P(\blambda). & \qedhere
\end{align*}
\end{proof}

\begin{proof}[Proof of ``(1) $\Rightarrow$ (2)'']
Note that
\begin{align*}
R_G(u; \blambda) = \lambda_u \frac{Z_G^u(\blambda)}{Z_G^{\bar{u}}(\blambda)} =
\lambda_u \frac{Z_G^u(\blambda) P(\blambda) }{Z_G^{\bar{u}}(\blambda) P(\blambda)}
=  \lambda_u \frac{Z_T^{\xi,u}(\blambda)}{Z_T^{\xi,\bar{u}}(\blambda)} = R_T^\xi(u; \blambda).
&\qedhere
\end{align*}
\end{proof}

\begin{proof}[Proof of ``(2) $\Rightarrow$ (3)'']
We deduce from \cref{clm:inf-derivative} and the chain rule that
\begin{align}
\Psi_G(u \sra v; \blambda)
&= \frac{\partial \log R_G(u; \blambda)}{\partial \log \lambda_v} \tag{\cref{clm:inf-derivative}} \\
&= \frac{\partial \log R_T^\xi(u; \blambda)}{\partial \log \lambda_v} \tag{Part (2)} \\
&= \sum_{w \in \CC_T(v)} \frac{\partial \log R_T^\xi(u; \blambda)}{\partial \log \lambda_w} \frac{\partial \log \lambda_w}{\partial \log \lambda_v} \tag{Chain rule} \\
&= \sum_{w \in \CC_T(v)} \Psi_T^\xi(u \sra w; \blambda) \tag{\cref{clm:inf-derivative}},
\end{align}
which shows Part (3). 
\end{proof}

Note, for a vertex $w$ in $T$, if $w$ is fixed by $\xi$, then $\Psi_T^\xi(u \sra w; \blambda) = 0$ by definition.

\begin{corollary}
\label{cor:tsaw}
By the triangle inequality, we deduce that
\[
\sum_{v \in V(G)} \left| \Psi_G(u \sra v) \right| 
\le 
\sum_{w \in V(T)} \left| \Psi_T(u \sra w) \right|.
\]
Hence, we reduce the problem to trees (albeit a possibly exponentially large tree). 
\end{corollary}

\subsection{Bounding influences on trees}
In this subsection, we bound the absolute sum of influences of the root on bounded-degree trees. 
In particular, we show the following result of Chen, Liu, and Vigoda~\cite{CLV20}.

\begin{lemma}
\label{lem:tree-inf}
Let $\delta\in(0,1)$ and
$\lambda\leq(1-\delta)\lambda_c(\Delta)$, and consider the
hard-core model on $T$ at fugacity $\lambda$.  Let $T=(V,E)$ be a tree of maximum degree at most $\Delta$ rooted at $u$. For an integer $k \in \N^+$ and a vertex $v \in V$, let $L_v(k)$ denote the set of all descendants at distance $k$ from $v$. 
Then for all $k \in \N^+$ we have
\[
\sum_{v \in L_u(k)} \left| \Psi_T(u \sra v) \right| \le  12 (1-\delta/8)^{k},
\]
\end{lemma}

The following claim is helpful to us. 

\begin{claim}
\label{clm:tree-inf}
Let $T=(V,E)$ be a tree and $u,v \in V$ be two distinct vertices. If $w$ is a vertex on the unique path from $u$ to $v$, then
\[
\Psi_T(u \sra v) = \Psi_T(u \sra w) \,\Psi_T(w \sra v).
\]
\end{claim}

\begin{proof}
Since $w$ is on the unique path from $u$ to $v$ we have that spins at $u$ and $v$ are conditionally independent given the spin at $w$, and hence we have the following:
\begin{align*}
\Psi_T(u \sra v)
={}& \mu_T(v \mid u) - \mu_T(v \mid \bar{u}) \\
={}& \mu_T(w \mid u) \mu_T(v \mid w) + \mu_T(\bar{w} \mid u) \mu_T(v \mid \bar{w}) \\
&- \mu_T(w \mid \bar{u}) \mu_T(v \mid w) - \mu_T(\bar{w} \mid \bar{u}) \mu_T(v \mid \bar{w}) \\
={}& \left( \mu_T(w \mid u) - \mu_T(w \mid \bar{u}) \right) \mu_T(v \mid w) - \left( \mu_T(\bar{w} \mid \bar{u}) - \mu_T(\bar{w} \mid u) \right) \mu_T(v \mid \bar{w}) \\
={}& \left( \mu_T(w \mid u) - \mu_T(w \mid \bar{u}) \right) \left( \mu_T(v \mid w) - \mu_T(v \mid \bar{w}) \right) \\
={}& \Psi_T(u \sra w) \,\Psi_T(w \sra v),
\end{align*}
as claimed.
\end{proof}

For an integer $k\geq 1$ and a vertex $u$, recall $L_u(k)$ is the set of all descendants at distance $k$ from $u$. 
We can apply \cref{clm:tree-inf} to build an inductive approach for bounding the sum of the influences from a vertex $u$ as follows:
\begin{align}
\nonumber
\sum_{v \in L_u(k)} \left| \Psi_T(u \sra v) \right| 
&= \sum_{w \in L_u(k-1)} \sum_{v \in L_w(1)} \left| \Psi_T(u \sra v) \right| \\
&= \sum_{w \in L_u(k-1)} \left| \Psi_T(u \sra w) \right| \sum_{v \in L_w(1)} \left| \Psi_T(w \sra v) \right|.
\label{eq:tree-expand}
\end{align}
Hence, an easy approach to establish \cref{lem:tree-inf} is to show for all $w$,
\begin{equation}\label{eq:ind?}
\sum_{v \in L_w(1)} \left| \Psi_T(w \sra v) \right| \le (1-\delta/8).
\end{equation}
Then \cref{lem:tree-inf} follows by induction.
But such an approach does not work since \cref{eq:ind?} is not true for all $\lambda \le (1-\delta) \lambda_c(\Delta)$, though it is easy to check that \cref{eq:ind?} holds when $\lambda \le \frac{1-\delta}{\Delta-2}$ which is below the uniqueness threshold.
Note, for the Ising model (see \cref{sub:Ising}), the analog of \cref{eq:ind?} holds for all $\beta$ where $|\beta| < \beta_c(\Delta)$, and hence we do not need to consider a potential function for the Ising model.

To prove \cref{lem:tree-inf} we use the \emph{potential function method}, which establishes a generalization of \cref{eq:ind?}.

\begin{proof}[Proof of \cref{lem:tree-inf}]
The multivariate tree recursion is a function $F = F_{d,\lambda}: \R_{\ge 0}^d \to \R_{\ge 0}$ given by 
\[
R = F(R_1,\dots,R_d) := \frac{\lambda}{\prod_{i=1}^d (1+R_i)},
\]
which means that for a tree rooted at $w$ with $d$ children $v_1,\dots,v_d$, the occupancy ratio $R_T(w)$ at the root is given by $R_T(w) = F(R_{T_1}(v_1),\dots,R_{T_d}(v_d))$ where $T_i$ is the subtree rooted at $v_i$ and $R_{T_i}(v_i)$ is the root occupancy ratio of $T_i$. 
It would be helpful to consider the logarithm of occupancy ratios in the spirit of \cref{clm:inf-derivative}. 
Writing $x = \log R$ and $x_i = \log R_i$, we define a multivariate function $H = H_{d,\lambda}: \R^d \to \R$ by
\[
x = H(x_1,\dots,x_d) := \log \lambda - \sum_{i=1}^d \log(1+e^{x_i}). 
\]
One can check that
\[
\frac{\partial H}{\partial x_i} = \Psi_T(w \sra v_i)
\]
which is similar to \cref{clm:inf-derivative}, 
and therefore
\[
\norm{\grad H}_1 = \sum_{i=1}^d \left| \Psi_T(w \sra v_i) \right|.
\]

Let $\varphi: \R \to \R$ be a suitable potential function that is monotone increasing, and we consider the tree recursion composed with $\varphi$. 
That is, let $y = \varphi(x)$ and $y_i = \varphi(x_i)$, and then for $H^\varphi = \varphi \circ H \circ \varphi^{-1}$ we have
\[
y = H^\varphi(y_1,\dots,y_d).
\]
Moreover, it is easy to check that
\begin{equation}\label{eq:deriv_H}
\norm{\grad H^\varphi}_1 = \sum_{i=1}^d \frac{\varphi'(x)}{\varphi'(x_i)} \left| \Psi_T(w \sra v_i) \right|. 
\end{equation}
Therefore, we have the following 
similar to \cref{eq:tree-expand}:
\begin{align}
\label{eq:tree-inf-chain}
\sum_{v \in L_u(k)} \frac{\varphi'(x_u)}{\varphi'(x_v)} \left| \Psi_T(u \sra v) \right|
= \sum_{w \in L_u(k-1)} \frac{\varphi'(x_u)}{\varphi'(x_w)} \left| \Psi_T(u \sra w) \right| \sum_{v \in L_w(1)} \frac{\varphi'(x_w)}{\varphi'(x_v)} \left| \Psi_T(w \sra v) \right|. 
\end{align}
It remains to choose a suitable potential function. 
We use an ingenious potential function given in \cite{LLY13}:
\[
\varphi(x) = \log\left( e^{x/2} + \sqrt{e^x+1} \right).
\]
Note that we have
\begin{equation}
\label{phi-interval}
\varphi'(x) = \frac{1}{2}\sqrt{\frac{e^{x}}{e^{x} +1 }}. 
\end{equation}
We are only interested in $x$ in the range
\begin{align*}
\frac{\lambda}{(1+\lambda)^\Delta} \le e^x = R \le \lambda,
\end{align*}
where $R$ denotes an achievable occupancy ratio. In this range, since $\varphi'$ is monotone,
we have
\begin{equation}\label{up:zzz}
\left(\frac{\max_x \varphi'(x)}{\min_x \varphi'(x)}\right)^2 \leq 
\frac{\lambda/(1+\lambda)}{(\lambda/(1+\lambda)^\Delta)/(1+\lambda/(1+\lambda)^\Delta)}
=(1+\lambda)^{\Delta-1} + \frac{\lambda}{1+\lambda}.
\end{equation}
We have $\lambda \leq \lambda_c(\Delta) = (\Delta-1)^{\Delta-1} / (\Delta-2)^\Delta$
and hence we can bound~\eqref{up:zzz} by 
\begin{equation}\label{up:zzz2}
(1+(\Delta-1)^{\Delta-1} / (\Delta-2)^\Delta)^{\Delta-1} + 1. 
\end{equation}
For $\Delta\in\{3,\dots,11\}$, one can check by explicit calculation, that the value of~\eqref{up:zzz2} is bounded by 26 (the maximum happens at $\Delta=3$). For $\Delta\geq 12$
we use $(1+\lambda)\leq\exp(\lambda)$ to bound~\eqref{up:zzz2} by
\begin{equation}\label{up:zzz3}
\exp\left( (\Delta-1)^{\Delta} / (\Delta-2)^\Delta \right) + 1 \leq 26, 
\end{equation}
where the last inequality for $\Delta=12$ follows by explicit calculation and for
$\Delta>12$ follows from the fact that $(\Delta-1)^{\Delta} / (\Delta-2)^\Delta$
is a decreasing function of $\Delta\geq 12$.
We established
\begin{equation}
\label{phi-derivative-ratio}
\frac{\max_x \varphi'(x)}{\min_x \varphi'(x)}\leq \sqrt{26} < 6.
\end{equation}

We will prove the following contraction for this potential function. 
\begin{claim}
    \label{claim:contraction}
For any integer $d \le \Delta -1$, and for any $x_i$ and $x = H(x_1,\dots,x_d)$, for any $\delta>0$, for $\lambda\leq (1-\delta)\lambda_c(\Delta)$, 
\begin{align}\label{eq:contraction}
\norm{\grad H^\varphi}_1 \le (1-\delta/8).
\end{align}
\end{claim}

We first prove \cref{lem:tree-inf} for the case $|L_u(1)| \le \Delta-1$ and then we will derive the bound for the case $|L_u(1)| = \Delta$ afterwards.

Using \cref{claim:contraction} we will prove by induction that if $|L_u(1)| \le \Delta-1$ then the following holds for every integer $k\geq 1$,
\begin{equation}\label{eq:induction}
\sum_{v \in L_u(k)} \frac{\varphi'(x_u)}{\varphi'(x_v)} \left| \Psi_T(u \sra v) \right| \le 
    (1-\delta/8)^k. 
\end{equation}

For the base case $k = 1$ we have the following bound by \cref{claim:contraction}:
\begin{equation}
    \label{eq:base-case-cor-decay}
\sum_{v \in L_u(1)} \frac{\varphi'(x_u)}{\varphi'(x_v)} \left| \Psi_T(u \sra v) \right| = \norm{\grad H^\varphi}_1 
\le
    1-\delta/8.
\end{equation}

Now suppose \cref{eq:induction} holds for $k-1$, and then we have:
\begin{align*}
\lefteqn{
\sum_{v \in L_u(k)} \frac{\varphi'(x_u)}{\varphi'(x_v)} \left| \Psi_T(u \sra v) \right| 
} \hspace*{.8in}
\\
={}& \sum_{w \in L_u(k-1)} \frac{\varphi'(x_u)}{\varphi'(x_w)} \left| \Psi_T(u \sra w) \right| \sum_{v \in L_w(1)} \frac{\varphi'(x_w)}{\varphi'(x_v)} \left| \Psi_T(w \sra v) \right| 
& \mbox{by \cref{eq:tree-inf-chain}}
\\
\le{}& \sum_{w \in L_u(k-1)} \frac{\varphi'(x_u)}{\varphi'(x_w)} \left| \Psi_T(u \sra w) \right| \cdot (1-\delta/8)
& \mbox{by \cref{claim:contraction}}
\\
\le{}& 
    (1-\delta/8)^k,
&\mbox{by induction} 
\end{align*}
which completes the proof of \cref{eq:induction}.

We can now complete the proof of \cref{lem:tree-inf} as follows:
\begin{align*}
\sum_{v \in L_u(k)} \left| \Psi_T(u \sra v) \right|  & = \sum_{v \in L_u(k)} \frac{\varphi'(x_v)} {\varphi'(x_u)}\frac{\varphi'(x_u)}{\varphi'(x_v)}  \left| \Psi_T(u \sra v) \right| 
\\
& \leq \frac{\max_x \varphi'(x)}{\min_x \varphi'(x)} \sum_{v \in L_u(k)} \frac{\varphi'(x_u)}{\varphi'(x_v)}  \left| \Psi_T(u \sra v) \right| 
\\
& \leq 6\sum_{v \in L_u(k)} \frac{\varphi'(x_u)}{\varphi'(x_v)}  \left| \Psi_T(u \sra v) \right| 
& \mbox{by \cref{phi-derivative-ratio}}
\\
& \leq 6(1-\delta/8)^k, & \mbox{by \cref{eq:induction}}
\end{align*}
as claimed.  This completes the proof of \cref{lem:tree-inf} for the case $|L_u(1)| \le \Delta-1$. Now we will derive the bound for the case $|L_u(1)| = \Delta$.

Suppose $|L_u(1)| = \Delta$, and let $L_u(1) = \{v_1,\dots,v_\Delta\}$.   
For $i\in\{1,2\}$, let $T_{v_i}$ be the subtree rooted at $v_i$, and let $T^{(i)}=T\setminus T_{v_i}$. 
For every $v\in V(T^{(i)})$, we have
\[
\Psi_T(u\sra v)=\Psi_{T^{(i)}}(u\sra v),
\]
since, conditional on the spin at $u$, the subtrees rooted at the
children of $u$ are independent. Moreover,
\[
L_u^{T^{(i)}}(k)
=
L_u^T(k)\setminus L_{v_i}^T(k-1),
\]
and the root $u$ has $\Delta-1$ children in both $T^{(1)}$ and $T^{(2)}$.

We can now complete the proof as follows:
\begin{align*}
    \sum_{v \in L_u(k)} \left| \Psi_T(u \sra v) \right| 
    &\le \sum_{v \in L_u(k) \setminus L_{v_1}(k-1)} \left| \Psi_T(u \sra v) \right|
    + \sum_{v \in L_u(k) \setminus L_{v_2}(k-1)} \left| \Psi_T(u \sra v) \right| \\
    &\le 6 (1-\delta/8)^{k}
    + 6 (1-\delta/8)^{k} \\
    &= 12 (1-\delta/8)^{k}.
\end{align*}
 This completes the proof of \cref{lem:tree-inf} for the case $|L_u(1)| = \Delta$.
\end{proof}

It remains to prove \cref{claim:contraction} which we do so now.

\begin{proof}[Proof of \cref{claim:contraction}]
The case $d=0$ is immediate. It suffices to prove the result for
$d=\Delta-1$. Indeed, for $1\leq d<\Delta-1$, we may add
$\Delta-1-d$ dummy children whose occupancy ratios tend to zero.
These children do not change the root recursion or the gradient
sum in the limit, and hence the result follows by continuity.  
Henceforth, we assume $d=\Delta-1$.

First, observe that:
\begin{equation}\label{eq:gradient-cor-decay}
\norm{\grad H^\varphi}_1 = \sum_{i=1}^d \sqrt{\frac{e^x}{e^x+1}} \sqrt{\frac{e^{x_i}}{e^{x_i} + 1}}.
\end{equation}
Recalling $x = \log R$, $x_i = \log R_i$ and writing $p = R/(R+1)$, $p_i = R_i/(R_i+1)$ for the occupancy probability, and letting $q = \frac{1}{d} \sum_{i=1}^d p_i$, then we deduce by the Cauchy-Schwarz inequality that
\begin{align}
\label{almost-111}
\norm{\grad H^\varphi}_1^2 
= \left( \sum_{i=1}^d \sqrt{p p_i} \right)^2 \le d p \sum_{i=1}^d p_i = d^2pq.
\end{align}
The tree recursion, in terms of the occupancy probability, is given by
\begin{align}
\label{almost-222}
\frac{p}{1-p} = \lambda \prod_{i=1}^d (1-p_i) \le \lambda \left( 1 - \frac{1}{d} \sum_{i=1}^d p_i \right)^d = \lambda(1-q)^d,
\end{align}
where the inequality is by the arithmetic-geometric mean inequality (AM-GM).  
Combining \cref{almost-111,almost-222} yields:
\begin{equation}
\label{ineq:DDD}
\norm{\grad H^\varphi}_1^2 \le d^2 p q \le d^2 q \cdot \frac{\lambda (1-q)^d}{\lambda (1-q)^d + 1},
\end{equation}
where we used the fact that $p/(1-p)\leq x$ is equivalent to $p\leq x/(1+x)$.

To bound \cref{ineq:DDD} we use the following property of the fixed points of the tree recursions. 
\begin{claim}
    \label{claim:pstar}
    Let $p^*$ be the fixed point of the (univariate) tree recursion such that $\lambda = p^*/(1-p^*)^{d+1}$, where $d=\Delta-1$.
For    $\lambda \le (1-\delta) \lambda_c(\Delta)$ we have
\[
    p^*\leq \frac{1-\delta/4}{d}.
\]
\end{claim}

 Now we prove that
\begin{align}\label{eq:cctech}
q \cdot \frac{\lambda (1-q)^d}{\lambda (1-q)^d + 1} \le \frac{p^*}{d}.
\end{align}
If $dq\leq p^*$, then the left-hand side of
\cref{eq:cctech} is at most $q\leq p^*/d$, and hence \cref{eq:cctech} is trivial in this case. Henceforth, assume $dq>p^*$, and then observe that~\cref{eq:cctech} is equivalent to the following: 
\begin{align*}
\lambda (1-q)^d \le \frac{p^*}{dq - p^*} 
\quad\Leftrightarrow \quad
(dq-p^*)(1-q)^d \le (1-p^*)^{d+1}.
\end{align*} 
By the AM--GM inequality, we have
\begin{align*}
(dq-p^*)(1-q)^d \le \left( \frac{(dq-p^*) + d(1-q)}{d+1} \right)^{d+1} = \left( \frac{d-p^*}{d+1} \right)^{d+1} \le (1-p^*)^{d+1},
\end{align*}
where the last inequality follows from $p^* \le 1/d$, which is implied by \cref{claim:pstar}. This establishes \cref{eq:cctech}.

Finally, \cref{eq:contraction}, and hence \cref{claim:contraction}, follows from plugging the bound from \cref{eq:cctech} and \cref{claim:pstar} into \cref{ineq:DDD}.
\begin{align*}
\norm{\grad H^\varphi}_1^2 & \le d^2 q \cdot \frac{\lambda (1-q)^d}{\lambda (1-q)^d + 1}.
& \mbox{by \cref{ineq:DDD}}
\\
& \leq dp^* 
& \mbox{by \cref{eq:cctech}}
\\
& \leq (1-\delta/4) 
& \mbox{by \cref{claim:pstar}.}
\end{align*}

Since $\norm{\grad H^\varphi}_1^2  \le (1-\delta/4)$ then we have that $\norm{\grad H^\varphi}_1  \le (1-\delta/8)$, which completes the proof of \cref{claim:contraction}.
\end{proof}

Using  \cref{lem:tree-inf} we can complete the proof of the main result (\cref{lem:plan}) of this section.

\begin{proof}[Proof of \cref{lem:plan}]
We can now conclude the main technical result \cref{eq:main_goal} in \cref{lem:plan}.
\begin{align*}
\sum_{v \in V(G)} \left| \Psi_{G,\lambda}(u \sra v) \right| & \le  
\sum_{w \in V(T)} \left| \Psi_T(u \sra w) \right|.
& \mbox{by \cref{cor:tsaw}}
\\
& = \sum_k \sum_{w\in L_u(k)} \left| \Psi_T(u \sra w) \right| 
\\
& \le 12 \sum_k (1-\delta/8)^k 
&\mbox{by \cref{lem:tree-inf}}
\\
& \le \frac{96}{\delta},
\end{align*}
which proves \cref{eq:main_goal} in \cref{lem:plan} with $1+\eta=96/\delta$.

To see why \cref{eq:main_goal} implies spectral independence, recall from \cref{lem:rowsum} in \cref{sub:rowsum} that for any graph $G$ of maximum degree at most $\Delta$, the maximum eigenvalue of $\Psi_{G,\lambda}$ is upper bounded by $\norm{\Psi_{G,\lambda}}_\infty = \max_{u \in V} \sum_{v \in V} \left| \Psi_{G,\lambda}(u \sra v) \right|$, which is at most $1+\eta$ by \cref{eq:main_goal}. 
Meanwhile, for any pinning $\tau \in \{0,1\}^\Lambda$, the conditional Gibbs distribution corresponds to the hard-core model on a smaller graph (removing $\Lambda$ and neighbors of occupied vertices in $\Lambda$) which also has maximum degree at most $\Delta$. 
Since \cref{eq:main_goal} applies to all graphs of maximum degree at most $\Delta$, the maximum eigenvalue of $\Psi_{G,\lambda}^\tau$ is also at most $1+\eta$.
Hence, $\eta$-spectral independence follows and then optimal mixing of the Glauber dynamics follows from \cref{thm:SI-constant-mix}.
\end{proof}

\begin{proof}[Proof of \cref{claim:pstar}]
It is enough to show that for $p = \frac{1-\delta/4}{\Delta-1}$ and $\lambda = (1-\delta)\frac{(\Delta-1)^{\Delta-1}}{(\Delta-2)^\Delta}$ we have
\begin{equation}\label{eeer}
p > \lambda (1 -p)^{\Delta},
\end{equation}
since the fixed point cannot occur in $[p,\infty)$ (note that the LHS of~\eqref{eeer} is increasing and the RHS of~\eqref{eeer} is decreasing).

Plugging in the values of $p$ and $\lambda$ in~\eqref{eeer} and simplifying we obtain the following equivalent inequality.
\begin{equation}\label{eeer2}
1-\delta/4 > (1-\delta) \left(1 + \frac{\delta/4}{\Delta-2}\right)^\Delta
\end{equation}
The RHS of~\eqref{eeer2} is maximized for $\Delta=3$ and hence it is enough to show
$$
1-\delta/4 > (1-\delta) \left(1 + \delta/4\right)^3,
$$
which is equivalent to $\delta^2 (\delta^2+11\delta+36)>0$.
\end{proof}

\begin{remark}
Another choice of the potential function is from \cite{RSTVY13}:
\[
\varphi(x) = \log\left( e^x + \frac{1}{\Delta} \right).
\]
\end{remark}

\section{Spectral Independence via Zero-Freeness}
\label{sec:stability}

In this section, we give a proof of a weaker version of \cref{lem:plan} using the zero-freeness of the partition function; the version proved here is weaker in the sense that it does not have an explicit bound on $\eta$.
The approach here is based on \cite{AASV21} with appropriate modifications and generalizations; see also \cite{CLV21zerofree} for a more general setting. 

Zero-freeness of the partition function characterizes the analytic region of the so-called \emph{free energy}, which is the logarithm of the partition function. It has been observed that the analyticity of the free energy is closely related to the uniqueness of the Gibbs measure and the correlation decay properties; for example, see \cite{DS85,DS87} for the complete analytical condition for lattice spin systems.
For (complex) analytic functions, their derivatives can be appropriately bounded. The implication from zero-freeness to spectral independence is realized by the fact that the second derivative of the free energy (with respect to the fugacities, or more generally, the external fields) is the covariance matrix of the Gibbs distribution.

\subsection{Some preliminaries}

For a complex number $\zeta \in \C$ and a real number $r > 0$, let
\[
\DD(\zeta,r) = \{z \in \C: |z-\zeta| < r\}
\]
be the open disk around $\zeta$ of radius $r$. 
Furthermore, for a subset $A \subseteq \C$ of complex numbers define
\[
\DD(A,r) = \bigcup_{\zeta \in A} \DD(\zeta,r). 
\]
Let $\overline{\DD}(\zeta,r)$ and $\overline{\DD}(A,r)$ denote their closure.

Consider the multivariate independence polynomial defined by \cref{eq:partition-function}. 
The following stability (zero-freeness) result was established by Peters and Regts~\cite{PR19} and is the basis of our approach.

\begin{theorem}[{\cite[Theorem 4.2]{PR19}}]
\label{thm:zero-free}
Let $\Delta \ge 3$ be an integer. 
For any $\delta \in (0,1)$, there exists $\eps = \eps(\delta) > 0$ such that for any graph $G$ of maximum degree at most $\Delta$, we have
$Z(\blambda) \neq 0$
whenever $\lambda_v \in \DD( [0,(1-\delta)\lambda_c(\Delta)], \eps )$ for each vertex $v$. 
\end{theorem}

The zero-freeness result of \cref{thm:zero-free} yielded an $\fptas$ for the partition function for all graphs of maximum degree $\Delta$ when $\lambda\leq (1-\delta)\lambda_c(\Delta)$ for $\delta>0$ via the algorithm of Patel and Regts~\cite{PR17} which builds upon the Taylor polynomial interpolation method introduced by Barvinok~\cite{Bar15,Bar16,Bar17book,PRsurvey}.

We also need the following lemma from complex analysis. 
Recall, a holomorphic function on an open set $U\subseteq\C$ is differentiable at every point in $U$.
\begin{lemma}[Schwarz--Pick Lemma]
\label{lem:Schwarz}
Let $f: \DD(0,1) \to \overline{\DD}(0,1)$ be a holomorphic function. 
Then 
\[
|f'(0)| \le 1 - |f(0)|^2 \le 1. 
\] 
\end{lemma}

\subsection{Proof approach}
\label{subsec:approach-0}

We need to show that for any graph $G$ of maximum degree at most $\Delta$ and any vertex $u \in V$, it holds
\[
\sum_{v \in V} \left| \Psi_{G,\lambda}(u \sra v) \right| = O(1). 
\]
At a high level, the proof proceeds in the following manner:

\begin{enumerate}
\item Consider the multivariate case where every vertex has its own fugacity. For a complex number $\zeta \in \C$, define $\blambda(\zeta)$ to be some perturbation of the fugacity vector such that: 
\begin{enumerate}
\item $\blambda(0) = \lambda \allone$ is the uniform fugacity vector.
\item Consider the complex function 
\[
f(\zeta) = \lambda \log \big( R_G\left(u;\blambda(\zeta)\right) \big),
\]
where $R_G$ is the occupancy ratio at $u$ defined in \cref{eq:ratio-pf}.
The derivative of $f$ is a rational function and \cref{thm:zero-free} implies that its denominator does not vanish in $\DD(0,\varepsilon)$ and hence $f$ is holomorphic on $\DD(0,\varepsilon)$.
\item For a suitable choice of $\blambda(\zeta)$ we will show, using \cref{clm:inf-derivative}, that we have
\[
f'(0) = \sum_{v \in V\setminus\{u\}} \left| \Psi_{G,\lambda}(u \sra v) \right|.
\]
\end{enumerate}
In the actual proof we define $f$ differently from above, so that it is easier to describe the image of $f$ as needed in the next step. 

\item We need to show the image of $f$ is contained in $\overline{\DD}(0,B)$ where $B > 0$ is a constant. 
Then, \cref{lem:Schwarz} applied to the function $g(z) = \frac{1}{B} f(\eps z)$, implies that
\[
\sum_{v \in V} |\Psi_{G,\lambda}(u \sra v)| = 1+f'(0) = 1+\frac{B}{\eps} g'(0) \le 1+\frac{B}{\eps}. 
\]
\end{enumerate}

\subsection{Proofs}

Now we present the proof of \cref{lem:plan}; our bound here on spectral independence depends on the radius of the zero-free region that appears in \cref{thm:zero-free} which implicitly depends on $\delta$.

\begin{proof}[Proof of (the weaker version of) \cref{lem:plan}]

Fix the graph $G$ and the vertex $u$. 
For each $v \neq u$ define
\[
s_v = \sgn\left(\Psi_{G,\lambda}(u \sra v) \right) :=
\begin{cases}
1, &\Psi_{G,\lambda}(u \sra v) \ge 0;\\
-1, &\Psi_{G,\lambda}(u \sra v) < 0.
\end{cases} 
\]
Note that $|\Psi_{G,\lambda}(u \sra v)| = s_v \Psi_{G,\lambda}(u \sra v)$. 
We then define the perturbed fugacity vector $\blambda(\zeta)$ by $\lambda_v(\zeta) = \lambda + s_v \zeta$ for $v \neq u$ and $\lambda_u(\zeta) = \lambda$.
Consider the complex function 
\[
f(\zeta) = \frac{\lambda}{R_{G,\lambda}(u)} R_G\big( u; \blambda(\zeta) \big).
\]

\begin{claim}
\label{clm:domain}
For $\eps$ given by \cref{thm:zero-free}, the complex function $f$ is holomorphic in $\DD(0,\eps)$. 
\end{claim}

\begin{proof}
As in \cref{eq:ratio-pf,eq:ratio-pf2}, we can write
\[
R_G\big( u; \blambda(\zeta) \big) = \frac{\lambda Z_G^u\big( \blambda(\zeta) \big)}{Z_G^{\bar{u}}\big( \blambda(\zeta) \big)}. 
\]
For any $\zeta \in \DD(0,\eps)$, we observe that $\lambda_v(\zeta) = \lambda + s_v \zeta \in \DD(\lambda,\eps) \subseteq \DD( [0,(1-\delta)\lambda_c(\Delta)], \eps )$ for any $v \neq u$, and hence, by \cref{thm:zero-free}, $Z_G^{\bar{u}}\big( \blambda(\zeta) \big) \neq 0$  (note that $Z_G^{\bar{u}}$ is the independence polynomial for the graph $G \setminus u$ of maximum degree at most $\Delta$). A rational function is holomorphic if the denominator is zero-free. 
Thus, $R_G\big( u; \blambda(\zeta) \big)$ is holomorphic in $\DD(0,\eps)$ and so is $f$ (since the term $\lambda/R_{G,\lambda}(u)$ is a real number).
\end{proof}

\begin{claim}
\label{clm:derivative-f}
We have
\[
f'(0) = \sum_{v \in V\setminus\{u\}} \left| \Psi_{G,\lambda}(u \sra v) \right|. 
\]
\end{claim}

\begin{proof}
By the chain rule, we have
\begin{align*}
f'(0) &= \frac{\lambda}{R_{G,\lambda}(u)} \frac{\dif}{\dif \zeta} R_G\big( u; \blambda(\zeta) \big) \bigg|_{\zeta = 0} \\
&= \sum_{v \in V \setminus \{u\}} 
\underbrace{ \frac{\lambda}{R_{G,\lambda}(u)} \left(\frac{\partial}{\partial \lambda_v} R_G\big( u; \blambda(\zeta) \big)\right) \bigg|_{\zeta = 0} }_{= \Psi_{G,\lambda}(u \text{\tiny$\,\rightarrow\,$} v) \text{~by \cref{clm:inf-derivative}}}
\underbrace{ \left( \frac{\dif \lambda_v}{\dif \zeta} \right) \bigg|_{\zeta = 0} }_{= s_v} \\
&= \sum_{v \in V \setminus \{u\}} \left| \Psi_{G,\lambda}(u \sra v) \right|,
\end{align*}
as claimed. 
\end{proof}

\begin{claim}
The image of $f$ is contained in $\overline{\DD}\big(0,\lambda^2/(\eps R_{G,\lambda}(u)) \big)$. 
\end{claim}

\begin{proof}
Observe that 
\begin{align*}
& R_G\big( u; \blambda(\zeta) \big) = y \\
\Longleftrightarrow\qquad& \frac{\lambda Z_G^u\big( \blambda(\zeta) \big)}{Z_G^{\bar{u}}\big( \blambda(\zeta) \big)} = y \\
\Longleftrightarrow\qquad& \left(-\frac{\lambda}{y}\right) Z_G^u\big( \blambda(\zeta) \big) + Z_G^{\bar{u}}\big( \blambda(\zeta) \big) = 0\\
\Longleftrightarrow\qquad& Z_G\big( \boldsymbol{\rho}(\zeta) \big) = 0,
\end{align*}
where 
\[
\rho_v(\zeta) 
= 
\begin{cases}
\lambda_v(\zeta), & v \neq u; \\
-\frac{\lambda}{y}, & v = u.
\end{cases}
\]
Hence, we deduce from \cref{thm:zero-free} that 
\[
-\frac{\lambda}{y} \notin \DD( [0,(1-\delta)\lambda_c(\Delta)], \eps ).
\]
In particular,
\[
-\frac{\lambda}{y} \notin \DD(0,\eps)
\quad\Longrightarrow\quad y \in \overline{\DD}\left(0,\frac{\lambda}{\eps}\right). 
\]
The claim then follows. 
\end{proof}

Let $B=\lambda^2/(\eps R_{G,\lambda}(u))$, and let $g(z) = \frac{1}{B} f(\eps z)$.
We can now conclude the proof of (the weaker version of) \cref{lem:plan} as follows:
\begin{align*}
\sum_{v \in V} |\Psi_{G,\lambda}(u \sra v)| 
&=
1 + \sum_{v \in V\setminus\{u\}} |\Psi_{G,\lambda}(u \sra v)| \\
&= 
1+f'(0) 
& \mbox{by \cref{clm:derivative-f}}
\\
&= 
1+\frac{B}{\eps} g'(0) 
& \mbox{by the definition of $g(z)$}
\\
& \le 1+\frac{B}{\eps} 
& \mbox{by \cref{lem:Schwarz}}
\\
& = 
1 + \frac{\lambda^2}{\eps^2 R_{G,\lambda}(u)} 
& \mbox{by the definition of $B$}
\\ &
= 1+O(\lambda/\eps^2), 
\end{align*}
where the last line follows from $R_{G,\lambda}(u) \ge \lambda/(1+\lambda)^\Delta = \Omega(\lambda)$ when $\lambda = O(1/\Delta)$, which is the case when $\lambda$ is in the tree-uniqueness regime. 
This completes the proof of (the weaker version of) \cref{lem:plan} as intended.
\end{proof}

\begin{remark}
Note the above proof establishes $\eta$-spectral independence for $\eta=O(\lambda/\eps^2)$ where $\eps=\eps(\delta)$ is the zero-free radius that appears in \cref{thm:zero-free}.
\end{remark}

\section{Spectral Independence via Coupling Independence}
\label{sec:CI}

The coupling independence approach was introduced in Chen and Zhang~\cite{CZ23}, though it was influenced by earlier works utilizing disagreement percolation to establish Gibbs uniqueness~\cite{vS94} and recursive coupling to prove strong spatial mixing~\cite{GMP05}.
Since its introduction, coupling independence has been widely applied due to its simplicity and convenience.
Roughly speaking, it aims to construct a coupling between two conditional Gibbs distributions $\mu_{G,\lambda}^u$ where we fix a vertex $u$ to be occupied, and $\mu_{G,\lambda}^{\bar{u}}$ where we fix $u$ to be unoccupied, such that the expected Hamming distance of the two samples is as small as possible.

For $\sigma,\tau \in \{0,1\}^n$, the Hamming distance between them is
\begin{align*}
    d_{\mathrm{H}}(\sigma,\tau) := \sum_{v \in V} \ind\{\sigma(v) \neq \tau(v)\} = \norm{\sigma-\tau}_1.
\end{align*}

The Wasserstein distance between two distributions $\mu,\nu$, denoted as $W_1(\mu,\nu)$, is the expected Hamming distance under the optimal coupling between $\mu$ and $\nu$.  Hence, to upper bound the Wasserstein distance it suffices to provide a coupling and upper bound the expected Hamming distance under this coupling.

\begin{definition}[Wasserstein Distance $W_1$]
    For two distributions $\mu,\nu$ over $\{0,1\}^n$, the \emph{$1$-Wasserstein distance} between $\mu$ and $\nu$ with respect to the Hamming distance is defined as
    \begin{align*}
        W_1(\mu,\nu) := \inf_{\pi \in \PP(\mu,\nu)} \E_{(X,Y) \sim \pi}\left[ d_{\mathrm{H}}(X,Y) \right],
    \end{align*}
    where $\PP(\mu,\nu)$ denotes the family of all couplings of $\mu$ and $\nu$.
\end{definition}

We can now define coupling independence. 
\begin{definition}[Coupling Independence (without pinning)]
    Let $\mu$ be a distribution over $\{0,1\}^n$ where $V$ is a finite set.
    For each $u \in V$, let $\mu^u(\cdot) = \mu(\cdot \mid \sigma(u) = 1)$ be the conditional distribution where $\sigma(u)$ is fixed to be $1$ and $\mu^{\bar{u}}(\cdot) = \mu(\cdot \mid \sigma(u) = 0)$ be the conditional distribution where $\sigma(u)$ is fixed to be $0$. 
    We say $\mu$ is $C_{\mathsf{CI}}$-coupling independent if for every $u \in V$ such that $\min\{\mu(u),\mu(\bar{u})\} > 0$, it holds
    \begin{align}\label{eq:CI-def}
        W_1\left( \mu^u, \mu^{\bar{u}} \right) \le C_{\mathsf{CI}}.
    \end{align}
\end{definition}

\begin{remark}
    Rigorously, just like spectral independence, one needs to define coupling independence under arbitrary pinnings, that is, \cref{eq:CI-def} holds true for any conditional Gibbs distribution $\mu_\tau$ 
    where $\tau: \Lambda \to \{0,1\}$ for $\Lambda\subset V$ is a valid pinning. For simplicity, we only consider the version without pinning in this section, which makes the connection clearer between coupling independence and the maximum eigenvalue of the influence matrix.
\end{remark}

It is immediate that coupling independence implies spectral independence, as can be seen by the following lemma.

\begin{lemma}\label{lem:CI-SI}
    Let $\mu$ be a distribution over $\{0,1\}^n$ where $V$ is a finite set.
    If $\mu$ is $C_{\mathsf{CI}}$-coupling independent, then 
    \begin{align*}
        \uplambda_{\max}(\Psi_\mu) \le C_{\mathsf{CI}}.
    \end{align*}
\end{lemma}
\begin{proof}
    Fix $u \in V$. Since $V$ is finite, there is a coupling $\pi$ achieving $W_1\left( \mu^u, \mu^{\bar{u}} \right)$, such that
    \begin{align*}
        \E_{(X,Y) \sim \pi}\left[ d_{\mathrm{H}}(X,Y) \right] \le C_{\mathsf{CI}}.
    \end{align*}
    It follows that
    \begin{align*}
        \sum_{v \in V} |\Psi_\mu(u \sra v)|
        &= \sum_{v \in V} |\mu^u(v) - \mu^{\bar{u}}(v)| \\
        &= \sum_{v \in V} \left| \E_{\pi}[X(v)] - \E_{\pi}[Y(v)] \right| \\
        &\le \sum_{v \in V} \E_{\pi}\left[ \left| X(v) - Y(v) \right| \right] \\
        &= \E_\pi\left[ d_{\mathrm{H}}(X,Y) \right] \le C_{\mathsf{CI}}.
    \end{align*}
    Hence, we conclude that $\uplambda_{\max}(\Psi_\mu) \le \norm{\Psi_\mu}_\infty \le C_{\mathsf{CI}}$ by \cref{lem:rowsum}.
\end{proof}

Thus, to establish spectral independence, it suffices to construct a coupling between $\mu^u$ and $\mu^{\bar{u}}$ which has small Hamming distance in expectation.

For the hard-core model, Chen and Feng~\cite{ChenFeng} showed that coupling independence holds up to the uniqueness threshold, and hence yields an alternative proof of~\cref{lem:plan}. 

\begin{theorem}[\cite{ChenFeng}]
Let $\Delta \ge 3$ and $\delta \in (0,1)$. 
For any graph $G=(V,E)$ of maximum degree at most $\Delta$, any vertex $u \in V$, and any $\lambda \le (1-\delta) \lambda_c(\Delta)$, it holds
\begin{equation*}
W_1\left( \mu_{G,\lambda}^u, \mu_{G,\lambda}^{\bar{u}} \right) = O\left( \frac{1}{\delta} \right).
\end{equation*}
\end{theorem}

The above result of \cite{ChenFeng} utilizes correlation-decay arguments based on self-avoiding walk trees, similar to what is done in~\cref{sec:Weitz}.

Since its introduction, coupling independence has seen great success in proving spectral independence for general multi-spin systems due to its simple combinatorial nature. Besides the hard-core model, 
example applications also include the ferromagnetic Ising model \cite{CZ23}, $k$-SAT \cite{CGGGHMM24}, $b$-matchings \cite{CG24}, vertex colorings \cite{CLMM23,ChenFeng,CFGZZ25}, and edge colorings \cite{CWZZ25}.

In addition to implying spectral independence and hence optimal mixing of the Glauber dynamics, coupling independence has also found applications in deterministic approximate counting algorithms~\cite{CFGZZ25}. 
In particular, Chen, Feng, Guo, Zhang, and Zou~\cite[Theorem 4]{CFGZZ25} showed that if $C$-coupling independence (under arbitrary pinnings) and $b$-marginal boundedness hold then there is an $\fptas$ for the partition function where the polynomial exponent in the running time depends on the constants $C$ and $b$ and also on the maximum degree~$\Delta$ of the input graph.

In the following theorem, we prove that coupling independence holds for the hard-core model for sufficiently small $\lambda$ as an illustrative example. 
Unlike \cref{lem:plan} which we proved earlier, the following \cref{lem:CI-HC} works for fugacity strictly below the tree-uniqueness threshold $\lambda_c(\Delta)$.

\begin{theorem}
\label{lem:CI-HC}
Let $\Delta \ge 3$ and $\delta \in (0,1)$.
For any graph $G=(V,E)$ of maximum degree at most $\Delta$, any vertex $u \in V$, and any $\lambda \le (1-\delta)/(\Delta-2)$, it holds
\begin{equation*}
W_1\left( \mu_{G,\lambda}^u, \mu_{G,\lambda}^{\bar{u}} \right) \le \frac{1+2\lambda}{\delta}.
\end{equation*}
\end{theorem}

\begin{proof}   
    Throughout the proof, we omit $\lambda$ in the subscript of the Gibbs distribution for simplicity.
    We first prove that for any graph $G$ of maximum degree at most $\Delta$ and any vertex $u \in V$ of degree $\le \Delta-1$, it holds
    \begin{equation}\label{eq:deg-Delta-1}
        W_1\left( \mu_G^u, \mu_G^{\bar{u}} \right) \le \frac{1+\lambda}{\delta}.
    \end{equation}
    We prove \cref{eq:deg-Delta-1} by induction on the number of vertices.
    If $|V|=1$ then \cref{eq:deg-Delta-1} is trivial.
    Now suppose \cref{eq:deg-Delta-1} holds true for any graph $G$ of maximum degree at most $\Delta$ and size at most $n-1$.
    Fix an $n$-vertex graph $G$ of maximum degree at most $\Delta$ and a vertex $u \in V$ of degree $\le \Delta-1$.
    Notice that if $\sigma(u)=1$ then $\sigma(v)=0$ for all $v \in N$ where $N = N(u)=\{v\in V: \{u,v\} \in E(G)\}$ is the set of neighbors of $u$. Meanwhile, if $\sigma(u)=0$ then we can decompose the conditional distribution $\mu_G^{\bar{u}}$ as 
    \begin{align*}
        \mu_G^{\bar{u}} (\cdot) = \sum_{\tau \in \{0,1\}^N} \mu_G^{\bar{u}} (\tau) \mu_G^{\bar{u},\tau} (\cdot),
    \end{align*}
    where $\mu_G^{\bar{u}} (\tau)$ is the probability that the neighborhood $N$ of $u$ receives the configuration $\tau$, and $\mu_G^{\bar{u},\tau} (\cdot)$ is the conditional Gibbs distribution where we fix $\sigma(u)=0$ and $\sigma(N) = \tau$.
    Let $H=G\setminus \{u\}$ denote the induced subgraph obtained by removing the vertex $u$.
    We then deduce that
    \begin{align*}
        W_1\left( \mu_G^u, \mu_G^{\bar{u}} \right)
        &\le \sum_{\tau \in \{0,1\}^N} \mu_G^{\bar{u}} (\tau) \cdot W_1\left( \mu_G^{u,\mathbf{0}}, \mu_G^{\bar{u},\tau} \right)\\
        &= 1 + \sum_{\tau \in \{0,1\}^N} \mu_G^{\bar{u}} (\tau) \cdot W_1\left( \mu_H^{\mathbf{0}}, \mu_H^{\tau} \right),
    \end{align*}
    where we write $\mu_G^{u,\mathbf{0}}=\mu_G^u$ to emphasize that the neighborhood $N$ receives the all-zero configuration~$\mathbf{0}$.
    Now for every feasible $\tau \in \{0,1\}^N$, we construct a sequence of configurations on $N$ denoted as $\mathbf{0}=\tau^{(0)},\tau^{(1)},\dots,\tau^{(k)} = \tau$ by flipping $0$'s to $1$'s one at a time, where $k=d_{\mathrm{H}}(\mathbf{0},\tau) = |\tau|$. Hence, we get
    \begin{align*}
        W_1\left( \mu_H^{\mathbf{0}}, \mu_H^{\tau} \right) 
        \le \sum_{i=1}^k W_1\left( \mu_H^{\tau^{(i-1)}}, \mu_H^{\tau^{(i)}} \right)
        \le \frac{(1+\lambda)k}{\delta},
    \end{align*}
    where the last inequality follows from the induction hypothesis since $\tau^{(i-1)}$ and $\tau^{(i)}$ differ by one vertex of degree $\le\Delta-1$ in the subgraph $H=G\setminus \{u\}$.
    Therefore, we conclude that
    \begin{align*}
        W_1\left( \mu_G^u, \mu_G^{\bar{u}} \right)
        \le 1 + \sum_{\tau \in \{0,1\}^N} \mu_G^{\bar{u}} (\tau) \cdot \frac{(1+\lambda)|\tau|}{\delta} 
        \le 1 + \frac{1+\lambda}{\delta} \cdot (\Delta-1) \cdot \frac{\lambda}{1+\lambda} \le \frac{1+\lambda}{\delta}.
    \end{align*}
    Finally, if the degree of $u$ is $\Delta$, then the same argument as above shows
    \begin{align*}
        W_1\left( \mu_G^u, \mu_G^{\bar{u}} \right)
        \le 1 + \frac{1+\lambda}{\delta} \sum_{\tau \in \{0,1\}^N} \mu_G^{\bar{u}} (\tau) \cdot |\tau|
        \le 1 + \frac{1+\lambda}{\delta} \cdot \Delta \cdot \frac{\lambda}{1+\lambda} 
        \le \frac{1+2\lambda}{\delta},
    \end{align*}
    as claimed.
\end{proof}

We end this section by showing that if there exists a contractive coupling of a Markov chain, then coupling independence and thus spectral independence hold.
The following lemma states that for any Markov chain, not necessarily the Glauber dynamics, a contractive coupling implies an upper bound on the Wasserstein distance.

\begin{lemma}[{\cite[Lemma 19]{CFGZZ25}}]
\label{lem:cc-CI}
    Let $\mu,\nu$ be distributions supported on $\Omega \subseteq [q]^n$ where $q \in \N^+$ and $V$ is finite. 
    Let $\rho>0$ 
    and $\kappa \in (0,1)$ be reals.
    Denote by $P,Q \in \R^{|\Omega| \times |\Omega|}$ the transition matrices of (not necessarily irreducible) Markov chains with stationary distributions $\mu,\nu$ respectively.
    Suppose the following two conditions are satisfied:
    \begin{enumerate}
        \item For any $\sigma \in \Omega$, it holds 
        \begin{align}\label{eq:cc-CI-cond1}
            W_1\left( P(\sigma,\cdot), Q(\sigma,\cdot) \right) \le \rho;
        \end{align}
        
        \item ($\kappa$-contractive coupling of $Q$) For any $\sigma,\tau \in \Omega$, it holds
        \begin{align}\label{eq:cc-CI-cond2}
            W_1\left( Q(\sigma,\cdot), Q(\tau,\cdot) \right) \le (1-\kappa) d_{\mathrm{H}}(\sigma,\tau).
        \end{align}
    \end{enumerate}
    Then, we have
    \begin{align*}
        W_1\left( \mu, \nu \right) \le \frac{\rho}{\kappa}.
    \end{align*}
\end{lemma}

\begin{proof}
    Since the Wasserstein distance is a metric, we deduce from the triangle inequality that
    \begin{align}\label{eq:cc-CI-term0}
         W_1\left( \mu, \nu \right)
         =  W_1\left( \mu P, \nu Q \right)
         \le  W_1\left( \mu P, \mu Q \right) +  W_1\left( \mu Q, \nu Q \right).
    \end{align}
    For the first term, we observe that
    \begin{align}\label{eq:cc-CI-term1}
        W_1\left( \mu P, \mu Q \right) \le \sum_{\sigma \in \Omega} \mu(\sigma) W_1\left( P(\sigma,\cdot), Q(\sigma,\cdot) \right) \le \rho,
    \end{align}
    which follows from \cref{eq:cc-CI-cond1}.
    For the second term, $\kappa$-contractive coupling \cref{eq:cc-CI-cond2} implies that
    \begin{align*}
        W_1\left( \mu Q, \nu Q \right)
        \le \sum_{\sigma,\tau} \pi(\sigma,\tau) W_1\left( Q(\sigma,\cdot), Q(\tau,\cdot) \right)
        \le (1-\kappa) \sum_{\sigma,\tau} \pi(\sigma,\tau) d_{\mathrm{H}}(\sigma,\tau),
    \end{align*}
    where $\pi$ is an arbitrary coupling of $\mu$ and $\nu$.
    Taking infimum over $\pi$, we thus obtain
    \begin{align}\label{eq:cc-CI-term2}
        W_1\left( \mu Q, \nu Q \right) \le (1-\kappa) W_1\left( \mu, \nu \right).
    \end{align}
    Combining \cref{eq:cc-CI-term0,eq:cc-CI-term1,eq:cc-CI-term2}, the lemma then follows.
\end{proof}

We remark that \cref{lem:cc-CI} holds more generally for any metric space $(\Omega,d)$ with the Wasserstein distance defined with respect to the metric $d$; see \cite{CFGZZ25}.


We provide here a simple illustration of the above techniques.  We show that there is a contractive coupling for the hard-core model when $\lambda$ is sufficiently small, and then we can apply \cref{lem:cc-CI} to conclude coupling independence.

\begin{lemma}
\label{lem:hc-small-lambda-coupling}
    Let $\Delta \ge 3$ and $\delta \in (0,1)$.
    For any $n$-vertex graph $G=(V,E)$ of maximum degree at most $\Delta$ and any $\lambda \le (1-\delta)/(\Delta-1)$, there is a $\kappa$-contractive coupling of the Glauber dynamics with stationary distribution $\mu_{G,\lambda}$ where $\kappa = \delta/((1+\lambda)n)$. Specifically, for any two independent sets $\sigma,\tau$ of $G$ it holds
    \begin{equation*}
    W_1\left( P_{\textsc{GD}}(\sigma,\cdot), P_{\textsc{GD}}(\tau,\cdot) \right) \le \left( 1 - \frac{\delta}{(1+\lambda)n} \right) d_{\mathrm{H}}(\sigma,\tau).
    \end{equation*}
\end{lemma}

\begin{proof}
    We use path coupling to establish a contractive coupling; we refer the reader to \cite[Chapter 4]{Jerrum:notes} and \cite[Chapter 14.2]{LevinPeresWilmer} for an introduction to path coupling.

    We first consider the case where $d_{\mathrm{H}}(\sigma,\tau)=1$, i.e., $\sigma$ and $\tau$ differ at only one vertex which we denote as $z$.
    We construct an identity coupling $\pi$ for the transition $(\sigma,\tau) \to (X,Y)$ where $X \sim P_{\textsc{GD}}(\sigma,\cdot)$ and $Y \sim P_{\textsc{GD}}(\tau,\cdot)$. 
    More specifically, from $(\sigma,\tau)$, we choose a random vertex $v$ for update in both chains, and then with probability $\lambda/(1+\lambda)$ both chains attempt to add $v$ to the current independent set (it might succeed in both chains, only one of the chains, or in neither chain), and with probability $1/(1+\lambda)$ we remove $v$ from both chains.

Note that if the updated vertex $v=z$ then $X=Y$ with probability 1, since $\sigma(w)=\tau(w)$ for all $w\in N(z)$. 
If $v \in V \setminus (\{z\} \cup N(z))$, then similarly $d_{\mathrm{H}}(X,Y) = 1$ with probability 1.
Finally, if $v\in N(z)$ then with probability $\leq\lambda/(1+\lambda)$ the Hamming distance increases by 1, and otherwise remains the same. 
Therefore, under this identity coupling $\pi$ we deduce that
\[
W_1\left( P_{\textsc{GD}}(\sigma,\cdot), P_{\textsc{GD}}(\tau,\cdot) \right) - 1
\leq \frac{1}{n}\left(\frac{\Delta\lambda}{1+\lambda}-1\right) = \frac{(\Delta-1)\lambda-1}{(1+\lambda)n} 
\leq -\frac{\delta}{(1+\lambda)n},
\]
which establishes a $\delta/((1+\lambda)n)$-contraction when $d_{\mathrm{H}}(\sigma,\tau)=1$.

For the general case, we construct a sequence of independent sets $\sigma = \sigma^{(0)}, \sigma^{(1)}, \dots, \sigma^{(k)} = \tau$ where $k = d_{\mathrm{H}}(\sigma,\tau)$, obtained by first removing vertices in $\sigma \setminus \tau$ one by one to reach $\sigma \cap \tau$, and then adding vertices in $\tau \setminus \sigma$ one by one.
Note that $d_{\mathrm{H}}(\sigma^{(i-1)},\sigma^{(i)}) = 1$ for each $i$.
It follows from the triangle inequality that
\begin{align*}
    W_1\left( P_{\textsc{GD}}(\sigma,\cdot), P_{\textsc{GD}}(\tau,\cdot) \right)
    \le \sum_{i=1}^k W_1\left( P_{\textsc{GD}}(\sigma^{(i-1)},\cdot), P_{\textsc{GD}}(\sigma^{(i)},\cdot) \right)
    \le \sum_{i=1}^k (1-\kappa) 
    = (1-\kappa) d_{\mathrm{H}}(\sigma,\tau),
\end{align*}
as claimed.
\end{proof}

By \cref{lem:cc-CI,lem:hc-small-lambda-coupling}, we deduce the following coupling independence result, where the regime of $\lambda$ is slightly 
worse than \cref{lem:CI-HC}.
\begin{theorem}
\label{lem:CI-HC-2}
Let $\Delta \ge 3$ and $\delta \in (0,1)$.
For any graph $G=(V,E)$ of maximum degree at most $\Delta$, any vertex $u \in V$, and any $\lambda \le (1-\delta)/(\Delta-1)$, it holds
\begin{equation*}
W_1\left( \mu_{G,\lambda}^u, \mu_{G,\lambda}^{\bar{u}} \right) \le \frac{1+\lambda}{\delta}.
\end{equation*}
\end{theorem}

\begin{proof}
    Let $\Omega$ denote the set of all independent sets of $G$. 
    Let $P,Q \in \R^{\Omega\times\Omega}$ be the transition matrices of the Glauber dynamics with stationary distributions $\mu_{G,\lambda}^u$ and $\mu_{G,\lambda}^{\bar{u}}$ respectively.
    Note that $P,Q$ are reducible; their transition rules are the same as the Glauber dynamics for $\mu_{G,\lambda}$ when updating a vertex $v \neq u$, while $P$ forces $\sigma(u)=1$ and $Q$ forces $\sigma(u)=0$ when $u$ is picked to be updated.
    Therefore, for any independent set $\sigma \in \Omega$, we have $W_1(P(\sigma,\cdot), Q(\sigma,\cdot)) \le 1/n$, establishing \cref{eq:cc-CI-cond1}. Meanwhile, \cref{lem:hc-small-lambda-coupling} shows a $\kappa$-contractive coupling \cref{eq:cc-CI-cond2} for $Q$ with $\kappa = \delta/((1+\lambda)n)$. 
    The theorem then follows from \cref{lem:cc-CI}.
\end{proof}

\section{Optimal relaxation time implies SI}
\label{sub:coupling-SI}

In this section, we show the result of Anari, Jain, Koehler, Pham, and Vuong~\cite{relax-optimal} that an optimal relaxation time for the Glauber dynamics, namely $O(n)$ relaxation time, implies spectral independence; this was stated in \cref{lem:opt-relax-SI} of  \cref{sub:mixing-via-SI}, which we restate now for convenience.

\optrelaxSI*

\begin{proof}
Recall (see equation~\cref{cov-SI-connection}) that $(C-1)$-spectral independence is equivalent to the following inequality:
\begin{equation}\label{eq:specin}
 \Cov_\mu\preceq C D,
\end{equation}
where $D$ is diagonal matrix with $D(i,i)=\mu(\sigma(i)=1)\mu(\sigma(i)=0)$. That is, for any $y\in\R^n$
we need to show $y^T \Cov_\mu y \leq C y^T D y$.

Recall (see~\cref{defn:approx-tensorization-var} and~\cref{lem:AT-gap}) that relaxation time being bounded by $Cn$ implies that for any $f:\{0,1\}^n\rightarrow\R$ we have
\begin{equation}\label{eq:variancefact}
\Var_\mu(f) \leq C\sum_{v\in V} \Exp_{\mu}[\Var_v(f)].
\end{equation}
Let $f(\sigma) = \sum_{v\in V} y_v \sigma(v)$. Note that we have 
\begin{equation}
    \label{var-cov}
\Var_\mu(f)  = y^T\Cov_\mu y
\end{equation}
and also the following pointwise and averaged identities:
\begin{align*}
\Var_v(f)
& =
y_v^2\,\chi\!\left(\Exp[X_v | X_1,\dots,X_{v-1},X_{v+1},\dots, X_n]\right)
\\
\Exp_\mu[\Var_v(f)]
& =
y_v^2\,
\Exp_\mu\!\left[
\chi\!\left(\Exp[X_v | X_1,\dots,X_{v-1},X_{v+1},\dots, X_n]\right)
\right],
\end{align*}
where $\chi(x)=x(1-x)$. Since $x(1-x)$ is concave we have
\begin{align*}
\Exp_\mu[ \chi( \Exp[X_v | X_1,\dots,X_{v-1},X_{v+1},\dots, X_n]) ]&\leq
\chi(\Exp_\mu[ \Exp[X_v | X_1,\dots,X_{v-1},X_{v+1},\dots, X_n] ])\\
& = \mu(\sigma(v)=1)\mu(\sigma(v)=0),
\end{align*}
and hence $\Exp_\mu[\Var_v(f)] \leq y_v^2\mu(\sigma(v)=1)\mu(\sigma(v)=0)$. Plugging this into~\eqref{eq:variancefact} (and applying~\cref{var-cov}) we obtain $y^T\Cov_\mu y\leq C y^T D y$. Since we can take an arbitrary 
$y$, we established~\eqref{eq:specin}, and hence $(C-1)$-spectral independence. 
\end{proof}

\subsection{Contractive coupling implies relaxation time}

Here, we show an example application of \cref{lem:opt-relax-SI}.
We show that a contractive coupling of an arbitrary Markov chain not only implies an upper bound on the mixing time but also provides a lower bound on the spectral gap.  Consequently, for a local Markov chain one obtains spectral independence, which provides an alternative proof approach compared with \cref{lem:cc-CI} which utilizes coupling independence.

\begin{restatable}{lemma}{couplingrelax}
\label{lem:coupling-relax}
Suppose that for a Markov chain on $\Omega$ with transition matrix $P$, we have a contractive coupling with contraction factor $\kappa \in (0,1)$; that is, we have a distance function $d:\Omega\times\Omega\rightarrow\R_{\geq 0}$ such that $d(\sigma,\tau)=0$ iff $\sigma=\tau$ and, for any $\sigma,\tau\in\Omega$, a coupling $\pi$ of $P(\sigma,\cdot)$ and $P(\tau,\cdot)$ such that 
\begin{equation}
    \label{contr-coupling-cond}
\Exp_{(X,Y) \sim \pi}[ d(X,Y) ] \leq (1-\kappa) d(\sigma,\tau).
\end{equation}
Then the absolute spectral gap of $P$ satisfies $\gamma=1-\uplambda_*\geq \kappa$, and hence the relaxation time of $P$ satisfies $\Trelax\leq 1/\kappa$.
\end{restatable}

For the Glauber dynamics we typically obtain $\kappa = \Omega(1/n)$.  If a Markov chain is local in the sense that the transitions update $O(1)$ vertices in each step then optimal relaxation time for this local chain implies optimal relaxation time for the Glauber dynamics, see~\cref{rem:comparison} for further discussion and associated references.

\begin{proof}
For simplicity, first consider a reversible Markov chain.
Let $\varphi:\Omega\rightarrow\R$ be an eigenvector
corresponding to a nontrivial eigenvalue $\uplambda$ of maximum
absolute value, so that $|\uplambda|=\uplambda_*$.
Denote
\begin{align*}
    L = \max_{\substack{\sigma,\tau \in \Omega \\ \sigma \neq \tau}}
    \frac{|\varphi(\sigma)-\varphi(\tau)|}{d(\sigma,\tau)},
\end{align*}
and let $\sigma,\tau$ be the maximizers; hence, $|\varphi(\sigma)-\varphi(\tau)| = L \cdot d(\sigma,\tau)$. 
Therefore, we have the following:
\begin{equation}
    \label{cond-eigenfunction}
|\varphi(X) - \varphi(Y)| \leq L \cdot d(X,Y).
\end{equation}
 Note that since $\varphi$ is an eigenvector with eigenvalue $\uplambda$, we have
 \begin{align*}
     \Exp[\varphi(X)] 
     = \sum_{\xi\in\Omega} P(\sigma,\xi) \varphi(\xi) 
     = (P\varphi) (\sigma)
     = \uplambda\varphi(\sigma),
 \end{align*}
and hence 
\begin{align*}
|\uplambda| \cdot|\varphi(\sigma)-\varphi(\tau)| 
&= \big| \Exp [ \varphi(X) - \varphi(Y) ] \big|
\\
&\leq  \Exp\left[ \big| \varphi(X) - \varphi(Y) \big| \right] 
\\
& \leq L \cdot \Exp[d(X,Y)] & \mbox{by \cref{cond-eigenfunction}}
\\
&\leq L \cdot (1-\kappa) d(\sigma,\tau) & \mbox{by \cref{contr-coupling-cond}}
\\
&= (1-\kappa) |\varphi(\sigma)-\varphi(\tau)|,
\end{align*}
and thus $|\uplambda|\leq 1-\kappa$, which proves the lemma for
reversible chains. For nonreversible chains, one instead considers
an eigenvector $\varphi:\Omega\rightarrow\C$ corresponding to a
nontrivial eigenvalue $\uplambda$ of maximum modulus, so that
$|\uplambda|=\uplambda_*$. Using the complex modulus throughout,
the same argument gives $\uplambda_*\leq 1-\kappa$.
\end{proof}

\begin{remark}\label{rem:connection-mixing-to-relax}
As noted above, for the Glauber dynamics, in the context of \cref{lem:coupling-relax} we typically obtain $\kappa = \alpha/n$ for some constant $\alpha>0$; see e.g. \cref{lem:hc-small-lambda-coupling}.  Note, a contractive coupling with such a contraction factor of $\kappa = \alpha/n$ implies a mixing time bound of $\Tmix(\eps)=O(n\log(n/\eps))$.  
The more general statement is that $\Tmix(\eps)=O(n\log(n/\eps))$ implies $\Trelax=O(n)$; we refer the interested reader to Levin and Peres~\cite[Chapter 14 notes]{LevinPeresWilmer} for a proof of the more general implication.
\end{remark}

\begin{remark}
\label{rem:comparison}
    Recall, \cref{lem:coupling-relax} applies to any Markov chain.
If the Markov chain is local in the sense that only $O(1)$ vertices are updated at each step, then $O(n)$ relaxation time for this local chain implies $O(n)$ relaxation time for the Glauber dynamics, see~\cite[Section 4.3]{BCCPSV22} for further details.  Applying~\cref{lem:opt-relax-SI} then establishes spectral independence and subsequently, assuming marginal boundedness, we obtain $O(n\log{n})$ mixing time for the Glauber dynamics for constant $\Delta$, via~\cref{thm:SI-constant-mix}.  This improves upon standard comparison arguments (comparing Dirichlet forms) which use $O(n)$ relaxation time of a local Markov chain to conclude $O(n^2)$ mixing time of the Glauber dynamics.
\end{remark}

\chapter{Critical Point Analysis using the Field Dynamics}
\label{sec:critical-point}

In this chapter we prove that the Glauber dynamics has polynomial mixing time at the critical point $\lambda=\lambda_c(\Delta)$ as stated in \cref{thm:hard-core-critical}.  This theorem was established by Chen, Chen, Yin, and Zhang \cite[Theorem 1.1]{CCYZ-critical-hard-core} who proved that the exponent satisfies $C\leq 2+4e + 4e/(\Delta-2)$; we simply show there exists a universal constant.

We restate \cref{thm:hard-core-critical} for convenience.

\criticalpoint*

The above result follows as a corollary of the following theorem, which says that to establish polynomial mixing at a particular $\lambda^*$ we need two properties: (1) for all $\alpha<1$, $(C_{SI}(\alpha)-1)$-spectral independence holds for $\lambda=(1-\alpha)\lambda^*$, and (2) fast mixing at some $\lambda'=(1-\theta)\lambda^*$ for some choice of $\theta<1$; then one obtains relaxation time $O(K(\theta)n)$ where $K(\theta)$ is defined in \cref{defn:K} below.  Note to obtain a polynomial upper bound on $K(\theta)$ and deduce \cref{thm:hard-core-critical}, we further use that $C_{SI}(\alpha)=O(1/\alpha)$.

\begin{theorem}
\label{thm:main-critical}
    Let $\lambda^*>0$. If the following hold:
    \begin{enumerate}
        \item For all \ $0\leq\alpha <1$, there exists $C_{SI}(\alpha)$ such that the measure $\mu_{(1-\alpha)\lambda^*}$ is $(C_{SI}(\alpha)-1)$-spectrally independent,
        \label{assume111}
        \item there exists $0\leq \theta < 1$ where the Glauber dynamics for the measure $\mu_{(1-\theta)\lambda^*}$ has relaxation time $\leq C_{Gl}n$ (on any graph with maximum degree $\Delta$),
        \label{assume222}
        \end{enumerate}
        then the relaxation time of the Glauber dynamics for the hard-core model at fugacity $\lambda^*$ for any graph of maximum degree $\Delta$ satisfies
\[
\Trelax \leq     \left(K(\theta)C_{Gl}\right)n, 
\]
where    
\begin{equation}
    \label{defn:K}
    K(\theta)= \exp\left(\int_0^\theta \frac{C_{SI}(\alpha)}{1-\alpha}d\alpha\right).
\end{equation}
\end{theorem}

We will use \cref{thm:main-critical} to prove \cref{thm:hard-core-critical}.
In \cref{sub:FD} we will give intuition for \cref{thm:main-critical} and then we will dive into its proof.  
At a high level, to prove \cref{thm:main-critical}, we will bound the relaxation time for the Glauber dynamics by establishing approximate tensorization of variance for $\mu_{\lambda^*}$.  The second hypothesis yields approximate tensorization of variance for $\mu_{\lambda'}$ for $\lambda' < \lambda^*$.  We use the first hypothesis, spectral independence at all smaller fugacities, to obtain an analog of approximate tensorization of variance for a type of block dynamics known as the field dynamics.  We then combine these approximate tensorization results to obtain the desired result at $\lambda^*$.

\begin{proof}[Proof of \cref{thm:hard-core-critical}]
Applying \cref{lem:hc-small-lambda-coupling} together with \cref{lem:AT-gap} yields \cref{assume222} by choosing $\theta$ so that $(1-\theta)\lambda_c(\Delta)=(1/2)/(\Delta-1)$ and taking $\delta=1/2$ in \cref{lem:hc-small-lambda-coupling}.
For \cref{assume111} we use \cref{lem:plan} in \cref{sec:Weitz}, which yields 
the following bound 
\begin{equation}\label{eq:specbound}
C_{SI}(\alpha)\leq \min\Big\{n,\frac{c^*}{\alpha}\Big\},
\end{equation}
where $c^*>0$ is an absolute constant (see \cref{lem:plan}) and the trivial upper bound of $n$ follows from the maximum row sum of the influence matrix, see \cref{lem:rowsum}. 
The bound in~\eqref{eq:specbound} is valid for all $\lambda\leq\lambda_c(\Delta)$, in fact, for $\lambda<\lambda_c(\Delta)$
one can get a better bound, see~\cref{rem:big-degree} below. Plugging in~\cref{defn:K} we obtain the following bound
\begin{align*}
\log K(\theta) = \int_0^\theta \frac{C_{SI}(\alpha)}{1-\alpha}d\alpha & \leq 
\int_0^{1/n} \frac{n}{1-\alpha}d\alpha + 
\int_{1/n}^\theta \frac{c^*}{\alpha(1-\alpha)}d\alpha \\
& \leq \frac{1}{1-\theta} + \frac{c^*}{1-\theta} \ln n,
\end{align*}
and hence $K(\theta) = O(n^C)$ where $C=\frac{c^*}{1-\theta}$. The result now follows from~\cref{thm:main-critical}.
\end{proof}

\begin{remark}
\label{rem:big-degree}
    When $\lambda^* = (1-\delta)\lambda_c(\Delta)$ for $\delta>0$ then, using \cref{thm:main-critical}, we will show that $\Trelax\leq C(\delta)n$, and hence the relaxation time is independent of $\Delta$.
    Let $C'_{SI}(\alpha)$ be the bound on spectral independence of $\mu_{(1-\alpha)\lambda_c(\Delta)}$ 
    and $C_{SI}(\alpha)$ be the bound on spectral independence of $\mu_{(1-\alpha)\lambda^*}$.
    In the application of \cref{thm:main-critical} we have 
    \[ C_{SI}(\alpha)=C'_{SI}(1-(1-\alpha)(1-\delta)) \leq \frac{c^*}{\delta+\alpha-\alpha\delta}.
    \]
    And then we obtain:
    \begin{align*}
\log K(\theta) = \int_0^\theta \frac{C_{SI}(\alpha)}{1-\alpha}d\alpha & \leq C' \log(1/\delta),
\end{align*}
for some constant $C'>0$.  This yields $\Trelax \leq C(\delta)n$ where $C(\delta)=\poly(1/\delta)$.
\end{remark}

\section{Field Dynamics}
\label{sub:FD}

In the remainder of \cref{sec:critical-point} we will prove \cref{thm:main-critical}.  To establish mixing at our target fugacity $\lambda^*$ we will use a bootstrapping approach, which utilizes \cref{assume222} in the hypotheses of \cref{thm:main-critical} which says that we have fast mixing of the Glauber dynamics at $\lambda'=(1-\theta)\lambda^*$ for $0<\theta<1$; we satisfy this hypothesis by using an easy argument such as the coupling argument in \cref{lem:hc-small-lambda-coupling}.  Before formally defining the field dynamics, which is an important tool in the proof of \cref{thm:main-critical}, we first attempt to motivate its definition and how it enables a bootstrapping approach.

Given fast mixing of the Glauber dynamics at
$\lambda'=(1-\theta)\lambda^*$, we want to deduce fast mixing at
$\lambda^*$.  In addition, by \cref{assume111} in the hypotheses
of \cref{thm:main-critical}, we know that spectral independence
holds for all $\lambda<\lambda^*$.

As in the proof of \cref{thm:SI-constant-relax}, we consider the
down-up walk.  For fixed $\alpha>0$, the $\alpha n$-block dynamics
at fugacity $\lambda^*$, which was used in the proof of
\cref{thm:SI-constant-relax}, operates as follows.  First, we apply
$\alpha n$ steps of the down-walk, where each step unassigns the
spin at a uniformly random vertex.  Let $S$ denote the resulting
set of $\alpha n$ unassigned vertices.  Subsequently, we apply
$\alpha n$ steps of the up-walk.

The up-walk assigns the spins at the vertices in $S$ one-by-one,
each time sampling from the conditional marginal given the fixed
configuration on $V\setminus S$ and the spins already assigned
within $S$.  Equivalently, by the chain rule for conditional
probabilities, the resulting configuration on $S$ has the same
distribution as a configuration sampled all at once from the Gibbs
distribution conditional on the configuration on $V\setminus S$.

We want to modify the above block dynamics so that in the resampling stage (up-walk steps) we are resampling from the Gibbs distribution at fugacity $\lambda'=(1-\theta)\lambda^*$ 
for $\theta<1$.  To this end, in the corresponding up-walk steps we will condition on a set of vertices which are occupied (in the independent set) and we will resample all of the remaining vertices (those dropped in the down-walk steps and those which are unoccupied initially).  Consequently in the corresponding up-walk steps we will sample from the smaller $\lambda'$ as we have conditioned on a set of occupied vertices (and have not conditioned on any unoccupied vertices) so we are biasing the distribution.  The following Markov chain is known as the {\em field dynamics} and was introduced by Chen, Feng, Yin, and Zhang~\cite{CFYZ22}.

In words, the field dynamics operates as follows. Consider an initial independent set $Y_t$; note, we are considering $Y_t\subseteq V$, we do not consider $Y_t$ to be in $\{0,1\}^n$.  We first remove each $v\in Y_t$ with probability $(1-\theta)$.  
Let $Y'$ be the remaining vertices. Since $Y'\subseteq Y_t$, the set $Y'$ is an independent set. We then resample the configuration on $V\setminus Y'$ from the Gibbs distribution at fugacity $\lambda'=(1-\theta)\lambda^*$ conditional on all vertices in $Y'$ being in the independent set. In the boundary case $\theta=1$, we have $Y'=Y_t$ and the resampling step has $\lambda'=0$, whereas when $\theta=0$, we have $Y'=\emptyset$ and the resampling step has $\lambda'=\lambda^*$.

We now formally define the field dynamics following \cite[Example 2.1]{CFYZ22}.  

\begin{definition}\label{def:field}
The (discrete-time) field dynamics with parameters $(\lambda,\theta)$ where $\lambda>0$ and $0\leq \theta\leq 1$ has the following transitions $Y_t\rightarrow Y_{t+1}$.   
    From an independent set $Y_t$, 
    \begin{enumerate}
    \item Let $Y'=\emptyset$. 
    \item 
  Independently, for each $v\in Y_t$, with probability $\theta$ add $v$ to set $Y'$.
\item 
    Sample an independent set $\widehat{Y}$ on the set $V\setminus (Y'\cup N(Y'))$ from the hard-core measure with fugacity $\lambda'=(1-\theta)\lambda$.
    \\
    (This is equivalent to conditioning on $Y'$ being in the independent set and sampling the remaining configuration from $\mu_{(1-\theta)\lambda}$.)
    \item Let $Y_{t+1}=Y'\cup \widehat{Y}$.
      \end{enumerate}
\end{definition}

We first state the observation that the field dynamics has the desired stationary distribution.

\begin{lemma}
\label{lem:FD-reversible}
    The field dynamics is reversible with respect to $\mu_{\lambda}$.
\end{lemma}

    We will prove the above lemma momentarily.  

    Chen et al. \cite{CFYZ22} utilize the field dynamics to establish a boosting (or inductive) statement, which expresses
the spectral gap of the Glauber dynamics at $\lambda^*$ as (roughly) the product of spectral gap of the field dynamics and the spectral gap of the Glauber dynamics at $\lambda'$; the first term (gap of the field dynamics) can be bounded using \cref{assume111} of \cref{thm:main-critical} and the second term (gap of the Glauber dynamics at $\lambda'$) is given by \cref{assume222}.  

To obtain mixing at the critical point we need a more refined approach.  We want to consider the mixing properties of the field dynamics at an infinitesimally small time and then we can integrate over time.  To enable us to take the derivative with respect to time we need to introduce a 
continuous-time analog of the field dynamics.

We now introduce a continuous-time process, which defines a down step and an up step and together they form the discrete-time field dynamics.   In the following our initial configuration will be an independent set $X_1$.  For $0\leq\theta\leq 1$, we define the down step $X_1\rightarrow X_\theta$ which is a generalization of the step $Y_t\rightarrow Y'$ in the first two steps of the field dynamics.  We then define the up step $X_\theta\rightarrow X'_1$ which generalizes $Y'\rightarrow Y_{t+1}$.

\begin{definition}\label{def:down} Let $X_1$ be from $\mu_\lambda$.
\begin{enumerate}
    \item 
     Independently, for every vertex $v$ choose a uniformly random number $\theta_v$ in $[0,1]$.
\item For $0\leq \theta\leq 1$, let $X_\theta$ be the set of vertices in $X_1$ that have $\theta_v < \theta$.
\end{enumerate}

Let $0\leq\theta\leq\alpha\leq 1$. 
The distribution of $X_\theta$ conditioned on $X_\alpha$ is denoted by $D_{\alpha\rightarrow\theta}(X_{\alpha},\cdot)$ and is referred to as the down step. The distribution of $X_\alpha$ conditioned on $X_\theta$ is denoted by $U_{\theta\rightarrow\alpha}(X_{\theta},\cdot)$ and is referred to as the up step.
\end{definition}

Note that the down step $D_{1\rightarrow\theta}$ is the following. For $Y'\subseteq Y$ 
\begin{align*}
\Pr(X_\theta = Y' | X_1 = Y ) & \propto \Pr\big(\forall v\in Y\setminus Y', \theta_v \in [\theta,1]\big) 
\times \Pr\big(\forall v\in Y', \theta_v \in [0,\theta]\big) \\
& = (1-\theta)^{|Y\setminus Y'|} \theta^{|Y'|},
\end{align*}
that is we drop each element of $Y$ independently with probability $1-\theta$.  Hence, in the discrete-time field dynamics, starting from $Y_t=Y$ then the distribution of $Y'$ is identical to $D_{1\rightarrow\theta}(Y,\cdot)$.

Similarly, for the up step, we have the following: 
\begin{align*}
\Pr(X_1=Y'\cup\widehat{Y} \mid X_\theta = Y')
&=
\frac{
\Pr(X_1=Y'\cup\widehat{Y} \wedge X_\theta = Y')
}{
\Pr(X_\theta = Y')
}
\\
&\propto
\lambda^{|Y'|+|\widehat{Y}|}
\Pr\big(\forall v\in\widehat{Y},\theta_v\in[\theta,1]\big)
\Pr\big(\forall v\in Y',\theta_v\in[0,\theta]\big)
\\
&\propto
(\lambda\theta)^{|Y'|}
\bigl(\lambda(1-\theta)\bigr)^{|\widehat{Y}|}
\\
& \propto
\bigl(\lambda(1-\theta)\bigr)^{|\widehat{Y}|}.
\end{align*}
Therefore, if we consider the discrete-time field dynamics, given $Y'$ at the end of step 2 (after dropping each vertex $v\in Y_t$ with probability $1-\theta$), then the distribution of $Y_{t+1}$ for the discrete-time field dynamics is identical to $X_1$ conditional on $X_\theta = Y'$.  

In summary, a transition of the discrete-time field dynamics is a down step $D_{1\rightarrow\theta}$ followed by an up step $U_{\theta\rightarrow 1}$ (starting from $Y_t$ we condition on $X_1=Y_t$ and use the down step to obtain $X_\theta$, then we 
resample $X'_1$ using the up step from $X'_\theta=X_\theta$, and let $Y_{t+1}=X'_1$).
Observe that, the down step $D_{1\rightarrow\theta}$ and the up step $U_{\theta\rightarrow 1}$ are adjoint to each other with respect to the distributions of $X_1$ and $X_\theta$.

The field dynamics is an example of a localization scheme, which was introduced by \cite{CE25}. At a high level, a localization scheme for a target distribution describes a stochastic process of distributions that is initialized from the target distribution and ends at a (random) delta distribution. The field dynamics corresponds to the \emph{negative-field localization scheme} which characterizes the evolution of $\Pr(X_1 = \cdot \mid X_\theta) = U_{\theta \rightarrow 1}(X_\theta, \cdot)$, a random hard-core model with smaller fugacity where the randomness comes from the pinning posed by $X_\theta$. Another important example of the localization scheme is the \emph{stochastic localization} \cite{Eld13}, a powerful tool that has been applied to sampling from log-concave densities in continuous space and to (almost) resolving the Kannan--Lovász--Simonovits conjecture~\cite{chen2021almost,Kla23}.

As promised earlier, the proof of \cref{lem:FD-reversible} is now straightforward to establish. 

\begin{proof}[Proof of \cref{lem:FD-reversible}]

Let $P_{FD}$ denote the transition matrix of the discrete-time field dynamics.  We have the ergodic flow from $x$ 
to $y$ is
\begin{align*}
\mu_\lambda(x)P_{FD}(x,y) & = \mu_\lambda(x) (D_{1\rightarrow\theta} U_{\theta\rightarrow 1})(x,y) 
\\
& = \sum_S \Pr(X_1 = x) \Pr(X_\theta = S | X_1=x ) \Pr(X_1 = y | X_\theta = S) \\
& = \sum_S \frac{1}{ \Pr(X_\theta = S) } \Pr(X_\theta = S \wedge X_1=x ) \Pr(X_1 = y \wedge X_\theta = S).
\end{align*}
Note the last expression does not change if we swap $x$ and $y$ and hence the ergodic flow from $y$ to $x$ is the same.
\end{proof}

\section{SI implies Fast Mixing of Field Dynamics}

Recall the definition of $C$-approximate tensorization of variance.
Here we present the analog for the field dynamics, known as $K$-approximate conservation of variance.  Whereas $C$-approximate tensorization of variance implies that the spectral gap of the Glauber dynamics is $\geq 1/(Cn)$ (see \cref{lem:AT-gap}), similarly, $K$-approximate conservation of variance implies that the field dynamics has spectral gap $\geq 1/K$.  

The following definition is from \cite[Definition 3.10 with $\phi=x^2$]{CCYZ-critical-hard-core}.

\begin{definition} Let $\lambda>0$, $0\leq\theta<1$ and $K>0$.
Let $\mu=\mu_{\lambda}$ and let $\mu'$ denote the distribution of $X_{\theta}$ from the down-walk $D_{1\rightarrow\theta}$ where $X_1\sim\mu$.
We say that {\em $K$-approximate conservation of variance up to time~$\theta$} holds if for every $f$ we have 
\begin{align}
\label{defn:var-conservation}
\Var_{\mu}(f) 
& \leq K \Exp_{\sigma\sim\mu'}\left[ \Var_{U_{\theta\rightarrow 1}(\sigma,\cdot)}(f) \right]
\end{align}
\end{definition}

Let us decipher the right-hand side of \cref{defn:var-conservation}.  Note, that
\[
\Exp_{\sigma\sim\mu'}\left[ \Var_{U_{\theta\rightarrow 1}(\sigma,\cdot)}(f) \right]
 = 
\frac{1}{2}\sum_{\sigma\in\Omega} \mu'(\sigma)\sum_{\eta_1,\eta_2\in\Omega}
U_{\theta\rightarrow 1}(\sigma,\eta_1)U_{\theta\rightarrow 1}(\sigma,\eta_2)\left(f(\eta_1)-f(\eta_2)\right)^2.
\]
In words, we first take a sample $\sigma$ from the down-walk of the field dynamics, then we consider the variance of $f$ from the up-walk starting from $X_\theta=\sigma$.  

Let
$$
f_\theta(y) := \E[f(X_1) | X_\theta=y],
$$
where the expectation is with respect to the up-walk $U_{\theta\rightarrow 1}(y,\cdot)$.
Another possible way of writing~\cref{defn:var-conservation} is 
\begin{equation}
\label{conservation-alternative}
\Var_{\mu}(f(X_1))  \leq K \Exp_{X_\theta} \left[ \Var[f(X_1) | X_\theta] \right].
\end{equation}

Our main technical result in this section is the following theorem.

\begin{theorem}
\label{thm:FD-critical}
    For $\lambda>0$, if for all $0\leq\alpha<1$, the measure $\mu_{(1-\alpha)\lambda}$ is $(C(\alpha)-1)$-spectrally independent, then for any $0\leq\theta<1$, variance conservation holds for the field dynamics at parameters $(\lambda,\theta)$ with constant $K=K(\theta)$ defined in \cref{defn:K}.
\end{theorem}

\begin{proof}
Let 
$$
Q(\alpha):=\E[\Var[f(X_1)|X_\alpha]],
$$
that is, sample $S$ from $X_\alpha$ and then consider the variance
 of $f(X_1)$ where $X_1\sim U_{\alpha\rightarrow 1}(S,\cdot)$.
Note, $Q(1)=0$, and $Q(0)=\Var_\mu(f)$ (since $X_0=\emptyset$).
To prove conservation of variance, recall \cref{conservation-alternative}.  To establish \cref{conservation-alternative} the rough intuition is that we will show that $Q(\alpha)/Q(0)$ is large for all $\alpha<1$ and hence $K(\alpha)$-approximate conservation of variance holds with small $K(\alpha)$ (this is an increasing function as $\alpha\rightarrow 1$).  Hence, we will prove that the log-derivative is large:
\begin{equation}
    \label{eqn:Q-derivative}
(\log Q(\alpha))' =   \frac{Q'(\alpha)}{Q(\alpha)} \geq -\frac{C(\alpha)}{1-\alpha}.
\end{equation}
We will now prove the theorem assuming \cref{eqn:Q-derivative} and then we will go back to prove \cref{eqn:Q-derivative}.

Fix $0\leq\theta<1$.
\begin{align*}
    (\log Q(\theta)) & = (\log Q(0)) + \int_0^\theta (\log Q(\alpha))' \,{\mathrm d}\alpha
 \geq (\log Q(0)) - \int_0^\theta \frac{C(\alpha)}{1-\alpha} \,{\mathrm d}\alpha.
\end{align*}
After exponentiating both sides, and using the fact that $Q(0)=\Var_\mu(f)$ we then have:
$$
\Var_\mu(f) \leq K Q(\theta),
$$
with $K$ given by \cref{defn:K}.  This completes the proof of the theorem modulo the proof of \cref{eqn:Q-derivative} which we now establish.

Recall, $f_\alpha(y) := \E[f(X_1) | X_\alpha=y]$. Note that we have $f_1 = f$.
For $0\leq s\leq 1$, let $\nu_s$ denote the distribution of $X_{s}$ from the down-walk $D_{1\rightarrow s}$ where $X_1\sim\mu$.
For $0\leq \alpha < s \leq 1$, by the law of total variance we have:
\begin{equation}
\label{eq:total-general}
\Var_{\nu_s}[f_s] = 
\Var_{\nu_\alpha}[f_\alpha]
+
\E_{\nu_\alpha}[\Var_{U_{\alpha\rightarrow s}}[f_s]]
\end{equation}

From \cref{eq:total-general} for $s=1$ we have:
\begin{equation}
    \label{total-variance}
    Q(\alpha) = \Var_\mu[f] - \Var[\E[f(X_1)|X_{\alpha}]],
\end{equation}
For $h>0$, from \cref{eq:total-general} we have:
\[
\Var_{\nu_{\alpha+h}} [f_{\alpha+h}] = \Var_{\nu_\alpha}[f_{\alpha}] +
\E_{S\sim\nu_\alpha}[\Var_{U_{\alpha\rightarrow \alpha+h}(S,\cdot)}[f_{\alpha+h}]].
\]
Note that $\E_{T\sim U_{\alpha\rightarrow\alpha+h}(S,\cdot)}[f_{\alpha + h}(T) ] = f_{\alpha} (S)$ and hence for any $S\in\Omega$:
\begin{align}
\nonumber
   \Var_{U_{\alpha\rightarrow\alpha+h}(S,\cdot)}[f_{\alpha+h}] 
    & = \sum_{T:T\supseteq S} U_{\alpha\rightarrow\alpha+h}(S,T)\left( f_{\alpha+h}(T)-f_\alpha(S)\right)^2
    \\
    \label{limit-derivative}
    & = 
    \frac{h}{1-\alpha}\sum_{v\in V\setminus S} \Pr_{R\sim U_{\alpha\rightarrow 1}(S,\cdot)}(v\in R)\left(f_{\alpha}(S\cup \{v\})-f_\alpha(S)\right)^2 + O(h^2),
\end{align}
where in \cref{limit-derivative} the factor $h/(1-\alpha)$ follows from the observation that each $v\in R$ is added at a uniformly chosen time in the interval $[\alpha,1]$ (and more than one vertex is added with probability $O(h^2)$), and the substitution of $f_{\alpha+h}(T)$ by $f_\alpha(T)$ with $T=S\cup\{v\}$ follows from  
$f_{\alpha+h}(T)= f_\alpha(T) + O(h)$ (see also \cite[(81) in Appendix A.1]{CCYZ-critical-hard-core}).

Recall, $\mu=\mu_{\lambda}$.  
Taking the derivative of variance in \cref{total-variance} we obtain:
\begin{align*}
-Q'&(\alpha) 
 = \frac{\partial}{\partial \alpha} \Var[f_{\alpha}]  
\\
\nonumber
&=
\lim_{h\rightarrow 0} \frac{1}{h}
\left[\Var_{\nu_{\alpha + h}}[f_{\alpha+h}]
- 
\Var_{\nu_\alpha}[f_{\alpha}]\right]
\\
\nonumber
&=
\lim_{h\rightarrow 0} \frac{1}{h}
\E_{S\sim\nu_\alpha}\left[\Var_{U_{\alpha\rightarrow\alpha + h}(S,\cdot)}[f_{\alpha+h}]\right], 
& \tag{\mbox{by \cref{eq:total-general}}}
\\
\nonumber
&=
\frac{1}{1-\alpha}\E_{S\sim\nu_\alpha}\left[\sum_{v\in V\setminus S} \Pr_{R\sim U_{\alpha\rightarrow 1}(S,\cdot)}(v\in R)\left(f_{\alpha}(S\cup \{v\})-f_\alpha(S)\right)^2\right]
& \tag{\mbox{by \cref{limit-derivative}}}
\\
\nonumber
&\leq 
\frac{1}{1-\alpha}\E_{S\sim\nu_\alpha}\left[\sum_{v\in V\setminus S}  \Var [ \E  [f | X_1[v]\, ]  | X_\alpha = S]\right],  
& \tag{\mbox{by \cref{eq:below} below}}
\\
&\leq
\nonumber
\frac{C(\alpha)}{1-\alpha} \E_{S\sim\nu_\alpha}\left[\Var [ f | X_\alpha = S]\right],
& \tag{\mbox{by  \cref{eq:subadditivity} and $(C(\alpha)-1)$-SI}}
\\
&\leq
\nonumber
\frac{C(\alpha)}{1-\alpha} Q(\alpha),
\end{align*}
which establishes \cref{eqn:Q-derivative}.

We will argue that the following is true (for any $v$, $S$, and $\alpha$):
\begin{equation}\label{eq:below}
\Pr_{R\sim U_{\alpha\rightarrow 1}}(v\in R)\left(f_{\alpha}(S\cup \{v\})-f_\alpha(S)\right)^2\leq \Var [ \E  [f | X_1[v]\, ]  | X_\alpha = S]. 
\end{equation}
To simplify the notation, let $A=\E  [f | v\in X_1, X_\alpha = S]$, $B=\E  [f | v\not\in X_1, X_\alpha = S\, ]$, and 
$p=\Pr_{R\sim U_{\alpha\rightarrow 1}(S,\cdot)}(v\in R)$. Equation~\eqref{eq:below} follows since the LHS is
$$
p( A - (pA + (1-p)B) )^2 = p(1-p)^2 (A-B)^2,
$$
and the RHS is 
\begin{align*}
    p(1-p)(A-B)^2. &\qedhere
\end{align*}
\end{proof}

 \section{Fast Mixing of Glauber Dynamics from Field Dynamics}
 \label{sub:proof-critical}

\begin{proof}[Proof of \cref{thm:main-critical}]
From \cref{thm:FD-critical}, we have $K$-approximate conservation of variance which states the following:
\begin{equation}
\label{A1A1}
\Var_{\mu}[f]\leq K\, \E_{\nu_\theta} [\Var[f(X_1) | X_\theta]] 
= K\, \E_{\nu_\theta}[\Var_{U_{\theta\rightarrow 1}(S,\cdot)}[f(X_1) | X_\theta = S]].
\end{equation}
For $i\in [n]$, let $\mathcal{P}_i = X_1[ [n]\setminus \{i\} ]$ be the configuration on all vertices but $i$. We want to prove approximate tensorization of variance for $\mu$ with constant $C_{Gl}K$, that is, establish the following inequality:
\begin{equation}\label{eq:variancetoprove}
\Var_{\mu}[f]\leq C_{Gl}K\, \sum_{i=1}^n\E_{\mu}[\Var[f(X_1) | \mathcal{P}_i]].
\end{equation}

By \cref{assume222} in the hypothesis of \cref{thm:main-critical} combined with~\cref{lem:AT-gap}, we know approximate tensorization of variance holds with constant $C_{Gl}$ for distribution $\mu'=\mu_{\lambda'}$ where $\lambda'=(1-\theta)\lambda^*$ and $\theta\in (0,1)$; hence we know the following:
\begin{equation}
    \label{eq:at-lambdaprime}
    \Var_{\mu'}[f]\leq C_{Gl}\sum_{i=1}^n\E_{\mu'}[\Var[f(X_1) | \mathcal{P}_i]].
\end{equation}
Moreover, \cref{eq:at-lambdaprime} holds for any induced subgraph $G'$ of $G$.   
For any $S\subseteq V$, considering the subgraph $G'$ induced by conditioning on the independent set $S$ then 
$\mu'$ on $G'$ is identical to the distribution $U_{\theta\rightarrow 1}(S,\cdot)$.
Hence, we have the following, for any $S\subseteq V$,
\[
\Var_{U_{\theta\rightarrow 1}(S,\cdot)}[f(X_1) | X_\theta = S] \leq C_{Gl} \sum_{i=1}^n 
\E_{U_{\theta\rightarrow 1}(S,\cdot)}[\Var[f(X_1) | X_\theta = S, \mathcal{P}_i]].
\]
Taking $S\sim\nu_\theta$, we have:
\begin{equation}
\label{B1B1}
\E_{S\sim\nu_\theta}[\Var_{U_{\theta\rightarrow 1}(S,\cdot)}[f(X_1) | X_\theta = S]] \leq C_{Gl} \sum_{i=1}^n 
\E_{S\sim\nu_\theta}[\E_{U_{\theta\rightarrow 1}}[\Var[f(X_1) | X_\theta = S, \mathcal{P}_i]]].
\end{equation}

We also need the following observation:
\begin{align}
\nonumber
\lefteqn{
\E_{S\sim\nu_\theta}[\E_{U_{\theta\rightarrow 1}(S,\cdot)}[\Var[f(X_1) | X_\theta = S, \mathcal{P}_i]]] 
}
\hspace{1.5in}
\\
\label{interBCBC}
& = 
\E_{\mu}\Big[\E_{S\sim D_{1\rightarrow\theta}}\Big[\Var[f(X_1) | X_\theta = S, \mathcal{P}_i]]\Big|\mathcal{P}_i\Big]\Big]
\\
\nonumber
&\leq \E_{\mu}\Big[\E_{S\sim D_{1\rightarrow\theta}}\Big[\Var[f(X_1) | X_\theta = S, \mathcal{P}_i]]\Big|\mathcal{P}_i\Big]\Big]
\\
\nonumber
& \hspace{1in}
+
\E_{\mu}\Big[\Var_{S\sim D_{1\rightarrow\theta}}\Big[\E[f(X_1) | X_\theta = S, \mathcal{P}_i]]\Big|\mathcal{P}_i\Big]\Big]
\\
\label{C1C1}
& = 
\E_{\mu}\Big[\Var[f(X_1) | \mathcal{P}_i]]\Big],
\end{align}
where \cref{interBCBC} is from the observation that sampling $S$ from $\nu_\theta$ and then applying the up-walk $U_{\theta\rightarrow 1}$ is the same as sampling
from $\mu$ and then applying the down-walk $D_{1\rightarrow\theta}$, and \cref{C1C1} is by the law of total variance.

Putting it all together we have:
\begin{align*}
    \Var_{\mu}[f]
    &\leq 
    K\, \E_{\nu_\theta}[\Var_{U_{\theta\rightarrow 1}(S,\cdot)}[f(X_1) | X_\theta = S]]
   &  \mbox{by \cref{A1A1}}
    \\
    & \leq 
    C_{Gl} K \sum_{i=1}^n 
\E_{S\sim\nu_\theta}[\E_{U_{\theta\rightarrow 1}(S,\cdot)}[\Var[f(X_1) | X_\theta = S, \mathcal{P}_i]]]
& \mbox{by \cref{B1B1}}
    \\
    &\leq 
        C_{Gl}K 
\sum_{i=1}^n        \E_{\mu}\Big[\Var[f(X_1) | \mathcal{P}_i]]\Big]
        & \mbox{by \cref{C1C1}},
\end{align*}
which establishes \cref{eq:variancetoprove} and completes the proof of the theorem.
\end{proof}

\chapter{Random Regular Graphs via Field Dynamics}
\label{sec:hard-core-random}

In this chapter we prove the result of Chen, Chen, Chen, Yin, and Zhang~\cite{CCCYZ-random-regular} establishing fast mixing of the Glauber dynamics on random $\Delta$-regular graphs beyond the uniqueness threshold.  We restate the formal statement of their result for convenience.

\randomregularthm*

\section{Proof Overview}

For simplicity we can assume $\lambda\geq 1/(2\Delta)$, otherwise the mixing time bound stated in \cref{thm:hard-core-random-graphs} follows from straightforward coupling arguments, cf., \cref{lem:hc-small-lambda-coupling}.

The proof of \cref{thm:hard-core-random-graphs} uses a very nice connection for the hard-core model between spectral independence and the minimum eigenvalue of the adjacency matrix $A(G)$ for sufficiently small $\lambda$. 
Let 
\[ f(\lambda)=\min \{\eta+1:\mbox{$\eta$-weak spectral independence holds}\}.
\]    We will show:  
\begin{equation}\label{eq:local}
f(\lambda) \leq 1 + \lambda\,|\uplambda_{\min}(A_G+I)| + o(\lambda).
\end{equation}
We then use the following lemma, which is a continuous analog of trickle-down, to conclude the bound on spectral independence for all $\lambda<1/|\uplambda_{\min}(A_G+I)|$. 
\begin{lemma}\label{lem:continuoustrickle}
For any $\lambda>0$ we have
\begin{equation}\label{eq:a1}
D^- f(\lambda) \leq \frac{f(\lambda)^2 - f(\lambda)}{\lambda},
\end{equation}
where $D^- f$ is the left derivative of $f$.
\end{lemma}
Intuitively, the lemma allows us
to extend the bound on $f(\lambda)$ from~\cref{eq:local} to larger values of~$\lambda$. More precisely, we will show that:
$$f(\lambda)\leq \frac{1}{1-\lambda|\uplambda_{\min}(A_G+I)|}$$ 
when $\lambda<1/|\uplambda_{\min}(A_G+I)|$. We then utilize the classical result of Friedman~\cite{Friedman} that, with high probability, a random $\Delta$-regular 
graph $G$ has $\uplambda_{\min}(A_G)\geq -2\sqrt{\Delta}$.

We first prove \cref{eq:local} and then we utilize \cref{eq:local} and \cref{lem:continuoustrickle} to prove \cref{thm:hard-core-random-graphs}; we then prove \cref{lem:continuoustrickle} in \cref{sub:continuoustrickle}.

Let us define $f(\lambda)$ more formally.
For $\lambda>0$, let $\widetilde{D}_\lambda={\mathrm{diag}}(\Exp(\mu_\lambda))$ be the diagonal matrix with entries  $\widetilde{D}_\lambda(u,u)=\mu_\lambda(\sigma(u)=1)$ for $u\in V$.
Recall from \cref{sub:CorrelationMatrix}, $\eta$-weak spectral independence is defined as $\uplambda_{\max}(\widetilde{\Psi}_\lambda)\leq 1+\eta$, which is equivalent to:
\[
    \tag{\ref{CM-semidefinite}}
    \Cov_{\mu_\lambda}\preceq (1+\eta) \widetilde{D}_\lambda.
    \]

 For a PSD matrix $A$, let $\rho(A)=\uplambda_{\max}(A)$ be the spectral radius of $A$, and then  for $\lambda>0$, let
\begin{equation}\label{eq:definef}
    f(\lambda) = \max_{\tau} \rho (\widetilde{D}_\lambda^{-1/2}\Cov_{\lambda} \widetilde{D}_\lambda^{-1/2})
\end{equation} 
where $\tau$ ranges over all pinnings and, inside the maximum,
$\Cov_{\lambda}$ and $\widetilde{D}_\lambda$ are defined with
respect to the conditional measure $\mu_\lambda^\tau$.
Equivalently, $f(\lambda)$ is the smallest constant such that
\[
\Cov_{\lambda}
\preceq
f(\lambda)\widetilde{D}_\lambda
\]
for every pinning $\tau$; hence $\mu_\lambda$ is
$(f(\lambda)-1)$-weakly spectrally independent.

\begin{proof}[Proof of \cref{eq:local}]

Fix the graph $G$ under consideration, and hence for notational convenience let $\mu_\lambda = \mu_{G,\lambda}$.
To establish the desired bound on $f(\lambda)$, 
 recall from \cref{sub:CorrelationMatrix}, we need to show $\Cov_{\mu_\lambda}\preceq f(\lambda) \widetilde{D}_\lambda$
where $\widetilde{D}_\lambda$ is the diagonal matrix with entries  $\widetilde{D}_\lambda(u,u)=\mu_\lambda(\sigma(u)=1)$ for $u\in V$.

To establish the bound on $f(\lambda)$ note that for $\lambda$ sufficiently small (e.g., $\lambda<<1/n^2$) then for $X_1\sim\mu_{\lambda}$ we have $\Prob{|X_1| \geq 3}=o(\lambda^2)$.  Hence, we have the following for any $u\neq v\in V$:
\begin{align*}
\Cov_{\mu_\lambda}(u,v) 
&= 
        \Pr[\{u,v\}\subseteq X_1]
        - \Pr[u\in X_1]\cdot\Pr[v\in X_1]
        \\
        & =
        (\Pr[X_1=\{u,v\}] + o(\lambda^2))  - (\Pr[X_1=\{u\}]+o(\lambda))(\Pr[X_1=\{v\}] + o(\lambda)) 
         \\  & =
        \Pr[X_1=\{u,v\}]  - \lambda^2 + o(\lambda^2).
        \end{align*}
Since $X_1$ is an independent set weighted by $\lambda^{|X_1|}$ we have:
\[\Pr[X_1=\{u,v\}]
 =         \begin{cases}
             0
             & \mbox{if $\{u,v\}\in E$} 
             \\
           \lambda^2 + o(\lambda^2) 
               & \mbox{if $\{u,v\}\notin E$} 
         \end{cases}
  \]
  In other words,
  \[ \Pr[X_1=\{u,v\}] -\lambda^2 = -\lambda^2A_G(u,v) + o(\lambda^2).
  \]
  Now for the diagonal entries of the covariance matrix we have:
  \[ \Cov_{\mu_\lambda}(u,u) = \Pr[u\in X_1]
        - \Pr[u\in X_1]^2 = \Pr[u\in X_1] - \lambda^2 + o(\lambda^2).
  \]
  Since $\widetilde{D}_{\lambda}(u,u)=\Pr[u\in X_1]$, we have: 
  \begin{equation}
      \label{cov-new}
\Cov_{\mu_\lambda} - \diag( \E[X_1] ) = -\lambda^2(A_G+I) + o(\lambda^2).
    \end{equation}

  Combining the above bounds we obtain the following:
\begin{align}
\nonumber
\Cov_{\mu_\lambda} & = \Cov[ X_1 ]  
& \mbox{since $X_1\sim\mu_{\lambda}$}
\\
\nonumber
&  =  \diag( \E[X_1] ) -\lambda^2 (A_G+I) + o(\lambda^2)   
& \mbox{by \cref{cov-new}}\\ 
\nonumber
& \preceq
\diag( \E[X_1] )  + \lambda^2 I |\uplambda_{\min} (A_G+I)| + o(\lambda^2) 
& \mbox{by \cref{A-preceq-111}}
\\
\nonumber
& = \diag( \E[X_1] ) (1 + \lambda\,|\uplambda_{\min} (A_G+I)| + o(\lambda)) &
\mbox{since }\diag(\E[X_1] ) = \lambda I + o(\lambda)
\\
& = 
\widetilde{D}_\lambda(1 + \lambda\,|\uplambda_{\min}(A_G+I)| + o(\lambda)).
\label{last-line-cov-111}
\end{align}

This establishes the bound on $f(\lambda)$ in the case when there is no pinning $\tau$.
For a pinning $\tau$, consider the corresponding induced subgraph $H$.
Note, 
$$
A_G \succeq
\uplambda_{\min}(A_G)I,$$
which implies:
\[
    -A_G \preceq |\uplambda_{\min}(A_G)|I,
\]
or equivalently,
\begin{equation}
\label{A-preceq-111}
-(A_G+I)
\preceq
|\uplambda_{\min}(A_G+I)|I.
\end{equation}
Note, for any induced subgraph $H$ of $G$ we also have:
\begin{equation}
\label{HtoG-111}
-(A_H+I)
\preceq
|\uplambda_{\min}(A_G+I)|I;
\end{equation}
this follows from $\uplambda_{\min}(A_G)\leq \uplambda_{\min}(A_H)$ (which follows from Cauchy's interlacing theorem, or directly the Courant--Fischer theorem).

Now we apply \cref{HtoG-111} in \cref{last-line-cov-111} and we obtain:
\begin{align*}
f(\lambda) \leq 1 + \lambda\,|\uplambda_{\min}(A_G+I)| + o(\lambda). &\qedhere
\end{align*}
\end{proof}

We can now prove the main result \cref{thm:hard-core-random-graphs}.

\begin{proof}[Proof of \cref{thm:hard-core-random-graphs}]
For any $c$, let $g(\lambda)=1/(1-\lambda c)$.  Then,
\begin{equation}\label{eq:a2}
g'(\lambda) =  D^- g(\lambda) = \frac{g(\lambda)^2 - g(\lambda)}{\lambda}.
\end{equation}
Hence if $f(\lambda')\geq \frac{1}{1-\lambda' c}$ for some $c$ and $\lambda'<1/c$  then from~\eqref{eq:a1} and~\eqref{eq:a2} we have that 
\begin{equation}\label{eq:difeq}
f(\lambda)\geq \frac{1}{1-\lambda c}\ \mbox{for all}\ \lambda\in [0,\lambda'].
\end{equation}

Let $c:=|\uplambda_{\min}(A_G+I)|$.  If there exists $\lambda'<1/c$ where $f(\lambda')>1/(1-\lambda'c)$ then for some $c'>c$ we have $f(\lambda')>1/(1-\lambda'c')$, and by \cref{eq:difeq}, for all 
$\lambda\leq\lambda'$ we have $f(\lambda)\geq 1/(1-\lambda c')$.  But \cref{eq:local} states that $f(\lambda)\leq 1+\lambda c + o(\lambda)$ and hence we need that
$1/(1-\lambda c')\leq 1+\lambda c + o(\lambda)$ which is false for all sufficiently small $\lambda$, and thus we have a contradiction for the assumption that $f(\lambda')>1/(1-\lambda' c)$ for some $\lambda'<1/c$.  Therefore, for all $\lambda < 1/|\uplambda_{\min}(A_G+I)|$, we have:
\begin{equation}
    \label{last-contradict}
f(\lambda) \leq \frac{1}{1-\lambda |\uplambda_{\min}(A_G+I)| }.
\end{equation}

Now we will apply \cref{thm:SI-constant-mix} to finish off the proof. 
By \cite{Friedman}, we know that for random $\Delta$-regular graphs, with probability $1-o(1)$ over the choice of the graph $G$, we have that $\uplambda_{\min}(A_G)\geq -2\sqrt{\Delta}$.

Let $\lambda^*(\Delta):=\frac{1}{3\sqrt{\Delta}}$. 
Fix $\lambda\leq\lambda^*(\Delta)$.
 From \cref{last-contradict}, we have that the measure
$\mu_{\lambda}$ 
is $(f(\lambda)-1)$-weakly spectrally independent.  
Since it is $b$-marginally bounded, \cref{rem:weak-vs-nonweak}
implies that it is
\[
\left(\frac{f(\lambda)}{b}-1\right)
\text{-spectrally independent},
\]
where 
\[
f(\lambda)
\leq
\frac{1}{1-\lambda|\uplambda_{\min}(A_G+I)|}
\leq
\frac{1}{1-2\lambda\sqrt{\Delta}}
\leq 3;
\]
the first inequality follows from \cref{last-contradict},
and the second follows since
\[
\lambda_{\min}(A_G+I)=\lambda_{\min}(A_G)+1
\]
and
\[
-2\sqrt{\Delta}\leq\lambda_{\min}(A_G)\leq-1.
\]
Note 
\begin{equation}
b=b(\Delta)\geq \frac{\min\{1,\lambda\}}{(1+\lambda)^{\Delta+1}}\geq \frac{1/(2\Delta)}{(1+1/(3\sqrt{\Delta}))^{\Delta+1}},
\label{b:bound}
\end{equation}
for $1/(3\sqrt{\Delta})>\lambda\geq 1/(2\Delta)$ (recall, we can assume the lower bound since the mixing-time
bound for smaller $\lambda$ follows from the straightforward
coupling argument in \cref{lem:hc-small-lambda-coupling}). 
Finally, applying \cref{thm:SI-constant-mix}, with the
$b=b(\Delta)$ marginal bound in \cref{b:bound}, establishes the
theorem.
\end{proof}

\section{Proof of the trickle-down for the field dynamics}
\label{sub:continuoustrickle}

Here we prove the main technical lemma which is a trickle-down result for the field dynamics.

\begin{proof}[Proof of \cref{lem:continuoustrickle}]
Fix a pinning $\tau$ attaining the maximum in~\cref{eq:definef}, and let $H$ be the corresponding induced subgraph.  For the remainder of the proof, suppress $\tau$ and $H$ from the notation, let $\mu_\lambda$ denote the hard-core measure on $H$, and let $X_1\sim\mu_\lambda$.

We begin by considering the first order expansion of $\Cov$ for the field dynamics (as in \cite[Proposition 4.1]{CCCYZ-random-regular}).
Recall the definition of the field dynamics in \cref{sub:FD}, and, in particular, the definition of $X_\theta$ for $0\leq\theta\leq 1$ in \cref{def:down}.  For $h>0$, we have the following, which we will prove at the end of this section:
\begin{equation}
    \label{cov-expansion}
\E[\Cov[ X_1|X_h ]] = \Cov [X_1] - h \Cov [X_1] \diag(\E[X_1])^{-1} \Cov [X_1] + o(h).
\end{equation}
Conditioned on $X_h$, the random variable $X_1$ is distributed according to a hard-core measure at fugacity $(1-h)\lambda$ under a further pinning. Therefore, by the definition of $f$, for every $0<h<1$,
\begin{equation}
\label{inter-cov}
\Cov[ X_1|X_h ] \preceq f((1-h)\lambda) \diag( \E[X_1|X_h] - \ind_{X_h}),
\end{equation}
where $\ind_{X_h}$ is one for coordinates $i\in X_h$ and zero otherwise. Taking the expectation over $X_h$ for both sides in \cref{inter-cov} yields:
\begin{align}
\nonumber
\E[\Cov[ X_1|X_h ]]\ 
& \preceq f((1-h)\lambda) \diag( \E[\E[X_1|X_h]]  - h \E[X_1]) 
\\
& =  f((1-h)\lambda) (1-h) \diag( \E[X_1]).\label{eq:ineq}
\end{align}

 Let $v$ be the eigenvector in~\eqref{eq:definef} for eigenvalue $f(\lambda)$; w.l.o.g. assume that $\|v\|_2=1$. We can rewrite~\cref{cov-expansion} as
\begin{equation}\label{eq:rew}
 \Cov [X_1] = \E[\Cov[ X_1|X_h ]] +  h \Cov [X_1] \diag(\E[X_1])^{-1} \Cov [X_1] + o(h).
\end{equation}
We have (combining~\eqref{eq:rew} and~\eqref{eq:ineq})
\begin{align*}
f(\lambda) & = v^T \, \diag(\E[X_1])^{-1/2}\Cov [X_1]\diag(\E[X_1])^{-1/2} v
\\ &\leq 
f((1-h)\lambda) (1-h) \|v\|_2^2 + h v^T (\diag(\E[X_1])^{-1/2}\Cov [X_1]\diag(\E[X_1])^{-1/2})^2 v + o(h).
\end{align*}
Since $\|v\|_2=1$ and $\diag(\E[X_1])^{-1/2}\Cov [X_1]\diag(\E[X_1])^{-1/2} v = f(\lambda) v$ (because $v$ is the eigenvector in \cref{eq:definef} with eigenvalue $f(\lambda)$), we have
\begin{equation}\label{eq:fff}
f(\lambda) \leq (1-h) f( (1-h)\lambda ) + f(\lambda)^2 h +o(h).
\end{equation}
We can rewrite~\eqref{eq:fff} as follows
$$
(1-h) \frac{f(\lambda) - f( (1-h)\lambda )}{h\lambda} \leq   \frac{f(\lambda)^2 - f(\lambda)}{\lambda}  +o(1).
$$
Taking the limit as $h\downarrow 0$ yields the upper bound  on the left derivative of $f$ as stated in \cref{lem:continuoustrickle}.
\end{proof}

Finally, it remains to prove \cref{cov-expansion}.  This corresponds to  \cite[Proposition 4.1]{CCCYZ-random-regular}, which is a special case of more general localization schemes presented in \cite{AKV24}.  We present a more pedestrian proof here.

\begin{proof}[Proof of \cref{cov-expansion}]
With probability $1-O(h^2)$ we have $|X_h|\leq 1$ and hence, in order to obtain an expression with error $O(h^2)$, we only need to consider the options $|X_h|=0$ and $|X_h|=1$. We
have 
$$\Pr(X_h=\{\ell\}) = h \Pr(\ell\in X_1) + O(h^2) = h \Exp[X_1]_\ell  + O(h^2).$$
For any $S\subset V$ we have
$$
\Pr( X_1 = S | X_h=\{\ell\} ) = \frac{ \Pr( X_1 = S, \ell\in  X_1) }{ \Pr(\ell\in X_1) } + O(h). 
$$
Hence for $i,j$ we have
\begin{align}
\Cov[ X_1|X_h=\{\ell\} ]_{i,j} 
& = \frac{\Pr(\{i,j,\ell\}\subseteq X_1)}{\Pr(\ell\in X_1)} - \frac{\Pr(\{i,\ell\}\subseteq X_1)}{\Pr(\ell\in X_1)}\frac{\Pr(\{j,\ell\}\subseteq X_1)}{\Pr(\ell\in X_1)} + O(h).\label{ecov}
\end{align}
Aggregating~\eqref{ecov} over $\ell$ weighted with probability $\Pr(X_h = \{\ell\}) = h \Pr(\ell\in X_1) + O(h^2)$ we obtain
\begin{align}
& \sum_{\ell} \Pr(X_h=\{\ell\}) \Cov[ X_1|X_h=\{\ell\} ]_{i,j} \nonumber \\
& = h \sum_\ell \Pr(\{i,j,\ell\}\subseteq X_1)  - h \sum_{\ell} \Pr(\{i,\ell\}\subseteq X_1)\frac{1}{\Pr(\ell\in X_1)} \Pr(\{j,\ell\}\subseteq X_1) + O(h^2)\label{esum}.
\end{align}
We also have
$$
\Pr(X_1=S | X_h = \emptyset) = \Pr(X_1 = S) \Big(1 + h\sum_\ell \Pr(\ell\in X_1) - h |S| \Big) + O(h^2).
$$
Note that we can write
$$
\sum_S \Pr(X_1 = S, i,j\in X_1) |S| = \sum_{\ell} \Pr(\{i,j,\ell\}\subseteq X_1). 
$$
and hence
\begin{equation}\label{eepro2}
\Pr(\{i,j\}\subseteq X_1|X_h=\emptyset) = \Pr( \{i,j\}\subseteq X_1) \Big(1 + h \sum_{\ell} \Pr(\ell\in X_1)\Big) - h \sum_{\ell} \Pr( \{i,j,\ell\}\subseteq X_1) + O(h^2). 
\end{equation}
Taking the analogous expressions to~\eqref{eepro2} for $\Pr(k\in X_1|X_h=\emptyset)$ for $k=i,j$ we obtain the following
\begin{align}
\Cov[ X_1|X_h=\emptyset ]_{i,j} & = \Pr(\{i,j\}\subseteq X_1) \left(1+ h \sum_\ell \Pr(\ell\in X_1)\right)\label{ecova} \\
&- \Pr(i\in X_1)\Pr(j\in X_1) \left(1 + 2h \sum_\ell \Pr(\ell\in X_1)\right)\nonumber \\
& + h \sum_\ell \left(\Pr(i\in X_1) \Pr(\{j,\ell\}\subseteq X_1) + \Pr(j\in X_1) \Pr(\{i,\ell\}\subseteq X_1) \right) +O(h^2).\nonumber
\end{align}
Note that
\begin{equation}\label{epzero}
\Pr(X_h=\emptyset) = 1 - h \sum_{\ell} \Pr(\ell\in X_1) + O(h^2).
\end{equation}
Combining~\eqref{esum}, \eqref{ecova}, and \eqref{epzero} we obtain
$$
\E[\Cov[ X_1|X_h ]]_{i,j}  =  \Cov [X_1]_{i,j} - h B_{i,j}  + O(h^2),
$$
where 
\begin{align*}
B_{i,j} = & \sum_{\ell} \frac{\Pr(\{i,\ell\}\subseteq X_1)\Pr(\{j,\ell\}\subseteq X_1)}{\Pr(\ell\in X_1)}  - \sum_\ell \Pr(i\in X_1) \Pr(\{\ell,j\}\subseteq X_1) \\
&- \sum_{\ell} \Pr(j\in X_1) \Pr(\{\ell,i\}\subseteq X_1) + \sum_{\ell} \Pr(j\in X_1) \Pr(i\in X_1) \Pr(\ell\in X_1).
\end{align*}
Finally, we note that $B_{i,j}=(\Cov[X_1] \diag(\E[X_1])^{-1} \Cov[X_1])_{i,j}$.
\end{proof}

\chapter{Extensions: Ising Model and Colorings}
\label{sec:other-models}

Here we detail the extension of the results to other models beyond the hard-core model on weighted independent sets.  In \cref{sub:Ising} we survey results for the Ising model, and in \cref{sub:colorings} we survey results for colorings and edge colorings.  The results stated in this chapter are not proved in this monograph.

\section{Ising Model}
\label{sub:Ising}

The Ising model is the classical model from statistical physics. 
Here we state the Ising model analogs of the hard-core model results from \cref{sub:hard-core-results}.

The Ising model is defined on a graph $G=(V,E)$ with a parameter $\beta\in\R$ which is referred to as the inverse temperature.  The state space of the model is $\Omega=\{+1,-1\}^n$.  
For a configuration $\sigma\in\Omega$, let $m(\sigma)=|\{\{v,w\}\in E: \sigma(v)=\sigma(w)\}|$ denote the number of monochromatic edges, and define its weight as 
\[ w(\sigma) = \exp(2\beta m(\sigma))=\exp\left(\beta\sum_{\{v,w\}\in E} \left(\sigma(v)\sigma(w)+1\right)\right) = \exp(\beta |E|) \exp\left(\beta\sum_{\{v,w\}\in E} \sigma(v)\sigma(w)\right).
\]
The Gibbs distribution is defined on $\Omega$ where for $\sigma\in\Omega$ let $\mu(\sigma)=w(\sigma)/Z$ and $Z=\sum_{\eta\in\Omega}w(\eta)$ is the partition function. 

The model is ferromagnetic (sometimes referred to as attractive) when $\beta\geq 0$ as neighboring spins prefer to align, and the model is antiferromagnetic (repulsive) when $\beta<0$.  The generalization to $q\geq 3$ spins is known as the Potts model.  The results stated in this section only hold for the Ising model (i.e., when $q=2$).  

Recall, the tree uniqueness threshold for the hard-core model occurs at $\lambda_c(\Delta)=(\Delta-1)^{\Delta-1}/(\Delta-2)^\Delta$.  In the Ising model the corresponding tree uniqueness threshold occurs at $\beta_c=\beta_c(\Delta)$ satisfying one of the following equivalent conditions:
\[ \exp(2|\beta_c|) = \frac{\Delta}{\Delta-2} \iff \tanh(|\beta_c|)=\frac{1}{\Delta-1}.
\]

Recall, the phase transition on $\Delta$-regular trees in the hard-core model corresponds to even-height versus odd-height trees.  In contrast, in the Ising model the phase transition on $\Delta$-regular trees corresponds to fixing all leaves to spin $+1$ versus fixing all leaves to spin $-1$.  For integer $\ell\geq 1$, recall $T_\ell$ denotes the complete $\Delta$-regular tree of height $\ell$.  For $s\in\{+1,-1\}$, let $\mu_\ell^s$ denote the Gibbs distribution on $T_\ell$ conditional on all leaves having the fixed spin $s$, and let $p_\ell^s$ denote the marginal probability that the root receives spin $+1$ in the conditional Gibbs distribution $\mu^s_\ell$.  

When $\beta$ satisfies $\exp(2|\beta|)\leq\frac{\Delta}{\Delta-2}$ then 
\begin{equation}
\label{uniq-condition}
    \lim_{\ell\rightarrow\infty} p^{+1}_{2\ell} = \lim_{\ell\rightarrow\infty} p^{-1}_{2\ell}
\end{equation} 
and we say the model is in the tree uniqueness region, whereas when $\exp(2|\beta|)>\frac{\Delta}{\Delta-2}$ then the limits are different and we say the model is in the tree non-uniqueness region (as there are multiple Gibbs measures on the infinite $\Delta$-regular tree).
Notice that in \cref{uniq-condition} we only considered trees of even height;  for the ferromagnetic Ising model this is unnecessary, we can 
consider $\lim_{\ell\rightarrow\infty} p^{+1}_{\ell} = \lim_{\ell\rightarrow\infty} p^{-1}_{\ell}$ to obtain the same threshold.  However, for the antiferromagnetic Ising model we need to restrict attention to even height trees; alternatively, we can consider even height vs. odd height trees (both with all $+1$-leaf configuration) and obtain the same threshold.
We refer the interested reader to \cite{Baxter,Georgii,FV:book} for further exposition on phase transitions in the Ising model.

The analog of \cref{thm:hard-core-constant-degree} establishing fast mixing of the Glauber dynamics on any graph of maximum degree $\Delta$ in the tree uniqueness region was proved in Mossel and Sly~\cite{MS} for the ferromagnetic Ising model and in Chen, Liu, and Vigoda~\cite{CLV21} for the antiferromagnetic Ising model.

\begin{theorem}
    \label{thm:Ising-constant-degree-tree-uniqueness}
    For all $\Delta\geq 3$, all $\delta>0$, there exists $C(\Delta,\delta)$, for all $\beta$ where $\exp(2|\beta|)\leq \frac{\Delta-\delta}{\Delta-2+\delta}$, for any $n$-vertex graph $G=(V,E)$ of maximum degree $\Delta$, the mixing time of the Glauber dynamics for the Ising model satisfies for all $\eps>0$,
    \[   \Tmix(\eps)\leq C(\Delta,\delta)n\log(n/\eps).
    \]
\end{theorem}

    As in the case of the hard-core model, the independent works of Chen, Feng, Yin, and Zhang \cite{CFYZ22} and Chen and Eldan \cite{CE25} improved \cref{thm:Ising-constant-degree-tree-uniqueness} to obtain $\Tmix(\eps) \leq C(\delta)n\log(n/\eps)$ and thus extend the theorem to classes of graphs where the maximum degree $\Delta$ grows as a function of $n$.

At the critical point, Bauerschmidt and Dagallier~\cite{BD-critical-Ising} proved polynomial mixing time of the Glauber dynamics.  The upper bound on the critical exponent was improved to $3+O(1/\Delta)$ by Chen, Chen, Yin, and Zhang~\cite{CCYZ-critical-hard-core} as stated in the following theorem; they also established a lower bound on the critical exponent of $3/2$.  More general results were subsequently established by Caputo, Chen, and Parisi~\cite{CCP}.

 \begin{theorem}
         For all $\Delta\geq 3$, for any graph $G=(V,E)$ of maximum degree $\Delta$, for the Glauber dynamics for the Ising model with $\beta=\beta_c(\Delta)$ or $\beta=-\beta_c(\Delta)$, the mixing time satisfies
    \[
    \Tmix = O(n^{3+O(1/\Delta)}\log{n}).
    \]
 \end{theorem}

For the antiferromagnetic Ising model, a computational hardness result in the tree non-uniqueness region (analogous to \cref{thm:hard-core-hardness}) was established in \cite{SlySun,GSV16}.  In contrast for the ferromagnetic Ising model there is a polynomial-time algorithm for approximating the partition function for every graph $G=(V,E)$ and every $\beta\geq 0$ due to Jerrum and Sinclair~\cite{JSising}, which utilizes the so-called high-temperature expansion of the Ising model to consider a computationally equivalent model over even-degree subgraphs; an efficient sampler was obtained via various methods~\cite{RandallWilson,GJ09,GJ17}.

The analog of \cref{thm:hard-core-random-graphs} for the antiferromagnetic Ising model was proved by Anari, Koehler, and Vuong~\cite[Corollary 104]{AKV24} establishing fast mixing of the Glauber dynamics on random $\Delta$-regular graphs beyond the uniqueness threshold.

\section{Multispin Systems: Colorings and Edge Colorings}
\label{sec:multispin}
\label{sub:colorings}

In this section we outline the extensions of spectral independence to multispin systems; these are models where each vertex has $q\geq 3$ possible labels.

A prominent example of multispin systems is the colorings problem, namely proper vertex $k$-colorings.  For simplicity, we focus attention on the colorings problem instead of considering arbitrary multispin systems.

Given a graph $G=(V,E)$ of maximum degree $\Delta$, and an integer $k\geq 2$, let $\Omega$ denote the collection of proper vertex $k$-colorings of $G$; more formally, let
\[
\Omega = \{\sigma\in [k]^V: \mbox{for all }\{v,w\}\in E, \sigma(v)\neq\sigma(w)\},
\]
where $[k]=\{1,\dots,k\}$.
Then the Gibbs distribution $\mu$ is uniform over $\Omega$, and hence
$\mu(\sigma)=1/|\Omega|$ for every $\sigma\in\Omega$.
 For a coloring $\sigma\in\Omega$ and a vertex $v\in V$, denote the available colors for $v$ in $\sigma$ as  
\[ A_\sigma(v)=[k]\setminus\sigma(N(v))=\{c\in [k]:\forall w\in N(v), \sigma(w)\neq c\}.
\]

The transitions of the Glauber dynamics for colorings are as follows.  From a coloring $X_t\in\Omega$, choose a vertex $v$ uniformly at random from $V$, and a color $c$ uniformly at random from the available colors $A_{X_t}(v)$.  Set $X_{t+1}(v)=c$, and for all $w\neq v$, set $X_{t+1}(w)=X_t(w)$. 

There are two natural generalizations of spectral influence to colorings which were introduced in the independent works of Chen, Galanis, {\v{S}}tefankovi\v{c}, and Vigoda~\cite{CGSV21} and Feng, Guo, Yin, and Zhang~\cite{FGYZ21}; both of these approaches extend to arbitrary multispin systems.  

The two notions of spectral independence differ in their definition of the influence matrix.  Both approaches generalize the definition of the influence matrix in \cref{defn:inf-matrix} for the 
special case of $q=2$ spins.

In \cite{CGSV21} they consider a $kn\times kn$ modified influence matrix, which we denote as $\widetilde{\Psi}$; it generalizes the definition of the modified influence matrix in \cref{defn:correlation-matrix}.  For each pair of vertices $v,w$, and pair of colors $c_1,c_2\in[k]$, 
let
\begin{equation}
    \label{defn:multi-Linfty}
\widetilde{\Psi}((v,c_1),(w,c_2)) = 
\mu[\sigma(w) = c_2 | \sigma(v)=c_1] - \mu[\sigma(w)=c_2 ],
\end{equation}
denote the influence of $(\mathrm{vertex},\mathrm{color})$ pair $(v,c_1)$ on $(w,c_2)$.

In contrast, \cite{FGYZ21} considers an $n\times n$ influence matrix $\Psi$, which generalizes the definition of the influence matrix in \cref{defn:inf-matrix}.  For each pair of vertices $v,w$, let
\begin{equation}
\label{eqn:multi-L1}
\Psi(v,w) = 
\max_{c,c'\in[k]}
\|\mu[\sigma(w)=\cdot| \sigma(v)=c] -\mu[\sigma(w)=\cdot | \sigma(v)=c']\|_{\mathrm{TV}},
\end{equation}
denote the total variation distance for the marginal distribution at $w$ conditional on the worst-case pair of colors at $v$.

We say a pinning is a fixed assignment $\tau$ to a subset of vertices $S\subset V$ (i.e., $\tau:S\rightarrow [k]$) such that there is a valid coloring $\sigma$ of the entire graph $G$ which is the same as $\tau$ on $S$.  We then generalize the above definitions as $\widetilde{\Psi}_\tau$ and $\Psi_\tau$ where $\mu$ is replaced by $\mu_\tau$, which is the Gibbs measure conditional on the fixed assignment $\tau$. 
For a pinning $\tau$, we restrict $\widetilde{\Psi}_\tau$ to
vertex-color pairs $(v,c)$ satisfying
$\mu_\tau[\sigma(v)=c]>0$.
We restrict $\Psi_\tau$ to the unfrozen vertices, and in
\cref{eqn:multi-L1} the maximum is taken over colors having
positive marginal probability under $\mu_\tau$; 
for every such vertex $v$, we have $\Psi_\tau(v,v)=1$.

We say that $\eta$-weak spectral independence holds if  we have $\uplambda_{\max}(\widetilde{\Psi}_\tau)\leq 1+\eta$ for all pinnings $\tau$, and $\eta$-spectral independence holds if  we have $\uplambda_{\max}(\Psi_\tau)\leq 1+\eta$ for all pinnings $\tau$.
As in the case of two-spin systems (see \cref{rem:weak-vs-nonweak}), spectral independence implies weak spectral independence, see \cref{lem:multi-weak-vs-nonweak} below for a detailed statement and proof.

Chen, Liu, and Vigoda~\cite{CLV21} proved a more general
multispin optimal-mixing theorem: weak spectral independence
(i.e., a uniform bound on the largest eigenvalue of
$\widetilde{\Psi}_\tau$ for every pinning $\tau$), together with
marginal boundedness, implies optimal mixing-time bounds for the Glauber dynamics. Since spectral independence implies
weak spectral independence, this result includes
\cref{thm:SI-constant-mix} as a special case.

The big conjecture in the area is that whenever $k\geq\Delta+2$ then the mixing time of the Glauber dynamics is $O(n\log{n})$.  We review the main results for colorings.

\begin{theorem}[\cite{Jerrum}]
For $k\geq 2\Delta+1$, for any graph $G=(V,E)$ of maximum
degree $\Delta$, the Glauber dynamics has mixing time
\[
\Tmix=O(kn\log n).
\]
\end{theorem}

Note, the maximum degree $\Delta$ can be a function of $n$ in the above result of Jerrum~\cite{Jerrum}.

This was improved to $k>(11/6)\Delta$ by Vigoda~\cite{Vig00} using the following chain known as the flip dynamics.  A transition $X_t\rightarrow X_{t+1}$ of the flip dynamics is as follows.  From $X_t\in\Omega$, choose a vertex $v$ uniformly at random from $V$, and a color $c$ uniformly at random from $[k]$.  Let $S(v,c)$ denote the maximal $(c,X_t(v))$-colored component in $X_t$ which contains $v$, and let $\ell=|S(v,c)|$ denote its size.  With probability $p_\ell/\ell$ we ``flip'' $S(v,c)$ to obtain $X_{t+1}$, where flipping $S(v,c)$ means interchanging colors $c$ and $X_t(v)$; and with probability $1-p_\ell/\ell$ we set $X_{t+1}=X_t$.  Note, when $p_1=1$ and $p_j=0$ for all $j\geq 2$ then the flip dynamics is equivalent to a variant of the Glauber dynamics (more specifically, it corresponds to the Metropolis Glauber dynamics rather than the heat-bath Glauber dynamics defined above).

Vigoda~\cite{Vig00} introduced a set of flip probabilities with the following three properties:  (i) $p_1=1$, (ii) $p_i\geq p_{i+1}$ for all $i\geq 1$, and (iii) $p_i=0$ for all $i\geq 7$. Later improvements~\cite{CDMPP19,CarlsonVigoda25} used slightly different flip probabilities but still maintained these three properties (i)-(iii).  
Under property (iii), the flip dynamics is local in that at most 6 vertices are recolored in every step, and hence the approach of \cref{rem:comparison} applies.  Therefore, $O(n\log(n/\eps))$ mixing time of the flip dynamics implies $O(n)$ relaxation time of the Glauber dynamics for constant $\Delta$, and the generalization of \cref{lem:opt-relax-SI} implies $O(n\log{n})$ mixing time of the Glauber dynamics.

The result of Vigoda~\cite{Vig00} was improved by Chen, Delcourt, Moitra, Perarnau, and Postle~\cite{CDMPP19} to $k>(11/6-\eps)\Delta$ for $\eps\approx 10^{-5}$.  This was further improved by Carlson and Vigoda~\cite{CarlsonVigoda25} to $k\geq 1.809\Delta$ as stated in the following result.

\begin{theorem}[\cite{CarlsonVigoda25}]
\label{thm:CV-colorings}
For $k\geq 1.809\Delta$ and $\Delta\geq 125$, the flip dynamics chain has $\Tmix=O(n\log{n})$ under a setting of the flip probabilities where $p_i=0$ for all $i\geq 7$.  Consequently, for constant $\Delta\geq 125$, using the approach of \cref{lem:opt-relax-SI} and \cref{rem:comparison}, the Glauber dynamics has $\Tmix=O_{\Delta}(n\log{n})$.
\end{theorem}

In \cite{CGSV21,FGYZ21} it was shown that there is a constant $\eta>0$ and $\alpha^*\approx 1.763$ such that $\eta$-spectral independence holds when $k>\alpha^*\Delta$ for any triangle-free graph with maximum degree $\Delta$, where the constant $\alpha^*$ is the root of $x=\exp(1/x)$.

\begin{theorem}[\cite{CGSV21,FGYZ21}]
\label{colorings-triangle-free}
Let $\alpha^*\approx1.763\dots$ denote the solution of
$x=\exp(1/x)$. For every $\alpha>\alpha^*$ and $\Delta\geq 3$, there exists $C=C(\alpha,\Delta)>0$ such that for any $k\geq \alpha\Delta+1$, for any triangle-free graph $G=(V,E)$ of maximum degree $\Delta$, the Glauber dynamics  has mixing time 
\[ \Tmix \leq Cn\log{n}.
\]
\end{theorem}

Chen and Feng~\cite{ChenFeng} extended \cref{thm:CV-colorings,colorings-triangle-free} to obtain optimal relaxation time for graphs with non-constant $\Delta$.

In contrast to \cref{thm:CV-colorings,colorings-triangle-free}, the following result of Chen, Liu, Matni, and Moitra~\cite{CLMM23} holds down to $k\geq\Delta+3$, which essentially matches the conjectured optimal bound of $k\geq\Delta+2$, but at the expense of a lower bound on the girth.

\begin{theorem}[{\cite[Theorem 1.5]{CLMM23}}]
For all $\Delta\geq 3$, there exists $g=g(\Delta),C(\Delta)>0$ such that for any graph $G=(V,E)$ of maximum degree $\Delta$ and girth $\geq g$, for any $k\geq\Delta+3$ the Glauber dynamics has mixing time 
\[ \Tmix\leq Cn\log{n}.
\]
\end{theorem}

It is unclear how large the lower bound on the girth is in general, however in the case when $k>(1+\eps)\Delta$ for $\eps>0$ then the girth requirement is roughly $g=C'(\eps)\log^2{\Delta}$ for some constant $C'(\eps)>0$, see \cite[Remark 5.16]{CLMM23}.

For edge colorings, which correspond to colorings of line graphs, recent work of Wang, Zhang, and Zhang~\cite{WZZ24-edge-colorings} achieves close to the threshold of twice the maximum degree, since each edge shares an endpoint with $\leq 2(\Delta-1)$ edges so the  degree of the line graph is $\leq 2(\Delta-1)$.

\begin{theorem}[{\cite[Theorem 1]{WZZ24-edge-colorings}}]
\label{edge-colorings}
For all $\Delta\geq 3$, for any $k\geq (2+o_{\Delta}(1))\Delta$, for any $G=(V,E)$ with maximum degree $\Delta$ the Glauber dynamics for $k$-edge-colorings has mixing time $O(|E|\log|E|)$.
\end{theorem}

In \cite{WZZ24-edge-colorings} the $o_{\Delta}(1)$-term in the bound on $k$ is more precisely quantified as $O(1/\log{\Delta})$.
They utilize a technique known as matrix trickle-down which was introduced in \cite{ALG-edge-colorings}.

By Vizing's Theorem, we can construct a $k$-edge-coloring for all $k\geq\Delta+1$.  It is a fascinating open problem to determine the computational complexity of approximating the number of $k$-edge colorings for $\Delta+1\leq k\leq2\Delta$.  More precisely, an interesting open problem is to improve \cref{edge-colorings} by establishing an $\fpras$ for $k>(2-\eps_0)\Delta$ for a fixed $\eps_0>0$, or to establish computational hardness for $k>(1+\eps_1)\Delta$ for a fixed $\eps_1>0$. 

\section{Multispin: SI implies Weak-SI}

As in the case of two-spin systems
(see \cref{rem:weak-vs-nonweak}), spectral independence implies
weak spectral independence.  The following lemma quantifies this
implication in the multispin setting.

\begin{lemma}
\label{lem:multi-weak-vs-nonweak}
For every pinning $\tau$,
\[
\uplambda_{\max}(\widetilde{\Psi}_\tau)-1
\leq
k\left(\uplambda_{\max}(\Psi_\tau)-1\right).
\]
Consequently, $\eta$-spectral independence implies
$(k\eta)$-weak spectral independence.
\end{lemma}

\begin{proof}
Fix a pinning $\tau$ and suppress it from the notation.  We
restrict attention to vertices that are not frozen and to
vertex--color pairs having positive marginal probability.

Let $D$ be the diagonal matrix with
\[
D_{(v,c),(v,c)}
=
\mu[\sigma(v)=c],
\]
and let
\[
B
=
D^{1/2}\widetilde{\Psi}D^{-1/2}.
\]
The same covariance calculation as in
\cref{sub:CorrelationMatrix} shows that $B$ is symmetric.
Moreover, $B$ is similar to $\widetilde{\Psi}$, and hence these
two matrices have the same eigenvalues.

Write
\[
B=B_{\mathrm{diag}}+B_{\mathrm{off}},
\]
where $B_{\mathrm{diag}}$ consists of the diagonal vertex blocks
and $B_{\mathrm{off}}$ consists of the off-diagonal vertex blocks.

Let $\widetilde{\Psi}_{\mathrm{off}}$ denote the matrix obtained
from $\widetilde{\Psi}$ by setting all diagonal vertex blocks to
zero. Then
\[
B_{\mathrm{off}}
=
D^{1/2}\widetilde{\Psi}_{\mathrm{off}}D^{-1/2}.
\]

For a vertex $v$, let
\[
p_v
=
\bigl(\mu[\sigma(v)=c]\bigr)_{c\in[k]},
\]
where colors of zero marginal probability are omitted.  The
diagonal block of $B_{\mathrm{diag}}$ corresponding to $v$ is
\[
I-\sqrt{p_v}\sqrt{p_v}^{\,T},
\]
and hence
\[
B_{\mathrm{diag}}\preceq I.
\]
Therefore, for every vector $x$,
\begin{align*}
x^TBx
&\leq
x^Tx+x^TB_{\mathrm{off}}x\\
&\leq
x^Tx+|x|^T|B_{\mathrm{off}}||x|\\
&\leq
\left(1+\uplambda_{\max}(|B_{\mathrm{off}}|)\right)x^Tx.
\end{align*}
Consequently,
\begin{equation}
\label{eq:multi-diagonal-comparison}
\uplambda_{\max}(\widetilde{\Psi})
\leq
1+\uplambda_{\max}(|B_{\mathrm{off}}|).
\end{equation}

For distinct vertices $v,w$ and colors $c_1,c_2\in[k]$, the
unconditional marginal distribution at $w$ is a convex
combination of the conditional marginal distributions given the
color at $v$.  Hence,
\begin{align*}
\left|
\mu[\sigma(w)=c_2\mid\sigma(v)=c_1]
-
\mu[\sigma(w)=c_2]
\right|
&\leq
\max_{c,c'\in[k]}
\left\|
\mu[\sigma(w)=\cdot\mid\sigma(v)=c]
-
\mu[\sigma(w)=\cdot\mid\sigma(v)=c']
\right\|_{\mathrm{TV}}\\
&=
\Psi(v,w).
\end{align*}

Let $J_k$ denote the $k\times k$ all-ones matrix.  It follows that
$|\widetilde{\Psi}_{\mathrm{off}}|$ is entrywise bounded by the
corresponding principal submatrix of
\[
(\Psi-I)\otimes J_k.
\]
Since
\[
|B_{\mathrm{off}}|
=
D^{1/2}|\widetilde{\Psi}_{\mathrm{off}}|D^{-1/2},
\]
the matrices $|B_{\mathrm{off}}|$ and
$|\widetilde{\Psi}_{\mathrm{off}}|$ have the same eigenvalues.
Therefore, by the Perron--Frobenius theorem and monotonicity under
entrywise domination of nonnegative matrices,
\begin{align*}
\uplambda_{\max}(|B_{\mathrm{off}}|)
&\leq
\uplambda_{\max}\left((\Psi-I)\otimes J_k\right)\\
&=
k\left(\uplambda_{\max}(\Psi)-1\right).
\end{align*}
Combining this with
\cref{eq:multi-diagonal-comparison} completes the proof.
\end{proof}

\chapter*{Acknowledgements}

This monograph originated from lectures given at the UCSB Summer School on ``New tools for optimal mixing of Markov chains: Spectral independence and entropy decay'' in August, 2022. 
The summer school lectures were presented by Nima Anari, Pietro Caputo, Zongchen Chen, Heng Guo, Tali Kaufman, and Kuikui Liu, and the associated lecture notes were initially prepared by the following students: Yuzhou Gu, Tushant Mittal, Amanda Priestley, and Juspreet Singh Sandhu.  
 For more information on the summer school, including lecture videos and lecture notes, see:
\href{https://sites.cs.ucsb.edu/~vigoda/School/}{https://sites.cs.ucsb.edu/$\sim$vigoda/School/}.
We also thank 
Catherine Greenhill and Daniel Zhang for useful comments on an earlier version of this text, and the anonymous referees for their careful reading of the entire monograph and many helpful suggestions.

\newpage

\bibliographystyle{alpha}
\bibliography{biblio}

\end{document}